\documentclass[
  11pt,
  onecolumn,
  prd,
  aps,
  amsfonts,
  eqsecnum,
  superscriptaddress,
  nofootinbib,
  longbibliography,
  notitlepage
]{revtex4-2}

\usepackage{lmodern}    
\usepackage{microtype}  
\usepackage[english]{babel}
\usepackage{setspace}

\usepackage{fancyhdr}

\usepackage{caption}
\usepackage{graphicx}
\usepackage{xcolor}
\usepackage{subcaption}
\usepackage{rotating}
\usepackage{amsmath, amssymb, graphics, amsthm, isomath}
\usepackage{physics}
\usepackage{verbatim}
\usepackage{tikz}
\usepackage{tikz-cd}
\usepackage{circuitikz}
\usetikzlibrary{patterns}
\usetikzlibrary{through,hobby}
\usetikzlibrary{decorations.markings}
\usetikzlibrary{fadings,shadings}
\usetikzlibrary{plotmarks}
\usetikzlibrary{positioning}
\usetikzlibrary{decorations,arrows}
\usetikzlibrary{decorations.pathreplacing}
\usetikzlibrary{chains, scopes, positioning, backgrounds, shapes, fit, shadows, calc, arrows.meta, decorations.pathreplacing}
\usepackage{blkarray}
\usepackage{svg}
\usepackage{mathrsfs}
\usepackage{algorithm}
\usepackage{algpseudocode}

\usepackage{xcolor}
\definecolor{mycitecolor}{rgb}{0.0, 0.45, 0.85}   

\usepackage[
  colorlinks=true,
  linkcolor=mycitecolor,
  citecolor=mycitecolor,
  urlcolor=mycitecolor,
  hyperindex=true,
  linktocpage=true
]{hyperref}

\makeatletter\def\@bibstyle{apsrev4-2-etal6}\makeatother

\usepackage[capitalise,compress]{cleveref}
\usepackage{tcolorbox}

\newtcolorbox{shadedtheorem}{
  colback=gray!15,    
  colframe=white,     
  boxrule=0pt,        
  arc=0pt,            
  left=5pt,           
  right=5pt,
  top=5pt,
  bottom=5pt
}

\newcommand{\e}{\mathrm{e}}
\newcommand{\ii}{\mathrm{i}}

\renewcommand\epsilon{\varepsilon}

\newcommand{\qcode}[1]{[\![#1]\!]}

\newtheoremstyle{upright-hang}
  {6pt}   
  {6pt}   
  {\normalfont\hangindent=1.5em\hangafter=1\relax} 
  {0pt}   
  {\bfseries} 
  {:}     
  {0.5em} 
  {}      

\theoremstyle{upright-hang}
\newtheorem{thm}{Theorem}
\numberwithin{thm}{section}
\newtheorem{cor}[thm]{Corollary}

\newtheorem{lem}[thm]{Lemma}

\newtheorem{prop}[thm]{Proposition}
\newtheorem{eg}[thm]{Example}

\newtheorem{defn}[thm]{Definition}

\renewcommand{\thesection}{\arabic{section}}
\renewcommand{\thesubsection}{\thesection.\arabic{subsection}}
\renewcommand{\thesubsubsection}{\thesubsection.\arabic{subsubsection}}

\makeatletter
\renewcommand{\p@subsection}{}
\renewcommand{\p@subsubsection}{}
\makeatother

\makeatletter
\fancypagestyle{plain}{%
  \fancyhf{}%
  \fancyfoot[C]{\thepage}%
}
\fancypagestyle{revtexfooter}{%
  \fancyhf{}%
  \fancyfoot[C]{\thepage}%
}
\makeatother
\AtBeginDocument{\setstretch{1.0}\pagestyle{revtexfooter}}
\makeatletter
\newcommand{\AndySectionStar}[1]{%
  \par\addpenalty\@secpenalty
  \addvspace{3.5ex \@plus 1ex \@minus .2ex}%
  \noindent{\normalfont\Large\bfseries #1\par}%
  \nobreak\vspace{2.3ex \@plus .2ex}%
  \@afterheading}
\newcommand{\AndySubsectionStar}[1]{%
  \par\addpenalty\@secpenalty
  \addvspace{3.25ex \@plus 1ex \@minus .2ex}%
  \noindent{\normalfont\large\bfseries #1\par}%
  \nobreak\vspace{1.5ex \@plus .2ex}%
  \@afterheading}
\newcommand{\AndySubsubsectionStar}[1]{%
  \par\addpenalty\@secpenalty
  \addvspace{3.25ex \@plus 1ex \@minus .2ex}%
  \noindent{\normalfont\normalsize\bfseries #1\par}%
  \nobreak\vspace{1.5ex \@plus .2ex}%
  \@afterheading}
\renewcommand\@seccntformat[1]{\csname the#1\endcsname\quad}
\def\@hangfrom@section#1#2#3{\@hangfrom{#1#2}#3}
\def\@hangfroms@section#1#2{#1#2}
\renewcommand\section{\@ifstar{\AndySectionStar}{\@startsection {section}{1}{\z@}%
  {-3.5ex \@plus -1ex \@minus -.2ex}%
  {2.3ex \@plus .2ex}%
  {\normalfont\Large\bfseries}}}
\renewcommand\subsection{\@ifstar{\AndySubsectionStar}{\@startsection{subsection}{2}{-\parindent}%
  {-3.25ex\@plus -1ex \@minus -.2ex}%
  {1.5ex \@plus .2ex}%
  {\normalfont\large\bfseries}}}
\renewcommand\subsubsection{\@ifstar{\AndySubsubsectionStar}{\@startsection{subsubsection}{3}{-\parindent}%
  {-3.25ex\@plus -1ex \@minus -.2ex}%
  {1.5ex \@plus .2ex}%
  {\normalfont\normalsize\bfseries}}}
\renewcommand*\l@section[2]{%
  \ifnum \c@tocdepth >\z@
    \addpenalty\@secpenalty
    \addvspace{0.25em \@plus\p@}%
    \begingroup
      \parindent\z@
      \rightskip\@pnumwidth
      \parfillskip-\@pnumwidth
      \leavevmode\bfseries
      \setlength\@tempdima{2.3em}%
      \advance\leftskip\@tempdima
      \hskip-\leftskip
      #1\nobreak\hfil\nobreak\hb@xt@\@pnumwidth{\hss #2}\par
    \endgroup
  \fi}
\renewcommand*\l@subsection{\@dottedtocline{2}{1.5em}{3.2em}}
\renewcommand*\l@subsubsection{\@dottedtocline{3}{3.8em}{4.1em}}
\makeatother

\usepackage{mathtools}
\tikzstyle{densely dashed}= [dash pattern=on 4pt off 3pt]

\newcommand{\supp}{\operatorname{supp}}

\usepackage{dsfont}

\allowdisplaybreaks

\usepackage{centernot}

\DeclareFontFamily{U}{wasy}{}
\DeclareFontShape{U}{wasy}{m}{n}{<-6>wasy5<6-7>wasy6<7-8>wasy7<8-9>wasy8<9-10>wasy9<10->wasy10}{}
\DeclareRobustCommand{\clockmark}{\TextOrMath{{\fontencoding{U}\fontfamily{wasy}\selectfont\char28}}{\text{{\fontencoding{U}\fontfamily{wasy}\selectfont\char28}}}}
\makeatletter
\renewcommand\@fnsymbol[1]{\ifcase#1\or \clockmark\or \TextOrMath\textdagger\dagger\or
  \TextOrMath\textdaggerdbl\ddagger\or \TextOrMath\textsection\mathsection\or
  \TextOrMath\textparagraph\mathparagraph\or \TextOrMath\textbardbl\|\or
  \TextOrMath{\textasteriskcentered\textasteriskcentered}{**}\or
  \TextOrMath{\textdagger\textdagger}{\dagger\dagger}\or
  \TextOrMath{\textdaggerdbl\textdaggerdbl}{\ddagger\ddagger}\else\@ctrerr\fi}
\makeatother

\begin{document}

\definecolor{darkred}{RGB}{140,0,0}
\definecolor{vvacolor}{rgb}{0.35, 0.0, 0.42}  
\newcommand{\andycomment}[1]{{\color{red} \textbf{Andy:}~#1}}
\newcommand{\jinkangcomment}[1]{{\color{darkred} \textbf{Jinkang:}~#1}}
\newcommand{\aranya}[1]{{\color{green} \textbf{AC:}~#1}}
\newcommand{\yiqiucomment}[1]{{\color{blue} \textbf{yiqiu:}~#1}}
\newcommand{\chatgptcomment}[1]{{\color{green!50!black} \textbf{chatgpt:}~#1}}
\newcommand{\vvacomment}[1]{{\color{red} \textbf{VVA:}~#1}}
\definecolor{shubhamcolor}{RGB}{0,110,140}  
\newcommand{\shubhamcomment}[1]{{\color{shubhamcolor} \textbf{SJ:}~#1}}
\newcommand{\andyinsert}[1]{{\color{black}#1}}
\newcommand{\jinkanginsert}[1]{{\color{black}#1}}
\newcommand{\yiqiuinsert}[1]{{\color{black}#1}}
\newcommand{\chatgptinsert}[1]{{\color{black}#1}}
\newcommand{\vvainsert}[1]{{\color{black}#1}}
\definecolor{vvanewcolor}{RGB}{200,0,120} 
\newcommand{\vvanew}[1]{{\color{vvanewcolor}\hypersetup{citecolor=vvanewcolor,linkcolor=vvanewcolor}#1}}  
\newcommand{\shubhaminsert}[1]{{\color{shubhamcolor}#1}}  
\makeatletter
\DeclareRobustCommand{\andy}{\@ifstar\andyinsert\andycomment}
\DeclareRobustCommand{\jinkang}{\@ifstar\jinkanginsert\jinkangcomment}
\DeclareRobustCommand{\yiqiu}{\@ifstar\yiqiuinsert\yiqiucomment}
\DeclareRobustCommand{\chatgpt}{\@ifstar\chatgptinsert\chatgptcomment}
\DeclareRobustCommand{\vva}{\@ifstar\vvainsert\vvacomment}
\DeclareRobustCommand{\shubham}{\@ifstar\shubhaminsert\shubhamcomment}
\makeatother

\newcommand{\seniorauthornote}{V.V.A. and A.L. are senior authors.
Electronic addresses: \href{mailto:vva@umd.edu}{vva@umd.edu},
\href{mailto:andrew.j.lucas@colorado.edu}{andrew.j.lucas@colorado.edu}}

\title{Quantum codes in the Lee metric}

\author{Jinkang Guo}
\affiliation{Department of Physics and Center for Theory of Quantum Matter, University of Colorado, Boulder, CO}

\author{Yiqiu Han}
\affiliation{Department of Physics and Center for Theory of Quantum Matter, University of Colorado, Boulder, CO}

\author{Aranya Chakraborty}
\affiliation{Joint Center for Quantum Information and Computer Science, NIST/University of Maryland, College Park, MD}

\author{Shubham~P.~Jain}
\affiliation{Joint Center for Quantum Information and Computer Science, NIST/University of Maryland, College Park, MD}

\author{Tianhao~Liu}
\affiliation{Department of Physics and Center for Theory of Quantum Matter, University of Colorado, Boulder, CO}

\author{Victor V. Albert}
\thanks{\seniorauthornote}
\affiliation{Joint Center for Quantum Information and Computer Science, NIST/University of Maryland, College Park, MD}

\author{Andrew Lucas}
\thanks{\seniorauthornote}
\affiliation{Department of Physics and Center for Theory of Quantum Matter, University of Colorado, Boulder, CO}

\date{\today}

\begin{abstract}

We introduce a quantum coding framework for discrete small-shift noise, in which errors on qudits are modeled as low-weight \(X\)- and \(Z\)-type Pauli shifts, analogous to small phase-space displacements in continuous-variable systems. This structure is approximately respected by nuclear-spin noise and captured by the Lee metric, motivating a quantum extension of classical Lee-metric coding theory.
We develop such a formalism for stabilizer codes over \(\mathbb{Z}_q\) for arbitrary \(q\), using joint and separate Lee metrics for the \(X\)- and \(Z\)-components of errors. For qubits, the joint metric counts \(Y\) errors twice and can yield codes that detect and correct \(X\) and \(Z\) errors with fewer physical qubits than codes designed for the conventional Hamming metric. We discuss Lee-weight spreading under Clifford gates and qubitize codes over $\mathbb{Z}_4$ via the Gray map, finding codes with two-fold transversal non-qubit-Clifford gates.
For quantum CSS Lee-LDPC codes on \(n\) qudits, we prove that the Lee distance cannot exceed  \(O(n)\), uniformly in \(q\), demonstrating an unexpected obstruction to using the large internal Hilbert space of a large-$q$ qudit to make high-Lee-distance codes. Under local Metropolis dynamics, certain classical \(q\)-ary ``helical repetition codes'' have exponentially long memory times at fixed temperature for sufficiently large \(q\) (that grows with system length), which can be understood as spontaneous symmetry breaking at finite temperature, even in local one-dimensional models.  Hypergraph products of these helical repetition codes provide local two-dimensional quantum codes that inherit self-correction for \(Z\) errors, but self-correction for \(X\) errors remains an open question.

\end{abstract}

\maketitle

\definecolor{leedark}{RGB}{8,69,148}
\colorlet{lee1}{leedark!15}
\colorlet{lee2}{leedark!30}
\colorlet{lee3}{leedark!45}
\vfill
\begin{center}
    \begin{tikzpicture}[x=1cm, y=1cm]
        \def\dr{0.95}   
        \def\RR{3.00}   
        \foreach \j in {0,...,6} {
            \pgfmathtruncatemacro{\w}{min(\j,7-\j)}
            \pgfmathsetmacro{\hang}{90-\j*360/7}
            \pgfmathsetmacro{\sweepend}{(\j<4) ? 90-\w*360/7 : 90+\w*360/7}
            \pgfmathsetmacro{\px}{\RR*cos(\hang)}
            \pgfmathsetmacro{\py}{\RR*sin(\hang)}
            \begin{scope}[shift={(\px,\py)}]
                \ifnum\w>0 \fill[lee\w] (0,0) -- (90:\dr) arc[start angle=90, end angle=\sweepend, radius=\dr] -- cycle; \fi
                \draw[gray!70!black, line width=0.5pt] (0,0) circle[radius=\dr];
                \draw[gray!70!black, line width=0.7pt] (90:0.80*\dr) -- (90:\dr);
                \foreach \k in {1,...,6} {\draw[gray!70!black, line width=0.5pt] ({90-\k*360/7}:0.88*\dr) -- ({90-\k*360/7}:\dr);}
                \fill ({\hang+90}:0.07) -- (\hang:0.74*\dr) -- ({\hang-90}:0.07) -- cycle;
                \fill (0,0) circle[radius=0.09];
            \end{scope}
        }
    \end{tikzpicture}
\end{center}
\vfill

\clearpage
{\setstretch{0.93}\tableofcontents}  

\newpage 

\section{Introduction \& summary} \label{sec:intro}

\subsection{Qudits, done right}

The past five years have seen rapid progress towards small-scale demonstrations of quantum error correction in scalable quantum codes~\cite{ryan-anderson2021realization,krinner2022realizing,postler2022demonstration,bluvstein2024logicala,acharya2025quantumb}.
Coupled with the discovery of many new and exciting families of quantum low-density parity-check (LDPC) codes~\cite{panteleev2022asymptotically, breuckmann2021balanced, leverrier2022quantum}, these advancements offer a route to low-overhead fault-tolerant quantum computation~\cite{gottesman2013fault, bravyi2024high}.  

The vast majority of recent theoretical work assumes that the system contains interacting two-level systems, i.e., qubits.  In experiment, of course, each quantum system has more than two levels, but for the most part, only two are used. For example, a transmon qubit made out of superconducting circuits typically only uses the lowest two energy levels out of a tower of accessible states in the circuit~\cite{koch2007charge}.  A neutral atom or trapped ion quantum computer uses a small number out of an enormous number of internal atomic states for storing quantum information~\cite{isenhower2010demonstration, harty2014high}, with others potentially used for error detection, such as erasure flagging~\cite{wu2022erasure, ma2023high, kang2023quantum}.   

It may seem a shame, however, that we only use two out of the many potential states in these experimental platforms.  A practical issue is that scaling hardware to $10^4$ or more qubits remains a huge technical challenge, especially when one wants enough individual qubit control to efficiently implement complex logical gate sets.  Using a qudit with $q\gg 2$ internal levels might reduce some of the scaling challenge.

If one were to look for a scalable qudit code in the existing literature, most likely one would simply adopt existing qubit codes in a natural way.   For example, the surface/toric code immediately generalizes to group-valued qudit codes known as quantum double models~\cite{Kitaev2003anyons,Bombin2007homological,Bullock2007qudit,Watanabe2023toric,McDonough:2026ngv}.  Many existing papers on LDPC codes can immediately be generalized from binary codes to Galois-qudit codes on Hilbert spaces of dimension $p^\ell$, for $p$ prime and $\ell$ an integer, because the properties of these codes follow from homology over Galois fields.  These LDPC codes, as we observed above, do inherit certain natural fault-tolerance properties.  

At the same time, we will argue in this paper that the existing analyses of these codes have overlooked a critical fact about qudit systems: not all errors on a qudit are equal!  The standard metric by which one abstractly characterizes the robustness of a code -- how many errors we can tolerate before an undetectable logical error occurs, known as the code distance -- is defined in the above literature under the assumption that all $q^2-1$ independent errors on a $q$-dimensional qudit are equally likely to occur.   This is not realistic: for example, a spin-$j$ qudit [$q=2j+1$] may be dominated by $J_x/J_y/J_z$ errors.   We should incorporate this fact into our framework for characterizing codes.  Indeed, this has already been done in the context of single-qudit spin and molecular codes~\cite{albert2020robustf,jain2024absorptionemissiona,Gross2021spin,Gross2024nuclear,Omanakuttan2024spincat}.  However, we do not know how to scale such single-qudit codes, and whether there are genuine advantages to scaling such codes relative to qubit codes.

Recent results on the statistical physics of high-dimensional on-site configuration spaces~\cite{Han:2025eiw,Huang:2025gvi,Tsao:2026hle,Andriolo:2026udg,mehta2026} 
further motivate us to consider realistic noise models on \(q\)-level systems.
In those settings, weighting all \(q^2-1\) Pauli shifts as equally likely misses out on \emph{qualitative} performance improvements
Strikingly,  ``entropic-ordered'' solids have been constructed, which stabilize ordered phases of matter at arbitrarily high temperature by utilizing an infinite on-site state space.   If the on-site state space had $q$ levels, then the solid would persist to some temperature $T_{\mathrm{c}}\sim q^a$ for some model-dependent exponent $a$.    Given the well-established connections between ordered phases of matter and passive decoding (self-correction), it is very natural to hope that high-dimensional configuration spaces may lead to codes that can be more robust against noise.  Could it be, for instance, that the threshold of such a code family might reach $\frac{1}{2} - O(T_{\mathrm{c}}^{-1})$?    The literature on entropic order has not carefully considered the coding perspective: for example, one must explicitly furnish a decoder under a ``realistic'' error model and explain if, and why, the decoder can tolerate higher noise rates when $q$ is large.  

Nevertheless, the prospect of such high-threshold codes is enticing, and certainly warrants a detailed investigation, in particular, into whether such codes exhibit a local, self-correcting~\cite{Brown2016RMP} decoder.  In the Hamming metric, the classic example of a self-correcting quantum code is the four-dimensional toric code, whose quantum memory is protected by thermal topological order~\cite{Dennis2002topological,Alicki2010thermal}.  More recently, it has been emphasized~\cite{Hong2025quantummemory,Placke:2024wey} that more general quantum expander LDPC codes are self-correcting. A self-correcting quantum memory has also been constructed from a Pauli stabilizer Hamiltonian that is geometrically local in three dimensions~\cite{Balasubramanian2026selfcorrecting}. Two-dimensional stabilizer qubit codes, on the other hand, cannot be self-correcting~\cite{BravyiTerhal2009,Stengele2026MLSI}. We should also stress that even for these codes that are formally self-correcting, the error rates they can tolerate are disappointingly small in practice; any systematic way to increase the noise tolerance parametrically in a large parameter (like $q$) is quite desirable.   By definition, if a large-$q$ qudit code had high-temperature topological order, it would be a self-correcting code under idealized syndrome extraction, and a tantalizing candidate for a high-threshold code more generally as well.

With all of this optimism, we must note the potential challenges and complications that could arise when studying quantum codes in the Lee metric, which weights an error by the total size of its $X$ and $Z$ shifts rather than by the number of qudits touched by those shifts (see below).  From a practical perspective, by far the most important is that the central advantage of LDPC codes -- inherent fault tolerance -- is no longer as manifest with qudit codes.  For example: when $q$ is large, can a low-Lee-weight error be converted into a large-Lee-weight (but still localized) error by the syndrome extraction circuit?  Although early work established stabilizer and fault-tolerance
frameworks for higher-dimensional systems~\cite{Rains1999nonbinary, Gottesman1999higher, hostens2005stabilizer,AshikhminKnill2001}, these frameworks do not by themselves characterize the advantages or disadvantages
of increasing $q$ when noise predominantly causes small-Lee-weight errors.
This motivates a systematic study of code distance, error propagation,
and sparse-check constructions adapted to this relevant noise model.

The purpose of this paper is to begin to fill the major gap in the literature described above.  We present a natural formalism for studying qudit codes and characterizing their robustness under a more realistic error model.  For concreteness, we focus on a specific error model which nicely generalizes the Pauli error model used to analyze qubits.   Indeed, one may naturally generalize Pauli $X$ and $Z$ to a $q$-dimensional ``modular'' qudit:  defining operators $X$ and $Z$ that obey $XZ=\mathrm{e}^{-2\pi \mathrm{i}/q}ZX$, one finds that the operators $X^aZ^b$ for $a,b\in \mathbb{Z}_q$ form a basis for single-qudit operators. If the fundamental errors acting on the system are $X_i^{\pm 1}$ and $Z_i^{\pm 1}$ on each individual qudit $i$, then we should count not only how many qudits an error acts on, but how many powers of $X$ and $Z$ act on each qudit.  This is known as the ``Lee metric'' in the classical coding literature~\cite{Lee1958}.    For $q>6$, the errors $X_1X_2^2$ and $X_1^3$ should both be counted as having Lee weight 3.  The existing characterizations of quantum codes, however, use Hamming weight (the number of qudits an error acts on): these errors have weights 2 and 1, respectively.  This paper will develop a systematic formalism for studying quantum codes in the Lee metric.   

While the Lee metric can already be a useful characterization of errors in qudit systems, such as the spin-$j$ system described above (see Section~\ref{subsec:spinnoise}), we have also chosen to focus on it because it is simple enough that we may borrow many sophisticated mathematical tools from the existing literature to characterize codes.  Our work can nevertheless serve as an important starting point for a more systematic theory of scalable families of qudit codes designed for other error models as well, eventually leading to a ``co-design'' problem where quantum codes may be designed with specific experimental error models in mind.

Appendix~\ref{app:background} gives further motivation and background, including a comparison with prior work.

\subsection{Summary of results} \label{subsec:summary}

\textit{Two Lee metrics and qubit codes.}  There are two ways to measure a Pauli string $X^xZ^z$, with $x,z\in\mathbb{Z}_q^n$, in the Lee metric (Section~\ref{subsec:qubits}).  The separate metric records the pair $(|x|_{\mathrm L},|z|_{\mathrm L})$ and suits Calderbank--Shor--Steane (CSS) codes.  The joint metric records the sum $|x|_{\mathrm L}+|z|_{\mathrm L}$.  It differs from the Hamming metric already for qubits, since it assigns weight two to $Y$ and weight one to $X$ and $Z$.  This is the appropriate metric when bit flips and phase flips arise from independent processes.  Using the classification of Refs.~\cite{Cross2025smallcodes,qiskitqec}, we maximize the joint Lee distance of every stabilizer code on at most nine qubits over single-qubit Clifford frames.  In many cases, the largest Lee distance exceeds the largest Hamming distance for the same numbers of physical and logical qubits (Table~\ref{tab:qubit_lee}).  For example, a $\qcode{9,1,5}$ code corrects any two $X$ or $Z$ errors or any single $Y$ error.  The shortest code that corrects any two errors in the Hamming metric has eleven qubits~\cite{Grassl2007codetables}.  Similarly, a $\qcode{6,1,4}$ code reaches a distance that no single-logical-qubit stabilizer code on at most nine qubits reaches in the Hamming metric.

\textit{Stabilizer and perfect codes.}  Section~\ref{sec:stabilizer} defines stabilizer codes on $\mathbb{Z}_q$ with a joint Lee distance $d_\diamond$ and a separate Lee distance $d_\square$.  The two satisfy $d_\square\le d_\diamond\le 2d_\square$ and coincide for CSS codes.  We derive sphere-packing bounds in both metrics and list the perfect codes that saturate them in Table~\ref{tab:perfect}.  In the joint metric, these include a single-qudit code on $\mathbb{Z}_{10}$ with stabilizer $X^2Z^4$, a two-qudit family on $\mathbb{Z}_{3m}$, and a $\qcode{6,4,3_\diamond}_5$ code.  In the separate metric, we find perfect CSS code families of increasing distance on one and on two qudits.  The one-qudit codes are the perfect shift-resistant codes of Ref.~\cite{Gottesman2001GKP} -- the first quantum Lee codes.  We also find perfect distance-three CSS codes on $n\ge3$ qudits of dimension $2n+1$ whenever $n\bmod 3\neq1$.  The two-qudit and $n$-qudit families are built from the classical perfect Lee codes of Golomb and Welch~\cite{GolombWelch1970}.  For the quantum Hamming bound, by contrast, the nontrivial pure perfect quantum codes over prime-power alphabets all have distance three~\cite{LiXing2009perfect}.  Finally, Lee distances are bounded by the Hamming distance $d_{\mathrm H}$ times a factor of order $q$ (Section~\ref{subsec:hamming}).  For example, $d_{\mathrm H}\le d\le\lfloor q/2\rfloor d_{\mathrm H}$ for a CSS code of Lee distance $d$.  Lee distances can therefore grow without bound with $q$ at fixed $n$, and the relative distance $d/n$ can be unbounded.

\textit{CSS codes with torsion.}  Section~\ref{sec:css} develops CSS codes on $\mathbb{Z}_q$ for composite as well as prime $q$.  The Smith normal form of a parity-check matrix describes the codespace of a classical code, which can contain ``torsion factors'', i.e., logical qudits of dimension less than \(q\).  For a quantum CSS code, we express the logical quotient $\ker(H_X)/\mathrm{im}(H_Z^{\mathsf T})$ as the homology of an integer chain complex (Theorem~\ref{thm:quantumcodespace}).  Its Smith normal form shows that the code encodes logical qudits whose dimensions divide $q$.  It also yields an explicit conjugate pair of logical operators for each of them.

\textit{Code families.}  Section~\ref{subsec:smallexamples} constructs CSS codes from classical Lee codes.  Shortening the self-dual Hamming and Golay codes over $\mathbb{Z}_{2^r}$~\cite{ConwaySloane1993,Bonnecaze1995,CalderbankSloane1995} yields self-dual CSS codes~\cite{JainAlbert2025transversal}.  We call them lifted Hamming codes $\qcode{7,1,d}_{2^r}$ and lifted Golay codes.  The distance $d$ of the lifted Hamming codes grows from three at $q=2$ to seven at $q=16$, and it cannot exceed seven for any $q$.  The member with $q=4$ is the quantum heptacode $\qcode{7,1,4}_4$.  The lifted Golay codes include $\qcode{23,1,10}_4$, $\qcode{23,1,12}_8$, and $\qcode{23,1,17}_{16}$ codes.  A quantum Reed--Solomon code on $\mathbb{Z}_5$~\cite{Cleve1999secret,Grassl2004optimal} has Hamming distance three but Lee distance four.  This is the largest Lee distance of any CSS code with the same parameters.  The two-qudit perfect code on $\mathbb{Z}_{10}$ shows that the conclusion of the Eastin--Knill theorem~\cite{EastinKnill2009} fails in the Lee metric.  This code corrects every error of Lee weight one, yet it admits the continuous family $\exp(\mathrm{i}\theta X_1^5)$ of transversal logical gates.

\textit{GKP realization.}  
The Lee distance can be thought of as a discretized version of the displacement distance for oscillator and rotor Gottesman--Kitaev--Preskill (GKP) codes~\cite{Gottesman2001GKP}.
As such, the features of Lee stabilizer codes over $\mathbb{Z}_q$ carry over when these codes are concatenated with qudit GKP codes (Section~\ref{subsec:gkp}). In particular, the minimal $1$-norm of a logical displacement is $\sqrt{2\pi/q}$ times the minimal Lee weight of the corresponding logical Pauli string.  In this sense, one unit of Lee weight becomes one elementary displacement on the GKP grid.

\textit{Fault tolerance.}  Section~\ref{sec:ft} quantifies how Clifford gates spread Lee weight -- a direct analogy to Gaussian gates amplifying displacement errors in the GKP setting.  
The Lee error spread $w_{\mathrm B}(U)$ of a Clifford gate $U$ is the largest factor by which $U$ can increase the joint Lee weight of a Pauli string.  We show that it equals the largest Lee weight of a column of the symplectic matrix of $U$ (Lemma~\ref{lem:spread}).  For example, the phase gate and the $\mathrm{SUM}$ gate have spread two, while the multiplication gate $\mathrm{MUL}_2$ has spread $(q-1)/2$ for odd $q$.  The spread bounds the fault distance of a logical Clifford gate from below by $\lceil d_\diamond/(2w_{\mathrm B}(U))\rceil$. 
Modulo Pauli strings and phases, the gates of spread one 
are generated by qudit permutations and single-qudit Fourier gates.  These gates define the natural equivalence of Lee stabilizer codes.

Transversal gates are no longer automatically fault tolerant.  For instance, on the perfect $n$-qudit codes, a single fault before $\mathrm{MUL}_a^{\otimes n}$ with $a\neq\pm1$ can cause a logical error for every decoder that corrects all errors of Lee weight one.  On the other hand, for a CSS code with $H_X=H_Z$, every single-qudit Clifford gate $G$ yields a transversal logical gate $G^{\otimes n}$, up to a Pauli correction.  The transversal Fourier, phase, and $\mathrm{SUM}$ gates of such a code tolerate as many faults as in the Hamming metric.  On the Reed--Solomon code on $\mathbb{Z}_5$, these gates realize the entire logical Clifford group.  Modulo Pauli operators and phases, this group is isomorphic to the binary icosahedral group of order $120$.  On the perfect two-qudit codes, our lower bound on the fault distance of any logical Clifford gate of bounded spread grows as $\sqrt{q}$.  The transversal Fourier and phase gates reach the optimal fault distance $t+1$, where $2t+1$ is the distance.  The GKP realizations of these codes are pairs of square GKP qudits mixed by a beam splitter.  The physical Fourier and phase gates are the same Gaussian unitaries as the logical Fourier and phase gates of these square GKP qudits.

\textit{Qubitization by Gray maps.}  Section~\ref{sec:gray} converts codes on $\mathbb{Z}_4$ into qubit codes.  The classical Gray map $0,1,2,3\mapsto00,01,11,10$ is an isometry from the Lee metric on $\mathbb{Z}_4^n$ to the Hamming metric on $\mathbb{Z}_2^{2n}$~\cite{Hammons1994}.  We promote it to a unitary that relabels each four-dimensional qudit by two qubits.  This unitary sends $\mathbb{Z}_4$ Pauli operators to qubit Clifford gates and $\mathbb{Z}_4$ Clifford gates to the third level of the Clifford hierarchy (Theorem~\ref{thm:graylevels}).  A qubitized stabilizer code is therefore in general a Clifford stabilizer code rather than a Pauli one, and we characterize the CSS codes whose images are ordinary stabilizer codes (Proposition~\ref{prop:grayimage}).  The joint Lee distance and the Hamming distance of the qubitized code are bounded from below by a weighted Lee distance of the $\mathbb{Z}_4$ code (Theorem~\ref{thm:graydistance}).  For CSS codes, the qubitized distance lies between $\min(d_X,\lceil d_Z/2\rceil)$ and $\min(d_X,d_Z)$ (Theorem~\ref{thm:graycss}), and it equals $\min(d_X,d_Z)$ whenever every $X$ check is also a $Z$ check (Corollary~\ref{cor:semiselfdual}).  The self-dual CSS codes of Section~\ref{subsec:smallexamples} thus qubitize without loss.  The quantum heptacode becomes a nonadditive $((14,4,4))_{\mathrm H}$ code, and the $\qcode{23,1,10}_4$ code becomes a nonadditive $((46,4,10))_{\mathrm H}$ code (Table~\ref{tab:qubitized}).  Their single-block transversal $\mathbb{Z}_4$ Clifford gates become depth-one circuits of two-qubit gates, among them the non-Clifford gate $(T\otimes T)\,\mathrm{CS}$ and the two-qubit Fourier transform, and transversal $\mathrm{SUM}$ becomes a depth-one circuit of four-qubit gates built from the two-bit adder.

\textit{Lee-LDPC codes.}  
Section~\ref{sec:LDPC} introduces $(w_{\mathrm B},w_{\mathrm C})$-Lee-LDPC codes, codes whose check matrices have rows of Lee weight at most $w_{\mathrm C}$ and columns of Lee weight at most $w_{\mathrm B}$, with both bounds independent of $n$ and $q$.  They extend the classical low-Lee-density parity-check codes of Ref.~\cite{Santini2020LLDPC} to quantum CSS codes.
As an interesting example, we construct a symmetry-enriched topological order by taking the hypergraph product~\cite{TillichZemor2014} of the lifted Hamming code of Example~\ref{eg:liftedhamming} with itself (Example~\ref{eg:leetoric}). 
We prove two bounds on quantum Lee-LDPC codes with constants independent of $q$.  First, every nontrivial CSS Lee-LDPC code has a nontrivial logical operator whose energy barrier is at most linear in $n$ (Theorem~\ref{thm:energybarrier}).  Second, 
every noncommuting pair of logical operators has a noncommuting logical witness of Lee weight at most linear in $n$ (Theorem~\ref{thm:baddistance}).  In particular, such a code has Lee distance $O(n)$, even though Lee distances of general codes can grow with $q$ at fixed $n$.  The proof relaxes the problem from $\mathbb{Z}_q$ to the real numbers and applies transference theorems from the geometry of numbers~\cite{transference,2017thesis}.  Appendix~\ref{app:noncss} extends the second bound to non-CSS codes (Theorem~\ref{thm:noncsswitness}).  

\textit{High-temperature memory.}
Section~\ref{sec:memory} asks whether Lee-LDPC codes can store information passively at arbitrarily high temperature when $q$ grows with $n$.  
We introduce classical helical repetition codes (Proposition~\ref{prop:helicalproperties}) on a chain of $n$ qudits.  Each check acts on two neighboring sites and requires the entries $x_j$ and $x_{j+1}$ of a codeword to satisfy $ax_j=bx_{j+1}$ modulo $q$, for fixed integers $a>b\ge1$ that are coprime to each other and to $q$.  Every codeword is thus a geometric progression with ratio $a/b$ modulo $q$, and the Lee distance of the code grows exponentially with $n$ for sufficiently large $q$. For $b=1$, the periodic checks are the classical couplings of a generalized clock model~\cite{HuWatanabe2023clock}.  
We bound the memory time of single-site Metropolis dynamics, in which each step changes one site by $\pm1$.  We identify two mechanisms for an exponentially long memory time.  The first is energetic.  The periodic helical repetition code with $q=a^n-b^n$ is linearly sound and confining~\cite{Campbell2019singleshot,Quintavalle2021singleshot}.  In other words, the Lee weight of the syndrome is within constant factors of the Lee distance to the nearest codeword.
Its memory time is exponentially large in \(q\) at every fixed temperature. Consistently, the mean hitting time measured in Fig.~\ref{fig:PBC_tmix} grows approximately exponentially with \(q\). Up to a constant factor, this hitting time is a lower bound on the mixing time when the region it targets carries a constant fraction of the Gibbs weight~\cite{PeresSousi2015,Oliveira2012}.
Unlike expander-based constructions~\cite{Sipser1996expander,rakovszky2023goodldpcI,rakovszky2024goodldpcII,Hong2025quantummemory}, this code has geometrically local checks.  It evades the usual argument against spontaneous breaking of a discrete symmetry in one dimension~\cite{landau1980statistical,BergamaschiChen2025fastmixing} because $q$ grows with the length of the chain.  The second mechanism is entropic.  The open-boundary helical code has an energy barrier of only $O(n)$, yet its memory time is at least $\mathrm{e}^{cn}$ once $q$ exceeds a temperature-dependent multiple of $n$.  Along the logical direction, the energy is flat but the path to another codeword is exponentially long.  Leaving this direction instead costs an energy of order $q$.  We also find that the equilibrium Gibbs weight along the logical direction is highly nonuniform, yet spread over exponentially many values (Appendix~\ref{app:freeenergylandscape}), reminiscent of a ``glassy'' complex energy landscape, although this is not an optimal interpretation: see \cite{andy_dlr}.

Hypergraph products of helical codes carry both mechanisms over to quantum CSS codes, where they protect only one of the two error sectors, which is the \(Z\)-error sector in the examples studied in Section~\ref{sec:memory}.  In each case, the memory time for $Z$ errors is exponentially large at every fixed temperature once $q\gg n^2\log q$, where $n$ is the linear size of these codes, which have about $2n^2$ qudits.  The bounds of Section~\ref{sec:LDPC} do not, by themselves, rule out arbitrarily high thresholds for passive quantum memories.  Section~\ref{subsec:quantumobstruction} illustrates this with two examples.  The first is a code with weight-two checks and distance $q$ whose elementary errors do not obey the $\mathbb{Z}_q$ Pauli algebra.  The second is a classical repetition code over $\mathbb{Z}$ with an energy barrier of one, which survives infinite-temperature dynamics for a time of order $n^2$.  Whether a quantum Lee-LDPC code can protect both error sectors at arbitrarily high temperature remains open, and settling it will require more than an obvious extension of existing bottleneck theorems.

Section~\ref{sec:outlook} discusses extensions to more general nice error bases~\cite{Knill1996nice,KlappeneckerRoetteler2002I,KlappeneckerRoetteler2002II}.

\section{Motivating the Lee metric}

\subsection{$q$-dimensional qudits with clock and shift errors}
In this paper, we will build quantum codes out of $q$-dimensional qudits, i.e., ``$q$-dits''.   $q$ need not be a prime number.   For any $q$, there is a natural generalization of the Pauli algebra as follows.   We can label the basis elements of a $q$-dimensional Hilbert  space as $|0\rangle, \ldots , |q-1\rangle$ and define Pauli $X$ and $Z$ operators as \begin{subequations}
    \begin{align}
        X|k\rangle &:= |k+1\text{ (mod $q$)}\rangle, \\
        Z |k\rangle &:= \omega^k |k\rangle ,
    \end{align}
\end{subequations}
where $\omega := \mathrm{e}^{2\pi \mathrm{i}/q}$ is a primitive $q$th root of unity.  One often calls $X$ the ``shift" operator and $Z$ the ``clock" operator.  One easily shows that \begin{equation}
    X^aZ^b = \omega^{-ab} Z^b X^a.
\end{equation}
Once we are studying multiple qudits, we can simply tensor together these Pauli operators as per usual.   As is well known, these identities are sufficient to generalize the symplectic stabilizer formalism to codes on $\mathbb{Z}_q$, and we will review this in Section~\ref{sec:stabilizer}.

We are interested in situations where, for example, $X$ is a more likely error than $X^2$, and $Z$ is more likely than $Z^2$.   In this situation, we should not think about measuring the weight of an error based on the number of qudits it acts on.  This motivates the following definition:

\begin{defn}[Lee vs.\ Hamming weight]
    For a vector $x\in\mathbb{Z}_q^n$, there are two relevant notions of its size.  The Hamming weight \begin{equation}
        |x|_{\mathrm{H}} := \sum_{i=1}^n (1-\delta_{x_i,0})
    \end{equation}
    counts the number of nonzero entries of $x$.  The Lee weight is \begin{equation}
        |x|_{\mathrm{L}} := \sum_{i=1}^n \min_{m_i\in \mathbb{Z}} |x_i-m_iq| =  \sum_{i=1}^n \min (x_i,q-x_i).
    \end{equation}
    In the last equality above, we are assuming that each $x_i\in\lbrace 0,1,\ldots, q-1\rbrace$.  The Lee distance between $x,y\in\mathbb{Z}_q^n$ is $|x-y|_{\mathrm{L}}$, and the resulting metric on $\mathbb{Z}_q^n$ is the Lee metric.  The Hamming distance is defined analogously.
\end{defn}

\noindent As an example, for $q> 12$ in a two-qudit code, the error $X_1X_2$ has Hamming weight 2 and Lee weight 2, while the error $X_1^6$ has Hamming weight 1 and Lee weight 6.   If $X_1X_2$ is more likely as an error than the large shift $X_1^6$, it makes sense to think in terms of Lee weight rather than Hamming weight.

The Lee weight (or Lee norm) is the natural adaptation to $\mathbb{Z}_q$ of the 1-norm of a vector $y\in \mathbb{Z}^n$, \begin{equation}
    |y|_1 := \sum_{i=1}^n |y_i|.
\end{equation}
To avoid confusion, especially since later proofs lift $\mathbb{Z}_q$ to $\mathbb{Z}$, we always denote Lee weights and 1-norms differently.

\subsection{Warm-up with qubits: $Y$ errors are rarer than $X$ or $Z$ errors}
\label{subsec:qubits}

There are two ways to carry the Lee metric over to quantum codes.  The first, which we call the separate Lee metric, treats $X$-type and $Z$-type errors separately: a Pauli string $X^{x}Z^{z}$ with $x,z\in\mathbb{Z}_q^n$ is assigned the pair of Lee weights $|x|_{\mathrm{L}}$ and $|z|_{\mathrm{L}}$.  This is the natural choice for CSS codes, whose distances $d_X$ and $d_Z$ are defined in Section~\ref{sec:css}.  
For qubits, the first notion contains nothing new, since the Lee weight of a binary vector is its Hamming weight. The two weights then count bit flips and phase flips separately, as for asymmetric quantum codes~\cite{IoffeMezard2007asymmetric,Sarvepalli2009asymmetric,Wang2010asymmetric,Ezerman2011asymmetric}.

The second, which we call the joint Lee metric, assigns to the same string the single number \begin{equation}
    |X^{x}Z^{z}|_{\mathrm{L}} := |x|_{\mathrm{L}} + |z|_{\mathrm{L}}, \label{eq:jointlee}
\end{equation}
i.e., the Lee weight of the symplectic vector $(x\mid z)\in\mathbb{Z}_q^{2n}$.  The corresponding code distance is defined in Section~\ref{sec:stabilizer}. This notion differs from the Hamming metric even for qubits. Each $X$ or $Z$ factor of a Pauli string contributes one unit to \eqref{eq:jointlee}, while each $Y\propto XZ$ factor contributes two.  In other words, the joint Lee weight of a qubit Pauli string exceeds the number of qubits it acts on by the number of its $Y$ factors.  This is the appropriate metric when bit flips and phase flips arise from independent processes, so that a $Y$ error, being both, occurs only at second order in the error rate~\cite{Robertson2017tailored,Rigby2020asymmetric}.

Since every non-identity single-qubit Pauli string has Lee weight one or two, the joint Lee distance of a qubit code lies between its Hamming distance and twice its Hamming distance.  Lee distances can therefore exceed Hamming distances, and a code can protect against all errors of a given Lee weight using fewer qubits than the Hamming metric would require.  

The simplest example is the code with the single check $Y^{\otimes n}$, which we call the $Y$-parity code.  A single $X$ or $Z$ on any qubit anticommutes with this check, while a single $Y$ commutes with it and has Lee weight two.  Thus, the $Y$-parity code has Lee distance two but Hamming distance one.  Throughout this subsection, we write $\qcode{n,k,d}$ for a qubit code of joint Lee distance $d$, and $\qcode{n,k,d}_{\mathrm{H}}$ when $d$ is instead the Hamming distance, so that the two metrics can be compared side by side.  In this notation, the $Y$-parity code is an $\qcode{n,n-1,2}$ code for every $n\ge 2$.  Its Hamming-metric counterpart, the iceberg code with checks $X^{\otimes n}$ and $Z^{\otimes n}$~\cite{Gottesman1997thesis}, is an $\qcode{n,n-2,2}_{\mathrm{H}}$ code that exists only for even $n$.

To see how far this advantage extends, we went through the database of small stabilizer codes~\cite{Cross2025smallcodes,qiskitqec}, which lists every stabilizer code on $n\le 9$ qubits up to qubit permutations and single-qubit Clifford gates.  We call the frame in which a code is listed in Refs.~\cite{Cross2025smallcodes,qiskitqec} its Hamming frame.  Hamming distance is invariant under such gates, but joint Lee distance is not.  The Hadamard gate only swaps $X$ and $Z$ and hence preserves Lee weights, whereas the phase gate $S$ maps $X$ to $Y$ and can change them. Likewise, applying single-qubit Clifford gates to a surface code can change its logical error rate under biased noise by orders of magnitude~\cite{Dua2024clifford}.  For each code, we maximized the Lee distance over all single-qubit Clifford frames, which amounts to applying one of $I$, $S$, or $SH$ to each qubit and keeping the best of the resulting $3^n$ frames.  The resulting Lee distances and Lee-optimal frames are available in Ref.~\cite{eczoocodebase}.  Table~\ref{tab:qubit_lee} lists the largest Lee distance obtained in this way for each pair $(n,k)$, alongside the largest Hamming distance for the same pair.

\begin{table}[t]
    \centering
    \caption{\textbf{Largest joint Lee distance of qubit stabilizer codes with $n\le 9$.}  For each number of physical qubits $n$ and logical qubits $k$, the row labeled ``Lee'' lists the largest joint Lee distance \eqref{eq:jointlee} of any $\qcode{n,k}$ stabilizer code, maximized over single-qubit Clifford frames, and the row labeled ``Ham.'' lists the largest Hamming distance.  Parentheses give the number of codes attaining the respective maximum, counted up to qubit permutations and single-qubit Clifford gates.  The Hamming rows are extracted from the classification of Refs.~\cite{Cross2025smallcodes,qiskitqec}, and the Lee rows from its extension in Ref.~\cite{eczoocodebase}.  Bold entries exceed the Hamming maximum for the same $n$ and $k$.}
    \label{tab:qubit_lee}
    \vspace{0.5em}
    \small
    \begin{tabular}{cl|cccccccc}
        \hline\hline
        $n$ & & $k=1$ & $k=2$ & $k=3$ & $k=4$ & $k=5$ & $k=6$ & $k=7$ & $k=8$ \\
        \hline
        2 & Lee  & $\mathbf{2}$ (1) & & & & & & & \\
          & Ham. & 1 (2) & & & & & & & \\
        \hline
        3 & Lee  & $\mathbf{2}$ (3) & $\mathbf{2}$ (1) & & & & & & \\
          & Ham. & 1 (5) & 1 (3) & & & & & & \\
        \hline
        4 & Lee  & 2 (10) & 2 (6) & $\mathbf{2}$ (1) & & & & & \\
          & Ham. & 2 (2) & 2 (1) & 1 (4) & & & & & \\
        \hline
        5 & Lee  & 3 (5) & 2 (27) & $\mathbf{2}$ (8) & $\mathbf{2}$ (1) & & & & \\
          & Ham. & 3 (1) & 2 (3) & 1 (19) & 1 (5) & & & & \\
        \hline
        6 & Lee  & $\mathbf{4}$ (1) & $\mathbf{3}$ (2) & 2 (69) & 2 (13) & $\mathbf{2}$ (1) & & & \\
          & Ham. & 3 (2) & 2 (29) & 2 (5) & 2 (1) & 1 (6) & & & \\
        \hline
        7 & Lee  & $\mathbf{4}$ (17) & $\mathbf{3}$ (122) & $\mathbf{3}$ (6) & 2 (158) & $\mathbf{2}$ (16) & $\mathbf{2}$ (1) & & \\
          & Ham. & 3 (19) & 2 (241) & 2 (67) & 2 (7) & 1 (48) & 1 (7) & & \\
        \hline
        8 & Lee  & $\mathbf{4}$ (471) & $\mathbf{4}$ (195) & $\mathbf{4}$ (12) & 2 (3111) & 2 (347) & 2 (23) & $\mathbf{2}$ (1) & \\
          & Ham. & 3 (178) & 3 (20) & 3 (1) & 2 (264) & 2 (11) & 2 (1) & 1 (8) & \\
        \hline
        9 & Lee  & $\mathbf{5}$ (8) & $\mathbf{4}$ (23676) & $\mathbf{4}$ (1541) & $\mathbf{3}$ (9261) & 2 (14482) & 2 (730) & $\mathbf{2}$ (28) & $\mathbf{2}$ (1) \\
          & Ham. & 3 (3609) & 3 (4445) & 3 (222) & 2 (24117) & 2 (768) & 2 (13) & 1 (99) & 1 (9) \\
        \hline\hline
    \end{tabular}
\end{table}

Several entries of Table~\ref{tab:qubit_lee} deserve comment.  Two inequivalent $\qcode{6,2,3}$ codes correct any single $X$ or $Z$ error, and detect any single $Y$ error, while encoding two logical qubits.  In the Hamming metric, distance three on six qubits leaves room for only one logical qubit, as in the two inequivalent $\qcode{6,1,3}_{\mathrm{H}}$ codes, and two logical qubits with distance three first appear in the twenty inequivalent $\qcode{8,2,3}_{\mathrm{H}}$ codes.  

Similarly, six inequivalent $\qcode{7,3,3}$ codes encode two more logical qubits than the Steane code.  The one with the largest automorphism group, of order 1344, is a CSS code in its Hamming frame: its three $Z$-type checks generate the $[7,3,4]$ simplex code, and its single $X$-type check is $X^{\otimes 7}$.  Its Hamming distance is only two, since $Z_iZ_j$ commutes with all checks for any pair of qubits.  In a Lee-optimal frame, its stabilizer tableau is
\begin{equation}
    \begin{array}{ccccccc}
        Y & Y & Y & I & I & I & Y \\
        I & Y & I & Y & I & Y & Y \\
        Y & Y & I & I & Y & Y & I \\
        Z & Z & Z & Z & Z & Z & Z
    \end{array}\ ,
    \label{eq:lee733}
\end{equation}
where each row is one check and each column is one qubit.

The unique $\qcode{6,1,4}$ code is the unique non-decomposable $\qcode{6,1,3}_{\mathrm{H}}$ code, viewed in a single-qubit Clifford frame that raises its Lee distance to four (the only other $\qcode{6,1,3}_{\mathrm{H}}$ code is the five-qubit code tensored with a fixed single-qubit stabilizer state).  It is degenerate, with a check of Hamming weight two.  No stabilizer code on at most nine qubits reaches Hamming distance four with one logical qubit, whereas the joint Lee metric achieves this with six qubits.  Twelve inequivalent $\qcode{8,3,4}$ codes similarly reach Lee distance four with three logical qubits, whereas Hamming distance four with three logical qubits is again out of reach on at most nine qubits.  One of the twelve is a CSS code in its Hamming frame, with $Z$-type checks generating the $[8,4,4]$ extended Hamming code and the single $X$-type check $X^{\otimes 8}$, and its Hamming distance is two.  Another is Gottesman's eight-qubit code~\cite{Gottesman1997thesis}, the unique $\qcode{8,3,3}_{\mathrm{H}}$ code, whose Lee distance in a suitable frame is four.

Finally, eight inequivalent $\qcode{9,1,5}$ codes attain the largest Lee distance in the table and come from $\qcode{9,1,3}_{\mathrm{H}}$ codes.  Exactly one of the eight is nondegenerate in the Hamming sense, meaning that no nontrivial element of its stabilizer group has Hamming weight below three, and its Lee-optimal frame, as reported in Ref.~\cite{eczoocodebase}, is
\begin{equation}
    \begin{array}{ccccccccc}
        I & Y & I & I & I & Y & I & Y & Y \\
        I & I & Y & I & I & I & Y & Y & Y \\
        I & I & I & Y & I & Y & Y & I & Y \\
        I & I & I & I & Y & Y & Y & Y & I \\
        I & Z & I & Z & Z & Z & I & I & I \\
        I & I & Z & I & I & I & X & Z & X \\
        X & I & Z & Z & Z & I & Z & Y & I \\
        Y & Z & Z & I & Z & I & I & X & Y
    \end{array}\ .
    \label{eq:lee915}
\end{equation}
This code corrects any error of Lee weight at most two, e.g., any two $X$ or $Z$ errors or any single $Y$ error, using nine qubits.  The shortest code that corrects any two errors in the Hamming metric is the $\qcode{11,1,5}_{\mathrm{H}}$ code~\cite{Grassl2007codetables}.  On the other hand, the five-qubit code gains nothing from the Lee metric: its Lee distance is three in every single-qubit Clifford frame.  Section~\ref{subsec:spinnoise} examines when noise of low Lee weight arises for $q>2$.  The rest of this paper, starting with Section~\ref{sec:stabilizer}, develops the corresponding theory for general $q$, where, for $q\ge4$, the Lee metric departs from the Hamming metric even for single-qudit errors of one type, such as $X^2$.

\subsection{Spin noise expansion} \label{subsec:spinnoise}

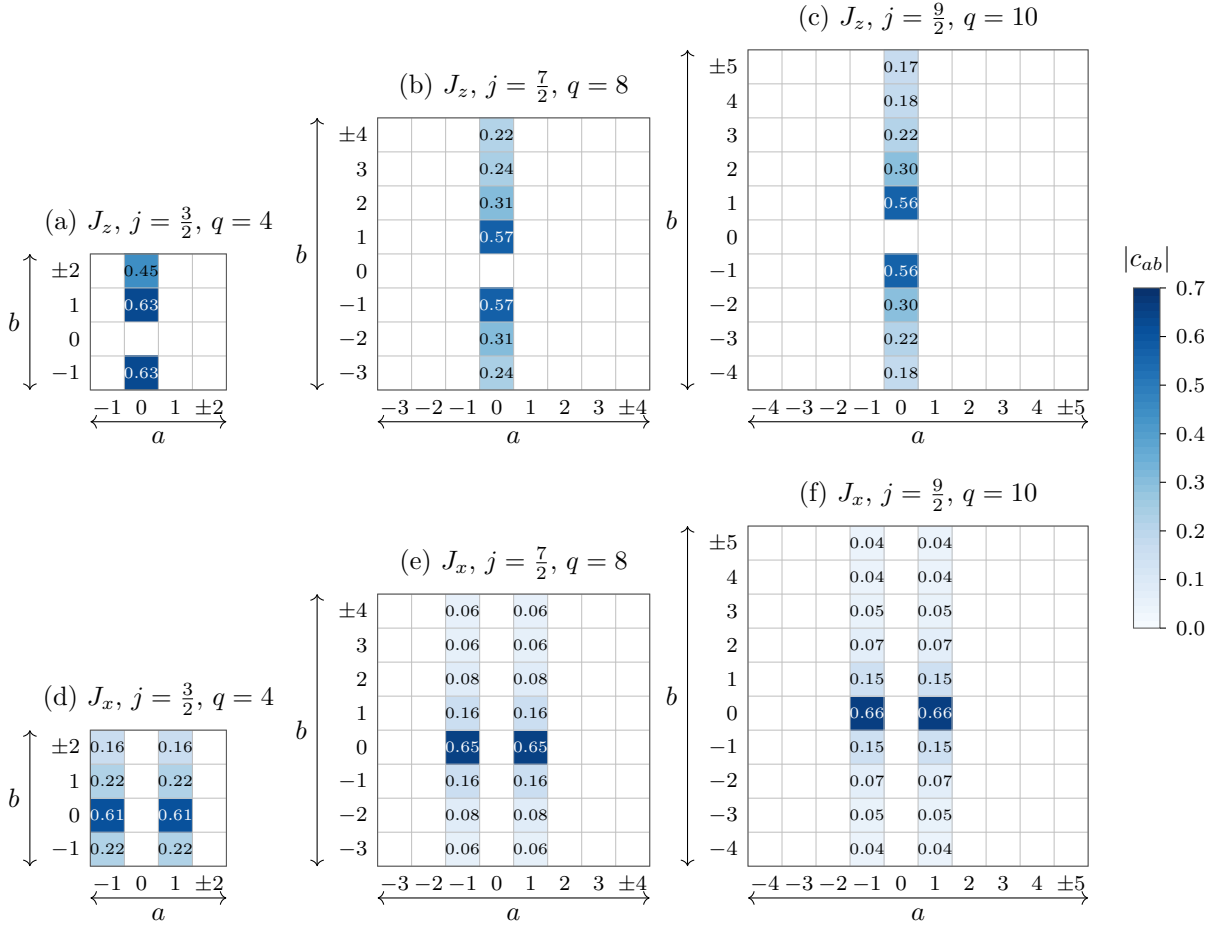
\begin{figure}[t]
    \centering
    \begin{tikzpicture}[x=1cm, y=1cm, font=\small]
        \begin{scope}[shift={(0.000,6.300)}]
            \fill[white] (0,0) rectangle (1.800,1.800);
            \fill[fill={rgb,255:red,8;green,73;blue,144}] (0.450,0.000) rectangle (0.900,0.450);
            \fill[fill={rgb,255:red,8;green,73;blue,144}] (0.450,0.900) rectangle (0.900,1.350);
            \fill[fill={rgb,255:red,62;green,142;blue,196}] (0.450,1.350) rectangle (0.900,1.800);
            \foreach \i in {0,...,4} {\draw[gray!50, very thin] (\i*0.45,0) -- (\i*0.45,1.800); \draw[gray!50, very thin] (0,\i*0.45) -- (1.800,\i*0.45);}
            \draw[gray!70!black, thin] (0,0) rectangle (1.800,1.800);
            \node[font=\tiny, text=white] at (0.675,0.225) {0.63};
            \node[font=\tiny, text=white] at (0.675,1.125) {0.63};
            \node[font=\tiny, text=black] at (0.675,1.575) {0.45};
            \node[font=\scriptsize, below=0pt] at (0.225,0) {$-1$}; \node[font=\scriptsize, below=0pt] at (0.675,0) {$0$}; \node[font=\scriptsize, below=0pt] at (1.125,0) {$1$}; \node[font=\scriptsize, below=0pt] at (1.575,0) {$\pm2$};
            \node[font=\scriptsize, left=0pt] at (0,0.225) {$-1$}; \node[font=\scriptsize, left=0pt] at (0,0.675) {$0$}; \node[font=\scriptsize, left=0pt] at (0,1.125) {$1$}; \node[font=\scriptsize, left=0pt] at (0,1.575) {$\pm2$};
            \draw[<->] (0,-0.42) -- (1.800,-0.42) node[midway,below=0pt] {$a$};
            \draw[<->] (-0.8,0) -- (-0.8,1.800) node[midway,left=0pt] {$b$};
            \node[above=2pt] at (0.900,1.800) {(a) $J_z$, $j=\tfrac{3}{2}$, $q=4$};
        \end{scope}
        \begin{scope}[shift={(0.000,0.000)}]
            \fill[white] (0,0) rectangle (1.800,1.800);
            \fill[fill={rgb,255:red,176;green,210;blue,231}] (0.000,0.000) rectangle (0.450,0.450);
            \fill[fill={rgb,255:red,8;green,81;blue,156}] (0.000,0.450) rectangle (0.450,0.900);
            \fill[fill={rgb,255:red,176;green,210;blue,231}] (0.000,0.900) rectangle (0.450,1.350);
            \fill[fill={rgb,255:red,202;green,221;blue,240}] (0.000,1.350) rectangle (0.450,1.800);
            \fill[fill={rgb,255:red,176;green,210;blue,231}] (0.900,0.000) rectangle (1.350,0.450);
            \fill[fill={rgb,255:red,8;green,81;blue,156}] (0.900,0.450) rectangle (1.350,0.900);
            \fill[fill={rgb,255:red,176;green,210;blue,231}] (0.900,0.900) rectangle (1.350,1.350);
            \fill[fill={rgb,255:red,202;green,221;blue,240}] (0.900,1.350) rectangle (1.350,1.800);
            \foreach \i in {0,...,4} {\draw[gray!50, very thin] (\i*0.45,0) -- (\i*0.45,1.800); \draw[gray!50, very thin] (0,\i*0.45) -- (1.800,\i*0.45);}
            \draw[gray!70!black, thin] (0,0) rectangle (1.800,1.800);
            \node[font=\tiny, text=black] at (0.225,0.225) {0.22};
            \node[font=\tiny, text=white] at (0.225,0.675) {0.61};
            \node[font=\tiny, text=black] at (0.225,1.125) {0.22};
            \node[font=\tiny, text=black] at (0.225,1.575) {0.16};
            \node[font=\tiny, text=black] at (1.125,0.225) {0.22};
            \node[font=\tiny, text=white] at (1.125,0.675) {0.61};
            \node[font=\tiny, text=black] at (1.125,1.125) {0.22};
            \node[font=\tiny, text=black] at (1.125,1.575) {0.16};
            \node[font=\scriptsize, below=0pt] at (0.225,0) {$-1$}; \node[font=\scriptsize, below=0pt] at (0.675,0) {$0$}; \node[font=\scriptsize, below=0pt] at (1.125,0) {$1$}; \node[font=\scriptsize, below=0pt] at (1.575,0) {$\pm2$};
            \node[font=\scriptsize, left=0pt] at (0,0.225) {$-1$}; \node[font=\scriptsize, left=0pt] at (0,0.675) {$0$}; \node[font=\scriptsize, left=0pt] at (0,1.125) {$1$}; \node[font=\scriptsize, left=0pt] at (0,1.575) {$\pm2$};
            \draw[<->] (0,-0.42) -- (1.800,-0.42) node[midway,below=0pt] {$a$};
            \draw[<->] (-0.8,0) -- (-0.8,1.800) node[midway,left=0pt] {$b$};
            \node[above=2pt] at (0.900,1.800) {(d) $J_x$, $j=\tfrac{3}{2}$, $q=4$};
        \end{scope}
        \begin{scope}[shift={(3.800,6.300)}]
            \fill[white] (0,0) rectangle (3.600,3.600);
            \fill[fill={rgb,255:red,170;green,207;blue,229}] (1.350,0.000) rectangle (1.800,0.450);
            \fill[fill={rgb,255:red,132;green,188;blue,219}] (1.350,0.450) rectangle (1.800,0.900);
            \fill[fill={rgb,255:red,20;green,96;blue,168}] (1.350,0.900) rectangle (1.800,1.350);
            \fill[fill={rgb,255:red,20;green,96;blue,168}] (1.350,1.800) rectangle (1.800,2.250);
            \fill[fill={rgb,255:red,132;green,188;blue,219}] (1.350,2.250) rectangle (1.800,2.700);
            \fill[fill={rgb,255:red,170;green,207;blue,229}] (1.350,2.700) rectangle (1.800,3.150);
            \fill[fill={rgb,255:red,179;green,211;blue,232}] (1.350,3.150) rectangle (1.800,3.600);
            \foreach \i in {0,...,8} {\draw[gray!50, very thin] (\i*0.45,0) -- (\i*0.45,3.600); \draw[gray!50, very thin] (0,\i*0.45) -- (3.600,\i*0.45);}
            \draw[gray!70!black, thin] (0,0) rectangle (3.600,3.600);
            \node[font=\tiny, text=black] at (1.575,0.225) {0.24};
            \node[font=\tiny, text=black] at (1.575,0.675) {0.31};
            \node[font=\tiny, text=white] at (1.575,1.125) {0.57};
            \node[font=\tiny, text=white] at (1.575,2.025) {0.57};
            \node[font=\tiny, text=black] at (1.575,2.475) {0.31};
            \node[font=\tiny, text=black] at (1.575,2.925) {0.24};
            \node[font=\tiny, text=black] at (1.575,3.375) {0.22};
            \node[font=\scriptsize, below=0pt] at (0.225,0) {$-3$}; \node[font=\scriptsize, below=0pt] at (0.675,0) {$-2$}; \node[font=\scriptsize, below=0pt] at (1.125,0) {$-1$}; \node[font=\scriptsize, below=0pt] at (1.575,0) {$0$}; \node[font=\scriptsize, below=0pt] at (2.025,0) {$1$}; \node[font=\scriptsize, below=0pt] at (2.475,0) {$2$}; \node[font=\scriptsize, below=0pt] at (2.925,0) {$3$}; \node[font=\scriptsize, below=0pt] at (3.375,0) {$\pm4$};
            \node[font=\scriptsize, left=0pt] at (0,0.225) {$-3$}; \node[font=\scriptsize, left=0pt] at (0,0.675) {$-2$}; \node[font=\scriptsize, left=0pt] at (0,1.125) {$-1$}; \node[font=\scriptsize, left=0pt] at (0,1.575) {$0$}; \node[font=\scriptsize, left=0pt] at (0,2.025) {$1$}; \node[font=\scriptsize, left=0pt] at (0,2.475) {$2$}; \node[font=\scriptsize, left=0pt] at (0,2.925) {$3$}; \node[font=\scriptsize, left=0pt] at (0,3.375) {$\pm4$};
            \draw[<->] (0,-0.42) -- (3.600,-0.42) node[midway,below=0pt] {$a$};
            \draw[<->] (-0.8,0) -- (-0.8,3.600) node[midway,left=0pt] {$b$};
            \node[above=2pt] at (1.800,3.600) {(b) $J_z$, $j=\tfrac{7}{2}$, $q=8$};
        \end{scope}
        \begin{scope}[shift={(3.800,0.000)}]
            \fill[white] (0,0) rectangle (3.600,3.600);
            \fill[fill={rgb,255:red,230;green,240;blue,249}] (0.900,0.000) rectangle (1.350,0.450);
            \fill[fill={rgb,255:red,224;green,236;blue,248}] (0.900,0.450) rectangle (1.350,0.900);
            \fill[fill={rgb,255:red,203;green,222;blue,241}] (0.900,0.900) rectangle (1.350,1.350);
            \fill[fill={rgb,255:red,8;green,65;blue,132}] (0.900,1.350) rectangle (1.350,1.800);
            \fill[fill={rgb,255:red,203;green,222;blue,241}] (0.900,1.800) rectangle (1.350,2.250);
            \fill[fill={rgb,255:red,224;green,236;blue,248}] (0.900,2.250) rectangle (1.350,2.700);
            \fill[fill={rgb,255:red,230;green,240;blue,249}] (0.900,2.700) rectangle (1.350,3.150);
            \fill[fill={rgb,255:red,231;green,240;blue,250}] (0.900,3.150) rectangle (1.350,3.600);
            \fill[fill={rgb,255:red,230;green,240;blue,249}] (1.800,0.000) rectangle (2.250,0.450);
            \fill[fill={rgb,255:red,224;green,236;blue,248}] (1.800,0.450) rectangle (2.250,0.900);
            \fill[fill={rgb,255:red,203;green,222;blue,241}] (1.800,0.900) rectangle (2.250,1.350);
            \fill[fill={rgb,255:red,8;green,65;blue,132}] (1.800,1.350) rectangle (2.250,1.800);
            \fill[fill={rgb,255:red,203;green,222;blue,241}] (1.800,1.800) rectangle (2.250,2.250);
            \fill[fill={rgb,255:red,224;green,236;blue,248}] (1.800,2.250) rectangle (2.250,2.700);
            \fill[fill={rgb,255:red,230;green,240;blue,249}] (1.800,2.700) rectangle (2.250,3.150);
            \fill[fill={rgb,255:red,231;green,240;blue,250}] (1.800,3.150) rectangle (2.250,3.600);
            \foreach \i in {0,...,8} {\draw[gray!50, very thin] (\i*0.45,0) -- (\i*0.45,3.600); \draw[gray!50, very thin] (0,\i*0.45) -- (3.600,\i*0.45);}
            \draw[gray!70!black, thin] (0,0) rectangle (3.600,3.600);
            \node[font=\tiny, text=black] at (1.125,0.225) {0.06};
            \node[font=\tiny, text=black] at (1.125,0.675) {0.08};
            \node[font=\tiny, text=black] at (1.125,1.125) {0.16};
            \node[font=\tiny, text=white] at (1.125,1.575) {0.65};
            \node[font=\tiny, text=black] at (1.125,2.025) {0.16};
            \node[font=\tiny, text=black] at (1.125,2.475) {0.08};
            \node[font=\tiny, text=black] at (1.125,2.925) {0.06};
            \node[font=\tiny, text=black] at (1.125,3.375) {0.06};
            \node[font=\tiny, text=black] at (2.025,0.225) {0.06};
            \node[font=\tiny, text=black] at (2.025,0.675) {0.08};
            \node[font=\tiny, text=black] at (2.025,1.125) {0.16};
            \node[font=\tiny, text=white] at (2.025,1.575) {0.65};
            \node[font=\tiny, text=black] at (2.025,2.025) {0.16};
            \node[font=\tiny, text=black] at (2.025,2.475) {0.08};
            \node[font=\tiny, text=black] at (2.025,2.925) {0.06};
            \node[font=\tiny, text=black] at (2.025,3.375) {0.06};
            \node[font=\scriptsize, below=0pt] at (0.225,0) {$-3$}; \node[font=\scriptsize, below=0pt] at (0.675,0) {$-2$}; \node[font=\scriptsize, below=0pt] at (1.125,0) {$-1$}; \node[font=\scriptsize, below=0pt] at (1.575,0) {$0$}; \node[font=\scriptsize, below=0pt] at (2.025,0) {$1$}; \node[font=\scriptsize, below=0pt] at (2.475,0) {$2$}; \node[font=\scriptsize, below=0pt] at (2.925,0) {$3$}; \node[font=\scriptsize, below=0pt] at (3.375,0) {$\pm4$};
            \node[font=\scriptsize, left=0pt] at (0,0.225) {$-3$}; \node[font=\scriptsize, left=0pt] at (0,0.675) {$-2$}; \node[font=\scriptsize, left=0pt] at (0,1.125) {$-1$}; \node[font=\scriptsize, left=0pt] at (0,1.575) {$0$}; \node[font=\scriptsize, left=0pt] at (0,2.025) {$1$}; \node[font=\scriptsize, left=0pt] at (0,2.475) {$2$}; \node[font=\scriptsize, left=0pt] at (0,2.925) {$3$}; \node[font=\scriptsize, left=0pt] at (0,3.375) {$\pm4$};
            \draw[<->] (0,-0.42) -- (3.600,-0.42) node[midway,below=0pt] {$a$};
            \draw[<->] (-0.8,0) -- (-0.8,3.600) node[midway,left=0pt] {$b$};
            \node[above=2pt] at (1.800,3.600) {(e) $J_x$, $j=\tfrac{7}{2}$, $q=8$};
        \end{scope}
        \begin{scope}[shift={(8.700,6.300)}]
            \fill[white] (0,0) rectangle (4.500,4.500);
            \fill[fill={rgb,255:red,195;green,218;blue,238}] (1.800,0.000) rectangle (2.250,0.450);
            \fill[fill={rgb,255:red,180;green,211;blue,233}] (1.800,0.450) rectangle (2.250,0.900);
            \fill[fill={rgb,255:red,138;green,191;blue,221}] (1.800,0.900) rectangle (2.250,1.350);
            \fill[fill={rgb,255:red,21;green,98;blue,169}] (1.800,1.350) rectangle (2.250,1.800);
            \fill[fill={rgb,255:red,21;green,98;blue,169}] (1.800,2.250) rectangle (2.250,2.700);
            \fill[fill={rgb,255:red,138;green,191;blue,221}] (1.800,2.700) rectangle (2.250,3.150);
            \fill[fill={rgb,255:red,180;green,211;blue,233}] (1.800,3.150) rectangle (2.250,3.600);
            \fill[fill={rgb,255:red,195;green,218;blue,238}] (1.800,3.600) rectangle (2.250,4.050);
            \fill[fill={rgb,255:red,199;green,219;blue,239}] (1.800,4.050) rectangle (2.250,4.500);
            \foreach \i in {0,...,10} {\draw[gray!50, very thin] (\i*0.45,0) -- (\i*0.45,4.500); \draw[gray!50, very thin] (0,\i*0.45) -- (4.500,\i*0.45);}
            \draw[gray!70!black, thin] (0,0) rectangle (4.500,4.500);
            \node[font=\tiny, text=black] at (2.025,0.225) {0.18};
            \node[font=\tiny, text=black] at (2.025,0.675) {0.22};
            \node[font=\tiny, text=black] at (2.025,1.125) {0.30};
            \node[font=\tiny, text=white] at (2.025,1.575) {0.56};
            \node[font=\tiny, text=white] at (2.025,2.475) {0.56};
            \node[font=\tiny, text=black] at (2.025,2.925) {0.30};
            \node[font=\tiny, text=black] at (2.025,3.375) {0.22};
            \node[font=\tiny, text=black] at (2.025,3.825) {0.18};
            \node[font=\tiny, text=black] at (2.025,4.275) {0.17};
            \node[font=\scriptsize, below=0pt] at (0.225,0) {$-4$}; \node[font=\scriptsize, below=0pt] at (0.675,0) {$-3$}; \node[font=\scriptsize, below=0pt] at (1.125,0) {$-2$}; \node[font=\scriptsize, below=0pt] at (1.575,0) {$-1$}; \node[font=\scriptsize, below=0pt] at (2.025,0) {$0$}; \node[font=\scriptsize, below=0pt] at (2.475,0) {$1$}; \node[font=\scriptsize, below=0pt] at (2.925,0) {$2$}; \node[font=\scriptsize, below=0pt] at (3.375,0) {$3$}; \node[font=\scriptsize, below=0pt] at (3.825,0) {$4$}; \node[font=\scriptsize, below=0pt] at (4.275,0) {$\pm5$};
            \node[font=\scriptsize, left=0pt] at (0,0.225) {$-4$}; \node[font=\scriptsize, left=0pt] at (0,0.675) {$-3$}; \node[font=\scriptsize, left=0pt] at (0,1.125) {$-2$}; \node[font=\scriptsize, left=0pt] at (0,1.575) {$-1$}; \node[font=\scriptsize, left=0pt] at (0,2.025) {$0$}; \node[font=\scriptsize, left=0pt] at (0,2.475) {$1$}; \node[font=\scriptsize, left=0pt] at (0,2.925) {$2$}; \node[font=\scriptsize, left=0pt] at (0,3.375) {$3$}; \node[font=\scriptsize, left=0pt] at (0,3.825) {$4$}; \node[font=\scriptsize, left=0pt] at (0,4.275) {$\pm5$};
            \draw[<->] (0,-0.42) -- (4.500,-0.42) node[midway,below=0pt] {$a$};
            \draw[<->] (-0.8,0) -- (-0.8,4.500) node[midway,left=0pt] {$b$};
            \node[above=2pt] at (2.250,4.500) {(c) $J_z$, $j=\tfrac{9}{2}$, $q=10$};
        \end{scope}
        \begin{scope}[shift={(8.700,0.000)}]
            \fill[white] (0,0) rectangle (4.500,4.500);
            \fill[fill={rgb,255:red,235;green,243;blue,251}] (1.350,0.000) rectangle (1.800,0.450);
            \fill[fill={rgb,255:red,233;green,242;blue,250}] (1.350,0.450) rectangle (1.800,0.900);
            \fill[fill={rgb,255:red,227;green,238;blue,249}] (1.350,0.900) rectangle (1.800,1.350);
            \fill[fill={rgb,255:red,206;green,224;blue,242}] (1.350,1.350) rectangle (1.800,1.800);
            \fill[fill={rgb,255:red,8;green,62;blue,129}] (1.350,1.800) rectangle (1.800,2.250);
            \fill[fill={rgb,255:red,206;green,224;blue,242}] (1.350,2.250) rectangle (1.800,2.700);
            \fill[fill={rgb,255:red,227;green,238;blue,249}] (1.350,2.700) rectangle (1.800,3.150);
            \fill[fill={rgb,255:red,233;green,242;blue,250}] (1.350,3.150) rectangle (1.800,3.600);
            \fill[fill={rgb,255:red,235;green,243;blue,251}] (1.350,3.600) rectangle (1.800,4.050);
            \fill[fill={rgb,255:red,235;green,243;blue,251}] (1.350,4.050) rectangle (1.800,4.500);
            \fill[fill={rgb,255:red,235;green,243;blue,251}] (2.250,0.000) rectangle (2.700,0.450);
            \fill[fill={rgb,255:red,233;green,242;blue,250}] (2.250,0.450) rectangle (2.700,0.900);
            \fill[fill={rgb,255:red,227;green,238;blue,249}] (2.250,0.900) rectangle (2.700,1.350);
            \fill[fill={rgb,255:red,206;green,224;blue,242}] (2.250,1.350) rectangle (2.700,1.800);
            \fill[fill={rgb,255:red,8;green,62;blue,129}] (2.250,1.800) rectangle (2.700,2.250);
            \fill[fill={rgb,255:red,206;green,224;blue,242}] (2.250,2.250) rectangle (2.700,2.700);
            \fill[fill={rgb,255:red,227;green,238;blue,249}] (2.250,2.700) rectangle (2.700,3.150);
            \fill[fill={rgb,255:red,233;green,242;blue,250}] (2.250,3.150) rectangle (2.700,3.600);
            \fill[fill={rgb,255:red,235;green,243;blue,251}] (2.250,3.600) rectangle (2.700,4.050);
            \fill[fill={rgb,255:red,235;green,243;blue,251}] (2.250,4.050) rectangle (2.700,4.500);
            \foreach \i in {0,...,10} {\draw[gray!50, very thin] (\i*0.45,0) -- (\i*0.45,4.500); \draw[gray!50, very thin] (0,\i*0.45) -- (4.500,\i*0.45);}
            \draw[gray!70!black, thin] (0,0) rectangle (4.500,4.500);
            \node[font=\tiny, text=black] at (1.575,0.225) {0.04};
            \node[font=\tiny, text=black] at (1.575,0.675) {0.05};
            \node[font=\tiny, text=black] at (1.575,1.125) {0.07};
            \node[font=\tiny, text=black] at (1.575,1.575) {0.15};
            \node[font=\tiny, text=white] at (1.575,2.025) {0.66};
            \node[font=\tiny, text=black] at (1.575,2.475) {0.15};
            \node[font=\tiny, text=black] at (1.575,2.925) {0.07};
            \node[font=\tiny, text=black] at (1.575,3.375) {0.05};
            \node[font=\tiny, text=black] at (1.575,3.825) {0.04};
            \node[font=\tiny, text=black] at (1.575,4.275) {0.04};
            \node[font=\tiny, text=black] at (2.475,0.225) {0.04};
            \node[font=\tiny, text=black] at (2.475,0.675) {0.05};
            \node[font=\tiny, text=black] at (2.475,1.125) {0.07};
            \node[font=\tiny, text=black] at (2.475,1.575) {0.15};
            \node[font=\tiny, text=white] at (2.475,2.025) {0.66};
            \node[font=\tiny, text=black] at (2.475,2.475) {0.15};
            \node[font=\tiny, text=black] at (2.475,2.925) {0.07};
            \node[font=\tiny, text=black] at (2.475,3.375) {0.05};
            \node[font=\tiny, text=black] at (2.475,3.825) {0.04};
            \node[font=\tiny, text=black] at (2.475,4.275) {0.04};
            \node[font=\scriptsize, below=0pt] at (0.225,0) {$-4$}; \node[font=\scriptsize, below=0pt] at (0.675,0) {$-3$}; \node[font=\scriptsize, below=0pt] at (1.125,0) {$-2$}; \node[font=\scriptsize, below=0pt] at (1.575,0) {$-1$}; \node[font=\scriptsize, below=0pt] at (2.025,0) {$0$}; \node[font=\scriptsize, below=0pt] at (2.475,0) {$1$}; \node[font=\scriptsize, below=0pt] at (2.925,0) {$2$}; \node[font=\scriptsize, below=0pt] at (3.375,0) {$3$}; \node[font=\scriptsize, below=0pt] at (3.825,0) {$4$}; \node[font=\scriptsize, below=0pt] at (4.275,0) {$\pm5$};
            \node[font=\scriptsize, left=0pt] at (0,0.225) {$-4$}; \node[font=\scriptsize, left=0pt] at (0,0.675) {$-3$}; \node[font=\scriptsize, left=0pt] at (0,1.125) {$-2$}; \node[font=\scriptsize, left=0pt] at (0,1.575) {$-1$}; \node[font=\scriptsize, left=0pt] at (0,2.025) {$0$}; \node[font=\scriptsize, left=0pt] at (0,2.475) {$1$}; \node[font=\scriptsize, left=0pt] at (0,2.925) {$2$}; \node[font=\scriptsize, left=0pt] at (0,3.375) {$3$}; \node[font=\scriptsize, left=0pt] at (0,3.825) {$4$}; \node[font=\scriptsize, left=0pt] at (0,4.275) {$\pm5$};
            \draw[<->] (0,-0.42) -- (4.500,-0.42) node[midway,below=0pt] {$a$};
            \draw[<->] (-0.8,0) -- (-0.8,4.500) node[midway,left=0pt] {$b$};
            \node[above=2pt] at (2.250,4.500) {(f) $J_x$, $j=\tfrac{9}{2}$, $q=10$};
        \end{scope}
        \begin{scope}[shift={(13.800,3.150)}]
            \fill[fill={rgb,255:red,245;green,250;blue,254}] (0,0.000) rectangle (0.35,0.100);
            \fill[fill={rgb,255:red,241;green,247;blue,253}] (0,0.100) rectangle (0.35,0.200);
            \fill[fill={rgb,255:red,236;green,244;blue,251}] (0,0.200) rectangle (0.35,0.300);
            \fill[fill={rgb,255:red,232;green,241;blue,250}] (0,0.300) rectangle (0.35,0.400);
            \fill[fill={rgb,255:red,227;green,238;blue,249}] (0,0.400) rectangle (0.35,0.500);
            \fill[fill={rgb,255:red,223;green,235;blue,247}] (0,0.500) rectangle (0.35,0.600);
            \fill[fill={rgb,255:red,219;green,233;blue,246}] (0,0.600) rectangle (0.35,0.700);
            \fill[fill={rgb,255:red,214;green,230;blue,244}] (0,0.700) rectangle (0.35,0.800);
            \fill[fill={rgb,255:red,210;green,227;blue,243}] (0,0.800) rectangle (0.35,0.900);
            \fill[fill={rgb,255:red,205;green,224;blue,241}] (0,0.900) rectangle (0.35,1.000);
            \fill[fill={rgb,255:red,202;green,221;blue,240}] (0,1.000) rectangle (0.35,1.100);
            \fill[fill={rgb,255:red,196;green,218;blue,238}] (0,1.100) rectangle (0.35,1.200);
            \fill[fill={rgb,255:red,189;green,215;blue,236}] (0,1.200) rectangle (0.35,1.300);
            \fill[fill={rgb,255:red,183;green,212;blue,234}] (0,1.300) rectangle (0.35,1.400);
            \fill[fill={rgb,255:red,175;green,209;blue,231}] (0,1.400) rectangle (0.35,1.500);
            \fill[fill={rgb,255:red,168;green,206;blue,228}] (0,1.500) rectangle (0.35,1.600);
            \fill[fill={rgb,255:red,161;green,203;blue,226}] (0,1.600) rectangle (0.35,1.700);
            \fill[fill={rgb,255:red,153;green,199;blue,224}] (0,1.700) rectangle (0.35,1.800);
            \fill[fill={rgb,255:red,143;green,194;blue,222}] (0,1.800) rectangle (0.35,1.900);
            \fill[fill={rgb,255:red,135;green,189;blue,220}] (0,1.900) rectangle (0.35,2.000);
            \fill[fill={rgb,255:red,125;green,184;blue,218}] (0,2.000) rectangle (0.35,2.100);
            \fill[fill={rgb,255:red,116;green,179;blue,216}] (0,2.100) rectangle (0.35,2.200);
            \fill[fill={rgb,255:red,106;green,174;blue,214}] (0,2.200) rectangle (0.35,2.300);
            \fill[fill={rgb,255:red,100;green,169;blue,211}] (0,2.300) rectangle (0.35,2.400);
            \fill[fill={rgb,255:red,92;green,164;blue,208}] (0,2.400) rectangle (0.35,2.500);
            \fill[fill={rgb,255:red,84;green,159;blue,205}] (0,2.500) rectangle (0.35,2.600);
            \fill[fill={rgb,255:red,78;green,154;blue,203}] (0,2.600) rectangle (0.35,2.700);
            \fill[fill={rgb,255:red,70;green,149;blue,200}] (0,2.700) rectangle (0.35,2.800);
            \fill[fill={rgb,255:red,63;green,143;blue,197}] (0,2.800) rectangle (0.35,2.900);
            \fill[fill={rgb,255:red,58;green,138;blue,194}] (0,2.900) rectangle (0.35,3.000);
            \fill[fill={rgb,255:red,52;green,132;blue,191}] (0,3.000) rectangle (0.35,3.100);
            \fill[fill={rgb,255:red,46;green,126;blue,188}] (0,3.100) rectangle (0.35,3.200);
            \fill[fill={rgb,255:red,41;green,121;blue,185}] (0,3.200) rectangle (0.35,3.300);
            \fill[fill={rgb,255:red,34;green,114;blue,182}] (0,3.300) rectangle (0.35,3.400);
            \fill[fill={rgb,255:red,29;green,108;blue,177}] (0,3.400) rectangle (0.35,3.500);
            \fill[fill={rgb,255:red,25;green,103;blue,173}] (0,3.500) rectangle (0.35,3.600);
            \fill[fill={rgb,255:red,21;green,97;blue,169}] (0,3.600) rectangle (0.35,3.700);
            \fill[fill={rgb,255:red,16;green,91;blue,164}] (0,3.700) rectangle (0.35,3.800);
            \fill[fill={rgb,255:red,11;green,85;blue,159}] (0,3.800) rectangle (0.35,3.900);
            \fill[fill={rgb,255:red,8;green,80;blue,155}] (0,3.900) rectangle (0.35,4.000);
            \fill[fill={rgb,255:red,8;green,74;blue,145}] (0,4.000) rectangle (0.35,4.100);
            \fill[fill={rgb,255:red,8;green,68;blue,136}] (0,4.100) rectangle (0.35,4.200);
            \fill[fill={rgb,255:red,8;green,62;blue,129}] (0,4.200) rectangle (0.35,4.300);
            \fill[fill={rgb,255:red,8;green,56;blue,119}] (0,4.300) rectangle (0.35,4.400);
            \fill[fill={rgb,255:red,8;green,50;blue,110}] (0,4.400) rectangle (0.35,4.500);
            \draw[gray!70!black, thin] (0,0) rectangle (0.35,4.5);
            \draw[thin] (0.35,0.000) -- (0.43,0.000) node[font=\scriptsize, right=0pt] {$0.0$};
            \draw[thin] (0.35,0.643) -- (0.43,0.643) node[font=\scriptsize, right=0pt] {$0.1$};
            \draw[thin] (0.35,1.286) -- (0.43,1.286) node[font=\scriptsize, right=0pt] {$0.2$};
            \draw[thin] (0.35,1.929) -- (0.43,1.929) node[font=\scriptsize, right=0pt] {$0.3$};
            \draw[thin] (0.35,2.571) -- (0.43,2.571) node[font=\scriptsize, right=0pt] {$0.4$};
            \draw[thin] (0.35,3.214) -- (0.43,3.214) node[font=\scriptsize, right=0pt] {$0.5$};
            \draw[thin] (0.35,3.857) -- (0.43,3.857) node[font=\scriptsize, right=0pt] {$0.6$};
            \draw[thin] (0.35,4.500) -- (0.43,4.500) node[font=\scriptsize, right=0pt] {$0.7$};
            \node[above=2pt] at (0.175,4.5) {$|c_{ab}|$};
        \end{scope}
    \end{tikzpicture}
    \caption{\textbf{Pauli expansion of spin noise.}  Each cell of a $q\times q$ grid is the single-qudit Pauli string $X^aZ^b$ in the plane of exponents $(a,b)$, in the same layout as Fig.~\ref{fig:leemetrics} below, with $a,b\in\{-\lfloor (q-1)/2\rfloor,\ldots,\lfloor q/2\rfloor\}$.  For even $q$, the exponent $q/2$ is its own negative, e.g., $X^{5}=X^{-5}$ for $q=10$, and its tick is labeled $\pm q/2$.  The color and the number in each cell give the coefficient magnitude $|c_{ab}|$ in the expansion \eqref{eq:pauliexpansion} of the spin operator named above the panel.  Each operator is normalized so that $\|E\|=1$, i.e., the squares of the numbers in each grid sum to one.  Blank cells (those without a number) have $c_{ab}=0$ exactly.  Panels (a)--(c) show the dephasing generator $J_z$, and panels (d)--(f) show the transverse-rotation generator $J_x$, for spins $j=3/2$, $7/2$, and $9/2$, i.e., $q=4$, $8$, and $10$.  The expansion of $J_z$ is supported on the column $a=0$ and decays as $1/\sin(\pi|b|_{\mathrm{L}}/q)$ by \eqref{eq:jzexpansion}.  The expansion of $J_x$ is supported on the columns $a=\pm1$ and is concentrated at $b=0$.  In both cases, most of the squared norm lies at low Lee weight.  The fraction at joint Lee weight one is $0.80$, $0.65$, and $0.63$ for $J_z$ and $0.75$, $0.85$, and $0.87$ for $J_x$, from left to right.}
    \label{fig:hyperfine}
\end{figure}

The advantage of the Lee metric in Section~\ref{subsec:qubits} rests on the assumption that errors of low Lee weight are the most likely ones.  There is often a sharp divide between noise that is simple to simulate and the noise present in a particular platform, but sometimes one can have the best of both worlds.  Displacement noise, i.e., noise composed of small shifts in position and momentum, is the natural setting for GKP codes~\cite{Gottesman2001GKP}.  Displacement operators form a complete basis for operators on an oscillator, so any physical error, such as a Kraus operator of a noise channel, can be expanded in this basis~\cite{Albert2018bosonic}.  The expansion provides a test of the displacement model: if it is mostly supported on small displacements, then the model approximates the noise well.  This turned out to be the case, quite surprisingly, for photon loss.  The products $E_\ell^\dagger E_{\ell'}$ of the Kraus operators of the loss channel, which govern error correction, expand into displacements whose size is of order the square root of the loss rate.  At moderate loss rates, GKP codes even outperform codes designed for loss~\cite{Albert2018bosonic}.

We call noise in which the probability of a Pauli error decreases with its Lee weight \emph{Lee noise}.  It is the direct analogue of displacement noise, and the same test applies to it.  The $q^2$ Pauli strings $X^aZ^b$ with $a,b\in\mathbb{Z}_q$ form an orthonormal basis for operators on a qudit with respect to the inner product $\frac{1}{q}\operatorname{tr}(A^\dagger B)$.  Any operator $E$ on the qudit, for example a Kraus operator of a physical noise channel, therefore expands as
\begin{equation}
    E = \sum_{a,b\in\mathbb{Z}_q} c_{ab}\,X^aZ^b, \qquad c_{ab} = \frac{1}{q}\operatorname{tr}\big[(X^aZ^b)^\dagger E\big], \qquad \sum_{a,b\in\mathbb{Z}_q}|c_{ab}|^2 = \frac{1}{q}\operatorname{tr}(E^\dagger E) =: \|E\|^2. \label{eq:pauliexpansion}
\end{equation}
The last identity defines the normalized Hilbert--Schmidt norm $\|E\|$.  Grouping the terms by the joint Lee weight $w=|a|_{\mathrm{L}}+|b|_{\mathrm{L}}$ of $X^aZ^b$ from \eqref{eq:jointlee} yields the distribution
\begin{equation}
    p_w := \frac{1}{\|E\|^2}\sum_{|a|_{\mathrm{L}}+|b|_{\mathrm{L}}=w}|c_{ab}|^2, \qquad \sum_{w\ge 0}p_w = 1, \label{eq:leedistribution}
\end{equation}
which records the fraction of $E$ that lies at each Lee weight $w$.  A code of joint Lee distance $2t+1$ corrects every Pauli string of Lee weight at most $t$ (Section~\ref{sec:stabilizer}).  Thus, if $p_w$ is concentrated on $w\le t$, such a code corrects the dominant part of $E$, and the tail $\sum_{w>t}p_w$ quantifies the part that it neglects.  The separate metric is handled in the same way, with $\max(|a|_{\mathrm{L}},|b|_{\mathrm{L}})$ in place of the sum.  The expansion extends to $n$ qudits, with $n$-qudit Pauli strings $X^{x}Z^{z}$, $x,z\in\mathbb{Z}_q^n$, in place of single-qudit ones.

We apply this test to spin qudits, such as hyperfine and nuclear spins.  A spin $j$ has $q=2j+1$ levels $|m\rangle_{j}$ with $m=-j,\ldots,j$, which we identify with the qudit basis $|k\rangle$ via $|m\rangle_{j}\mapsto|m+j\rangle$.  Examples include the nuclear spins of donors in silicon, such as $^{75}$As with $j=3/2$~\cite{Gross2024nuclear} and $^{123}$Sb with $j=7/2$~\cite{Asaad2020nucleus}, the nuclear spin $j=9/2$ of $^{87}$Sr~\cite{Omanakuttan2024spincat}, and the hyperfine manifolds of alkali atoms.  The dominant errors in such systems are linear, and in some platforms also quadratic, in the angular momentum operators $J_x$, $J_y$, and $J_z$~\cite{Gross2021spin,Gross2024nuclear,Omanakuttan2024spincat,jain2024absorptionemissiona}.  We treat the linear terms here.  In particular, dephasing is the primary error in nuclear and semiconductor spin qubits~\cite{Gross2024nuclear,Hendrickx2024sweetspot}.  It consists of random rotations $\mathrm{e}^{-\mathrm{i}\theta J_z}$ about the quantization axis, caused by magnetic field noise, charge noise, and, for electronic spins, hyperfine coupling to host nuclei.  To first order in the rotation angle $\theta$, the error operator is $J_z$ itself.  Transverse field fluctuations and control errors produce rotations generated by $J_x$ and $J_y$ in the same way.

The operator $J_z$ is diagonal in the qudit basis, so its expansion contains only powers of $Z$.  A direct computation gives
\begin{equation}
    J_z = \sum_{b\in\mathbb{Z}_q\setminus\{0\}}\left(-\frac{1}{2}+\frac{\mathrm{i}}{2}\cot\frac{\pi b}{q}\right)Z^b, \qquad |c_{0b}| = \frac{1}{2\sin(\pi|b|_{\mathrm{L}}/q)} \approx \frac{q}{2\pi|b|_{\mathrm{L}}}, \label{eq:jzexpansion}
\end{equation}
where the approximation holds for $|b|_{\mathrm{L}}\ll q$.  In other words, $J_z$ is the sawtooth function $k-j$ of the level index $k$, and the coefficients of its Fourier series decay inversely with Lee weight.  The slow decay is caused by the jump of the sawtooth between the levels $m=j$ and $m=-j$, which are adjacent on the cycle $\mathbb{Z}_q$ but far apart for the spin.  Since $\|J_z\|^2=\|J_x\|^2=\frac{1}{q}\sum_{m=-j}^{j}m^2=j(j+1)/3$, the resulting Lee-weight distribution is $p_w = 3/[2j(j+1)\sin^2(\pi w/q)]$ for $1\le w<q/2$, and half of that value at $w=q/2$ for even $q$.  As $q\to\infty$, it tends to $p_w = 6/(\pi^2w^2)$.  Thus, Lee weight one carries more than $6/\pi^2\approx 61\%$ of the squared norm of $J_z$ for every spin, while the tail beyond Lee weight $t$ decays only as $\sum_{w>t}p_w\approx 6/(\pi^2 t)$ for $1\ll t\ll q$.

The operator $J_x=(J_++J_-)/2$ shifts the level index by one, so its expansion is supported on the strings $X^{\pm1}Z^b$.  Explicitly, $J_+=XD$ with $D=\sum_{k=0}^{q-1}d_k|k\rangle\langle k|$ and $d_k=\sqrt{(k+1)(q-1-k)}$, so that $J_x=\frac{1}{2}(XD+DX^{-1})$.  The coefficients of $X^{\pm1}Z^b$ are therefore, up to phases, half of the Fourier coefficients of $d_k$.  Unlike the sawtooth, the profile $d_k$ is a semicircle, with zeros at $k=-1$ and $k=q-1$ that coincide on the cycle.  It therefore has no jump of order $q$ at the wraparound, only a square-root singularity, and its Fourier coefficients decay approximately as $|b|_{\mathrm{L}}^{-3/2}$.  The area of the semicircle gives $|c_{\pm1,0}|\approx\pi q/16$ for large $q$, and dividing by $\|J_x\|^2\approx q^2/12$ shows that the fraction of $J_x$ at joint Lee weight one tends to $3\pi^2/32\approx 93\%$ as $q\to\infty$.  The operator $J_y$ has the same coefficient magnitudes as $J_x$.

Figure~\ref{fig:hyperfine} shows the coefficient magnitudes of $J_z$ and $J_x$ for $j=3/2$, $7/2$, and $9/2$ on the grid of exponents $(a,b)$.  For $j=9/2$, the fractions of $J_z$ and $J_x$ at joint Lee weight one are $0.63$ and $0.87$, and the fraction of $J_x$ at Lee weight at most two is $0.96$.  Both expansions are dominated by terms of low Lee weight, and the Lee weight of $J_x$ is set almost entirely by its $X$ part.  The approximation by Lee noise is not exact.  The tails of both distributions are nonzero, and higher orders in $\theta$ contribute products of these operators, whose supports extend to larger Lee weights.  Nevertheless, the concentration of spin noise at low Lee weight provides the same type of promising avenue for Lee codes that the displacement expansion of loss provides for GKP codes~\cite{Albert2018bosonic}.

\section{Stabilizer codes} \label{sec:stabilizer}

This section sets up general stabilizer codes on $\mathbb{Z}_q$ in the Lee metric and identifies the perfect codes among them.  Recall from Section \ref{subsec:qubits} that there are two ways to measure a Pauli string $X^{x}Z^{z}$, with $x,z\in\mathbb{Z}_q^n$, in the Lee metric.  The separate metric records the pair of Lee weights $|x|_{\mathrm{L}}$ and $|z|_{\mathrm{L}}$, and the joint metric records their sum \eqref{eq:jointlee}.  Figure~\ref{fig:leemetrics} shows the two metrics for a single qudit.  The errors of joint Lee weight at most $t$ form a Lee ball, which is a discretized rhombus in the plane of exponents $(a,b)$ of $X^aZ^b$.  The errors with both Lee weights at most $t$ instead form a square of side length $2t+1$.  The joint metric was used for the qubit codes of Section \ref{subsec:qubits}, and it applies without change to $q>2$, where the errors of joint Lee weight one are $X^{\pm 1}$ and $Z^{\pm 1}$ on a single qudit.  The separate metric is the natural one for CSS codes, whose distances $d_X$ and $d_Z$ are defined in Section \ref{sec:css}.  For $q=2$ and $q=3$, every nonzero element of $\mathbb{Z}_q$ has Lee weight one, so $d_X$ and $d_Z$ coincide with the Hamming distances of the CSS code.  The case $q=4$ is the first in which the two can differ, since $X^2$ then has Lee weight two.

\newcommand{\leegridlabels}{%
    \foreach \a/\b/\lab in {%
        -2/2/{X^{-2}Z^{2}}, -1/2/{X^{-1}Z^{2}}, 0/2/{Z^{2}}, 1/2/{XZ^{2}}, 2/2/{X^{2}Z^{2}},%
        -2/1/{X^{-2}Z}, -1/1/{X^{-1}Z}, 0/1/{Z}, 1/1/{XZ}, 2/1/{X^{2}Z},%
        -2/0/{X^{-2}}, -1/0/{X^{-1}}, 0/0/{I}, 1/0/{X}, 2/0/{X^{2}},%
        -2/-1/{X^{-2}Z^{-1}}, -1/-1/{X^{-1}Z^{-1}}, 0/-1/{Z^{-1}}, 1/-1/{XZ^{-1}}, 2/-1/{X^{2}Z^{-1}},%
        -2/-2/{X^{-2}Z^{-2}}, -1/-2/{X^{-1}Z^{-2}}, 0/-2/{Z^{-2}}, 1/-2/{XZ^{-2}}, 2/-2/{X^{2}Z^{-2}}}
        {\node at (\a,\b) {$\lab$};}%
}
\newcommand{\leegridlines}{%
    \foreach \i in {-2.5,-1.5,-0.5,0.5,1.5,2.5}{%
        \draw[gray!50, very thin] (\i,-2.5) -- (\i,2.5);
        \draw[gray!50, very thin] (-2.5,\i) -- (2.5,\i);
    }%
}
\newcommand{\leerhombusouter}{(-0.5,2.5) -- (0.5,2.5) -- (0.5,1.5) -- (1.5,1.5) -- (1.5,0.5) -- (2.5,0.5) -- (2.5,-0.5) -- (1.5,-0.5) -- (1.5,-1.5) -- (0.5,-1.5) -- (0.5,-2.5) -- (-0.5,-2.5) -- (-0.5,-1.5) -- (-1.5,-1.5) -- (-1.5,-0.5) -- (-2.5,-0.5) -- (-2.5,0.5) -- (-1.5,0.5) -- (-1.5,1.5) -- (-0.5,1.5) -- cycle}
\newcommand{\leerhombusinner}{(-0.5,1.5) -- (0.5,1.5) -- (0.5,0.5) -- (1.5,0.5) -- (1.5,-0.5) -- (0.5,-0.5) -- (0.5,-1.5) -- (-0.5,-1.5) -- (-0.5,-0.5) -- (-1.5,-0.5) -- (-1.5,0.5) -- (-0.5,0.5) -- cycle}

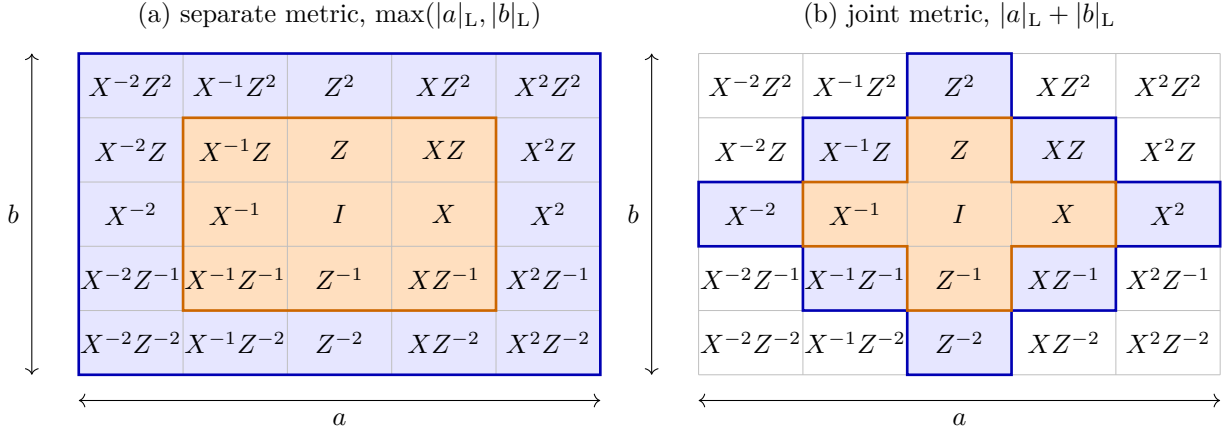
\begin{figure}[t]
    \centering
    \begin{tikzpicture}[x=1.38cm, y=0.85cm, font=\small]
        \begin{scope}
            \fill[blue!10] (-2.5,-2.5) rectangle (2.5,2.5);
            \fill[orange!25] (-1.5,-1.5) rectangle (1.5,1.5);
            \leegridlines
            \draw[blue!70!black, line width=1pt] (-2.5,-2.5) rectangle (2.5,2.5);
            \draw[orange!80!black, line width=1pt] (-1.5,-1.5) rectangle (1.5,1.5);
            \leegridlabels
            \node at (0,3.1) {(a) separate metric, $\max(|a|_{\mathrm{L}},|b|_{\mathrm{L}})$};
            \draw[<->] (-2.5,-2.9) -- (2.5,-2.9) node[midway,below=1pt] {$a$};
            \draw[<->] (-2.95,-2.5) -- (-2.95,2.5) node[midway,left=1pt] {$b$};
        \end{scope}
        \begin{scope}[xshift=8.2cm]
            \fill[blue!10] \leerhombusouter;
            \fill[orange!25] \leerhombusinner;
            \leegridlines
            \draw[blue!70!black, line width=1pt] \leerhombusouter;
            \draw[orange!80!black, line width=1pt] \leerhombusinner;
            \leegridlabels
            \node at (0,3.1) {(b) joint metric, $|a|_{\mathrm{L}}+|b|_{\mathrm{L}}$};
            \draw[<->] (-2.5,-2.9) -- (2.5,-2.9) node[midway,below=1pt] {$a$};
            \draw[<->] (-2.95,-2.5) -- (-2.95,2.5) node[midway,left=1pt] {$b$};
        \end{scope}
    \end{tikzpicture}
    \caption{\textbf{Two Lee metrics on Pauli errors.}  Each cell of the $5\times5$ grid is the single-qudit Pauli error $X^aZ^b$ with $a,b\in\{-2,-1,0,1,2\}$ ($q\ge5$).  (a) The separate metric assigns to $X^aZ^b$ the pair $(|a|_{\mathrm{L}},|b|_{\mathrm{L}})$, and the errors with both entries at most $t$ form a square of side length $2t+1$.  (b) The joint metric \eqref{eq:jointlee} assigns the single number $|a|_{\mathrm{L}}+|b|_{\mathrm{L}}$, and the errors of joint Lee weight at most $t$ form a Lee ball of radius $t$.  For $q\ge2t+1$, the Lee ball is a polyomino with $2t^2+2t+1$ cells that approximates a rhombus~\cite{GolombWelch1970}.  For $t=1$, it is the plus-shaped pentomino.  In both panels, orange marks $t=1$ and blue marks $t=2$, with the boundary of each region drawn in the matching color.  A code with $d_\square\ge 2t+1$ corrects every error in the square of side length $2t+1$, and a code with $d_\diamond\ge 2t+1$ corrects every error in the Lee ball of radius $t$.  For the same $t$, the Lee ball lies inside the square, so correcting the square is the more demanding task, in accordance with the inequality $d_\square\le d_\diamond$ of \eqref{eq:twodistances}.  For $n$ qudits, the same statements hold with the Lee weights summed over qudits.}
    \label{fig:leemetrics}
\end{figure}

Every stabilizer code has a distance in each metric.  We write $d_\diamond$ for the minimum joint Lee weight of a nontrivial logical operator, and $d_\square$ for the minimum over nontrivial logical operators of $\max(|x|_{\mathrm{L}},|z|_{\mathrm{L}})$, as made precise in Definition \ref{def:stabilizer} below.  A code with $d_\diamond\ge 2t+1$ corrects every error of joint Lee weight at most $t$, while a code with $d_\square\ge 2t+1$ corrects every error whose $X$ and $Z$ parts each have Lee weight at most $t$.  Since $\max(|x|_{\mathrm{L}},|z|_{\mathrm{L}})\le |x|_{\mathrm{L}}+|z|_{\mathrm{L}}\le 2\max(|x|_{\mathrm{L}},|z|_{\mathrm{L}})$ for every Pauli string, each distance bounds the other:
\begin{equation}
    d_\square \le d_\diamond \le 2 d_\square. \label{eq:twodistances}
\end{equation}
For CSS codes, the two distances coincide and equal $\min(d_X,d_Z)$.  Indeed, if $X^xZ^z$ is a nontrivial logical operator of a CSS code, then $X^x$ or $Z^z$ is already a nontrivial logical operator.  The minimum in either metric is therefore attained by a logical operator of a single type.  For CSS codes, we write simply $d$ for the common value of the two distances, as in the rest of this paper, and reserve the subscripts for non-CSS codes.  Example \ref{eg:z10} below has $d_\diamond=3$ and $d_\square=2$, so both inequalities in \eqref{eq:twodistances} can be strict for non-CSS codes.

Relegating CSS codes to Section \ref{sec:css}, we discuss general Lee stabilizer codes here.  The stabilizer formalism on $\mathbb{Z}_q$ for composite $q$ is developed in Refs.~\cite{Gottesman1999higher, hostens2005stabilizer, Gheorghiu2014standard, SarkarYoder2024pauli, NadellaKlappenecker2012frobenius, GluesingLuerssenPllaha2019frobenius}, and we recall only the ingredients needed to define Lee stabilizer codes.

\begin{defn}[Quantum stabilizer code in the Lee metric] \label{def:stabilizer}
A linear $\qcode{n,k,d_\diamond}_q$ quantum stabilizer code on $\mathbb{Z}_q$ is specified
by a stabilizer matrix
\(
    H\in\mathbb{Z}_q^{m\times 2n}.
\)
Its $a$-th row is written as
\(
    h_a=(x_a\mid z_a)\in\mathbb{Z}_q^{2n}
\)
and lists the exponents of the stabilizer generator $X^{x_a}Z^{z_a}$, up to a phase.
Let
\begin{equation}
    J=
    \begin{pmatrix}
        0 & I_n \\
        -I_n & 0
    \end{pmatrix}.
\end{equation}
The stabilizer commutation condition is
\begin{equation}
    HJH^{\mathsf T}=0
    \quad(\mathrm{mod}\ q). \label{eq:stabcommute}
\end{equation}
The stabilizer group is $\mathcal{S}=\mathrm{im}(H^{\mathsf T})$, the logical group is
\begin{equation}
    \mathcal{L}=\ker(HJ)/\mathrm{im}(H^{\mathsf T}),
\end{equation}
and
\begin{equation}
    k=\frac{1}{2}\log_q|\mathcal{L}|=\log_q K,
    \qquad
    K=\frac{q^n}{|\mathcal{S}|}, \label{eq:stabk}
\end{equation}
where $K$ is the dimension of the codespace.  The joint Lee distance is
\begin{equation}
    d_\diamond=
    \min_{\ell\in\ker(HJ),\ \ell\notin\mathrm{im}(H^{\mathsf T})}
    |\ell|_{\mathrm{L}}, \label{eq:ddiamond}
\end{equation}
and the separate Lee distance is
\begin{equation}
    d_\square=
    \min_{\ell=(x\mid z)\in\ker(HJ),\ \ell\notin\mathrm{im}(H^{\mathsf T})}
    \max\left(|x|_{\mathrm{L}},|z|_{\mathrm{L}}\right). \label{eq:dsquare}
\end{equation}
\end{defn}

\noindent Given \eqref{eq:stabcommute}, the generators commute, and their phases can be chosen so that they generate an abelian group isomorphic to $\mathcal{S}$ that contains no nontrivial multiple of the identity.  The codespace is the common $+1$ eigenspace of this group.  Its dimension is $q^n/|\mathcal{S}|$~\cite{Gheorghiu2014standard} because every non-identity Pauli string is traceless, which justifies \eqref{eq:stabk}.   A Pauli error $\ell\in\mathbb{Z}_q^{2n}$ has syndrome $HJ\ell^{\mathsf T}$.  Since $J$ is invertible and $|\mathrm{im}(H)|=|\mathrm{im}(H^{\mathsf T})|$, the number of distinct syndrome values is $|\mathcal{S}|=q^n/K$.  Hence $|\ker(HJ)|=q^{2n}/|\mathcal{S}|$ and $|\mathcal{L}|=K^2$, which gives the first equality in \eqref{eq:stabk}.  The errors with zero syndrome form $\ker(HJ)$, so \eqref{eq:ddiamond} and \eqref{eq:dsquare} minimize over the undetectable errors that are not stabilizers, i.e., over the nontrivial logical operators.  From now on, the subscript on the bracket in $\qcode{n,k,d}_q$ is the qudit dimension, and $k=\log_q K$ need not be an integer.  For non-CSS codes, we write $d_\diamond$ or $d_\square$ in place of $d$ to indicate the metric.

The simplest error-correcting case is a joint Lee distance of three.  Such a code corrects the four errors $X$, $Z$, $X^{-1}$, and $Z^{-1}$ on any one qudit, along with the absence of an error, and it detects every error of joint Lee weight two.

\begin{eg}[Perfect single-qudit code on $\mathbb{Z}_{10}$] \label{eg:z10}
Let $q=10$ and $n=1$, and let $H=(2\mid 4)$, so that the stabilizer group is generated by $X^2Z^4$.  This generator has order five, so $|\mathcal{S}|=5$, $K=2$, and $k=\log_{10}2$.  The syndrome of the error $X^aZ^b$ is $2b-4a$ modulo $10$, which is always even.  The five errors $I$, $X$, $X^{-1}$, $Z$, and $Z^{-1}$ have syndromes $0$, $6$, $4$, $2$, and $8$, respectively.  These are distinct and exhaust the five available syndrome values, so the code corrects every error of joint Lee weight one and is perfect in the sense of Section \ref{subsec:perfect}.  Errors of weight two are not corrected: $X^2$ shares its syndrome with $Z$, and $XZ$ shares its syndrome with $Z^{-1}$.  Accordingly, $d_\diamond=3$, and the weight-three operator $XZ^2$ is a logical operator.  Its conjugate logical operator $Z^5$ has Lee weight five.  In the separate metric, $XZ^2$ has weight two, so $d_\square=2$.
\end{eg}

\subsection{Perfect codes} \label{subsec:perfect}

A stabilizer code detects and corrects errors by means of its error spaces.  Each syndrome value labels the error space onto which the errors with that syndrome map the codespace, and there are $q^n/K$ such values.  Two distinct correctable errors can map the codespace onto the same error space.  This happens only when they differ by a stabilizer, in which case the same recovery operation corrects both and the code is degenerate.  Conversely, some syndrome values may be left over after all errors of weight at most $t=\lfloor(d-1)/2\rfloor$ have been assigned, and such values can be used to correct some errors of weight above $t$.

Perfect codes are those for which neither happens, i.e., the number of correctable errors equals the number of syndrome values.  Counting syndrome values yields a sphere-packing bound.  Let $V_N(t)$ denote the number of vectors in $\mathbb{Z}_q^N$ of Lee weight at most $t$.  For $q\ge 2t+1$, distinct integer vectors in the ball remain distinct modulo $q$, and the ball contains
\begin{equation}
    V_N(t)=\sum_{j=0}^{\min(N,t)} 2^j\binom{N}{j}\binom{t}{j} \label{eq:leeball}
\end{equation}
vectors~\cite[Prop.~10.10]{Roth2006}.  The $j$-th term counts the vectors with exactly $j$ nonzero entries: there are $\binom{N}{j}$ choices for their positions, $2^j$ choices for their signs, and $\binom{t}{j}$ choices for their absolute values $w_1,\dots,w_j\ge1$ subject to $w_1+\dots+w_j\le t$.  There are $V_{2n}(t)$ Pauli errors of joint Lee weight at most $t$, and $V_n(t)^2$ errors whose $X$ and $Z$ parts each have Lee weight at most $t$.  If the errors in the relevant set have pairwise distinct syndromes, the corresponding bound is
\begin{subequations} \label{eq:spherepacking}
\begin{align}
    K\,V_{2n}(t) &\le q^n \qquad\text{(joint metric)}, \label{eq:spherepackingdiamond}\\
    K\,V_{n}(t)^2 &\le q^n \qquad\text{(separate metric)}. \label{eq:spherepackingsquare}
\end{align}
\end{subequations}
A code is perfect in the joint metric if it saturates \eqref{eq:spherepackingdiamond} with $t=\lfloor (d_\diamond-1)/2\rfloor$ and the errors of joint Lee weight at most $t$ have distinct syndromes.  Perfect codes in the separate metric are defined in the same way using \eqref{eq:spherepackingsquare} and $t=\lfloor (d_\square-1)/2\rfloor$, which for CSS codes is $t=\lfloor (d-1)/2\rfloor$.  Note also that if the errors of weight at most $t$ have distinct syndromes, then no nontrivial logical operator has weight at most $2t$, since such an operator would be a product of two distinct errors of weight at most $t$ with equal syndromes.  The bound \eqref{eq:spherepackingdiamond} is the Lee-metric analogue of the quantum Hamming bound~\cite{EkertMacchiavello1996}.  Its classical counterpart $|\mathcal{C}|\,V_n(t)\le q^n$ has the same form with the dimension of the ball halved, and classical codes saturating it are the perfect Lee codes of Golomb and Welch~\cite{GolombWelch1970}.

\begin{table}[t]
    \centering
    \caption{\textbf{Perfect quantum Lee codes.}  Each row is a stabilizer code, or a family of such codes, that saturates the sphere-packing bound \eqref{eq:spherepacking} in the indicated metric.  In each case, the correctable errors, namely those of joint Lee weight at most $t=\lfloor(d_\diamond-1)/2\rfloor$ in the joint metric, or those whose $X$ and $Z$ parts each have Lee weight at most $t=\lfloor(d-1)/2\rfloor$ in the separate metric, have distinct syndromes.  The upper block saturates the joint bound \eqref{eq:spherepackingdiamond}, and the two lower blocks consist of CSS codes with $H_X=H_Z$ that saturate the separate bound \eqref{eq:spherepackingsquare}.  Stabilizer generators are written as Pauli strings, and $K=q^k$ is the codespace dimension.  The listed ranges exclude the values $m=1$ and $K=1$, for which the families reduce to perfect stabilizer states.  Every entry with $q\ge4$ is perfect for all listed parameter values, as shown in Examples \ref{eg:z10}, \ref{eg:perfect643}, and \ref{eg:perfect3m} and in Section \ref{subsec:smallexamples}. The bottom block lists codes with $q\le3$, where Lee and Hamming weights agree.  Their checks $h$ are the rows of the parity-check matrix of a perfect classical Hamming or Golay code, which contains its dual for the listed parameters.  For $q=2$, where $V_n(1)=1+n$, the quantum Hamming codes saturate \eqref{eq:spherepackingsquare}, as Steane observed~\cite{Steane1996simple}.  They include the $\qcode{7,1,3}_2$ Steane code~\cite{Steane1996PRL} and the $\qcode{15,7,3}_2$ code~\cite{ChaoReichardt2018few}.  The first ternary member is a $\qcode{13,7,3}_3$ code, with $K\,V_{13}(1)^2=3^7\cdot27^2=3^{13}$.  The quantum Golay codes saturate \eqref{eq:spherepackingsquare} since $V_{23}(3)=2^{11}$ for $q=2$ and $V_{11}(2)=3^5$ for $q=3$.}
    \label{tab:perfect}
    \vspace{0.5em}
    \small
    \renewcommand{\arraystretch}{1.15}
    \begin{tabular}{l|l|l|l}
        \hline\hline
        Metric & Code & Stabilizer generators & Parameters \\
        \hline
        joint & $\qcode{1,\,\log_{10}2,\,3_\diamond}_{10}$ & $X^{2}Z^{4}$ & Example \ref{eg:z10} \\
        joint & $\qcode{2,\,2\log_{3m}m,\,3_\diamond}_{3m}$ & $X^{m}\otimes(XZ)^{-m}$,\ \ $(XZ)^{m}\otimes Z^{m}$ & $m\ge2$, Example \ref{eg:perfect3m} \\
        joint & $\qcode{6,\,4,\,3_\diamond}_{5}$ & Eq.~\eqref{eq:perfect643} & Example \ref{eg:perfect643} \\
        \hline
        separate & $\qcode{1,\,\log_{q}K,\,(2t+1)}_{q}$ & $X^{K(2t+1)}$,\ \ $Z^{K(2t+1)}$ & $q=K(2t+1)^2$, $K\ge2$, Ref.~\cite{Gottesman2001GKP} \\
        separate & $\qcode{2,\,2\log_{q}m,\,(2t+1)}_{q}$ & $X^{mt}\otimes X^{m(t+1)}$,\ \ $Z^{mt}\otimes Z^{m(t+1)}$ & $q=m(2t^2+2t+1)$, $m\ge2$ \\
        separate & $\qcode{n,\,n-2,\,3}_{2n+1}$ & $X\otimes X^{2}\otimes\cdots\otimes X^{n}$,\ \ $Z\otimes Z^{2}\otimes\cdots\otimes Z^{n}$ & $n\ge3$, $n\not\equiv1 \pmod 3$ \\
        \hline
        separate & $\qcode{n,\,n-2r,\,3}_{2}$ & $X^{h}$,\ \ $Z^{h}$, $h$ a binary Hamming check & $n=2^r-1$, $r\ge3$, Ref.~\cite{Steane1996simple} \\
        separate & $\qcode{n,\,n-2r,\,3}_{3}$ & $X^{h}$,\ \ $Z^{h}$, $h$ a ternary Hamming check & $n=(3^r-1)/2$, $r\ge3$ \\
        separate & $\qcode{23,\,1,\,7}_{2}$ & $X^{h}$,\ \ $Z^{h}$, $h$ a binary Golay check & Section \ref{subsec:smallexamples} \\
        separate & $\qcode{11,\,1,\,5}_{3}$ & $X^{h}$,\ \ $Z^{h}$, $h$ a ternary Golay check & Ref.~\cite{Prakash2020golay} \\
        \hline\hline
    \end{tabular}
\end{table}

Table \ref{tab:perfect} lists all perfect Lee codes with $q\ge4$ that we have found, in the joint as well as in the separate metric, together with known perfect CSS codes with $q\le3$.  Example \ref{eg:z10} is perfect, and so are Examples \ref{eg:perfect643} and \ref{eg:perfect3m}.  The CSS entries with $q\ge4$ are constructed in Section \ref{subsec:smallexamples}.

\begin{eg}[Perfect six-qudit code on $\mathbb{Z}_5$] \label{eg:perfect643}
The $\qcode{6,4,3_\diamond}_5$ code is generated by the two stabilizers
\begin{equation}
    \begin{aligned}
        S_1 &= Z\otimes Z\otimes XZ\otimes XZ^{2}\otimes X^{2}Z^{2}\otimes X^{-2}Z^{-2},\\
        S_2 &= X\otimes X^{2}Z\otimes X^{2}Z^{-2}\otimes X^{-1}\otimes XZ^{2}\otimes XZ^{2},
    \end{aligned}
    \label{eq:perfect643}
\end{equation}
i.e., by the rows $h_1=(0,0,1,1,2,-2\mid 1,1,1,2,2,-2)$ and $h_2=(1,2,2,-1,1,1\mid 0,1,-2,0,2,2)$ of $H$.  The symplectic product $h_1Jh_2^{\mathsf T}=-5$ vanishes modulo five, so the generators commute.  Each has order five, so $|\mathcal{S}|=25$, $K=5^6/25=5^4$, and $k=4$.  The Lee ball of radius one in $\mathbb{Z}_5^{12}$ contains $V_{12}(1)=25$ errors, one for each syndrome value, and their syndromes are distinct.  Thus, the code is perfect with $d_\diamond=3$.  It corrects any single $X^{\pm1}$ or $Z^{\pm1}$ error on any of its six qudits while encoding four qudits.
\end{eg}

\begin{eg}[Perfect two-qudit family on $\mathbb{Z}_{3m}$] \label{eg:perfect3m}
Let $q=3m$ with $m\ge 2$.  The two stabilizers (with phases suppressed)
\begin{equation}
    S_1 = X^{m}\otimes (XZ)^{-m}, \qquad S_2 = (XZ)^{m}\otimes Z^{m}, \label{eq:perfect3m}
\end{equation}
with rows $h_1=(m,-m\mid 0,-m)$ and $h_2=(m,0\mid m,m)$, commute over the integers, and each has order three.  Thus, $\mathcal{S}\cong\mathbb{Z}_3^2$ for every $m$, and $K=(3m)^2/9=m^2$.  The syndrome of the error $(a\mid b)$ is
\begin{equation}
    HJ(a\mid b)^{\mathsf T} = m\,\big(b_1-b_2+a_2,\ b_1-a_1-a_2\big) \quad (\mathrm{mod}\ 3m),
\end{equation}
which depends only on the error modulo three.  The nine errors of joint Lee weight at most one map onto the nine possible values in $m\mathbb{Z}_{3m}^2\cong\mathbb{Z}_3^2$, so the code is perfect with $d_\diamond=3$.  The undetectable error $X_1^3$ is not a stabilizer for any $m\ge2$, and it is a logical operator of weight three.  For $m=2$, the code is a $\qcode{2,\,2\log_6 2,\,3_\diamond}_6$ code that stores a four-dimensional logical space in two six-dimensional qudits.
\end{eg}

\subsection{Relation to Hamming distance} \label{subsec:hamming}

Lee metrics introduce a new parameter, the qudit dimension $q$.  The corresponding Lee distances can thus grow with $q$ in addition to their more conventional growth with the number of qudits $n$.  In the Hamming metric, the distance $d_{\mathrm{H}}$ of a code, i.e., the minimum number of qudits on which a nontrivial logical operator acts, can grow at most linearly with $n$.  Lee distances are bounded in terms of both the Hamming distance and $q$:
\begin{subequations}\label{eq:leevshamming}
\begin{align}
    d_{\mathrm{H}} &\le d_\diamond \le 2\lfloor q/2\rfloor\, d_{\mathrm{H}}, \label{eq:leevshammingdiamond}\\
    \lceil d_{\mathrm{H}}/2\rceil &\le d_\square \le \lfloor q/2\rfloor\, d_{\mathrm{H}}. \label{eq:leevshammingsquare}
\end{align}
\end{subequations}
To see this, consider a Pauli string $X^xZ^z$ acting nontrivially on $w$ qudits.  Each of these qudits contributes between one and $2\lfloor q/2\rfloor$ to $|x|_{\mathrm{L}}+|z|_{\mathrm{L}}$, and at most $\lfloor q/2\rfloor$ to each of $|x|_{\mathrm{L}}$ and $|z|_{\mathrm{L}}$.  Therefore, $w\le |x|_{\mathrm{L}}+|z|_{\mathrm{L}}\le 2\lfloor q/2\rfloor w$ and $w/2\le\max(|x|_{\mathrm{L}},|z|_{\mathrm{L}})\le\lfloor q/2\rfloor w$, and minimizing over nontrivial logical operators yields \eqref{eq:leevshamming}.  The factor of two in \eqref{eq:leevshammingdiamond} is already necessary for qubits, where the $Y$-parity code of Section \ref{subsec:qubits} has $d_\diamond=2d_{\mathrm{H}}$.  The lower bound in \eqref{eq:leevshammingsquare} is attained by the five-qubit code, which has $d_{\mathrm{H}}=d_\diamond=3$ but $d_\square=2$.  Indeed, in the presentation with stabilizer generator $XZZXI$ and its cyclic shifts, $IZXZI$ is a logical operator.  

For CSS codes, the Hamming distance is the smaller of the minimum Hamming weights of nontrivial $X$-type and $Z$-type logical operators, and the bounds on the Lee distance $d=\min(d_X,d_Z)$ tighten to
\begin{equation}
    d_{\mathrm{H}}\le d\le\lfloor q/2\rfloor\, d_{\mathrm{H}}, \label{eq:leevshammingcss}
\end{equation}
because $|x|_{\mathrm{H}}\le|x|_{\mathrm{L}}\le\lfloor q/2\rfloor|x|_{\mathrm{H}}$ for every $x\in\mathbb{Z}_q^n$.

Lee distances can indeed grow without bound as functions of $q$ at fixed $n$.  The one-qudit and two-qudit CSS families of Table \ref{tab:perfect} are examples: their distance $2t+1$ grows as $\sqrt{q}$ at fixed $K$ or $m$, while their Hamming distance is one.  More generally, the relative distance $d/n$, which never exceeds one in the Hamming metric, can become unbounded for families of codes in which both $n$ and $q$ grow.  In Section \ref{sec:LDPC}, we use the geometry of numbers to determine how fast Lee distances can grow for Lee-LDPC codes: at least one of $d_X$ and $d_Z$ stays at most linear in $n$, regardless of $q$.

\section{CSS codes} \label{sec:css}
Like in the case of ordinary qubit codes ($q=2$), it is often helpful to design codes that have a special CSS structure~\cite{CalderbankShor1996,Steane1996PRL,Steane1996multiple}, where each stabilizer generator contains only powers of $X$ or only powers of $Z$.  This section further develops the theory of parity-check codes in this setting,  which will be the focus of the rest of the paper.

We proceed in four steps. Section~\ref{subsec:classicalcodes} reviews how the Smith normal form
determines a classical codespace over $\mathbb Z_q$. Section~\ref{subsec:quantumcss} extends this
description to quantum CSS codes by replacing the modular quotient with an
integer chain complex. Section~\ref{subsec:smallexamples} gives small examples, including perfect and lifted codes, and Section~\ref{subsec:gkp} realizes every $\mathbb Z_q$ stabilizer code as a multimode GKP code.

\subsection{Classical codes on $\mathbb{Z}_q$} \label{subsec:classicalcodes}
We are interested in classical linear codes over the alphabet $\mathbb{Z}_q$, for general $q$.   A linear code over $n$ variables is the submodule $\mathcal{C}=\ker(H)\subseteq\mathbb{Z}_q^n$ of right null vectors of a parity-check matrix $H$.

\begin{defn}[Classical linear code on $\mathbb{Z}_q$] \label{def : Classical_linear_Code}
    A parity-check matrix $H\in \mathbb{Z}_q^{m\times n}$ defines an $[n,k,d]_q$ code, with codespace $\mathcal{C}=\ker(H)$.   Define $ k := \log_q |\mathcal{C}|$  as the ``number" of logical variables.  $k$ need not be an integer.  The code distance $d$ is the minimum Lee weight of a nonzero codeword: \begin{equation}
        d := \min_{x\in \mathcal{C}\setminus0} |x|_{\mathrm{L}}. \label{eq:classical_d}
    \end{equation}
\end{defn}

We can easily understand the codespace $\mathcal{C}=\ker(H)$ as follows.  Without loss of generality, we may first interpret $H$ as a $\mathbb{Z}$-valued matrix.  The Smith normal form (SNF) will be an extremely useful tool here and in the rest of this paper.  By performing this Smith decomposition~\cite{stanley2016smith} \begin{equation}
        H = UDV^{-1},
    \end{equation}
    where $U\in \mathrm{GL}(m,\mathbb{Z})$ and $V\in\mathrm{GL}(n,\mathbb{Z})$ are unimodular matrices (determinant $\pm 1$, so invertible over the integers), while $D$ is diagonal and takes the form \begin{equation}
        D = \left(\begin{array}{ccccccc} d_1 &\ 0 &\ \cdots &\ 0 &\ 0 &\ \cdots &\ 0\\ 0 &\ d_2 &\ \cdots  &\ 0 &\ 0 &\ \cdots  &\ 0 \\ \vdots &\ \vdots &\ \ddots  &\ \vdots &\ \vdots &\ \cdots  &\ \vdots \\ 0 &\ 0 &\ \cdots  &\ d_r &\ 0 &\ \cdots  &\ 0 \\  0 &\ 0 &\ \cdots  &\ 0 &\ 0 &\ \cdots  &\ 0 \\ \vdots &\ \vdots &\ \cdots  &\ \vdots &\ \vdots &\ \cdots  &\ \vdots \\ 0 &\ 0 &\ \cdots  &\ 0 &\ 0 &\ \cdots  &\ 0 \\ \end{array}\right),
    \end{equation}
where each $d_{i-1}$ is a divisor of $d_{i}$ (for $i\le r$), we can easily deduce the form of the codespace: since $U$ is invertible, $x\in\mathcal{C}$, i.e., $UDV^{-1}x=0$, if and only if $DV^{-1}x=0$.  Writing $x=Vy$, we deduce that the components $y_i$ obey $d_i y_i = 0 \text{ (mod $q$)}$.  For $i>r$, clearly any $y_i\in\mathbb{Z}_q$ is possible; for later convenience we define $\widehat{d}_i=1$ for these components.   For $i\le r$, letting $\widehat{d}_i = q/\gcd(q,d_i)$, we see that $y_i = c_i \widehat{d}_i$ are allowed, and on $\mathbb{Z}_q$, the choices $c_i=0,\ldots,\gcd(q,d_i)-1$ give distinct solutions.  Putting this together, we conclude that the codespace always takes the following form over $\mathbb{Z}_q$~\cite{Park2009modular,Dougherty2017rings}:
\begin{prop}[Codespace of classical codes]\label{prop:classicalcodespace}
    Let $\{e_i\}_{i=1}^n$ be the canonical basis of $\mathbb{Z}_q^n$.  Given the Smith factors $d_i$ of the parity-check matrix $H$, the codespace is \begin{equation}
        \mathcal{C} = \mathrm{span}\left\lbrace \widehat{d}_i \cdot Ve_i \right\rbrace_{i=1}^n \cong \mathbb{Z}_q^{n-r}\oplus \bigoplus_{j=1}^r \mathbb{Z}_{\gcd(q,d_j)}  . \label{eq:classicalcodespace}
    \end{equation}
\end{prop}

If $\mathcal{C}$ contains any factors $\mathbb{Z}_{\gcd(q,d_j)}$ with $\gcd(q,d_j)\notin \lbrace 1,q\rbrace$, then we say that there is torsion. 
  Torsion complicates the analysis of the codespace and, especially in the quantum context, is usually neglected as unwanted.  We will carefully consider classical and quantum (Theorem \ref{thm:quantumcodespace}) codes with torsion in this paper. 

In classical coding theory, it does make mathematical sense to define a classical code on $\mathbb{Z}$, but for quantum codes on finite-dimensional qudits we work with $\mathbb{Z}_q$, using integer lifts only as an auxiliary tool (Theorem~\ref{thm:quantumcodespace}).  For a classical code on $\mathbb{Z}$, one modifies \eqref{eq:classicalcodespace} by replacing $\mathbb{Z}_q^{n-r}$ with $\mathbb{Z}^{n-r}$ and discarding all finite factors $\mathbb{Z}_{\gcd(q,d_j)}$: $\mathcal{C}\cong\mathbb{Z}^{n-r}$.  

In the special case when $q$ is a prime number, we have $\gcd(d_i,q)\in \lbrace 1,q\rbrace$.  We are therefore led to the standard result: \begin{prop}
    For $q$ prime, the codespace is $\mathcal{C}\cong\mathbb{Z}_q^{n-r+s}$, where $s$ is the number of Smith factors $d_i$ that are multiples of $q$.  In particular, $\mathcal{C}$ is a vector space with integer dimension, so the number of logical dits $k$ is an integer.\footnote{This is a consequence of the fact that $\mathbb{Z}_q$ is a finite field, not just a ring.  In the coding literature it is thus usually denoted as $\mathbb{F}_q$.}
\end{prop}

\subsection{Quantum CSS codes} \label{subsec:quantumcss}
Now we turn to quantum CSS codes on $\mathbb{Z}_q$, mirroring closely the classical discussion.  Earlier literature~\cite{amini2003quantum,GuendaGulliver2014rings} has also discussed the generalization of quantum CSS codes to $\mathbb{Z}_q$.  We will be general and systematic, first discussing arbitrary CSS codes under the Lee metric.  Later, Section \ref{sec:LDPC} specializes to the Lee-LDPC setting, where more nontrivial results can be found.

\begin{defn}[Quantum CSS code in the Lee metric]
    A linear $\qcode{n,k,d}_q$ quantum CSS code on $\mathbb{Z}_q$ is given by matrices $H_X\in \mathbb{Z}_q^{m_X\times n}$ and $H_Z\in\mathbb{Z}_q^{m_Z\times n}$ obeying \begin{equation}
        H_XH_Z^{\mathsf{T}} = 0\ \text{(mod $q$)}.
    \end{equation}
    $k$ is defined analogously to Definition~\ref{def : Classical_linear_Code}: \begin{equation}
        k = \log_q \frac{|\ker(H_X)|}{|\mathrm{im}(H_Z^{\mathsf{T}})|}= \log_q \frac{|\ker(H_Z)|}{|\mathrm{im}(H_X^{\mathsf{T}})|}. \label{eq:quantumk}
    \end{equation}
    The Lee distances $d_X$ and $d_Z$ are defined analogously to \eqref{eq:classical_d}: \begin{subequations}\begin{align}
        d_X &= \min_{x\in \ker(H_Z), \;  x\notin \mathrm{im}(H_X^{\mathsf{T}})}|x|_{\mathrm{L}},\\
        d_Z &= \min_{z\in \ker(H_X), \;  z\notin \mathrm{im}(H_Z^{\mathsf{T}})}|z|_{\mathrm{L}}.
    \end{align}\end{subequations}
    The Lee distance of the code is $d=\min(d_X,d_Z)$.
\end{defn}

For $x\in\ker(H_Z)$, its reduced Lee weight is $\min_{s\in\mathbb Z_q^{m_X}}|x+H_X^{\mathsf T}s|_{\mathrm L}$, the smallest Lee weight in its stabilizer coset. The reduced Lee weight of $z\in\ker(H_X)$ is defined analogously, with $H_Z^{\mathsf T}s$, $s\in\mathbb Z_q^{m_Z}$, in place of $H_X^{\mathsf T}s$.

It is familiar from the theory of quantum codes on fields (applicable when $q$ is prime) that the two formulas for $k$ in \eqref{eq:quantumk} are equivalent, but it may be less obvious when $q$ is composite and $k$ is not an integer.  Of course, on physical grounds, we know that $k$ must be well-defined, as $q^k$ is the dimension of the codespace, which is a well-defined Hilbert space.   A simple way to see it purely from homology on $\mathbb{Z}_q$ is to notice that for both $H_X$ and $H_Z$, \begin{equation}
    q^n = |\ker(H_X)|\,|\mathrm{im}(H_X^{\mathsf{T}})| = |\ker(H_Z)|\,|\mathrm{im}(H_Z^{\mathsf{T}})|.
\end{equation}
This fact follows from linear algebra; it can be readily derived by generalizing the discussion of the Smith decomposition (Proposition \ref{prop:classicalcodespace}) to count the size of $\mathrm{im}(H_X^{\mathsf{T}})$ as well.

To see this explicitly, let $d_1,\ldots,d_r$ be the nonzero Smith factors of an
$m\times n$ matrix $H$. Proposition~\ref{prop:classicalcodespace} gives
\begin{equation}
    |\ker(H)|
    =
    q^{n-r}\prod_{i=1}^r \gcd(q,d_i).
\end{equation}
Since $H^{\mathsf T}$ has the same nonzero Smith factors as $H$, we also have
\begin{equation}
    |\operatorname{im}(H^{\mathsf T})|
    =
    \prod_{i=1}^r \frac{q}{\gcd(q,d_i)}.
\end{equation}
Consequently, $|\ker(H)|\,|\operatorname{im}(H^{\mathsf T})|=q^n$. Applying
this identity to $H_X$ and $H_Z$ shows that the two expressions in \eqref{eq:quantumk} agree.

Throughout this paper, we assign energies to configurations by the Lee weight of the measured syndrome.
For a classical code with parity-check matrix $H$, the energy of a configuration
$x$ is the Lee weight of its syndrome,
$E_{\mathrm{cl}}(x)=|Hx|_{\mathrm{L}}$. For a quantum CSS code, an $X$-type error $x$ has
energy $E_X(x)=|H_Zx|_{\mathrm{L}}$, while a $Z$-type error $z$ has energy
$E_Z(z)=|H_Xz|_{\mathrm{L}}$. 

The next step in our analysis is to describe the codespace of a general quantum CSS code in detail.  We summarize the essential points in the (lengthy) statement of the main theorem below, and provide a proof in Appendix \ref{app:codespace}, using standard tools from homological algebra. Sarkar and Yoder found bases of logical operators for stabilizer codes on qudits of any dimension using normal forms over $\mathbb{Z}_q$~\cite{SarkarYoder2024pauli}.  For CSS codes, the theorem below obtains such bases from a single Smith normal form over $\mathbb{Z}$.

\begin{thm}[Codespace of a quantum CSS code on $\mathbb{Z}_q$]\label{thm:quantumcodespace}
Consider the quantum code defined by $H_X : \mathbb{Z}_q^n \rightarrow \mathbb{Z}_q^{m_X}$ and $H_Z: \mathbb{Z}_q^n \rightarrow \mathbb{Z}_q^{m_Z}$ obeying $H_XH_Z^{\mathsf{T}}=0$ on $\mathbb{Z}_q$.  Regard $H_X$ and $H_Z$ as integer matrices by choosing integer lifts of their entries. By the isomorphisms below, the result does not depend on this choice.  Define the abstract chain complex over $\mathbb{Z}$ given by three matrices $D_{1,2,3}$ obeying \begin{equation}
    D_1D_2 = D_2D_3=0,
\end{equation}
where the integer-valued matrices $D_{1,2,3}$ are:\footnote{Because $H_XH_Z^{\mathsf{T}}=0\; \text{(mod $q$)}$, the matrix $q^{-1}H_XH_Z^{\mathsf{T}}$ has integer entries.} \begin{subequations}
    \begin{align}
        D_1 &= \left(\begin{array}{cc} H_X &\ - qI_{m_X} \end{array}\right), \label{eq:D1def} \\
        D_2 &= \left(\begin{array}{cc} H_Z^{\mathsf{T}} &\ qI_{n} \\ q^{-1}H_XH_Z^{\mathsf{T}} &\ H_X \end{array}\right), \label{eq:D2def} \\
        D_3 &= \left(\begin{array}{c} qI_{m_Z} \\ -H_Z^{\mathsf{T}} \end{array}\right).
    \end{align}
\end{subequations}
The codespace of the quantum code is \begin{align}
    \mathcal{C} &:= \ker_{\mathbb{Z}_q} (H_X)/\mathrm{im}_{\mathbb{Z}_q}(H_Z^{\mathsf{T}}) \cong \ker_{\mathbb{Z}}(D_1)/\mathrm{im}_{\mathbb{Z}}(D_2)  \notag \\
    &\cong \ker_{\mathbb{Z}_q} (H_Z)/\mathrm{im}_{\mathbb{Z}_q}(H_X^{\mathsf{T}}) \cong \ker_{\mathbb{Z}}(D_3^{\mathsf{T}})/\mathrm{im}_{\mathbb{Z}}(D_2^{\mathsf{T}}).
\end{align}
Following Proposition~\ref{prop:classicalcodespace}, we call $\mathcal{C}$ the codespace. Its order $|\mathcal{C}|=q^k$ is the dimension of the code's Hilbert space. It should not be confused with the logical group $\mathcal{C}$ of Definition~\ref{def:stabilizer}, whose order is $q^{2k}$.
Moreover, the Smith decomposition of $D_2$ takes the form \begin{equation} \label{eq:D2SNF}
    D_2 = \widetilde{U}\left(\begin{array}{cc} \widetilde{D}_{n\times n} &\ 0_{n\times m_Z} \\ 0_{m_X\times n} &\ 0 _{m_X\times m_Z} \end{array}\right)\widetilde{V},
\end{equation}
where the subscripts give the dimensions of each block of the middle matrix, and $\widetilde{D} = \mathrm{diag}(d_1,\ldots, d_n)$ has nonzero factors $1 \le d_1,\ldots, d_n\le q$, with each $d_j \; | \; q$. The codespace is \begin{equation}
    \mathcal{C} \cong \bigoplus_{i=1}^n \mathbb{Z}_{d_i}. \label{eq:quantumcodespace}
\end{equation}
Notice that any factor with $d_i=1$ does not contribute meaningfully to the codespace. For any $d_j>1$ in the SNF $\widetilde{D}$,  a generator for the corresponding $Z$-logical operator is given by \begin{equation}
    \left( z_{\mathrm{L}}^{(j)}\right)_a = \widetilde{U}_{a,j} \; \text{(mod $q$)} \;\;\; (a=1,\ldots,n), \label{eq:zLogical}
\end{equation}
where $\widetilde{U}_{a,j}$ are the appropriate matrix elements of $\widetilde{U}\in \mathbb{Z}^{(n+m_X) \times (n+m_X)}$. The conjugate $X$-logical operator is \begin{equation}
    \left( x_{\mathrm{L}}^{(j)}\right)_a = \widetilde{V}_{j,m_Z+a} \; \text{(mod $q$)} \;\;\; (a=1,\ldots,n). \label{eq:xLogical}
\end{equation}
As vectors in $\mathbb{Z}_q^n$, these logical operators generate a logical $\mathbb{Z}_{d_j}$-qudit: on $\mathbb{Z}_q$, \begin{equation}
    x_{\mathrm{L}}^{(j)} \cdot z_{\mathrm{L}}^{(k)} = \frac{q}{d_j} \delta_{jk}. \label{eq:logicalConjugates}
\end{equation}
\end{thm}

Note that when $q$ is prime, each $d_i$ is either 1 or $q$, so the codespace is isomorphic to $\mathbb Z_q^k$.

\subsection{Small examples} \label{subsec:smallexamples}

The construction in Theorem~\ref{thm:quantumcodespace} is abstract because it packages a modular CSS quotient into an integer chain complex. The following one-qudit example makes all of its ingredients explicit and shows how torsion, a factor $\mathbb{Z}_{c'}$ with $1<c'<q$, can arise when $q$ is composite.

\begin{eg}[Discrete GKP/torus code] \label{eg:toruscode}
    Consider the case $m_X=m_Z=n=1$ for positive integers $a,b,c$ with $c\mid ab$ and $q=ab/c\ge2$.    Now consider $H_X = a$ and $H_Z =  b$.   $H_XH_Z^{\mathsf{T}} = ab = 0 \; (\text{mod $q$})$, so this is a valid CSS code.  We can explicitly carry out the procedure in Theorem \ref{thm:quantumcodespace}: the matrix \begin{equation}
        D_2 = \left(\begin{array}{cc} b &\ q \\ c &\ a \end{array}\right) = \left(\begin{array}{cc} q/a^\prime &\ a_1 \\ a/a^\prime &\ a_2 \end{array}\right)\left(\begin{array}{cc} c^\prime &\ 0 \\ 0 &\ 0 \end{array}\right)\left(\begin{array}{cc} ba^\prime/qc^\prime &\ a^\prime/c^\prime \\ b_1 &\ b_2 \end{array}\right),
    \end{equation}
    where $c^\prime = \gcd(a,b,c,q)$, $a^\prime = \gcd(a,q)$, $b^\prime = \gcd(b,q)$, and $a_{1,2}$ and $b_{1,2}$ are B\'ezout coefficients chosen to ensure that $\widetilde{U}$ and $\widetilde{V}$ are unimodular.  We deduce that the codespace is $\mathcal{C}\cong\mathbb{Z}_{c^\prime}$ and that the logical operators correspond to $z_{\mathrm{L}} = q/a^\prime$ and $x_{\mathrm{L}} = a^\prime/c^\prime$.  A short exercise 
shows that $a^\prime/c^\prime = q/b^\prime$.

The logical operators can also be recovered directly from the stabilizer
conditions. A logical $X$ operator lies in $\ker(H_Z)$, which is generated by
$q/b'$, and hence one may take $x_{\mathrm{L}}=q/b'$. Since $q/b'=a'/c'$, this agrees with
the logical operator obtained from Theorem \ref{thm:quantumcodespace}.

This example is the finite-dimensional analogue of the ideal
single-mode GKP code, whose lattice is square when $a'=b'$ and which encodes a qubit when $c'=2$. For a bosonic mode with canonical quadratures
$[\hat q,\hat p]=\ii$, the displacement operator
\begin{equation} \label{eq:displacementconvention}
    D(\Delta q,\Delta p)
    =
    \exp\!\left[\ii(\Delta p\,\hat q-\Delta q\,\hat p)\right]
\end{equation}
translates a state by $(\Delta q,\Delta p)$ in phase space. The square GKP code
is the simultaneous $+1$ eigenspace of the commuting displacements
$D(2\sqrt{\pi},0)$ and $D(0,2\sqrt{\pi})$, while the half-lattice displacements
$D(\sqrt{\pi},0)$ and $D(0,\sqrt{\pi})$ act as logical $X$ and logical $Z$,
respectively. Here the continuous phase plane is replaced by the discrete torus
$\mathbb Z_q\times \mathbb Z_q$. The stabilizers $X^a$ and $Z^b$ generate the
two discrete lattice translations, while the logical operators $Z^{q/a'}$ and
$X^{a'/c'}=X^{q/b'}$ generate the residual finite quotient. Small Pauli shifts
correspond to short translations on the discrete torus, whose cyclic displacement
is measured by the Lee metric. The corresponding exact single-mode GKP realization is given in Section \ref{subsec:gkp} after the general CSS/GKP correspondence is established.

    In the notation of Section \ref{sec:stabilizer}, this is a $\qcode{1,\log_q c',d}_q$ code.  The $X$-type logical operators are the powers of $X^{q/b'}$ that are not powers of $X^{a'}$, and likewise for $Z$, so $d_X=q/b'$ and $d_Z=q/a'$ whenever $c'>1$, and $d=\min(q/a',q/b')$.  The syndromes of $X$-type errors take $q/b'$ values, and those of $Z$-type errors take $q/a'$ values.  The code is therefore perfect in the sense of Section \ref{subsec:perfect} exactly when $a'=b'$ and $q/a'=2t+1$ is odd.  In that case, the $2t+1$ errors of each type with Lee weight at most $t$ exhaust the syndrome values of that type.  In particular, the choice $a=b=K(2t+1)$ and $q=K(2t+1)^2$ gives $a'=b'=K(2t+1)$ and $c'=K$, and yields the perfect $\qcode{1,\log_qK,(2t+1)}_q$ codes of Table \ref{tab:perfect}. The codes of this example are the shift-resistant codes of Ref.~\cite{Gottesman2001GKP}, which already noted that some of them are perfect. Members of this family that encode a qubit were studied in Refs.~\cite{Pirandola2008minimal,Cafaro2012embedding}.
\end{eg}

CSS codes also allow us to borrow from the classical perfect Lee codes of Golomb and Welch~\cite{GolombWelch1970} to construct perfect quantum codes in the separate Lee metric.  Both of the following families stem from that work.  The first uses the two-dimensional Lee-ball tilings, which exist for every radius, and the second uses the radius-one tilings, which exist in every dimension.

\begin{eg}[Perfect two-qudit CSS codes] \label{eg:perfectcss2}
    Let $q=m(2t^2+2t+1)$ with $m\ge2$ and $t\ge1$, and let $H_X=H_Z=(mt,\ m(t+1))$, so that the checks are $X^{mt}\otimes X^{m(t+1)}$ and $Z^{mt}\otimes Z^{m(t+1)}$.  The CSS condition holds because $m^2t^2+m^2(t+1)^2=mq$.  Each check has order $2t^2+2t+1$, so the logical dimension is $K=q^2/(2t^2+2t+1)^2=m^2$ and $k=2\log_q m$.  The syndrome of an $X$-error $a\in\mathbb{Z}_q^2$ is $m\varphi(a)$, where $\varphi(a)=ta_1+(t+1)a_2$ is taken modulo $2t^2+2t+1$.  The kernel of $\varphi$ on $\mathbb{Z}^2$ is the lattice generated by the orthogonal vectors $(t,t+1)$ and $(-(t+1),t)$, whose nonzero elements have $1$-norm at least $2t+1$.  Thus, $\varphi$ is injective on the Lee ball of radius $t$, which has exactly as many elements as there are syndrome values.  This is the two-dimensional Lee-ball tiling of Ref.~\cite{GolombWelch1970}, and the code is a perfect $\qcode{2,2\log_qm,(2t+1)}_q$ code.  The vector $(t+1,-t)$ lies in the kernel of $\varphi$ and is not a stabilizer for $m\ge2$, so $d_X=d_Z=2t+1$ exactly.

    The member with $t=1$ and $m=2$ is a $\qcode{2,2\log_{10}2,3}_{10}$ code with checks $X^2\otimes X^4$ and $Z^2\otimes Z^4$.  Its Hamming distance is one, since $X_1^5$ commutes with both checks and is not a stabilizer.  This operator squares to the identity and is therefore Hermitian, so $\exp(\ii\theta X_1^5)=\cos\theta+\ii\sin\theta\,X_1^5$ is a transversal unitary that preserves the codespace and acts on it as a nontrivial logical gate for generic $\theta$. The code thus admits a continuous family of transversal logical gates, with the Lee weight of the generator \(X_1^5\) being at least the Lee distance by definition.  The Eastin--Knill theorem~\cite{EastinKnill2009,Faist2020continuous} forbids such a family for any code that detects every single-qudit error in the Hamming metric. Its conclusion does not extend to codes required to detect only small Lee-weight errors: a code can correct every error of Lee weight one, as this one does, while having Hamming distance one.
\end{eg}

\begin{eg}[Perfect $n$-qudit CSS codes] \label{eg:perfectcssn}
    Let $q=2n+1$ with $n\ge3$ and $n\not\equiv1\pmod 3$, and let $H_X=H_Z=(1,2,\dots,n)$, so that the checks are $X\otimes X^2\otimes\cdots\otimes X^n$ and $Z\otimes Z^2\otimes\cdots\otimes Z^n$.  The CSS condition requires $\sum_{j=1}^nj^2=n(n+1)(2n+1)/6$ to vanish modulo $2n+1$, which holds if and only if $6$ divides $n(n+1)$, i.e., exactly for these $n$.  Each check has order $q$, so $K=q^{n-2}$ and $k=n-2$.  The syndromes of the $X$-errors of Lee weight at most one, of which there are $2n+1$, are $0$ and $\pm j$ for $j=1,\dots,n$, which are distinct modulo $2n+1$ and exhaust $\mathbb{Z}_q$.  The code is thus a perfect $\qcode{n,n-2,3}_{2n+1}$ code, and its classical ingredient is the radius-one perfect Lee code of Ref.~\cite{GolombWelch1970}.  Moreover, $X_1^2X_2^{-1}$ is a logical operator of Lee weight three that is not a stabilizer for $n\ge3$, so $d_X=d_Z=3$ exactly.  The first members are the $\qcode{3,1,3}_7$, $\qcode{5,3,3}_{11}$, and $\qcode{6,4,3}_{13}$ codes.
\end{eg}

Other well-known classical Lee codes are the self-dual codes over $\mathbb{Z}_4$, such as the octacode~\cite{ConwaySloane1993,Hammons1994} and the $\mathbb{Z}_4$ Golay code~\cite{Bonnecaze1995}.  The former is the quaternary analogue of the $[8,4,4]$ extended Hamming code, and the latter of the $[24,12,8]$ extended Golay code.  Each is obtained by Hensel lifting the generator polynomial of the binary cyclic code from $\mathbb{Z}_2$ to $\mathbb{Z}_4$ and extending the resulting code by one coordinate.  More generally, both families extend to every $q$ that is a power of two, as projections of the 2-adic Hamming and Golay codes of Calderbank and Sloane~\cite{CalderbankSloane1995}.  The extended codes are self-dual for every such $q$, and their Lee distances grow with $q$.  The Hamming family has Lee distances $4,6,8,12,14,18$ for $q=2,4,8,16,32,64$, respectively~\cite{CalderbankSloane1995}, and we find that the Golay family has Lee distance $8,12,14,24$ for $q=2,4,8,16$.

Shortening such a self-dual code at one coordinate yields a code that is contained in its dual, namely the punctured code.  The CSS construction with $H_X=H_Z$ built from the shortened code therefore yields a self-dual CSS code on one fewer qudit~\cite{JainAlbert2025transversal}.  Its codespace dimension is $K=\mu^2/q$, where $\mu$ is the number of values that codewords take at the shortened coordinate.  In particular, the code encodes a single qudit ($K=q$) whenever some codeword takes a unit value at that coordinate, as in the two families below.  Its distance is the minimum Lee weight of the punctured codewords that are not shortened codewords.  For $q=2$, the extended Hamming and Golay codes yield the $\qcode{7,1,3}_2$ Steane code and the $\qcode{23,1,7}_2$ quantum Golay code, respectively.  These two cases generalize to the following two families.

\begin{eg}[Lifted Hamming codes] \label{eg:liftedhamming}
    Let $q=2^r$, and let $\lambda\in\mathbb{Z}_q$ satisfy $\lambda^2-\lambda+2=0$, e.g., $\lambda=0,2,6,6,6,38$ for $q=2,4,8,16,32,64$.  The cyclic code $\mathcal{H}_q$ of length seven generated by (using polynomial notation for cyclic codes~\cite{macwilliams1977theory})
    \begin{equation}
        g_q(x)=x^3+\lambda x^2+(\lambda-1)x-1 \label{eq:liftedhamming}
    \end{equation}
    is the Hensel lift of the $[7,4,3]$ Hamming code, to which it reduces modulo two.  It is also the puncturing of the extended code of Ref.~\cite{CalderbankSloane1995} at one coordinate.  Its dual is generated by $(x-1)g_q(x)$ and is contained in $\mathcal{H}_q$.  The CSS construction with $H_X=H_Z$ given by the parity-check matrix of $\mathcal{H}_q$ therefore yields a $\qcode{7,1,d}_q$ code, where $d=d_X=d_Z$ is the minimum Lee weight of $\mathcal{H}_q\setminus\mathcal{H}_q^\perp$.  By exhaustive enumeration, we find
    \begin{equation}
        d=3,\,4,\,5,\,7,\,7,\,7 \qquad \text{for} \qquad q=2,\,4,\,8,\,16,\,32,\,64. \label{eq:liftedhammingd}
    \end{equation}
    The $q=4$ member is the quantum heptacode $\qcode{7,1,4}_4$, which stems from the octacode.
The distance cannot exceed seven for any $q$.  Indeed, the all-ones vector lies in $\mathcal{H}_q$, because $g_q$ divides $1+x+\cdots+x^6$, but not in $\mathcal{H}_q^\perp$, because $g_q(1)=2\lambda-1$ is a unit.  Thus, the growth of the quantum distance with $q$ saturates at $q=16$, in contrast with the classical Lee distance of the extended codes, which keeps growing.  Shortening the extended code at two coordinates instead of one yields a $\qcode{6,2,3}_4$ code for $q=4$.
\end{eg}

\begin{eg}[Lifted Golay codes] \label{eg:liftedgolay}
    The cyclic code of length $23$ over $\mathbb{Z}_4$ generated by
    \begin{equation}
        \pi_4(x)=x^{11}+2x^{10}+3x^{9}+3x^{7}+3x^{6}+3x^{5}+2x^{4}+x+3 \label{eq:liftedgolay}
    \end{equation}
    is the Hensel lift of the $[23,12,7]$ Golay code.  Its extension by one coordinate is the self-dual $\mathbb{Z}_4$ Golay code~\cite{Bonnecaze1995}, whose Lee distance is $12$.  Like its binary counterpart, the cyclic code contains its dual, which is generated by $(x-1)\pi_4(x)$, and the CSS construction yields a $\qcode{23,1,10}_4$ code.  Shortening the extended code at two coordinates yields a $\qcode{22,2,8}_4$ code.  For higher powers of two, the lifts of Ref.~\cite{CalderbankSloane1995} yield $\qcode{23,1,12}_8$ and $\qcode{23,1,17}_{16}$ codes, whose distances we computed by exhaustive enumeration.  As for the Hamming family, the all-ones vector is a logical operator, so the distance of this family never exceeds $23$.
\end{eg}

The lifted codes above require $q$ to be a power of two.
Among odd $q$, the Lee metric first differs from the Hamming metric at $q=5$, where a Reed--Solomon code that contains its dual yields the following example.

\begin{eg}[Quantum Reed--Solomon code on $\mathbb{Z}_5$] \label{eg:rs5}
    Let $q=5$, and let $H_X=H_Z$ have the rows $(1,1,1,1,1)$ and $(0,1,2,3,4)$.  These rows are the evaluations of the polynomials $1$ and $x$ at the five points of $\mathbb{Z}_5$.  They span the Reed--Solomon code $\mathrm{RS}_2$ of polynomials of degree less than two, whose dual is the code $\mathrm{RS}_3$ of polynomials of degree less than three.  The CSS construction yields a quantum Reed--Solomon code~\cite{AharonovBenOr2008,Cleve1999secret,Grassl2004optimal}, which encodes one qudit.

    Its nontrivial logical operators are the evaluations of polynomials of degree two, modulo $\mathrm{RS}_2$.  Its Hamming distance is three, but by exhaustive enumeration its Lee distance is four.  It is thus a $\qcode{5,1,4}_5$ code. By the sphere-packing bound \eqref{eq:spherepackingsquare}, no nondegenerate CSS code with these parameters has a larger Lee distance.  Example \ref{eg:rs5ft} discusses its transversal gates.
\end{eg}

\subsection{Realization of $\mathbb Z_q$ stabilizer codes as multimode GKP codes} \label{subsec:gkp}
In this subsection, we show that every $\mathbb Z_q$ stabilizer code can be realized as an ideal multimode GKP code, with one oscillator for each physical qudit.  Under this correspondence, one unit of Lee weight becomes one elementary displacement on the GKP grid. The construction is the $q$-ary analogue of the binary Construction-A formula of Ref.~\cite{Conrad2022GKP}; arbitrary-dimensional GKP-qudit concatenations
with qudit stabilizer codes were also studied in Ref.~\cite{Schmidt2022HigherGKP}.

For $n$ bosonic modes, let
$\hat q=(\hat q_1,\ldots,\hat q_n)$ and
$\hat p=(\hat p_1,\ldots,\hat p_n)$ satisfy
\begin{equation}
    [\hat q_j,\hat p_k]=\ii\delta_{jk}.
\end{equation}

For phase-space vectors
$\xi=(\Delta q\mid\Delta p)\in\mathbb R^{2n}$, we use the displacement convention
\begin{equation}
    D(\xi)
    =
    D(\Delta q,\Delta p)
    :=
    \exp\!\left[
        \ii\left(
            \Delta p\cdot\hat q-\Delta q\cdot\hat p
        \right)
    \right].
\end{equation}

With
\begin{equation}
    J
    =
    \begin{pmatrix}
        0 & I_n \\
        -I_n & 0
    \end{pmatrix},
\end{equation}
we have
\begin{equation}
    D(\xi)D(\eta)
    =
    \exp\!\left(
        -\ii\xi J\eta^{\mathsf T}
    \right)
    D(\eta)D(\xi).
\end{equation}

Set $\lambda_q:=\sqrt{2\pi/q}$.
On mode $j$, first impose the two local GKP stabilizers
\begin{subequations} \label{eq:largestabilizers}
    \begin{align}
            S_{q,j}^{(q)}
    :=
    D(q\lambda_q e_j,0)
    =
    D(\sqrt{2\pi q}\,e_j,0), \\
    S_{p,j}^{(q)}
    :=
    D(0,q\lambda_q e_j)
    =
    D(0,\sqrt{2\pi q}\,e_j),
    \end{align}
\end{subequations}
where $e_j$ is the $j$-th standard basis vector of $\mathbb R^n$. Their ideal
joint eigenspace is a $q$-dimensional square GKP qudit. On this subspace, define
\begin{equation}
    \overline X_j
    :=
    D(\lambda_q e_j,0),
    \qquad
    \overline Z_j
    :=
    D(0,\lambda_q e_j).
\end{equation}
These operators obey
\begin{equation}
    \overline X_j\overline Z_j
    =
    \e^{-2\pi\ii/q}\overline Z_j\overline X_j,
    \qquad
    \overline X_j^{\,q}
    =
    S_{q,j}^{(q)},
    \qquad
    \overline Z_j^{\,q}
    =
    S_{p,j}^{(q)}.
\end{equation}

We next impose the stabilizers of a $\mathbb Z_q$ CSS code on the $q^n$-dimensional space encoded in these $n$ local GKP qudits. General stabilizer codes are treated at the end of this subsection.

Let $H_X\in\mathbb Z_q^{m_X\times n}$ and
$H_Z\in\mathbb Z_q^{m_Z\times n}$ obey
$H_XH_Z^{\mathsf T}=0\pmod q$. Choose arbitrary integer lifts of their entries.
For every row $h_X$ of $H_X$ and every row $h_Z$ of $H_Z$, impose the mapped
stabilizers
\begin{equation}
    D(\lambda_q h_X,0),
    \qquad
    D(0,\lambda_q h_Z),
\end{equation}
together with all local stabilizers in \eqref{eq:largestabilizers}. Equivalently, an overcomplete
row-generator matrix for the resulting GKP lattice is
\begin{equation} \label{eq:GKPgenerator}
    M_{\mathrm{GKP}}
    =
    \lambda_q
    \begin{pmatrix}
        H_X & 0 \\
        0 & H_Z \\
        qI_n & 0 \\
        0 & qI_n
    \end{pmatrix}.
\end{equation}
The first two blocks contain the mapped CSS checks, whereas the last two blocks
contain the $2n$ local large-displacement stabilizers. The displayed set is generally redundant and depends on the chosen rows and lifts, but the lattice it generates does not: changing a lift by a multiple of $q$ only adds a combination of the $qI_n$ rows.

It is useful to write the same construction without choosing generators. Let
\begin{equation}
    \mathcal S_{\mathrm{CSS}}
    :=
    \operatorname{im}(H_X^{\mathsf T})
    \oplus
    \operatorname{im}(H_Z^{\mathsf T})
    \subseteq
    \mathbb Z_q^{2n}
\end{equation}
be the module of CSS stabilizer labels, with the two summands placed in the
$X$ and $Z$ components. We write \begin{equation}
    \mathcal S_{\mathrm{CSS}}^\perp
    :=
    \left\{
        v\in\mathbb Z_q^{2n}:
        vJu^{\mathsf T}=0\pmod q
        \text{ for every }u\in \mathcal S_{\mathrm{CSS}}
    \right\},
\end{equation}
so the logical Pauli module is\begin{equation} \label{eq:S_SCC_logical}
    \frac{\mathcal S_{\mathrm{CSS}}^\perp}{\mathcal S_{\mathrm{CSS}}} = \frac{\ker(H_Z)}{\operatorname{im}(H_X^{\mathsf T})}
    \oplus
    \frac{\ker(H_X)}{\operatorname{im}(H_Z^{\mathsf T})}.
\end{equation}

If
$\pi_q:\mathbb Z^{2n}\to\mathbb Z_q^{2n}$ denotes reduction modulo $q$, then
the full GKP stabilizer lattice is
\begin{equation} \label{eq:GKPlattice}
    \Lambda_q(\mathcal S_{\mathrm{CSS}})
    =
    \lambda_q\pi_q^{-1}(\mathcal S_{\mathrm{CSS}})
    =
    \lambda_q
    \left\{
        (x\mid z)\in\mathbb Z^{2n}:
        \begin{array}{l}
            x\bmod q\in\operatorname{im}(H_X^{\mathsf T}),\\
            z\bmod q\in\operatorname{im}(H_Z^{\mathsf T})
        \end{array}
    \right\}.
\end{equation}
For the displacement centralizer of $\Lambda_q(\mathcal S_{\mathrm{CSS}})$,
\begin{equation}
   \Lambda_q(\mathcal S_{\mathrm{CSS}})^{\perp_{2\pi}}
    :=
    \left\{
        \eta\in\mathbb R^{2n}:
        \eta J\xi^{\mathsf T}\in 2\pi\mathbb Z
        \text{ for every }\xi\in\Lambda_q(\mathcal S_{\mathrm{CSS}})        
    \right\},
\end{equation}
we have
\begin{equation} \label{eq:lambdaperp}
    \Lambda_q(\mathcal S_{\mathrm{CSS}})^{\perp_{2\pi}}
    =
    \lambda_q\pi_q^{-1}(\mathcal S_{\mathrm{CSS}}^\perp).
\end{equation}

\begin{prop}[Exact CSS/GKP correspondence]
The logical-operator quotient of the multimode GKP code described by \eqref{eq:GKPlattice} is isomorphic to that of the corresponding $\mathbb Z_q$ code,
\begin{equation} \label{eq:GKP_codespace}
    \frac{\Lambda_q(\mathcal S_{\mathrm{CSS}})^{\perp_{2\pi}}}
    {\Lambda_q(\mathcal S_{\mathrm{CSS}})}
    \cong
    \frac{\ker(H_Z)}{\operatorname{im}(H_X^{\mathsf T})}
    \oplus
    \frac{\ker(H_X)}{\operatorname{im}(H_Z^{\mathsf T})}.
\end{equation}
More explicitly, there is a one-to-one correspondence between the logical classes of the two codes: every logical operator $v_q\in \mathcal S_{\mathrm{CSS}}^\perp/\mathcal S_{\mathrm{CSS}}$ of the $\mathbb Z_q$ code can be mapped to a unique logical operator of the multimode GKP code\begin{equation} \label{eq:vGKP}
    v_{\text{GKP}} = \lambda_q \pi_q^{-1}(v_q) \in \frac{\Lambda_q(\mathcal S_{\mathrm{CSS}})^{\perp_{2\pi}}}
    {\Lambda_q(\mathcal S_{\mathrm{CSS}})},
\end{equation}
and their reduced weights, i.e., the minimum weights over all representatives of the logical class, are proportional to each other \begin{equation} \label{eq:GKPreducedweight}
    \min_{u\in \Lambda_q(\mathcal S_{\mathrm{CSS}})}
    \left|v_{\text{GKP}}+u\right|_1
    =
    \lambda_q
    \min_{w\in \mathcal S_{\mathrm{CSS}}}
    |v_q+w|_{\mathrm L}.
\end{equation}
\end{prop}

\begin{proof}
Since $\lambda_q qe_k\in\Lambda_q(\mathcal S_{\mathrm{CSS}})$ for every standard basis vector $e_k$ of $\mathbb R^{2n}$, and $\lambda_q^2q=2\pi$, every $\eta\in\Lambda_q(\mathcal S_{\mathrm{CSS}})^{\perp_{2\pi}}$ lies in $\lambda_q\mathbb Z^{2n}$. Writing $\eta=\lambda_q v$ and $\xi=\lambda_q u$, the condition $\eta J\xi^{\mathsf T}\in2\pi\mathbb Z$ becomes $vJu^{\mathsf T}\equiv0\pmod q$, which gives \eqref{eq:lambdaperp} and, with \eqref{eq:S_SCC_logical}, \eqref{eq:GKP_codespace}. Since $\lambda_q q\mathbb Z^{2n}\subseteq\Lambda_q(\mathcal S_{\mathrm{CSS}})$, the minimum in \eqref{eq:GKPreducedweight} runs over all integer lifts of the classes $v_q+w$, and $\min_{x\equiv a\ (\mathrm{mod}\ q)}|x|=|a|_{\mathrm L}$ entrywise.
\end{proof}

The construction above is not limited to CSS codes. For a general
$\mathbb Z_q$ stabilizer code, we map each stabilizer generator $X^xZ^z$, with integer lifts $x$ and $z$ and with its phase chosen as in the remark after Definition~\ref{def:stabilizer}, to the GKP operator $\overline X^{x}\overline Z^{z}=\e^{-\ii\pi x\cdot z/q}D(\lambda_q(x\mid z))$ with the same phase, where $\overline X^{x}:=\prod_j\overline X_j^{\,x_j}$ and $\overline Z^{z}:=\prod_j\overline Z_j^{\,z_j}$. Together with the local stabilizers in \eqref{eq:largestabilizers}, these operators generate the GKP stabilizer group, which is the image of the qudit stabilizer group.
Indeed, consider any $\mathbb Z_q$ stabilizer code, and let
$\mathcal S\subseteq\mathbb Z_q^{2n}$ denote its stabilizers in the symplectic
representation. Since all stabilizers commute,
\begin{equation}
    uJv^{\mathsf T}=0\pmod q
    \qquad
    \text{for all }u,v\in \mathcal S.
\end{equation}
Replacing $\mathcal S_{\mathrm{CSS}}$ by $\mathcal S$ gives
\begin{equation}
    \Lambda_q(\mathcal S)
    =
    \lambda_q\pi_q^{-1}(\mathcal S),
    \qquad
    \Lambda_q(\mathcal S)^{\perp_{2\pi}}
    =
    \lambda_q\pi_q^{-1}(\mathcal S^\perp),
\end{equation}
and hence
\begin{equation}
    \Lambda_q(\mathcal S)^{\perp_{2\pi}}/\Lambda_q(\mathcal S)
    \cong
    \mathcal S^\perp/\mathcal S,
\end{equation}
with the same proportionality between the corresponding reduced weights.
Therefore, the construction extends to arbitrary $\mathbb Z_q$ stabilizer
codes, including non-CSS codes as in Definition~\ref{def:stabilizer}.

Returning to the one-qudit code of Example \ref{eg:toruscode}, we can now give its exact single-mode GKP realization. With \(a'=\gcd(a,q)\), \(b'=\gcd(b,q)\), and \(c'=\gcd(a,b,c,q)\) as defined there, the complete one-mode GKP realization has the four stabilizer generators
\begin{equation}
    D(\lambda_q a,0),
    \qquad
    D(0,\lambda_q b),
    \qquad
    D(\lambda_q q,0),
    \qquad
    D(0,\lambda_q q).
\end{equation}
Along the two coordinate axes, the generated integer subgroups are
\begin{equation}
    a\mathbb Z+q\mathbb Z=a'\mathbb Z,
    \qquad
    b\mathbb Z+q\mathbb Z=b'\mathbb Z.
\end{equation}
Hence $D(\lambda_q a',0)$ and $D(0,\lambda_q b')$ form a minimal set of stabilizer generators.
The fundamental-cell area is
\begin{equation}
    \lambda_q^2 a'b'
    =
    \frac{2\pi}{q}a'b'
    =
    2\pi c',
\end{equation}
so the ideal GKP codespace has dimension $c'$, exactly as in Example \ref{eg:toruscode}. For $c'>1$, a conjugate logical pair is
\begin{equation}
    \overline X_{\mathrm L}
    =
    D\left(\lambda_q\frac{q}{b'},0\right),
    \qquad
    \overline Z_{\mathrm L}
    =
    D\left(0,\lambda_q\frac{q}{a'}\right),
\end{equation}
which agrees with \eqref{eq:vGKP}.

\section{Fault tolerance} \label{sec:ft}

A gadget is fault tolerant if it does not propagate correctable errors into uncorrectable ones~\cite{spencer2026quantuma}.  Transversal gates are automatically fault tolerant in the Hamming metric because a product of single-qudit unitaries never enlarges the support of an error.  In the Lee metric, a small error has low Lee weight, and a gate on a single qudit can already increase it.  Transversality thus no longer guarantees fault tolerance.  As Example~\ref{eg:perfectcss2} shows, the usual Eastin--Knill conclusion need not apply when only errors of small Lee weight are required to be detectable, because a single-qudit operator can have Lee weight as large as the distance. 

We measure the fault tolerance of a gate by its spread, the largest factor by which it can increase the Lee weight of a Pauli error.  A code that corrects all errors of Lee weight at most $t$ still corrects, after a gate of spread $w$, every error whose Lee weight before the gate was at most $\lfloor t/w\rfloor$.  We compute the spread of Clifford gates, bound the fault distance of a logical Clifford gate by its spread, and determine the gates of spread one.  These gates preserve Lee weight exactly and play the role of transversal gates in the Hamming metric.  We then apply the bound to transversal gates and to a family of perfect codes.  The GKP realization turns these perfect codes and their gates into standard GKP codes and Gaussian gates.

\subsection{Lee error spread} \label{subsec:leespread}

Clifford generators, listed in Table~\ref{tab:spread}, spread Lee weight differently from Hamming weight. The phase gate $S$ maps $X^xZ^z$ to $X^xZ^{x+z}$ up to a phase.  It thus maps $X$, of joint Lee weight one, to $XZ$, of joint Lee weight two.  It also maps $XZ$, whose $X$ and $Z$ parts each have Lee weight one, to $XZ^2$, whose $Z$ part has Lee weight two when $q\ge4$.  For qubits, this is the observation of Section \ref{subsec:qubits} that $S$ maps $X$ to $Y$.  The same happens in GKP codes, where squeezing a single mode amplifies displacements along one of its quadratures~\cite{Gottesman2001GKP}.  In both cases, a gate can enlarge a magnitude without enlarging a support.

We quantify the spread with the symplectic formalism over $\mathbb{Z}_q$~\cite{Gottesman1999higher,Appleby2005SIC,hostens2005stabilizer,deBeaudrap2013linearized}.  A Clifford unitary $U$ on $n$ qudits maps Pauli strings to Pauli strings under conjugation.  Up to phases, its action on the label $p=(x\mid z)\in\mathbb{Z}_q^{2n}$ of $X^xZ^z$ is linear:
\begin{equation}
    U X^{x}Z^{z} U^\dagger \propto X^{x'}Z^{z'}, \qquad (x'\mid z')^{\mathsf T} = M\,(x\mid z)^{\mathsf T}. \label{eq:cliffordaction}
\end{equation}
Here $M\in\mathrm{Sp}(2n,\mathbb{Z}_q)$ is a symplectic matrix, $M^{\mathsf T}JM=J$ modulo $q$, with $J$ as in Definition \ref{def:stabilizer}.  Conversely, every symplectic matrix comes from a Clifford unitary that is unique up to Pauli strings and phases.  In this section, we treat Pauli labels as column vectors and abbreviate \eqref{eq:cliffordaction} as $p\mapsto Mp$.

Table \ref{tab:spread} lists the standard generators of the Clifford group~\cite{Gottesman1999higher,hostens2005stabilizer,Farinholt2014clifford}.  They are the Fourier gate $F$, the qudit permutations, the phase gate $S$, the two-qudit gate $\mathrm{SUM}$, which generalizes the controlled-NOT gate, and the multiplication gates $\mathrm{MUL}_a$ for units $a\in\mathbb{Z}_q^\times$.  The square of $F$ is the inversion $|j\rangle\mapsto|-j\rangle$, which is also the multiplication gate $\mathrm{MUL}_{-1}$.

\begin{table}[t]
    \centering
    \caption{\textbf{Lee error spread of the generators of the Clifford group.}  Each row lists a Clifford gate on one, two, or $n$ qudits, its action on the computational basis, its action on Pauli labels $(x\mid z)$ up to phases, and its Lee error spread $w_{\mathrm B}$ \eqref{eq:leespread}.  The spread is the largest Lee weight of a column of the symplectic matrix by Lemma \ref{lem:spread}.  A qudit permutation $\pi$ permutes the entries of $x$ and of $z$ in the same way.  The phase gate is written for odd $q$.  For even $q$, the phase $\omega^{j(j-1)/2}$ is replaced by $\e^{\ii\pi j^2/q}$, which has the same action on labels~\cite{Gottesman1999higher,hostens2005stabilizer}.  The gates in the first three rows preserve Lee weight exactly.  The multiplication gate is defined for units $a\in\mathbb{Z}_q^\times$, and its spread equals one only for $a=\pm1$.  For example, $\mathrm{MUL}_3$ on $\mathbb{Z}_{10}$ has spread three, and $\mathrm{MUL}_2$ for odd $q$ has spread $(q-1)/2$.}
    \label{tab:spread}
    \vspace{0.5em}
    \footnotesize
    \renewcommand{\arraystretch}{1.3}
    \begin{tabular}{l|l|l|c}
        \hline\hline
        Gate & Action on basis states & Action on Pauli labels & $w_{\mathrm B}$ \\
        \hline
        Fourier $F$ & $|j\rangle\mapsto q^{-1/2}\sum_{k}\omega^{jk}|k\rangle$ & $(x\mid z)\mapsto(-z\mid x)$ & $1$ \\
        Inversion $F^2$ & $|j\rangle\mapsto|-j\rangle$ & $(x\mid z)\mapsto(-x\mid -z)$ & $1$ \\
        Permutation & $|j_1,\dots,j_n\rangle\mapsto|j_{\pi(1)},\dots,j_{\pi(n)}\rangle$ & $(x\mid z)\mapsto(\pi x\mid \pi z)$ & $1$ \\
        Phase $S$ & $|j\rangle\mapsto\omega^{j(j-1)/2}|j\rangle$ & $(x\mid z)\mapsto(x\mid x+z)$ & $2$ \\
        $\mathrm{SUM}$ & $|j,k\rangle\mapsto|j,j+k\rangle$ & $(x_1,x_2\mid z_1,z_2)\mapsto(x_1,x_1+x_2\mid z_1-z_2,z_2)$ & $2$ \\
        Multiplication $\mathrm{MUL}_a$ & $|j\rangle\mapsto|aj\rangle$ & $(x\mid z)\mapsto(ax\mid a^{-1}z)$ & $\max(|a|_{\mathrm L},|a^{-1}|_{\mathrm L})$ \\
        \hline\hline
    \end{tabular}
\end{table}

\begin{defn}[Lee error spread] \label{def:leespread}
    The Lee error spread of a Clifford gate $U$ with symplectic matrix $M$ is
    \begin{equation}
        w_{\mathrm B}(U) := \max_{p\in\mathbb{Z}_q^{2n}\setminus 0}\frac{|Mp|_{\mathrm L}}{|p|_{\mathrm L}}, \label{eq:leespread}
    \end{equation}
    the largest factor by which conjugation by $U$ can increase the joint Lee weight \eqref{eq:jointlee} of a Pauli string.
\end{defn}

\noindent The spread depends only on $M$, so we also write $w_{\mathrm B}(M)$.  It does not change when $U$ is multiplied by a Pauli string.  Since $M$ is invertible, $w_{\mathrm B}(U)\ge1$.  The spread is submultiplicative, $w_{\mathrm B}(UV)\le w_{\mathrm B}(U)\,w_{\mathrm B}(V)$, and the spread of a layer of gates on disjoint qudits is the largest spread in the layer.  A circuit of depth $D$ with gates of spread at most $w$ therefore has spread at most $w^D$.  This bound is the Lee-metric counterpart of the light cone that bounds the growth of Hamming weight in a circuit.

\begin{lem}[Column formula for the spread] \label{lem:spread}
    For every matrix $M\in\mathbb{Z}_q^{m\times N}$,
    \begin{equation}
        \max_{p\in\mathbb{Z}_q^{N}\setminus0}\frac{|Mp|_{\mathrm L}}{|p|_{\mathrm L}}=\max_{1\le j\le N}|Me_j|_{\mathrm L}, \label{eq:columnformula}
    \end{equation}
    where $e_j$ is the $j$-th standard basis vector.  The same holds for integer matrices with the $1$-norm in place of the Lee weight.  In particular, $w_{\mathrm B}(U)$ is a positive integer, namely the largest Lee weight of a column of the symplectic matrix of $U$.
\end{lem}

\begin{proof}
    Lee weight is subadditive and invariant under $v\mapsto-v$.  Every $p$ is a sum of $|p|_{\mathrm L}$ vectors of the form $\pm e_j$, because its $j$-th entry equals $\pm|p_j|_{\mathrm L}$ modulo $q$.  Hence $|Mp|_{\mathrm L}\le\sum_j|p_j|_{\mathrm L}\,|Me_j|_{\mathrm L}\le|p|_{\mathrm L}\max_j|Me_j|_{\mathrm L}$.  Equality holds for $p=e_j$ with $j$ a maximizing index.  The same argument applies to the $1$-norm on $\mathbb{Z}^N$.
\end{proof}

\noindent The symbol $w_{\mathrm B}$ anticipates the Lee-LDPC condition of Section~\ref{sec:LDPC}: by Lemma~\ref{lem:spread}, a check matrix $H$ whose columns have Lee weight at most $w_{\mathrm B}$ has a syndrome map $x\mapsto Hx$ of spread at most $w_{\mathrm B}$.

A code with joint Lee distance $d_\diamond\ge2t+1$ corrects every error of joint Lee weight at most $t$ (Section \ref{sec:stabilizer}).  An error of Lee weight $s$ that enters the gate $U$ leaves it with Lee weight at most $s\,w_{\mathrm B}(U)$.  It remains correctable if $s\le\lfloor t/w_{\mathrm B}(U)\rfloor$.

The same bound controls faults.  Let $U$ be a Clifford gate that maps the codespace to itself, up to a Pauli string.  This holds exactly when its symplectic matrix $M$ maps the stabilizer group $\mathcal S$ of Definition \ref{def:stabilizer} to itself, and we then call $U$ a logical Clifford gate.  We place ideal error correction before and after $U$, and we model each fault as a Pauli error of joint Lee weight one before or after $U$.  After $U$, a fault that occurred before $U$ has Lee weight at most $w_{\mathrm B}(U)$.  A fault that occurs after $U$ has Lee weight one.  Since $w_{\mathrm B}(U)\ge1$, any $s$ faults reach the final error correction as an error of Lee weight at most $s\,w_{\mathrm B}(U)$.  The fault distance $d_{\mathrm{fault}}(U)$ is the smallest number of faults that can cause a logical error.  With $t=\lfloor(d_\diamond-1)/2\rfloor$, it satisfies
\begin{equation}
    d_{\mathrm{fault}}(U)\ge\left\lfloor\frac{t}{w_{\mathrm B}(U)}\right\rfloor+1=\left\lceil\frac{d_\diamond}{2\,w_{\mathrm B}(U)}\right\rceil. \label{eq:faultdistance}
\end{equation}

The bound certifies that $U$ tolerates a fault only when $w_{\mathrm B}(U)\le t$.  The spread lowers the guaranteed fault distance from $t+1$ to about $t/w_{\mathrm B}(U)$ for every code.  For qubit codes of fixed length, $t$ is bounded, so the guarantee stays bounded.  For Lee codes of fixed length, $t$ can grow with $q$, as Example \ref{eg:perfectcss2ft} below shows. 

For CSS codes, a decoder can correct the $X$ and $Z$ parts of an error separately.  Since $d_\diamond=d_\square=d$, it corrects every error whose two parts each have Lee weight at most $t$.  The bound \eqref{eq:faultdistance} then holds with $w_{\mathrm B}(U)$ replaced by the largest Lee weight of the $X$ part or the $Z$ part of a column of $M$.  The argument does not require $U$ to map CSS codes to CSS codes.

The Clifford gates that preserve Lee weight exactly, i.e., the gates of spread one, correspond to those symplectic matrices \(M\) that are signed permutation matrices.  These play the role for Lee codes that qudit permutations and single-qudit Clifford gates play in the Hamming metric~\cite{GluesingLuerssenPllaha2019frobenius,Pllaha2020symplectic, koh2026achieving}. 
These matrices are the products of a permutation of the qudits with a power of the Fourier gate \(F\) on each qudit. Modulo Pauli strings and phases, the gates of spread one thus form the wreath product \(\mathbb{Z}_4 \wr S_n\) for \(q>2\), where the cyclic group is generated by \(F\) and \(S_n\) is the qudit permutation group.
They map the logical operators of any stabilizer code to logical operators of the same joint Lee weight of another stabilizer code.  They therefore preserve the joint Lee distance $d_\diamond$ of every code and define the natural notion of equivalence of Lee stabilizer codes.  

The matrices in \(\mathbb{Z}_4 \wr S_n\) that preserve the separate Lee weights \(|x|_{\mathrm L}\) and \(|z|_{\mathrm L}\), up to exchanging the two, form a subgroup.
This subgroup is generated by the qudit permutations, the single-qudit inversions \(F^2_i\), and the \textit{global} Fourier gate \(F^{\otimes n}\).
It preserves both $d_\diamond$ and $d_\square$ and maps CSS codes to CSS codes.  Its elements are the relevant equivalences of CSS codes.

\subsection{Codes with Clifford gates}

Every CSS code has the transversal gate $\mathrm{SUM}^{\otimes n}$ between two of its blocks.  A code with $H_X=H_Z$ also has the transversal Fourier gate $F^{\otimes n}$, which exchanges the $X$ checks and the $Z$ checks, up to inversion.  In this case, $S^{\otimes n}$ fixes each $Z$ check and maps each $X$ check to the product of an $X$ check and a $Z$ check.  It is thus a logical gate for every $q$.  Since $F$ and $S$ generate the single-qudit Clifford group modulo Pauli operators and phases, $G^{\otimes n}$ is then a logical gate for every single-qudit Clifford gate $G$.

These logical gates can need a Pauli correction.  For even $q$, $S^{\otimes n}$ needs none exactly when every row $h$ of $H_X$ also satisfies $h\cdot h\equiv0\pmod{2q}$.  For $q=2$, this is the doubly-even condition.  The lifted Hamming codes and the $\qcode{23,1,10}_4$ code of Examples \ref{eg:liftedhamming} and \ref{eg:liftedgolay} satisfy it.  For odd $q$, $S^{\otimes n}$ itself may need the correction $(Z^{2^{-1}})^{\otimes n}$, but the phase gate $|j\rangle\mapsto\omega^{2^{-1}j^2}|j\rangle=SZ^{2^{-1}}|j\rangle$ needs no correction on any code with $H_X=H_Z$, because $2$ is invertible modulo $q$~\cite{Prakash2020golay}.

A transversal gate acts on each qudit position separately.  A fault on qudit $i$ of one block thus spreads only to qudit $i$ of the other blocks.  For a CSS code, the CSS form of the bound then applies to each block separately.  The relevant weight is the largest Lee weight, within one block, of the $X$ part or the $Z$ part of a column of the symplectic matrix at one position.  This per-block weight is at most $w_{\mathrm B}(U)$.

The per-block weight is one for $F^{\otimes n}$, $S^{\otimes n}$, and $\mathrm{SUM}^{\otimes n}$.  These gates thus tolerate $t$ faults, as in the Hamming metric.  In the joint metric, $S$ and $\mathrm{SUM}$ have spread two, so the joint bound \eqref{eq:faultdistance} certifies only $\lfloor t/2\rfloor$ faults.  On a CSS code, this loss is only in the bound, since a decoder that minimizes the joint Lee weight also minimizes the Lee weights of the two parts separately.

For $\mathrm{MUL}_a^{\otimes n}$, the per-block weight is $\max(|a|_{\mathrm L},|a^{-1}|_{\mathrm L})$, so the bound certifies fewer faults when $a\neq\pm1$.
On the perfect codes of Example \ref{eg:perfectcssn}, one fault before $\mathrm{MUL}_a^{\otimes n}$ with $a\neq\pm1$ can cause a logical error for every decoder that corrects all errors of Lee weight one.

\begin{eg}[Transversal Clifford gates on the $\mathbb{Z}_5$ Reed--Solomon code] \label{eg:rs5ft}
    Consider the $\qcode{5,1,4}_5$ code of Example \ref{eg:rs5}.  Every single-qudit Clifford gate acts transversally on it, since $H_X=H_Z$.  The transversal gates thus realize the full Clifford group of the logical qudit.  Modulo Pauli operators and phases, this group is $\mathrm{Sp}(2,\mathbb{Z}_5)=\mathrm{SL}(2,5)$, which is isomorphic to the binary icosahedral group $2I$ of order $120$.  The smallest known qubit code whose transversal gates form the group $2I$ is the $((7,2,3))$ code of Pollatsek and Ruskai~\cite{Pollatsek2004permutation,Kubischta2023exotic}.  The two codes are otherwise not directly comparable, since our five subsystems have dimension five and our distance is measured in the Lee metric.  For quantum polynomial codes over prime fields, Aharonov and Ben-Or implemented the $\mathrm{SUM}$ and multiplication gates transversally, and the Fourier gate up to a change of code~\cite{AharonovBenOr2008}.

    The distance $d=4$ gives $t=1$.  The gates $F^{\otimes5}$, $S^{\otimes5}$, and $\mathrm{SUM}^{\otimes5}$ between two blocks each tolerate one fault.  A single fault before $\mathrm{MUL}_2^{\otimes5}$ becomes $X_j^{\pm2}$ or $Z_j^{\pm2}$, of Lee weight two.  This error has the same syndrome as an error of Lee weight two on two other qudits, and the two differ by a logical operator.  Whether $\mathrm{MUL}_2^{\otimes5}$ tolerates one fault thus depends on how the decoder breaks this tie.  For any decoder that corrects every error of Lee weight one, the logical $\mathrm{MUL}_2$ gate tolerates one fault when it is implemented as a product of $F^{\otimes5}$ and $S^{\otimes5}$ gates, with error correction between them.
\end{eg}

We conclude with a family of codes on which the fault distance of every logical Clifford gate of bounded spread grows with $q$.

\begin{eg}[Fault-tolerant gates on perfect two-qudit CSS codes] \label{eg:perfectcss2ft}
    Consider the codes of Example \ref{eg:perfectcss2}, with $q=m(2t^2+2t+1)$, $H_X=H_Z=(mt,\ m(t+1))$, and distance $d=2t+1$.  Let $a_1=(t,t+1)$ and $a_2=(-(t+1),t)$.  Since $a_i\cdot a_j=\delta_{ij}\,q/m$, the operators $X^{a_i}$ and $Z^{a_i}$ act as the Pauli operators of two logical qudits of dimension $m$.

    Because $H_X=H_Z$, the Fourier gate $F^{\otimes2}$ on both qudits exchanges the $X$ checks and the $Z$ checks, up to inversion.  It is a logical gate of spread one, and it acts as the Fourier gate on each logical qudit.  The phase gate $S^{\otimes2}$ on both qudits fixes each $Z$ check and maps each $X$ check to the product of an $X$ check and a $Z$ check.  The resulting logical gate has spread two and acts as a logical phase gate on each logical qudit, up to a logical Pauli string.

    At fixed $m$, the radius $t$ grows with $q$, since $(2t+1)^2=2q/m-1$, so that $t\approx\sqrt{q/(2m)}$.  By \eqref{eq:faultdistance}, every logical Clifford gate $U$ of spread $w=w_{\mathrm B}(U)$ satisfies
    \begin{equation}
        d_{\mathrm{fault}}(U)\ge\left\lceil\frac{\sqrt{2q/m-1}}{2w}\right\rceil. \label{eq:perfectcss2ft}
    \end{equation}
    Such a gate tolerates at least one fault once $t\ge w$, that is, once $q\ge m(2w^2+2w+1)$.  At fixed $m$ and $w$, the bound grows as $\sqrt q$, while each fault keeps Lee weight one.

    Like every transversal phase gate on a CSS code, $S^{\otimes2}$ tolerates $t$ faults, as many as a gate of spread one.  This is optimal.  After any gate, $t+1$ faults can produce the error $X_2^{t+1}$.  It has the same syndrome as $X_1^{-t}$, so the decoder turns it into the logical operator $X^{a_1}$.  Hence $d_{\mathrm{fault}}(F^{\otimes2})=d_{\mathrm{fault}}(S^{\otimes2})=t+1$.

    The GKP realization of Section \ref{subsec:gkp} turns these codes and gates into standard GKP objects.  The vectors $a_1$ and $a_2$ are orthogonal and have length $\sqrt{q/m}$, so $a_i=\sqrt{q/m}\,R_\theta e_i$.  Here $R_\theta$ is the rotation of the plane by the angle $\theta=\arctan((t+1)/t)$.  The integer lifts of the stabilizer labels of each type form the lattice generated by $ma_1$ and $ma_2$.  The GKP stabilizer lattice \eqref{eq:GKPlattice} is therefore
    \begin{equation}
        \Lambda_q(\mathcal S_{\mathrm{CSS}})=\sqrt{2\pi m}\,\left(R_\theta\oplus R_\theta\right)\mathbb Z^4. \label{eq:perfectcss2gkp}
    \end{equation}
    The map $R_\theta\oplus R_\theta$ rotates the position displacements $(\Delta q_1,\Delta q_2)$ and the momentum displacements $(\Delta p_1,\Delta p_2)$ in the same way, so it is a beam splitter $B_\theta$.  The GKP realization is thus $B_\theta$ applied to two $m$-dimensional square GKP qudits, with stabilizers $D(\sqrt{2\pi m}\,e_j,0)$ and $D(0,\sqrt{2\pi m}\,e_j)$~\cite{Gottesman2001GKP}.  The logical Pauli operators $X^{a_i}$ and $Z^{a_i}$ become the images under $B_\theta$ of the Pauli operators $D(\sqrt{2\pi/m}\,e_i,0)$ and $D(0,\sqrt{2\pi/m}\,e_i)$ of the square GKP qudits.

    The gates follow the same pattern.  In the GKP realization, $F^{\otimes2}$ rotates the phase space of each mode by a quarter turn, $(\Delta q_j,\Delta p_j)\mapsto(-\Delta p_j,\Delta q_j)$.  Up to a displacement, $S^{\otimes2}$ is the shear $(\Delta q_j,\Delta p_j)\mapsto(\Delta q_j,\Delta q_j+\Delta p_j)$ of each mode.  These Gaussian gates do not depend on the dimension of a square GKP qudit.  They are the Fourier gate and the phase gate both for the $q$-dimensional qudits of the grid and for the $m$-dimensional square GKP qudits.  They act on both modes in the same way, so they commute with $B_\theta$.  The physical gates $F^{\otimes2}$ and $S^{\otimes2}$ are thus the same Gaussian unitaries as the logical Fourier and phase gates of the square GKP qudits, for every $t$.  Similarly, swapping the two qudits and then inverting the second one is the beam splitter $B_{-\pi/2}$, which commutes with $B_\theta$.  It is a logical gate of spread one, and it swaps the two square GKP qudits up to an inversion.

    As $q\to\infty$ at fixed $m$, the angle $\theta$ tends to $\pi/4$, since $\tan(\theta-\pi/4)=1/(2t+1)$.  The GKP realizations thus converge to two $m$-dimensional square GKP qudits mixed by a balanced (50:50) beam splitter.  This limit code has the stabilizer lattice $\sqrt{\pi m}\,(D_2\oplus D_2)$, where $D_2=\{(u,v)\in\mathbb Z^2:u+v\text{ even}\}$ is the checkerboard lattice.  Its stabilizer group is generated by $D(\sqrt{\pi m}\,(e_1\pm e_2),0)$ and $D(0,\sqrt{\pi m}\,(e_1\pm e_2))$.  In the limit, the Lee-sphere tiling of Example \ref{eg:perfectcss2} becomes the tiling of each quadrature plane by $1$-norm balls of radius $\sqrt{\pi/m}$ centered on the logical lattice $\sqrt{\pi/m}\,D_2$.  Each fault of Lee weight one is a displacement by one grid step $\lambda_q=\sqrt{2\pi/q}$, which tends to zero while the GKP realizations converge.
\end{eg}

\section{Qubitizing codes with Gray maps} \label{sec:gray}

A Gray code lists the binary strings of a fixed length in a cyclic order in which consecutive strings differ in exactly one bit.  Read as an encoding of $\mathbb{Z}_q$ with $q=2^r$, it sends each integer to a string of $r$ bits that is one bit flip away from the strings of its two neighbors on the cycle $\mathbb{Z}_q$.  It gives a bijection between $\mathbb Z_{2^r}^n$ and binary strings of length $rn$, although it need not preserve distances. 

For $q=4$, the Gray map is an isometry from the Lee metric on $\mathbb{Z}_4^n$ to the Hamming metric on $\mathbb{Z}_2^{2n}$ \cite{Hammons1994,Wan1997quaternary}.  Hammons et al.~\cite{Hammons1994} used this fact to resolve a puzzle of classical coding theory.  The binary Kerdock and Preparata codes are nonlinear, yet their weight enumerators are related by the MacWilliams identity as if the two codes were dual linear codes.  The resolution is that the Kerdock code and a Preparata-like code with the same weight enumerator as the Preparata code are Gray images of \emph{linear} codes over $\mathbb{Z}_4$ that are dual to each other.  Ling and Sol\'e later used the same map to turn $\mathbb{Z}_4$-linear codes into nonadditive qubit codes \cite{LingSole2008}.  We develop further quantum applications of the Gray map.  

Section \ref{subsec:graymaps} defines quantum Gray maps, which relabel the basis of a four-dimensional qudit, or quartrit, by two qubits, and determines where they send Pauli and Clifford operators.  Section \ref{subsec:qubitizing} relates the Lee distance of a $\mathbb{Z}_4$ stabilizer code to the distances of its qubitized image, applies the result to the self-dual codes of Section \ref{subsec:smallexamples}, and translates their transversal gates.

\subsection{Quantum Gray maps} \label{subsec:graymaps}

The quaternary-to-binary Gray map converts, or ``qubitizes,'' each element of $\mathbb{Z}_4$ into two bits,
\begin{equation}
    \varphi:\qquad 0\mapsto 00,\qquad 1\mapsto 01,\qquad 2\mapsto 11,\qquad 3\mapsto 10, \label{eq:graymap}
\end{equation}
and acts entrywise on vectors, $\varphi:\mathbb{Z}_4^n\to\mathbb{Z}_2^{2n}$.  It preserves distances.  Two elements of $\mathbb{Z}_4$ at Lee distance one are neighbors on the cycle $0,1,2,3$, and their images differ in one bit, while the images of two elements at Lee distance two differ in both bits.  Hence, the Lee weight of a vector is the Hamming weight of its image, and the Lee distance between two vectors is the Hamming distance between their images \cite[Thm.~3.1]{Wan1997quaternary}.  A $\mathbb{Z}_4$-linear code $\mathcal{C}$ of Lee distance $d$ thus becomes a binary code $\varphi(\mathcal{C})$ with the same number of codewords and with Hamming distance $d$ \cite{Hammons1994}.  The binary code is usually nonlinear, because the images of $u$, $v$, and $u+v$ fail to add up at the coordinates where both $u$ and $v$ are odd: $\varphi(u+v)=\varphi(u)\oplus\varphi(v)\oplus\varphi(2\,u\ast v)$, where $u\ast v$ is the entrywise product.  Accordingly, $\varphi(\mathcal{C})$ is linear if and only if $2\,u\ast v\in\mathcal{C}$ for all $u,v\in\mathcal{C}$ \cite{Hammons1994}.

The reflected Gray code extends \eqref{eq:graymap} to $q=2^r$ with $r$ bits per element, but distance preservation fails for every $q>4$.  For $q=8$, for example, the element $4$ is encoded as $110$, of Hamming weight two, while its Lee weight is four.  The quantum maps below are confined to $q=4$ for a similar reason: for $q>4$, some single-qubit phase flip pulls back to an operator whose Pauli expansion contains $Z^b$ with $|b|_{\mathrm L}\ge2$, whichever labeling is used.

There are several ways to upgrade the Gray map to the quantum setting, because we are free to choose both the eigenbasis of the incoming quartrit that is relabeled and the eigenbases of the two outgoing qubits.  Let
\begin{equation}
    \phi_Z=\sum_{m\in\mathbb{Z}_4}|\varphi(m)\rangle\langle m|
    \qquad\text{and}\qquad
    \phi_X=(H\otimes H)\,\phi_Z\,F, \label{eq:quantumgray}
\end{equation}
where $H$ is the qubit Hadamard gate and $F$ is the quartrit Fourier gate of Table \ref{tab:spread}.  The $Z$-basis map $\phi_Z$ is a permutation matrix.  It relabels the computational basis $|m\rangle$ of the quartrit, on which $Z|m\rangle=i^m|m\rangle$, by the computational basis $|\varphi(m)\rangle$ of two qubits.  The $X$-basis map $\phi_X$ does the same for the $X$ eigenbases.
 On $n$ quartrits, the maps act as $\Phi_Z=\phi_Z^{\otimes n}$ and $\Phi_X=\phi_X^{\otimes n}$, sending the $j$-th quartrit to the $j$-th \emph{block} of two qubits.  The relabeling $\phi_Z$ is the four-level case of the Gray encoding proposed for simulating bosonic and spin systems on qubit-based quantum computers~\cite{Sawaya2020resource}.

More generally, a Pauli-basis Gray map is $\phi_{V,W}=W\phi_ZV^\dagger$ with $V$ a $\mathbb{Z}_4$ Clifford gate on one quartrit and $W$ a qubit Clifford gate on two qubits.  We mostly consider the two canonical choices \eqref{eq:quantumgray}.  Since $FXF^\dagger=Z$ and $FZF^\dagger=X^{-1}$ (Table \ref{tab:spread}) and $H$ exchanges the qubit bit-flip and phase-flip operators, every statement below about $\phi_Z$ has a counterpart for $\phi_X$ with the roles of $X$ and $Z$ exchanged.  We comment on other choices when they increase the distance of the qubitized code. 

Throughout this section, $X$ and $Z$ denote quartrit Pauli operators, while $\sigma_x$, $\sigma_y$, and $\sigma_z$ denote qubit Pauli operators.  A qubit Pauli string is written $\sigma_x^{a}\sigma_z^{b}$ with $a,b\in\mathbb{Z}_2^{2n}$.  The quartrit phase gate of Table \ref{tab:spread} is $S:|m\rangle\mapsto e^{i\pi m^2/4}|m\rangle$, and the qubit phase gate is written $S_2=\mathrm{diag}(1,i)$ to distinguish the two.  We also use $T=\mathrm{diag}(1,e^{i\pi/4})$, the controlled-$S_2$ gate $\mathrm{CS}=\mathrm{diag}(1,1,1,i)$, and the standard two-qubit gates $\mathrm{CZ}$, $\mathrm{CNOT}$, and $\mathrm{SWAP}$, with the first qubit of a block as the control.  The following dictionary is the basis of everything that follows.

\begin{lem}[Gray dictionary] \label{lem:graydictionary}
    Let $R=F^2$ be the inversion $|m\rangle\mapsto|-m\rangle$ of one quartrit, so that $R=\tfrac12(I+X^2+Z^2-X^2Z^2)$.  Under $\phi_Z$, the qubit Pauli operators pull back as $E\mapsto\phi_Z^\dagger E\phi_Z$, namely
    \begin{subequations} \label{eq:graydictionary}
    \begin{align}
        \sigma_x\otimes I&\mapsto X^{-1}R, &
        I\otimes\sigma_x&\mapsto XR, &
        \sigma_x\otimes\sigma_x&\mapsto X^2, \label{eq:graydictionaryx}\\
        \sigma_z\otimes I&\mapsto \frac{e^{-i\pi/4}Z+e^{i\pi/4}Z^{-1}}{\sqrt2}, &
        I\otimes\sigma_z&\mapsto \frac{e^{i\pi/4}Z+e^{-i\pi/4}Z^{-1}}{\sqrt2}, &
        \sigma_z\otimes\sigma_z&\mapsto Z^2. \label{eq:graydictionaryz}
    \end{align}
    \end{subequations}
    Consequently, a qubit Pauli string on $2n$ qubits pulls back to
    \begin{equation}
        \Phi_Z^\dagger\,\sigma_x^{a}\sigma_z^{b}\,\Phi_Z = X^{t}R_{\mathcal{O}(t)}\sum_{k}c_k Z^{k}, \label{eq:graypullback}
    \end{equation}
    where $t\in\mathbb{Z}_4^n$ has $t_j=0,1,-1,2$ when the two bits of $a$ in block $j$ are $00,01,10,11$, respectively, and $\mathcal{O}(t)=\{j:t_j\text{ odd}\}$ is the odd support of $t$.  The partial inversion $R_{\mathcal{O}}=\bigotimes_{j\in\mathcal{O}}R_j$ negates the coordinates in $\mathcal{O}$: $R_{\mathcal{O}}|v\rangle=|v'\rangle$ with $v'_j=-v_j$ for $j\in\mathcal{O}$ and $v'_j=v_j$ otherwise.  The sum runs over the $k\in\mathbb{Z}_4^n$ with $k_j=0$, $k_j\in\{1,-1\}$, or $k_j=2$ when the two bits of $b$ in block $j$ are $00$, contain a single $1$, or are $11$, respectively, and all coefficients $c_k$ are nonzero.  Every term satisfies $|t|_{\mathrm L}=|a|_{\mathrm H}$ and $|k|_{\mathrm L}=|b|_{\mathrm H}$.
\end{lem}

\begin{proof}
    The bit flip $\sigma_x\otimes I$ exchanges the labels $00\leftrightarrow10$ and $01\leftrightarrow11$, i.e., it maps $|m\rangle$ to $|-m-1\rangle=X^{-1}R|m\rangle$, and $I\otimes\sigma_x$ maps $|m\rangle$ to $|1-m\rangle=XR|m\rangle$.  The phase flips $\sigma_z\otimes I$ and $I\otimes\sigma_z$ are diagonal with entries $(1,1,-1,-1)$ and $(1,-1,-1,1)$ for $m=0,1,2,3$, which are the displayed combinations of $Z=\mathrm{diag}(1,i,-1,-i)$ and $Z^{-1}$.  The remaining entries of \eqref{eq:graydictionary} are products of these.  Since $\Phi_Z$ is a tensor product, \eqref{eq:graypullback} follows by multiplying the single-block pullbacks, using that $R_j$ commutes with the Pauli operators of every other quartrit.  The weight statements hold because each block with one flipped bit contributes Lee weight one, and each block with two flipped bits contributes Lee weight two.
\end{proof}

\noindent In words, a qubit phase flip pulls back to a superposition of quartrit phase shifts of the same Lee weight, and a qubit bit flip pulls back to a quartrit shift of the same Lee weight times a partial inversion.  The partial inversion $R_{\mathcal{O}}$, which mixes Lee weights, is the only obstacle to reading $\Phi_Z$ as a weight-preserving map, and it is responsible for the distance subtleties below.

The Clifford hierarchy \cite{GottesmanChuang1999} of $N$ qubits is defined recursively.  Its first level $\mathcal{C}_1$ is the Pauli group, and its level $\mathcal{C}_{k+1}$ consists of the unitaries $U$ with $U\sigma U^\dagger\in\mathcal{C}_k$ for every Pauli string $\sigma$.  The second level is the qubit Clifford group, and the third level contains the $T$ gate, the $\mathrm{CS}$ gate, and the Toffoli gate.  On the operator level, the $Z$-basis Gray map sends quartrit Pauli operators to qubit Clifford gates, and $\mathbb{Z}_4$ Clifford gates to the third level of the hierarchy.  We prove this in the following theorem and list the images of the standard generators in Table \ref{tab:grayimages}.

\begin{table}[t]
    \centering
    \caption{\textbf{Gray images of $\mathbb{Z}_4$ Pauli and Clifford gates.}  Each row lists a gate on one or two quartrits, the reduction $\bar M$ of its symplectic matrix \eqref{eq:cliffordaction} modulo two, its image under the $Z$-basis Gray map $\phi_Z$ of \eqref{eq:quantumgray}, and the level of the image in the qubit Clifford hierarchy, in accordance with Theorem \ref{thm:graylevels}.  The two qubits of a block are labeled by $b_1b_2=\varphi(m)$, and the first qubit is the control of $\mathrm{CZ}$, $\mathrm{CS}=\mathrm{diag}(1,1,1,i)$, and $\mathrm{CNOT}$.  The entry ``other'' marks a matrix $\bar M$ that is not the symplectic matrix of a permutation of the quartrits (for $F$, $\bar M$ swaps $x$ and $z$ within one quartrit).  Images marked ``binary'' are stated in the binary labeling $|2b_1+b_2\rangle\mapsto|b_1b_2\rangle$, which differs from the Gray labeling by one $\mathrm{CNOT}$ within each block, so that the level is the same in both labelings.  The images of $S$, $F$, and $\mathrm{SUM}$ are not Clifford gates.}
    \label{tab:grayimages}
    \vspace{0.5em}
    \footnotesize
    \renewcommand{\arraystretch}{1.3}
    \begin{tabular}{l|c|l|c}
        \hline\hline
        Quartrit gate & $\bar M$ & Image under $\phi_Z$ & Level \\
        \hline
        $X^2$, $Z^2$ & $I$ & $\sigma_x\otimes\sigma_x$, $\sigma_z\otimes\sigma_z$ & 1 \\
        $X$, $Z$ & $I$ & $(I\otimes\sigma_x)\,\mathrm{SWAP}$, $(S_2^\dagger\otimes S_2)\,\mathrm{CZ}$ & 2 \\
        Inversion $R=F^2$ & $I$ & $\mathrm{SWAP}$ & 2 \\
        $S^2$ & $I$ & $(S_2\otimes S_2)\,\mathrm{CZ}$ & 2 \\
        Quartrit permutation & permutation & the same permutation of the blocks & 2 \\
        Phase $S$ & other & $(T\otimes T)\,\mathrm{CS}$; binary: $(\sigma_z\otimes I)\,\mathrm{CZ}\,(I\otimes T)$ & 3 \\
        Fourier $F$ & other & binary: $(H\otimes I)\,\mathrm{CS}\,(I\otimes H)\,\mathrm{SWAP}$, the two-qubit Fourier transform & 3 \\
        $\mathrm{SUM}$ & other & binary: $\mathrm{CNOT}_{b_2\to b_2'}\,\mathrm{CNOT}_{b_1\to b_1'}\,\mathrm{Toffoli}_{b_2b_2'\to b_1'}$, the two-bit adder & 3 \\
        \hline\hline
    \end{tabular}
\end{table}

\begin{thm}[Gray images of $\mathbb{Z}_4$ Clifford gates] \label{thm:graylevels}
    Let $U$ be a $\mathbb{Z}_4$ Clifford gate on $n$ quartrits with symplectic matrix $M\in\mathrm{Sp}(2n,\mathbb{Z}_4)$, as in \eqref{eq:cliffordaction}, and let $\bar M$ be the reduction of $M$ modulo two.
    \begin{enumerate}
        \item $\Phi_ZU\Phi_Z^\dagger$ lies in the third level $\mathcal{C}_3$ of the Clifford hierarchy of $2n$ qubits.
        \item $\Phi_ZU\Phi_Z^\dagger$ is a qubit Clifford gate if and only if $\bar M$ is the symplectic matrix of a permutation of the quartrits.
        \item Every quartrit Pauli string maps to a qubit Clifford gate, and every even Pauli string $X^{2a}Z^{2b}$ maps to a qubit Pauli string.
    \end{enumerate}
    The same statements hold for every Pauli-basis Gray map $\Phi_{V,W}=\phi_{V,W}^{\otimes n}$.
\end{thm}

\begin{proof}
    Let $\mathcal{K}_n$ be the set of $\mathbb{Z}_4$ Clifford gates whose symplectic matrix is the identity modulo two.  Such a matrix has the form $M=I+2A$, and $(I+2A)(I+2B)=I+2(A+B)$ modulo four.  Thus, $\mathcal{K}_n$ is a group, and it is abelian modulo Pauli strings and phases.  It is normal in the $\mathbb{Z}_4$ Clifford group, because $M_U(I+2A)M_U^{-1}=I+2M_UAM_U^{-1}$.  For $M=I+2A$, the symplectic condition $M^{\mathsf T}JM=J$ modulo four says that $JA$ is symmetric modulo two.  The symmetric binary matrices are spanned by those with a single nonzero entry on the diagonal or a single symmetric pair of nonzero off-diagonal entries.  Each of these involves at most two of the $2n$ symplectic coordinates, hence at most two quartrits.  Modulo Pauli strings, $\mathcal{K}_n$ is therefore generated by gates on one or two quartrits, namely $R_j$, $S_j^2$, and $F_jS_j^2F_j^\dagger$ on one quartrit $j$, and $\mathrm{SUM}_{jk}^2$, $\mathrm{SUM}_{kj}^2$, $\mathrm{CZ}_{jk}^2$, and $(F_j\otimes F_k)\mathrm{CZ}_{jk}^2(F_j\otimes F_k)^\dagger$ on two quartrits $j$ and $k$.  Here, $\mathrm{CZ}=(I\otimes F)\,\mathrm{SUM}\,(I\otimes F^\dagger)$ is the quartrit controlled-$Z$ gate $|m,m'\rangle\mapsto i^{mm'}|m,m'\rangle$, distinct from the qubit $\mathrm{CZ}$ on the right-hand sides of \eqref{eq:kernelimages}.  Since $\Phi_Z$ is a tensor product, it suffices to compute the images of these generators on one or two blocks.  Label the two qubits of a block by $b_1b_2=\varphi(m)$, so that $m$ is even exactly when $b_1\oplus b_2=0$.  The inversion $R$ exchanges $m=1$ and $m=3$, i.e., the labels $01$ and $10$.  The gate $S^2=\mathrm{diag}(1,i,1,i)$ multiplies $|m\rangle$ by $i^{b_1\oplus b_2}$.  The gate $\mathrm{SUM}^2$ adds $2m$ to the target, which applies $X^2$ to the target when $m$ is odd, and $\mathrm{CZ}^2$ multiplies $|m,m'\rangle$ by $(-1)^{mm'}$, which depends only on the parities of $m$ and $m'$.  Finally, $FS^2F^\dagger=\tfrac{1+i}{2}I+\tfrac{1-i}{2}X^2$ and $(F\otimes F)\mathrm{CZ}^2(F\otimes F)^\dagger=\tfrac12(I+X^2\otimes I+I\otimes X^2-X^2\otimes X^2)$, since $FZ^2F^\dagger=X^2$.  Writing $j_p$ for qubit $p\in\{1,2\}$ of block $j$, we obtain
    \begin{equation}
    \begin{aligned}
        R&\mapsto\mathrm{SWAP}, &
        S^2&\mapsto(S_2\otimes S_2)\,\mathrm{CZ}, \\
        FS^2F^\dagger&\mapsto(H\otimes H)(S_2\otimes S_2)\,\mathrm{CZ}\,(H\otimes H), &
        \mathrm{SUM}_{jk}^2&\mapsto\prod_{p,p'=1}^{2}\mathrm{CNOT}_{j_p\to k_{p'}}, \\
        \mathrm{CZ}_{jk}^2&\mapsto\prod_{p,p'=1}^{2}\mathrm{CZ}_{j_pk_{p'}}, &
        (F_j\otimes F_k)\mathrm{CZ}_{jk}^2(F_j\otimes F_k)^\dagger&\mapsto H^{\otimes4}\prod_{p,p'=1}^{2}\mathrm{CZ}_{j_pk_{p'}}H^{\otimes4},
    \end{aligned} \label{eq:kernelimages}
    \end{equation}
    where $H^{\otimes4}$ acts on the four qubits of blocks $j$ and $k$.  These gates have the qubit Clifford images \eqref{eq:kernelimages}, and the quartrit Pauli strings have the qubit Clifford images listed in Table \ref{tab:grayimages}, which follow from \eqref{eq:graydictionary} (see below), so $\Phi_Z\mathcal{K}_n\Phi_Z^\dagger$ consists of qubit Clifford gates.

    By Lemma \ref{lem:graydictionary}, every qubit Pauli string pulls back to an element of $\mathcal{K}_n$ times a Pauli string: a bit flip on the first or second qubit of block $j$ pulls back to $X_j^{-1}R_j$ or $X_jR_j$, and a phase flip to $Z_j^{-1}S_j^2$ or $Z_jS_j^2$, respectively, since $Z^{-1}S^2=\mathrm{diag}(1,1,-1,-1)$ and $ZS^2=\mathrm{diag}(1,-1,-1,1)$.  For a $\mathbb{Z}_4$ Clifford gate $U$ and a qubit Pauli string $\sigma$, the operator $U(\Phi_Z^\dagger\sigma\Phi_Z)U^\dagger$ therefore lies in $\mathcal{K}_n$ times a Pauli string, by normality, and its image under $\Phi_Z$ is a qubit Clifford gate.  In other words, $(\Phi_ZU\Phi_Z^\dagger)\,\sigma\,(\Phi_ZU\Phi_Z^\dagger)^\dagger$ is a qubit Clifford gate for every qubit Pauli string $\sigma$, which is statement 1.  Statement 3 is the case $M=I$.  The images $X\mapsto(I\otimes\sigma_x)\,\mathrm{SWAP}$, $Z\mapsto(S_2^\dagger\otimes S_2)\,\mathrm{CZ}$, $X^2\mapsto\sigma_x\otimes\sigma_x$, and $Z^2\mapsto\sigma_z\otimes\sigma_z$ follow from \eqref{eq:graydictionary}.

    For statement 2, first suppose that $\bar M$ is the symplectic matrix of a permutation $\pi$ of the quartrits.  Then $U$ is the product of $\pi$, which $\Phi_Z$ maps to the same permutation of the blocks, and an element of $\mathcal{K}_n$, up to a Pauli string.  All three factors map to qubit Clifford gates.  Conversely, suppose that $\Phi_ZU\Phi_Z^\dagger$ is a qubit Clifford gate.  It then normalizes the qubit Pauli group, so $U$ normalizes the pullback of that group, which is generated, up to phases, by $X_j^{-1}R_j$, $X_jR_j$, $Z_j^{-1}S_j^2$, and $Z_jS_j^2$ for $j=1,\ldots,n$.  Their symplectic matrices are $I+2I_j$ for the bit flips and $I+2N_j$ for the phase flips.  Thus, the symplectic matrices of the pullback group are $I+2A$ with $A$ in the binary span $\mathcal{A}$ of the block identities $I_j$ and of the matrices $N_j$ that map $x_j$ to $z_j$.  Conjugation by $\bar M$ must preserve $\mathcal{A}$.  Now $\mathcal{A}$ is a commutative algebra isomorphic to the direct product of $n$ copies of $\mathbb{Z}_2[\epsilon]/(\epsilon^2)$, so its only primitive idempotents are the $I_j$.  An automorphism permutes them, so $\bar M$ maps the coordinate pair of quartrit $j$ onto that of quartrit $\pi(j)$ for some permutation $\pi$.  Moreover, $\bar M N_j\bar M^{-1}$ is a nonzero nilpotent element of $I_{\pi(j)}\mathcal{A}I_{\pi(j)}$, whose only such element is $N_{\pi(j)}$.  The resulting $2\times2$ block $B_j$ of $\bar M$ therefore commutes with $N=\left(\begin{smallmatrix}0&0\\1&0\end{smallmatrix}\right)$.  The invertible binary matrices commuting with $N$ are $I$ and $I+N$, and $I+N$ is the symplectic matrix of the phase gate $S$ modulo two.  Thus, $U$ is the product of the permutation $\pi$, an element of $\mathcal{K}_n$, a Pauli string, and $S_j$ on some set of quartrits.  Since the first three factors have qubit Clifford images, $\Phi_ZU\Phi_Z^\dagger$ is a qubit Clifford gate only if this set is empty, because $\phi_ZS\phi_Z^\dagger=(T\otimes T)\,\mathrm{CS}=\mathrm{diag}(1,e^{i\pi/4},e^{i\pi/4},-1)$ is not a qubit Clifford gate: conjugating the bit flip $\sigma_x\otimes I$ on any one of these blocks already fails to give a Pauli string.  Finally, $\Phi_{V,W}U\Phi_{V,W}^\dagger=W^{\otimes n}\,\Phi_Z(V^{\dagger\otimes n}UV^{\otimes n})\Phi_Z^\dagger\,W^{\dagger\otimes n}$, where $V^{\dagger\otimes n}UV^{\otimes n}$ is again a $\mathbb{Z}_4$ Clifford gate, whose symplectic matrix modulo two is the symplectic matrix of a quartrit permutation exactly when $\bar M$ is, because the symplectic matrix of $V^{\otimes n}$ commutes with that of every quartrit permutation, and conjugation by the qubit Clifford gate $W^{\otimes n}$ preserves every level of the hierarchy.
\end{proof}
 Since conjugation by $\phi_Z$ is a group isomorphism, the $\mathbb{Z}_4$ Clifford group of one quartrit, of order $768$ modulo phases, maps onto a group of two-qubit gates inside the third level of the hierarchy.  By exhaustive computation, we found that $16$ of these gates are Pauli strings, $112$ are other qubit Clifford gates, and $640$ have level exactly three.  The qubit Clifford part consists of the images of the gates with $\bar M=I$, in accordance with Theorem \ref{thm:graylevels}.

\subsection{Qubitizing $\mathbb{Z}_4$ stabilizer codes and transversal gates} \label{subsec:qubitizing}

A stabilizer code $Q$ on $n$ quartrits with codespace projector $\Pi$ is converted, or qubitized, into the code $\Phi_Z(Q)$ on $2n$ qubits with projector $\Phi_Z\Pi\Phi_Z^\dagger$.  This is a quantum version of the binary image $\varphi(\mathcal{C})$ of a $\mathbb{Z}_4$-linear code, and the results below are quantum versions of the classical distance preservation of Ref.~\cite{Hammons1994}.  

We consider the $Z$-basis map and note along the way what changes for the $X$-basis map.  We write $((N,K,d))$ for a qubit code of length $N$ and codespace dimension $K$, not necessarily a stabilizer code, whose distance $d$ is measured in the joint Lee metric of Section \ref{subsec:qubits}, and $((N,K,d_{\mathrm H}))_{\mathrm H}$ when the distance $d_{\mathrm H}$ is measured in the Hamming metric.  For qubit stabilizer codes, we write $\qcode{N,k,d}$ and $\qcode{N,k,d_{\mathrm H}}_{\mathrm H}$ as in Section~\ref{subsec:qubits}, so in this section the subscript ${\mathrm H}$ marks the Hamming metric rather than a qudit dimension.  Recall that an error $E$ is detectable by a code with projector $\Pi$ if $\Pi E\Pi$ is proportional to $\Pi$ \cite{KnillLaflamme1997}, and that the distance of a code with $K\ge2$ is the smallest weight of an undetectable Pauli string.  For stabilizer codes, this agrees with Definition \ref{def:stabilizer}, since the undetectable Pauli strings are the nontrivial logical operators.  The two distances of a qubit code satisfy $d_{\mathrm H}\le d\le 2d_{\mathrm H}$ by \eqref{eq:leevshammingdiamond}.

Table \ref{tab:grayimages} shows that the stabilizer generators of $Q$, which are quartrit Pauli strings, become qubit Clifford gates.  Thus, $\Phi_Z(Q)$ is in general not a stabilizer code but the joint eigenspace of a group of commuting qubit Clifford gates, i.e., a Clifford stabilizer code \cite{Kobayashi2025clifford}.  These gates are monomial matrices, i.e., permutations composed with diagonal phases, so $\Phi_Z(Q)$ belongs to the monomial stabilizer formalism of Ref.~\cite{VandenNest2011monomial}.  If every stabilizer generator is even, i.e., a product of $X_j^2$ and $Z_j^2$ operators, then the generators map to Pauli strings and $\Phi_Z(Q)$ is an ordinary stabilizer code.  This condition is sufficient but not necessary, and Proposition \ref{prop:grayimage} below gives the exact criterion for CSS codes.  We first relate the distances.

\begin{table}[t]
    \centering
    \caption{\textbf{Costs of quartrit Pauli operators under the $Z$-basis Gray map.}  For each single-quartrit Pauli operator $P$, the table lists its Lee weight $|P|_{\mathrm L}$ and its costs $c_\diamond(P)$ and $c_{\mathrm H}(P)$.  These are the smallest joint Lee weight and the smallest Hamming weight, respectively, of a two-qubit Pauli operator whose pullback under $\phi_Z$, expanded in quartrit Pauli operators by Lemma \ref{lem:graydictionary}, contains $P$.  The cost $c_\diamond$ equals the Lee weight except for $X^{\pm1}Z^2$, which is cheaper by two.  For the $X$-basis map, exchange $X$ and $Z$.}
    \label{tab:graycost}
    \vspace{0.5em}
    \small
    \renewcommand{\arraystretch}{1.2}
    \begin{tabular}{l|ccccccc}
        \hline\hline
        $P$ & $I$ & $X^{\pm1}$, $Z^{\pm1}$ & $X^{\pm1}Z^{\pm1}$ & $X^2$, $Z^2$ & $X^{\pm1}Z^2$ & $X^2Z^{\pm1}$ & $X^2Z^2$ \\
        \hline
        $|P|_{\mathrm L}$ & 0 & 1 & 2 & 2 & 3 & 3 & 4 \\
        $c_\diamond(P)$ & 0 & 1 & 2 & 2 & 1 & 3 & 4 \\
        $c_{\mathrm H}(P)$ & 0 & 1 & 1 & 2 & 1 & 2 & 2 \\
        \hline\hline
    \end{tabular}
\end{table}

\begin{thm}[Distance of qubitized stabilizer codes] \label{thm:graydistance}
    Let $Q$ be a stabilizer code on $n$ quartrits with $K\ge2$, and let $c_\diamond$ and $c_{\mathrm H}$ be the costs of Table \ref{tab:graycost}.  Then the joint Lee distance and the Hamming distance of $\Phi_Z(Q)$ satisfy
    \begin{equation}
        d\big(\Phi_Z(Q)\big)\ge\min_{L}\sum_{j=1}^n c_\diamond(L_j)
        \qquad\text{and}\qquad
        d_{\mathrm H}\big(\Phi_Z(Q)\big)\ge\min_{L}\sum_{j=1}^n c_{\mathrm H}(L_j), \label{eq:graydistance}
    \end{equation}
    where the minima run over the nontrivial logical Pauli strings $L=\bigotimes_jL_j$ of $Q$.  In other words, the joint Lee distance of the qubitized code is at least the joint Lee distance $d_\diamond$ of $Q$ recomputed with a discount of two at every site where the logical operator is $X^{\pm1}Z^2$.  For $\Phi_X$, the same holds with $X$ and $Z$ exchanged.
\end{thm}

\begin{proof}
    Let $E$ be a qubit Pauli string.  By Lemma \ref{lem:graydictionary}, its pullback $\Phi_Z^\dagger E\Phi_Z$ is a tensor product over blocks of single-quartrit operators, each of which is a linear combination of quartrit Pauli operators.  Expanding the product gives $\Phi_Z^\dagger E\Phi_Z=\sum_P\gamma_PP$, where every $P=\bigotimes_jP_j$ in the sum has $P_j$ in the expansion of the pullback of the part $E_j$ of $E$ on block $j$.  By the definition of the costs, $c(P_j)\le|E_j|$ in either metric, so $\sum_jc(P_j)\le|E|$.  For a quartrit Pauli string $P$, the operator $\Pi P\Pi$ vanishes if $P$ fails to commute with some stabilizer.  It is proportional to $\Pi$ if $P$ is a stabilizer up to a phase, and it is not proportional to $\Pi$ if $P$ is a nontrivial logical operator.  Hence, if $|E|$ is smaller than the minimum in \eqref{eq:graydistance}, then no nontrivial logical operator appears in the expansion, and $\Pi\Phi_Z^\dagger E\Phi_Z\Pi=\sum_P\gamma_P\Pi P\Pi$ is proportional to $\Pi$.  Every Pauli string of weight below the minimum is thus detectable, which proves \eqref{eq:graydistance}.  The costs in Table \ref{tab:graycost} are read off from the single-block pullbacks.  By \eqref{eq:graydictionary} and $R=\tfrac12(I+X^2+Z^2-X^2Z^2)$, a single bit flip contains $X^{\pm1}$ and $X^{\pm1}Z^2$, a single phase flip contains $Z^{\pm1}$, and a single $\sigma_y$ contains $X^{\pm1}Z^{\pm1}$.  The two-qubit operators $\sigma_x\otimes\sigma_x$, $\sigma_z\otimes\sigma_z$, $\sigma_y\otimes\sigma_y$, $\sigma_x\otimes\sigma_z$, and $\sigma_x\otimes\sigma_y$ contain $X^2$, $Z^2$, $X^2Z^2$, $X^{\pm1}Z^{\pm1}$, and $X^2Z^{\pm1}$, respectively.  Conversely, checking the sixteen two-qubit Pauli operators shows that none of smaller weight, in either metric, contains $P$, so the table entries are minima.  For $\Phi_X$, conjugation by $H$ exchanges $\sigma_x$ and $\sigma_z$ without changing weights, and conjugation by $F$ exchanges $X$ and $Z$ up to inversion.
\end{proof}

\begin{eg}[Qubitizing the smallest quartrit code] \label{eg:grayyparity}
    The $\qcode{1,\log_42,2_\diamond}_4$ code with the single stabilizer $X^2Z^2$ detects every error of joint Lee weight one.  Its logical operators are the Pauli operators $X^aZ^b$ with $a\equiv b\pmod 2$, and the lightest nontrivial ones, $X^{\pm1}Z^{\pm1}$, $X^2$, and $Z^2$, have joint Lee weight two.  Among them, $X^{\pm1}Z^{\pm1}$ has separate Lee weight one, so $d_\diamond=2$ and $d_\square=1$.  By Table \ref{tab:grayimages}, the stabilizer maps to $(\sigma_x\otimes\sigma_x)(\sigma_z\otimes\sigma_z)=-\sigma_y\otimes\sigma_y$.  The qubitized code is thus the two-qubit $Y$-parity code $\qcode{2,1,2}$ of Section \ref{subsec:qubits}, up to the sign of its check, whose joint Lee distance is two and whose Hamming distance is one.  Both bounds of Theorem \ref{thm:graydistance} are attained, since $X^{\pm1}Z^{\pm1}$ costs two in the joint Lee metric and one in the Hamming metric.  The joint Lee metric, which counts a $\sigma_y$ error twice, thus reproduces the distance of the quartrit code, whereas the Hamming metric counts it once and loses a factor of two.
\end{eg}

For CSS codes, the bound of Theorem \ref{thm:graydistance} can be evaluated explicitly, and the distance of the qubitized code is pinned between two nearby quantities.  Let $Q$ be a CSS code on $\mathbb{Z}_4$ with check matrices $H_X$ and $H_Z$.  Its $X$-type stabilizers are the operators $X^s$ with $s$ in the row module $\mathcal{D}=\mathrm{im}(H_X^{\mathsf T})$, and its codewords are the uniform superpositions $|\bar\mu\rangle\propto\sum_{c\in\mathcal{D}}|\mu+c\rangle$ over the cosets of $\mathcal{D}$ in $\ker(H_Z)$.  By Lemma \ref{lem:graydictionary}, a qubit bit-flip string whose pullback has $X$ part $s\in\mathcal{D}$ acts on the codespace as the partial inversion $R_{\mathcal{O}(s)}$ alone, since $X^s$ acts trivially.  Such an error is cheap, of weight $|s|_{\mathrm L}$, and it goes undetected exactly when $\Pi R_{\mathcal{O}(s)}\Pi$ is not proportional to $\Pi$.  The quantity $\mu_X$ below measures this failure, and the quantity $\nu_X$ evaluates the bound of Theorem \ref{thm:graydistance}.  For a set $\mathcal{A}$ of coordinates, $z_{\mathcal{A}}$ denotes the vector $z$ with the coordinates outside $\mathcal{A}$ set to zero, so that $z=z_{\mathcal{A}}+z_{\mathcal{A}^{\mathrm c}}$.

\begin{thm}[Distance of qubitized CSS codes] \label{thm:graycss}
    Let $Q$ be a CSS code on $\mathbb{Z}_4$ with $K\ge2$, distances $d_X$ and $d_Z$, codespace projector $\Pi$, and $X$-check module $\mathcal{D}=\mathrm{im}(H_X^{\mathsf T})$.  Define
    \begin{subequations} \label{eq:munu}
    \begin{align}
        \mu_X &= \min\big\{|s|_{\mathrm L} : s\in\mathcal{D},\ \Pi R_{\mathcal{O}(s)}\Pi\not\propto\Pi\big\}, \label{eq:muX}\\
        \nu_X &= \min\big\{|s|_{\mathrm L}+|z_{\supp(s)^{\mathrm c}}|_{\mathrm L} : s\in\mathcal{D}\setminus0,\ z\in\ker(H_X)\setminus\mathrm{im}(H_Z^{\mathsf T})\big\}, \label{eq:nuX}
    \end{align}
    \end{subequations}
    with the minimum of an empty set taken to be $\infty$.  Thus, $\mu_X=\infty$ if no $s$ qualifies, and $\nu_X=\infty$ if $\mathcal{D}=\{0\}$.  Then, in the joint Lee and Hamming metrics,
    \begin{equation}
        \min(d_X,d_Z,\nu_X)\le d_{\mathrm H}\big(\Phi_Z(Q)\big)\le d\big(\Phi_Z(Q)\big)\le\min(d_X,d_Z,\mu_X). \label{eq:graycss}
    \end{equation}
    Moreover, $\nu_X\ge\lceil d_Z/2\rceil$, and $\nu_X$ is at least the minimum Lee weight of a nonzero vector in $\mathcal{D}$.
\end{thm}

\begin{proof}
    The lower bound follows from Theorem \ref{thm:graydistance}.  A nontrivial logical operator of a CSS code is $X^tZ^z$ with $t\in\ker(H_Z)$ and $z\in\ker(H_X)$, where $t\notin\mathcal{D}$ or $z\notin\mathrm{im}(H_Z^{\mathsf T})$.  By Table \ref{tab:graycost}, $c_{\mathrm H}(X^{t_j}Z^{z_j})$ equals $|t_j|_{\mathrm L}$ whenever $t_j\neq0$, and equals $|z_j|_{\mathrm L}$ otherwise.  If $t\notin\mathcal{D}$, the Hamming cost of $X^tZ^z$ is at least $|t|_{\mathrm L}\ge d_X$.  If $t=0$, it is $|z|_{\mathrm L}\ge d_Z$.  If $t\in\mathcal{D}\setminus0$ and $z\notin\mathrm{im}(H_Z^{\mathsf T})$, it is $|t|_{\mathrm L}+|z_{\supp(t)^{\mathrm c}}|_{\mathrm L}\ge\nu_X$.  This proves the first inequality in \eqref{eq:graycss}, and the second is \eqref{eq:leevshammingdiamond}.

    For the upper bound, we use pure bit-flip and phase-flip strings, whose Hamming and joint Lee weights coincide.  First, let $Z^z$ be a nontrivial $Z$-type logical operator with $|z|_{\mathrm L}=d_Z$.  By \eqref{eq:graydictionaryz} read in reverse, $Z\mapsto\tfrac{1+\ii}{2}\,\sigma_z\otimes I+\tfrac{1-\ii}{2}\,I\otimes\sigma_z$ and $Z^2\mapsto\sigma_z\otimes\sigma_z$, so $\Phi_ZZ^z\Phi_Z^\dagger$ is a combination of phase-flip strings of Hamming weight $d_Z$.  Since $\Pi Z^z\Pi$ is not proportional to $\Pi$, at least one of these strings is undetectable.  Second, let $X^t$ be a nontrivial $X$-type logical operator with $|t|_{\mathrm L}=d_X$, and let $\sigma_x^{a}$ be the bit-flip string of Hamming weight $d_X$ whose pullback is $X^tR_{\mathcal{O}(t)}$.  Applied to the codeword $|\bar0\rangle$, the pullback yields $|\mathcal{D}|^{-1/2}\sum_{c\in\mathcal{D}}|t+c'\rangle$, where $c'$ is $c$ with the coordinates in $\mathcal{O}(t)$ negated.  The term with $c=0$ shows that this vector has a nonzero overlap with the codeword $|\bar t\rangle$, which is orthogonal to $|\bar0\rangle$ because $t\notin\mathcal{D}$.  Hence $\sigma_x^a$ is undetectable.  Third, if $\mu_X<\infty$, let $s\in\mathcal{D}$ attain $\mu_X$, and let $\sigma_x^{a}$ be the bit-flip string whose pullback is $X^sR_{\mathcal{O}(s)}$.  Since $X^{-s}$ is a stabilizer, $\Pi X^s=(X^{-s}\Pi)^\dagger=\Pi$, and $\Pi X^sR_{\mathcal{O}(s)}\Pi=\Pi R_{\mathcal{O}(s)}\Pi$ is not proportional to $\Pi$.

    For the last statement, the case $\mathcal{D}=\{0\}$ follows from the convention on empty minima.  Otherwise, let $s$ and $z$ attain $\nu_X$, and let $\mathcal{T}=\supp(s)$.  Then $|z_{\mathcal{T}^{\mathrm c}}|_{\mathrm L}\ge d_Z-|z_{\mathcal{T}}|_{\mathrm L}\ge d_Z-2|\mathcal{T}|$ and $|s|_{\mathrm L}\ge|\mathcal{T}|$, so $\nu_X\ge\max(|\mathcal{T}|,d_Z-|\mathcal{T}|)\ge d_Z/2$.  The bound $\nu_X\ge\min_{s\in\mathcal{D}\setminus0}|s|_{\mathrm L}$ is immediate from \eqref{eq:nuX}.
\end{proof}

\noindent Together with \eqref{eq:graycss}, the last statement gives $d_{\mathrm H}\big(\Phi_Z(Q)\big)\ge\min(d_X,\lceil d_Z/2\rceil)$: the $Z$-basis map loses at most half of $d_Z$, and nothing at all when $d_X\le\lceil d_Z/2\rceil$.  Consequently, if every nonzero vector of $\mathcal{D}$ has Lee weight at least $\min(d_X,d_Z)$, then $\nu_X\ge\min(d_X,d_Z)$, and $\Phi_Z(Q)$ has distance $\min(d_X,d_Z)$ in both metrics.  This holds in particular for pure codes, i.e., codes all of whose nontrivial stabilizers have Lee weight at least $\min(d_X,d_Z)$.  By the upper bound, the qubitized distance falls below $\min(d_X,d_Z)$ whenever $\mu_X<\min(d_X,d_Z)$, in contrast with the classical isometry. 

Besides pure codes, the following class of codes also qubitizes without loss.
Following Ref.~\cite{DasuBurton2025}, we call a CSS code \emph{semi-self-dual} if every $X$-check vector is also a $Z$-check vector, i.e., if $\mathrm{im}(H_X^{\mathsf T})\subseteq\mathrm{im}(H_Z^{\mathsf T})$.  Self-dual CSS codes, with $H_X=H_Z$, are the symmetric case.

\begin{cor}[Semi-self-dual codes qubitize without loss] \label{cor:semiselfdual}
    If $\mathrm{im}(H_X^{\mathsf T})\subseteq\mathrm{im}(H_Z^{\mathsf T})$, then $\nu_X\ge d_Z$, and the qubitized code $\Phi_Z(Q)$ has joint Lee distance and Hamming distance both equal to $\min(d_X,d_Z)$.  If instead $\mathrm{im}(H_Z^{\mathsf T})\subseteq\mathrm{im}(H_X^{\mathsf T})$, the same holds for $\Phi_X(Q)$.
\end{cor}

\begin{proof}
    Let $s\in\mathcal{D}\setminus0$ and $z\in\ker(H_X)\setminus\mathrm{im}(H_Z^{\mathsf T})$, and let $\mathcal{T}=\supp(s)$.  By assumption, $Z^s$ is a stabilizer, so $Z^{z+as}$ is a nontrivial $Z$-type logical operator for every $a\in\mathbb{Z}_4$, and $|z+as|_{\mathrm L}\ge d_Z$.  On each coordinate $j\in\mathcal{T}$, the four values $z_j+as_j$ have Lee weights summing to four: they exhaust $\mathbb{Z}_4$ if $s_j$ is odd, and they are $z_j$ and $z_j+2$ twice each if $s_j=2$.  Summing over $\mathcal{T}$ and averaging over $a$, we find some $a$ with $|(z+as)_{\mathcal{T}}|_{\mathrm L}\le|\mathcal{T}|\le|s|_{\mathrm L}$.  Since $z+as$ agrees with $z$ outside $\mathcal{T}$,
    \begin{equation}
        d_Z\le|z+as|_{\mathrm L}\le|z_{\mathcal{T}^{\mathrm c}}|_{\mathrm L}+|s|_{\mathrm L},
    \end{equation}
    which gives $\nu_X\ge d_Z$.  Theorem \ref{thm:graycss} then yields $\min(d_X,d_Z)\le d_{\mathrm H}\le d\le\min(d_X,d_Z)$.  The mirror statement follows by applying $F^{\otimes n}$, which exchanges the two check matrices up to inversion.
\end{proof}

Other Pauli-basis maps can push the qubit distance above the Lee distance when the two distances of a CSS code differ.  Consider the Bell-basis map $\phi_{I,W}$ with $W=\mathrm{CNOT}\,(H\otimes I)$, which sends $|0\rangle,|1\rangle,|2\rangle,|3\rangle$ to the Bell states $|\Phi^+\rangle,|\Psi^+\rangle,|\Psi^-\rangle,|\Phi^-\rangle$.  It maps $Z$ to $\tfrac{1+\ii}{2}\,\sigma_x\otimes\sigma_x+\tfrac{1-\ii}{2}\,\sigma_z\otimes\sigma_z$, so a quartrit phase shift becomes a correlated two-qubit error, and no single-qubit Pauli operator pulls back to a combination containing $Z^{\pm1}$.  In the language of Table \ref{tab:graycost}, the cost of $Z^{\pm1}$ becomes two, whereas the cost of $X^{\pm1}$ stays one, since a single bit flip still pulls back to a combination containing $X^{\pm1}$.  For the seven-quartrit CSS code with the single $X$-check $X^{\otimes7}$ and with $Z$-checks $Z^z$ for all $z$ orthogonal to both $(1,1,1,1,1,1,1)$ and $(1,1,1,0,0,0,0)$, we have $K=4$, $d_X=3$, and $d_Z=2$.  The $Z$-basis map yields a $((14,4,2))_{\mathrm H}$ code, while the Bell-basis map yields a $((14,4,3))_{\mathrm H}$ code, by exhaustive enumeration.  The gain is thus available when the $Z$-distance is the smaller one, and the $X$-basis analogue of the Bell map serves the opposite case.

Self-dual CSS codes are semi-self-dual, so every code of Examples \ref{eg:liftedhamming} and \ref{eg:liftedgolay} with $q=4$, and more generally every code obtained by shortening a self-dual $\mathbb{Z}_4$ code as in Section \ref{subsec:smallexamples}, qubitizes without loss.  Before listing these codes, we settle when the qubitized code is a stabilizer code.

\begin{prop}[When the qubitized code is a stabilizer code] \label{prop:grayimage}
    Let $Q$ be a CSS code on $\mathbb{Z}_4$ with $\mathcal{D}=\mathrm{im}(H_X^{\mathsf T})$ and $\mathcal{C}=\ker(H_Z)$.  Then $\Phi_Z(Q)$ is a qubit stabilizer code if and only if
    \begin{enumerate}
        \item $2\,u\ast v\in\mathcal{C}$ for all $u,v\in\mathcal{C}$, i.e., $\varphi(\mathcal{C})$ is linear, and
        \item $2\,u\ast c\in\mathcal{D}$ for all $u\in\mathcal{C}$ and $c\in\mathcal{D}$.
    \end{enumerate}
    In that case, $\Phi_Z(Q)$ is the qubit CSS code with $X$-type checks generating $\varphi(\mathcal{D})$ and $Z$-type checks generating $\varphi(\mathcal{C})^\perp$, and, if $K\ge2$, its distance is $\min(d_X,d_Z)$ in both metrics.  For $\Phi_X$, exchange the roles of $H_X$ and $H_Z$.
\end{prop}

\begin{proof}
    The codewords of $Q$ map to $|\mathcal{D}|^{-1/2}\sum_{c\in\mathcal{D}}|\varphi(\mu+c)\rangle$, uniform superpositions over the classes $\varphi(\mu+\mathcal{D})$, which partition $\varphi(\mathcal{C})$.  Applying the identity $\varphi(u+v)=\varphi(u)\oplus\varphi(v)\oplus\varphi(2\,u\ast v)$ twice gives $\varphi(\mu+c)=\varphi(\mu)\oplus\varphi(c+2\,\mu\ast c)$.  If condition 2 holds, then $c\mapsto c+2\,\mu\ast c$ maps $\mathcal{D}$ into itself and is an involution, since $(1+2\mu_j)^2=1$ modulo four on each coordinate.  It is therefore a bijection of $\mathcal{D}$, so $\varphi(\mu+\mathcal{D})=\varphi(\mu)\oplus\varphi(\mathcal{D})$, and condition 2 with $u\in\mathcal{D}$ makes $\varphi(\mathcal{D})$ linear.  With condition 1, $\varphi(\mathcal{C})$ is a linear code containing $\varphi(\mathcal{D})$, the classes are the cosets of $\varphi(\mathcal{D})$ in $\varphi(\mathcal{C})$, and the codewords are exactly those of the qubit CSS code built from $\varphi(\mathcal{D})\subseteq\varphi(\mathcal{C})$.

    Conversely, suppose that $\Phi_Z(Q)$ is a stabilizer code.  The set of computational basis states on which its projector has a nonzero diagonal entry is the set of strings on which all $Z$-type stabilizers have eigenvalue $+1$, which is an affine subspace.  Here this set is $\varphi(\mathcal{C})$, which contains $0$, so $\varphi(\mathcal{C})$ is linear, which is condition 1.  Moreover, for two strings $x$ and $y$ in this set, the matrix element $\langle x|\Pi|y\rangle$ of a stabilizer code is nonzero if and only if $x\oplus y$ is the bit-flip part of a stabilizer.  Here it is nonzero if and only if $x$ and $y$ lie in the same class.  Every class $\varphi(\mu+\mathcal{D})$ is therefore a coset of one fixed subspace, namely $\varphi(\mathcal{D})$, so $\varphi(\mu+c)\oplus\varphi(\mu)=\varphi(c)\oplus\varphi(2\,\mu\ast c)$ lies in $\varphi(\mathcal{D})$ for all $c\in\mathcal{D}$.  Hence $\varphi(2\,\mu\ast c)\in\varphi(\mathcal{D})$, which is condition 2.

    For the distance, assume $K\ge2$.  The $X$-type logical operators of the qubit CSS code are the strings in $\varphi(\mathcal{C})\setminus\varphi(\mathcal{D})$, whose minimal Hamming weight is $d_X$ because $\varphi$ is an isometry.  Its $Z$-type distance is at most $d_Z$ by the phase-flip construction in the proof of Theorem \ref{thm:graycss}, and at least $d_Z$ because a phase-flip string of Hamming weight below $d_Z$ pulls back to a combination of $Z^k$ with $|k|_{\mathrm L}<d_Z$, each of which is detectable or a stabilizer.
\end{proof}

\noindent Proposition \ref{prop:grayimage} has instances beyond the codes with even checks.  The smallest is the $\qcode{2,\log_42,2}_4$ code with $H_X=(1\ \,1)$ and $H_Z=(2\ \,2)$.  Here $\varphi(\mathcal{D})$ is spanned by $0101$ and $1010$, $\varphi(\mathcal{C})$ is the even-weight code, and the qubitized code is the $\qcode{4,1,2}_{\mathrm H}$ Chuang--Leung--Yamamoto code \cite{Chuang1997bosonic}, up to Hadamard gates on all qubits and an exchange of the middle two qubits, with stabilizers $\sigma_x\otimes I\otimes\sigma_x\otimes I$, $I\otimes\sigma_x\otimes I\otimes\sigma_x$, and $\sigma_z^{\otimes4}$.  Larger instances arise from the self-dual codes over $\mathbb{Z}_4$ in the database of Ref.~\cite{HaradaMunemasa}.  A self-dual code of length $16$ and Lee distance $8$ yields a $\qcode{15,1,6}_4$ code whose image is a $\qcode{30,2,6}_{\mathrm H}$ CSS code, and a Type II code of length $24$ and Lee distance $8$ yields a $\qcode{23,1,6}_4$ code whose image is a $\qcode{46,2,6}_{\mathrm H}$ CSS code.  The $\qcode{15,1,6}_4$ code has the same parameters as the code of Table \ref{tab:qubitized} at length $16$, whose image is nonadditive.

\begin{table}[t]
    \centering
    \caption{\textbf{Qubitized self-dual CSS codes.}  Each row starts from a Type II self-dual code over $\mathbb{Z}_4$ of length $n$ and Lee distance $d_{\mathrm L}$, i.e., a self-dual code all of whose Euclidean weights are divisible by eight \cite{Bonnecaze1995,Dougherty1999typeII}.  Shortening the classical code at one or at two coordinates yields the self-dual CSS codes in the last two columns, and the $Z$-basis Gray map turns them into the qubit codes after the arrows.  The quantum distances were computed by exhaustive enumeration and do not depend on the choice of the shortened coordinates.  By Corollary \ref{cor:semiselfdual}, the Hamming distance of each qubit code, marked by the subscript, equals the Lee distance of the corresponding quaternary quantum code, and so does its joint Lee distance.  For each shortened construction, the corresponding modules fail the stabilizer criterion of Proposition \ref{prop:grayimage}.  
    For the octacode, the quantum codes relate to its Gray image, the Nordstrom--Robinson code \cite{Hammons1994}.
    All codes in the table admit the two-fold transversal gates of \eqref{eq:qubitizedgates}.}
    \label{tab:qubitized}
    \vspace{0.5em}
    \small
    \renewcommand{\arraystretch}{1.2}
    \begin{tabular}{l|c|c|c}
        \hline\hline
        Classical code & $(n,d_{\mathrm L})$ & One logical quartrit & Two logical quartrits \\
        \hline
        Octacode \cite{ConwaySloane1993,Hammons1994} & $(8,6)$ & $\qcode{7,1,4}_4\to((14,4,4))_{\mathrm H}$ & $\qcode{6,2,3}_4\to((12,16,3))_{\mathrm H}$ \\
        Type II code found by search & $(16,8)$ & $\qcode{15,1,6}_4\to((30,4,6))_{\mathrm H}$ & $\qcode{14,2,4}_4\to((28,16,4))_{\mathrm H}$ \\
        $\mathbb{Z}_4$ Golay code \cite{Bonnecaze1995} & $(24,12)$ & $\qcode{23,1,10}_4\to((46,4,10))_{\mathrm H}$ & $\qcode{22,2,8}_4\to((44,16,8))_{\mathrm H}$ \\
        Extended $\mathbb{Z}_4$ quadratic residue code \cite{Bonnecaze1995} & $(32,14)$ & $\qcode{31,1,12}_4\to((62,4,12))_{\mathrm H}$ & $\qcode{30,2,11}_4\to((60,16,11))_{\mathrm H}$ \\
        \hline\hline
    \end{tabular}
\end{table}

Table \ref{tab:qubitized} lists the qubitized codes obtained from one Type II self-dual $\mathbb{Z}_4$ code at each of the lengths $8$, $16$, $24$, and $32$, i.e., a self-dual code all of whose Euclidean weights are divisible by eight \cite{Bonnecaze1995,Dougherty1999typeII}. Length $8$ gives the octacode and length $24$ the $\mathbb{Z}_4$ Golay code, corresponding to quaternary quantum codes from Examples \ref{eg:liftedhamming} and \ref{eg:liftedgolay}.  For length $32$, we use the extended quadratic residue code over $\mathbb{Z}_4$ \cite{Bonnecaze1995}, i.e., the Hensel lift of the binary $[31,16,7]$ quadratic residue code extended by a parity coordinate.  The length-$16$ entry is a Type II code of Lee distance $8$ that we found by a random search over generator matrices of the form $(I_8\mid A)$.  Its Lee distance is the largest possible at this length, because the Lee weight of a vector never exceeds its Euclidean weight and Type II codes of length $16$ have minimum Euclidean weight $8$ \cite{Dougherty1999typeII}.

For each classical code of length $n$ and Lee distance $d_{\mathrm L}$, shortening at one coordinate yields an $\qcode{n-1,1,d}_4$ code with $d\ge d_{\mathrm L}-2$, and shortening at two coordinates yields an $\qcode{n-2,2,d'}_4$ code with $d'\ge d_{\mathrm L}-4$, as in Examples \ref{eg:liftedhamming} and \ref{eg:liftedgolay}.  The bounds hold because puncturing removes Lee weight at most two per coordinate.   All quantum distances carry over to the qubit codes by Corollary \ref{cor:semiselfdual}.  

The transversal gates of these codes translate to level-three hierarchy gates on the qubitized codes.  By Section \ref{sec:ft}, every self-dual CSS code on $n$ quartrits admits $G^{\otimes n}$ as a logical gate for every single-quartrit $\mathbb{Z}_4$ Clifford gate $G$, as well as $\mathrm{SUM}^{\otimes n}$ between two code blocks.  The codes of Table \ref{tab:qubitized} need no Pauli correction for $S^{\otimes n}$, because their classical codes are Type II.  Under $\Phi_Z$, a transversal gate becomes a depth-one circuit of gates, each acting within one block of two qubits or, for $\mathrm{SUM}^{\otimes n}$, on a pair of corresponding blocks in the two codes, which we call a \emph{two-fold transversal} gate \cite{ChakrabortyGottesman2026nogo,Holmes2026logic,Albert2026beyond}.  By Table \ref{tab:grayimages}, the generators translate into
\begin{equation}
\begin{aligned}
    F^{\otimes n}&\mapsto\big[\mathrm{CNOT}\,(H\otimes I)\,\mathrm{CS}\,(I\otimes H)\,\mathrm{SWAP}\,\mathrm{CNOT}\big]^{\otimes n}, \\
    S^{\otimes n}&\mapsto\big[(T\otimes T)\,\mathrm{CS}\big]^{\otimes n}, \\
    R^{\otimes n}&\mapsto\mathrm{SWAP}^{\otimes n}, \\
    \mathrm{SUM}^{\otimes n}&\mapsto\big[(\mathrm{CNOT}\otimes\mathrm{CNOT})\,\mathrm{ADD}\,(\mathrm{CNOT}\otimes\mathrm{CNOT})\big]^{\otimes n},
\end{aligned} \label{eq:qubitizedgates}
\end{equation}
where each $\mathrm{CNOT}$ acts within a block.  The gate $\mathrm{ADD}=\mathrm{CNOT}_{2\to4}\,\mathrm{CNOT}_{1\to3}\,\mathrm{Toffoli}_{24\to3}$ is the two-bit adder, acting on the qubits $1,2$ of a block in the control code and the qubits $3,4$ of the corresponding block in the target code.  The logical Pauli operators translate in the same way, with $X\mapsto(I\otimes\sigma_x)\,\mathrm{SWAP}$ and $Z\mapsto(S_2^\dagger\otimes S_2)\,\mathrm{CZ}$ on each block.

\section{Lee-LDPC codes} \label{sec:LDPC}

Quantum LDPC codes are widely studied in the theory of binary quantum stabilizer codes. A family of binary quantum stabilizer codes is LDPC if one can choose a generating set for the stabilizer such that every generator acts nontrivially on at most a bounded number of physical qubits and every physical qubit is acted on nontrivially by at most a bounded number of generators, where both bounds are independent of the code length $n$. Quantum LDPC codes are of interest not only because bounded-weight stabilizer generators can be easier to measure than high-weight generators, but, more importantly, because the LDPC condition is closely connected to fault-tolerant quantum error correction.

One direct manifestation of this connection is error propagation during ancilla-based syndrome extraction. During the measurement of a stabilizer generator, an error on the ancilla can propagate through subsequent two-qubit gates into a correlated error on multiple data qubits; such an error is known as a hook error~\cite{Dennis2002topological}. Because every stabilizer generator of an LDPC code has weight bounded independently of $n$, the weight of a hook error from a single stabilizer is also bounded. The LDPC condition is therefore crucial for fault tolerance because it prevents a single ancilla error during stabilizer measurement from propagating into an error whose weight grows with $n$.

For codes over $\mathbb Z_q$, defining the LDPC condition solely in terms of Hamming weight is insufficient to control error propagation in the Lee metric: a stabilizer may have bounded Hamming weight and unbounded Lee weight. As a result, hook errors may have Lee weight that grows with $q$, and hence with $n$ whenever $q$ grows with $n$. Consequently, we need the following extension of the definition of an LDPC code to a Lee-LDPC code.

\subsection{Quantum Lee-LDPC codes}

We first formulate the Lee-LDPC condition at the level of a parity-check
matrix. Let $H\in\mathbb Z_q^{m\times n}$. We say that $H$ is a
$(w_{\mathrm B},w_{\mathrm C})$-Lee-LDPC matrix if
\begin{subequations}\label{eq:classicalLee}
        \begin{align}
            \sum_{j=1}^n |H_{ij}|_{\mathrm{L}} &\le w_{\mathrm{C}} \text{ for each $i$}, \\
            \sum_{i=1}^m |H_{ij}|_{\mathrm{L}} &\le w_{\mathrm{B}} \text{ for each $j$}.
        \end{align}
\end{subequations}
A classical linear code is a $(w_{\mathrm B},w_{\mathrm C})$-Lee-LDPC
code if it admits a parity-check matrix satisfying
Eq.~\eqref{eq:classicalLee}. A family of such matrices (or codes) is Lee-LDPC if $w_{\mathrm B}$ and $w_{\mathrm C}$ can be chosen independent of $n$ and $q$. Low-Lee-density parity-check codes, whose parity-check matrices over $\mathbb{Z}_q$ have rows and columns of low Lee weight, were introduced in Ref.~\cite{Santini2020LLDPC}.

\begin{defn}[Quantum Lee-LDPC code]
A quantum CSS code specified by
$H_X:\mathbb Z_q^n\rightarrow\mathbb Z_q^{m_X}$ and
$H_Z:\mathbb Z_q^n\rightarrow\mathbb Z_q^{m_Z}$ is a
$(w_{\mathrm B},w_{\mathrm C})$-Lee-LDPC code if both $H_X$ and $H_Z$
satisfy Eq.~\eqref{eq:classicalLee}.
\end{defn}

Whenever we write a matrix over $\mathbb Z_q$, we implicitly choose for each
entry its unique representative in
\(
    (-q/2,q/2]\cap \mathbb Z
\)
and regard the matrix as written in these representatives. When we lift such a
matrix from $\mathbb Z_q$ to $\mathbb Z$, we keep these integer entries
unchanged.

\begin{prop} \label{prop:commteoverZ}
    For any $(w_{\mathrm{B}},w_{\mathrm{C}})$-Lee-LDPC CSS code defined over $\mathbb{Z}_q$, if $q > w_{\mathrm{C}}^2$ and we lift each entry of $H_X$ and $H_Z$ to $\mathbb{Z}$, then $H_XH_Z^{\mathsf{T}}=0$ holds over $\mathbb{Z}$, not only $\mathbb{Z}_q$.
\end{prop}
\begin{proof}
Since both $X$-checks and $Z$-checks have Lee weight at most $w_{\mathrm{C}}$, over $\mathbb{Z}$ we know that the matrix $H_XH_Z^{\mathsf{T}}$ is entrywise bounded:  $|(H_XH_Z^{\mathsf{T}})_{a,b}|\le w_{\mathrm{C}}^2$.  
On the other hand, the CSS commutation condition says that each entry $(H_XH_Z^{\mathsf T})_{a,b}$, i.e., the inner product of an $X$-check with a $Z$-check,
vanishes modulo $q$. Hence, when $q>w_{\mathrm C}^2$, it must vanish exactly, even over $\mathbb{Z}$.
\end{proof}

Henceforth in this paper, we will usually simply refer to ``Lee-LDPC'' codes and drop the $(w_{\mathrm{B}},w_{\mathrm{C}})$-prefix.   In this setting, unless otherwise stated, it is assumed that $q>w_{\mathrm{C}}^2$.  Indeed, the limit where $q$ is large while $w_{\mathrm{C}}$ stays $O(1)$ is, potentially, the most interesting regime for Lee-LDPC codes.

Exact commutation over $\mathbb Z$ makes the integer chain structure
considerably simpler. In this regime, the logical decomposition can be obtained
directly from the separate Smith normal forms of $H_X$ and $H_Z$, without
constructing the larger matrices used in Theorem~\ref{thm:quantumcodespace}. Exact commutation also underlies homological rotor codes~\cite{Vuillot2024rotor} and local-dimension-invariant codes, which remain valid for every prime local dimension~\cite{Gunderman2020LDI,MoorthyGunderman2023LDICSS}.  For qudit CSS codes from cell complexes, in which $H_X$ contributes no torsion, Nov\'ak derived the same decomposition from the universal coefficient theorem~\cite{Novak2024torsion}.

\begin{prop}[Logical decomposition of integer-compatible CSS codes] \label{prop:decom}
For any quantum CSS code defined by parity-check matrices $H_X$ and $H_Z$, if $H_X H_Z^\mathsf{T}=0$ over $\mathbb{Z}$, we can directly obtain the decomposition of the codespace $\mathcal{C}$ from the SNFs of $H_X$ and $H_Z$,
\begin{align}\label{eq:SNF-H}
  U_X H_X V_X &= D_X
  = \operatorname{diag}\bigl(d_1^X,\ldots,d_{r_X}^X,0,\ldots,0\bigr), \nonumber\\
  U_Z H_Z V_Z &= D_Z
  = \operatorname{diag}\bigl(d_1^Z,\ldots,d_{r_Z}^Z,0,\ldots,0\bigr),
\end{align}
where $r_{X/Z}$ is the rank of $H_{X/Z}$. The codespace is isomorphic to   \begin{align}\label{eq:C-decomp_integer}
    \mathcal{C}
    \cong
    \mathbb{Z}_q^{\,n - r_X - r_Z}
    \;\oplus\;
    \bigoplus_{j=1}^{r_Z}
      \mathbb{Z}_{\gcd(d_{j}^Z,\,q)}
    \;\oplus\;
    \bigoplus_{i=1}^{r_X}
      \mathbb{Z}_{\gcd(d_{i}^X,\,q)}.
  \end{align}
We note that any factor $\mathbb{Z}_{\gcd(d_i^{X/Z},q)}$ with $\gcd(d_i^{X/Z},q)=1$ is trivial and can be omitted. 

There exist logical dual bases $x_L$ and $z_L$ such that
their pairing is diagonal, the entry for a cyclic summand $\mathbb Z_g$ being $q/g$ (so $1$ for each free summand), as in \eqref{eq:logicalConjugates}. Thus each cyclic summand in the
decomposition is associated with a unique logical pair. The dual bases also satisfy the following
properties. For each $i$ with $\gcd(d_i^X,q)>1$, the corresponding logical $X$ operator, denoted $x_{L,a}$,
lies in the real linear span of $X$ stabilizers, while its dual logical
$z_{L,a}$, viewed as an integer vector, satisfies
$H_X z_{L,a} = q c_a$ for some nonzero integer vector $c_a$. Similarly,
for each $j$ with $\gcd(d_j^Z,q)>1$, the corresponding logical $Z$ operator, denoted $z_{L,b}$,
lies in the real linear span of $Z$ stabilizers, while its dual
logical $x_{L,b}$ satisfies
$H_Z x_{L,b} = q c_b$ for some nonzero integer vector $c_b$.

\end{prop} 
\begin{proof}
The proof is provided in Appendix \ref{app:codespace_integer}. The same argument in fact applies more generally: it is enough that there exist
integer matrices congruent modulo $q$ to the present $H_X$ and $H_Z$ and
satisfying exact commutation over $\mathbb Z$. Replacing $H_X$ and $H_Z$
throughout the proof by those integer matrices gives exactly the same
decomposition result. To avoid introducing additional notation, however, we
state and use only the simpler version here.
\end{proof}

The decomposition in the preceding proposition also gives Lee-weight bounds for the logical operators associated with each torsion summand $\mathbb Z_{\gcd(d_i^{X/Z},q)}$ with $\gcd(d_i^{X/Z},q)>1$ in Eq.~\eqref{eq:C-decomp_integer}.

\begin{cor} \label{cor:reducedweightbound}

For any quantum $(w_{\mathrm{B}},w_{\mathrm{C}})$-Lee-LDPC code with $H_X H_Z^\mathsf{T}=0$ over $\mathbb{Z}$, let $\alpha_{X/Z}\le w_{\mathrm C}$ denote the largest Lee weight of a single entry of $H_{X/Z}$. Then the logical dual bases $x_L$ and $z_L$ described in Proposition \ref{prop:decom} have the following properties. For each $i$ with $\gcd(d_i^X,q)>1$, the reduced Lee weight (the smallest Lee weight up to stabilizers) of the corresponding logical $X$ operator $x_{L,a}$ is at most $w_{\mathrm C} n /2$, while the reduced Lee weight of its dual logical $z_{L,a}$ is at least $q/\alpha_X$. For each $j$ with $\gcd(d_j^Z,q)>1$, the reduced Lee weight of the corresponding logical $Z$ operator $z_{L,b}$ is at most $w_{\mathrm C} n /2$, while the reduced Lee weight of the dual logical $x_{L,b}$ is at least $q/\alpha_Z$.

\end{cor}

\begin{proof}
For every \(i\) with \(\gcd(d_i^X,q)>1\), the corresponding basis element \(x_{L,a}\) lies in \(\operatorname{span}_{\mathbb{R}}(H_X^{\mathsf T})\). Choose \(r_X\leq n\) linearly independent original \(X\)-check rows \(h_1,\ldots,h_{r_X}\) spanning this real row space, and write
\begin{equation}
x_{L,a}=\sum_{k=1}^{r_X}\lambda_k h_k.
\end{equation}
For each \(k\), choose an integer \(m_k\) such that \(|\lambda_k-m_k|\le 1/2\). Subtracting the stabilizer combination \(\sum_k m_k h_k\) gives a representative of the same logical class with Lee weight at most
\begin{equation}
\frac{1}{2}\sum_{k=1}^{r_X}|h_k|_{\mathrm{L}}
\leq
\frac{w_{\mathrm{C}}r_X}{2}
\leq
\frac{w_{\mathrm{C}}n}{2}.
\end{equation}

Now consider its dual logical $z_{L,a}$ that satisfies $H_X z_{L,a} = q c_a$. Here $c_a$ is a nonzero integer vector for any representative of $z_{L,a}$. If the reduced Lee weight of $z_{L,a}$ were strictly smaller than $q/\alpha_X$, we would get
\begin{equation}
\left|(H_X z_{L,a})_i\right|
\le
\sum_j |(H_X)_{ij}|_L\, |(z_{L,a})_j|_L
\le
\alpha_X \sum_j |(z_{L,a})_j|_L
=
\alpha_X |z_{L,a}|_L <q,
\end{equation}
which is impossible because $c_a \neq 0$, so the reduced Lee weight of $z_{L,a}$ is at least $q/\alpha_X$.

The same argument applies to every $j$ with $\gcd(d_j^Z,q)>1$. This completes the proof.

\end{proof}

We also record a criterion for the nontriviality of the codespace, which is used in the proof of Theorem~\ref{thm:energybarrier}.

\begin{prop} \label{prop:trivialcodespace}
For any quantum CSS code over $\mathbb Z_q$, suppose that $H_XH_Z^{\mathsf T}=0$
over $\mathbb Z$. Let the nonzero Smith factors of $H_X$ and $H_Z$ be
$d_1^X,\ldots,d_{r_X}^X$ and $d_1^Z,\ldots,d_{r_Z}^Z$. Then the codespace over
$\mathbb Z_q$ is trivial if and only if
\begin{equation}
    r_X+r_Z=n,
    \qquad
    \gcd(d_i^X,q)=1 \text{ for all } i,
    \qquad
    \gcd(d_j^Z,q)=1 \text{ for all } j.
\end{equation}
Moreover, if the $\mathbb Z_q$ codespace is nontrivial, then either there
exists a nontrivial logical $Z$ operator $z_L$ and an integer vector
$\widetilde z_L\in\ker_{\mathbb Z}H_X$ such that
\(
    \widetilde z_L \bmod q = z_L,
\)
or there exists a nontrivial logical $X$ operator $x_L$ and an integer vector
$\widetilde x_L\in\ker_{\mathbb Z}H_Z$ such that
\(
    \widetilde x_L \bmod q = x_L.
\)
\end{prop}

\begin{proof}
By \eqref{eq:C-decomp_integer} in Proposition~\ref{prop:decom}, the codespace is trivial exactly when $r_X+r_Z=n$ and all cyclic factors
in the two torsion blocks are trivial.

If the codespace is nontrivial, at least one of the three summands above is
nontrivial. If the free summand $\mathbb Z_q^{\,n-r_X-r_Z}$ is nontrivial, then
$n-r_X-r_Z>0$. By the decomposition \eqref{eq:decom_kerHX} established in the proof of Proposition~\ref{prop:decom} (Appendix~\ref{app:codespace_integer}), the
integer kernel of $H_X$ splits as
\begin{equation}
    \ker_{\mathbb Z} H_X
    =
    \operatorname{Sat}(\operatorname{im}_{\mathbb Z} H_Z^{\mathsf T})
    \oplus F_Z,
\end{equation}
where $F_Z$ is a free abelian group of rank $n-r_X-r_Z$. The image
$\operatorname{im}_{\mathbb Z}(H_Z^{\mathsf T})$ is contained in the first
summand, so any $\widetilde z_L\in F_Z$ whose reduction $z_L:=\widetilde z_L \bmod q$ is nonzero represents a nontrivial logical $Z$ operator in
\(
    \ker_{\mathbb Z_q} H_X / \operatorname{im}_{\mathbb Z_q} H_Z^{\mathsf T}.
\)
Alternatively, one could use the dual free summand in $\ker_{\mathbb Z}H_Z$ to obtain a nontrivial logical $X$ operator with an exact integer-kernel representative.

On the other hand, if $\gcd(d_j^Z,q)>1$ in the Smith decomposition \eqref{eq:SNF-H}, the vector
\(
    \widetilde z_L=(V_Z^{-1})^{\mathsf T}e_j
\)
satisfies $d_j^Z\widetilde z_L\in\operatorname{im}_{\mathbb Z}(H_Z^{\mathsf T})$ and
\(
    H_X\widetilde z_L=0.
\)
Thus $z_L :=\widetilde z_L\bmod q\in\ker_{\mathbb Z_q}H_X$.

It remains to check that $z_L$ is not a $Z$ stabilizer.
Suppose, for contradiction, that
\(
    z_L\in\operatorname{im}_{\mathbb Z_q}(H_Z^{\mathsf T}).
\)
Then there exist integer vectors $a$ and $b$ such that
\(
    \widetilde z_L=H_Z^{\mathsf T}a+qb.
\)
Multiplying by $V_Z^{\mathsf T}$, and looking at the $j$-th coordinate, gives
\(
    1=d_j^Z c+qb'
\)
for some integers $c,b'$, which forces $\gcd(d_j^Z,q)=1$, contradicting $\gcd(d_j^Z,q)>1$.
Therefore
$z_L \in \ker_{\mathbb Z_q} H_X / \operatorname{im}_{\mathbb Z_q} H_Z^{\mathsf T}$ is a nontrivial $Z$-logical class, and it has an
integer lift $\widetilde z_L\in\ker_{\mathbb Z}H_X$.

The $X$-side torsion case is identical. If $\gcd(d_i^X,q)>1$, there exists a nontrivial logical $X$ operator $x_L$ with an
integer lift $\widetilde x_L\in\ker_{\mathbb Z}H_Z$.

This proves the claim.
\end{proof}

\subsection{Bounds on logical operators and energy barriers} \label{subsec:quantumbounds}

In this subsection, we state the main theorems of the section: every nontrivial quantum CSS Lee-LDPC code has a nontrivial logical operator whose energy barrier is at most $O(n)$, with constants independent of $q$; independently, every noncommuting logical pair has a noncommuting logical witness of Lee weight at most $O(n)$.

We begin with the bound on the energy barrier. For a logical operator $\ell$ in the $Z$-error sector, its energy barrier $\Delta E(\ell)$ is the minimum, over paths $0=e^{(0)},\ldots,e^{(T)}$ in $\mathbb Z_q^n$ in which each step changes one coordinate by $\pm1$ and $e^{(T)}$ represents $\ell$, of $\max_t E_Z(e^{(t)})$. The $X$-error sector is analogous, with $E_X$. We first need a path construction for an integer codeword. The following lemma shows that any integer codeword of a Lee-LDPC code can be reached by single-site updates while keeping the syndrome energy bounded by $O(n)$.

\begin{lem}[Small barrier from an integer codeword] \label{lem:smallbarrier}
Let $H\in\mathbb Z^{m\times n}$ be an integer matrix satisfying
\begin{equation}
    \sum_{i=1}^m |H_{ij}| \le w_{\mathrm B}
    \qquad
    \text{for each } j=1,\ldots,n.
\end{equation}
For $x\in\mathbb Z^n$, define the energy
\begin{equation}
    E(x)=|Hx|_1.
\end{equation}
Then, for every $v\in\ker_{\mathbb Z}H$, there exists a path
\begin{equation}
    0=x^{(0)},x^{(1)},\ldots,x^{(T)}=v
\end{equation}
in $\mathbb Z^n$, where each step changes a single coordinate by $+1$ or $-1$,
such that
\begin{equation}
    E(x^{(t)})<2w_{\mathrm B} n
\end{equation}
for every $t$. Moreover, reducing the entire path modulo any $q\ge 2$ gives a
path in $\mathbb Z_q^n$ from $0$ to $v\bmod q$, and the Lee weight of the
syndrome along this reduced path is also bounded by $2w_{\mathrm B} n$.
\end{lem}

\begin{proof}
Fix $v\in\ker_{\mathbb Z}H$. Then $H\alpha v=0$ over $\mathbb{R}$ for any $\alpha\in\mathbb{R}$. For each $\alpha\in[0,1]$, let
\begin{equation}
    x(\alpha)=\lfloor \alpha v\rfloor,
\end{equation}
where the floor is taken componentwise. The point of this construction is that each coordinate of $x(\alpha)-\alpha v$ has absolute value less than one.
Therefore,
\begin{equation}
    |Hx(\alpha)|_1
    \le
    \sum_{a,j}|H_{aj}|\,|x_j(\alpha)-\alpha v_j|
    <
    \sum_j\sum_a |H_{aj}|
    \le
    w_{\mathrm B} n.
\end{equation}
Thus every rounded point $x(\alpha)$ has energy strictly smaller than $w_{\mathrm B} n$.

It remains to turn these rounded points into a single-site path. As $\alpha$
increases from $0$ to $1$, the vector $x(\alpha)$ changes only when one of the
numbers $\alpha v_j$ crosses an integer. Listing these changes in increasing
order gives a finite sequence of rounded configurations
\begin{equation}
    0=y^{(0)},y^{(1)},\ldots,y^{(M)}=v,
\end{equation}
such that each coordinate of $y^{(r+1)}-y^{(r)}$ lies in $\{-1,0,1\}$. We now
connect consecutive rounded configurations by an explicit single-site path. For
a fixed $r$, let $j_1,\ldots,j_s$ be the coordinates on which
$y^{(r+1)}-y^{(r)}$ is nonzero, and write
\begin{equation}
    \sigma_\ell
    =
    \bigl(y^{(r+1)}-y^{(r)}\bigr)_{j_\ell}
    \in \{-1,1\}.
\end{equation}
Define
\begin{equation}
    y^{(r,\ell)}
    =
    y^{(r)}
    +
    \sum_{\mu=1}^{\ell}\sigma_\mu e_{j_\mu},
    \qquad
    \ell=0,\ldots,s.
\end{equation}
Then $y^{(r,0)}=y^{(r)}$, $y^{(r,s)}=y^{(r+1)}$, and each step from
$y^{(r,\ell)}$ to $y^{(r,\ell+1)}$ changes exactly one coordinate by $+1$ or
$-1$. Concatenating these paths over all $r$ gives a single-site path from $0$
to $v$.

A single-site update at coordinate $j$ changes the syndrome by one column of
$H$, with a sign. Hence the energy can change by at most
\begin{equation}
    \sum_a |H_{aj}|\le w_{\mathrm B}.
\end{equation}

Between $y^{(r)}$ and $y^{(r+1)}$, at most $n$ coordinates are updated. Since
$E(y^{(r)})<w_{\mathrm B}n$, every intermediate configuration on this short segment has
energy less than
\begin{equation}
    w_{\mathrm B}n+w_{\mathrm B}n=2w_{\mathrm B}n.
\end{equation}

Concatenating these short single-site paths gives a path from $0$ to $v$ along
which every configuration has energy less than $2w_{\mathrm B}n$. Relabeling the concatenated path as $x^{(0)},\ldots,x^{(T)}$ gives the claimed path.

Finally, reducing this integer path modulo any $q\ge 2$ gives a single-site path
in $\mathbb Z_q^n$ from $0$ to $v\bmod q$. At each step, the Lee weight of the
reduced syndrome is no larger than the ordinary $\ell_1$ weight of the
corresponding integer syndrome, so the same $2w_{\mathrm B}n$ energy bound holds after
reduction modulo $q$.
\end{proof}

We now apply Lemma \ref{lem:smallbarrier} to quantum CSS Lee-LDPC codes.

\begin{thm}[No superlinear energy barrier for quantum CSS Lee-LDPC codes] \label{thm:energybarrier}
Consider a quantum $(w_{\mathrm B},w_{\mathrm C})$-Lee-LDPC CSS code over $\mathbb Z_q$ on $n$
physical qudits, defined by parity-check matrices $H_X$ and $H_Z$. Suppose that
the code has nontrivial logical operators. Then either the $Z$-error sector or
the $X$-error sector contains a nontrivial logical operator whose energy barrier is
bounded by
\begin{equation}
    \Delta E
    \le
    \begin{cases}
        w_{\mathrm B}w_{\mathrm C}^2 n/2, & q\le w_{\mathrm C}^2, \\
        2w_{\mathrm B}n, & q>w_{\mathrm C}^2.
    \end{cases}
\end{equation}
In particular, the energy barrier is at most
linear in $n$, with a constant not
growing with $q$.
\end{thm}

\begin{proof}
The proof is trivial when $q\le w_{\mathrm C}^2$ because every configuration has energy at most $w_{\mathrm B}n\lfloor q/2\rfloor\le w_{\mathrm B}w_{\mathrm C}^2n/2$: $H_X$ and $H_Z$ each have at most $w_{\mathrm B}n$ nonzero rows (their total Lee weight is at most $w_{\mathrm B}n$), and each syndrome entry has Lee weight at most $\lfloor q/2\rfloor$. When $q>w_{\mathrm C}^2$, from Proposition \ref{prop:commteoverZ}, the code must have $H_X H_Z^\mathsf{T}=0$ over $\mathbb{Z}$. In this case, by Proposition \ref{prop:trivialcodespace}, nontriviality of the codespace implies that either there is an integer vector $\widetilde z_L\in\ker_{\mathbb Z}H_X$
whose reduction modulo $q$ is a nontrivial logical $Z$ operator $z_L$, or there
is an integer vector $\widetilde x_L\in\ker_{\mathbb Z}H_Z$ whose reduction
modulo $q$ is a nontrivial logical $X$ operator $x_L$.

In the first case, apply Lemma~\ref{lem:smallbarrier} to $H_X$ and
$\widetilde z_L$. It gives a path from $0$ to
$\widetilde z_L$ over $\mathbb Z$, where each step changes a single coordinate by $+1$ or $-1$. The energy is always less than
$2w_{\mathrm B}n$ in the path. Reducing the entire path modulo $q$ gives a path from the zero error configuration to the nontrivial logical $Z$ operator $z_L$, and the energy on $\mathbb{Z}_q$ is still less than $2w_{\mathrm B} n$.
The second case is identical, which gives the claimed bound on the energy barrier.
\end{proof}

We now turn to the bound on the weight of a noncommuting logical witness, proved in Appendix~\ref{app:baddistance}. The following theorem shows that, no matter how $q$ grows with $n$, every nontrivial quantum Lee-LDPC CSS code has Lee distance at most $O(n)$, and any subsystem code obtained by retaining a collection of its noncommuting logical pairs as bare logical operators and treating the remaining logical
degrees of freedom as gauge degrees of freedom has dressed Lee distance at most $O(n)$. Appendix~\ref{app:noncss} extends it to general stabilizer codes.

\begin{thm}[Short noncommuting witnesses in quantum Lee-LDPC codes] \label{thm:baddistance}
 For any quantum $(w_{\mathrm{B}},w_{\mathrm{C}})$-Lee-LDPC code on $\mathbb{Z}_q$, let $L_X \in \ker(H_Z)$ and $L_Z \in \ker(H_X)$ be nontrivial noncommuting logical operators, namely $L_X\cdot L_Z\ne 0 \pmod{q}$. Then there exist either $\widetilde{L}_X\in \ker(H_Z)$ or $\widetilde{L}_Z\in \ker(H_X)$ such that  $\widetilde{L}_X\cdot L_Z\ne 0 \pmod{q}$ [or $L_X\cdot \widetilde{L}_Z\ne 0 \pmod{q}$] and 
\begin{equation} \label{eq:baddistancebound}
    \min\left(|\widetilde{L}_X|_{\mathrm{L}}, |\widetilde{L}_Z|_{\mathrm{L}}\right)\le
    \begin{cases}
        \tfrac12 w_{\mathrm C}^2 n, & q\le w_{\mathrm C}^2,\\[2mm]
        \left(2\sqrt{2}+\frac{w_{\mathrm C}}{2}\right)n,
        & w_{\mathrm C}^2<q.
    \end{cases}
\end{equation}
\end{thm}

\begin{proof}
    When $q\le w_{\mathrm C}^2$, the proof is trivial. For $q>w_{\mathrm C}^2$, from Proposition \ref{prop:commteoverZ}, the code must have $H_X H_Z^\mathsf{T}=0$ over $\mathbb{Z}$. The proof of this part is nontrivial and is given in Appendix~\ref{app:baddistance}; here we provide an executive summary.  The essential strategy is to start from the reduced weight of $L_X$,\begin{equation}
        \min_{k_X\in\mathbb Z_q^{m_X}}|L_X+H_X^{\mathsf{T}}k_X|_{\mathrm{L}},
    \end{equation} and relax the minimization by allowing $k_X$ to range over $\mathbb{R}^{m_X}$, which is easier to analyze mathematically.  This relaxation, along with relaxing norms from $\ell_1$ to $\ell_2$, incurs at most $O(n)$ loss in the final Lee-weight bound.  We then generalize mathematical results known as transference inequalities for dual lattices~\cite{transference, 2017thesis} to prove the bound \eqref{eq:baddistancebound} for every noncommuting pair of logical operators.
\end{proof}

Note that the $\mathbb{Z}_q$ Pauli algebra of the elementary physical errors and the corresponding Lee metric are essential to Theorem~\ref{thm:baddistance}. Example~\ref{eg:highTmemory} in Section~\ref{subsec:quantumobstruction} does not satisfy these assumptions and has distance $q$ despite having bounded-weight stabilizers.

\subsection{Hypergraph products and the lifted-Hamming toric code} \label{subsec:leetoric}

The hypergraph product (HGP)~\cite{TillichZemor2014} turns two classical codes into a quantum CSS code, and the construction carries over to $\mathbb{Z}_q$ once a sign is inserted in $H_Z$.

\begin{defn}[Hypergraph product code] \label{def:hgp}
Let $H_1\in\mathbb Z_q^{m_1\times n_1}$ and
$H_2\in\mathbb Z_q^{m_2\times n_2}$ be the parity-check matrices of two codes. The hypergraph-product
code associated with $(H_1,H_2)$ is the CSS code with check matrices
\begin{equation}
    H_X=
    \left(
        H_1\otimes I_{n_2}
        \,\middle|\,
        I_{m_1}\otimes H_2^{\mathsf T}
    \right),
    \qquad
    H_Z=
    \left(
        -I_{n_1}\otimes H_2
        \,\middle|\,
        H_1^{\mathsf T}\otimes I_{m_2}
    \right). \label{eq:hgp}
\end{equation}
These matrices naturally satisfy
\(
    H_XH_Z^{\mathsf T}=0,
\)
so they define a CSS code over $\mathbb Z_q$ on
\(
    n_1 n_2+m_1 m_2
\)
physical qudits.
\end{defn}

\noindent The $\mathbb{Z}_q$ toric code~\cite{Kitaev2003anyons} is the hypergraph product of the cyclic repetition code with itself, where the check matrix of the repetition code is the $L\times L$ circulant matrix $\mathrm{circ}_L(1,-1)$ whose rows are the cyclic shifts of $(1,-1,0,\dots,0)$.  Replacing the repetition code by the dual of lifted Hamming code of Example~\ref{eg:liftedhamming} yields a two-dimensional local code whose Lee distance exceeds its Hamming distance and whose topological order is enriched by translations.

\begin{eg}[Lifted-Hamming toric code: a stack of toric codes enriched by translation] \label{eg:leetoric}
    Let $q=2^r$, let $\lambda\in\mathbb{Z}_q$ satisfy $\lambda^2-\lambda+2=0$ as in Example~\ref{eg:liftedhamming}, and let $V=\mathrm{circ}_L(v)$ be the $L\times L$ circulant matrix whose rows are the cyclic shifts of
    \begin{equation}
        v=(1,\ \lambda,\ \lambda-1,\ -1,\ 0,\ \dots,\ 0), \label{eq:leetoricpattern}
    \end{equation}
    the coefficients of the reciprocal $x^3g_q(x^{-1})$ of the lifted-Hamming generator \eqref{eq:liftedhamming}.  The hypergraph product of $V$ with itself places one qudit on each horizontal and each vertical edge of an $L\times L$ torus.  Each $X$ check is a vertex term with two arms of four edges, and each $Z$ check is a plaquette term with two such arms, with the exponents shown in Fig.~\ref{fig:leetoric}.  Every check acts on eight edges (six for $q=2$, where $\lambda=0$) inside a $4\times4$ window, for every $q$, and for $q=4$, where $\lambda=2$, the arms read $X,X^2,X,X^{-1}$.

\definecolor{toricx}{RGB}{23,110,94}
\definecolor{toricz}{RGB}{86,70,168}
\definecolor{toricgrid}{RGB}{200,205,200}
    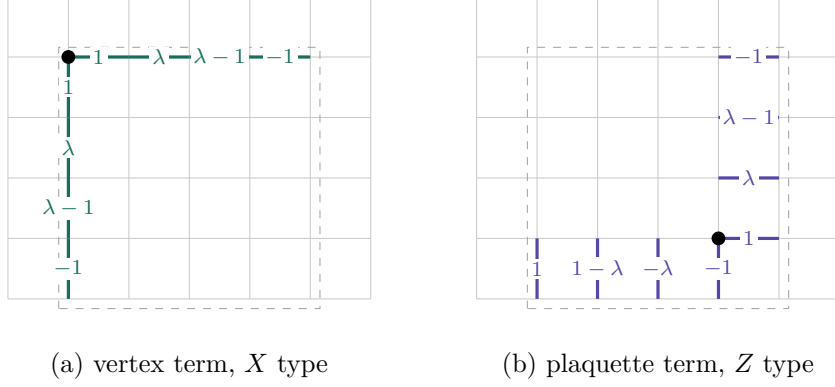
\begin{figure}[t]
        \centering
        \begin{tikzpicture}[x=0.8cm, y=0.8cm,
          op/.style={midway, fill=white, inner sep=1.1pt, font=\scriptsize}]
        \begin{scope}
        \foreach \yy in {-4,...,1}{\draw[toricgrid] (-1,\yy) -- (5,\yy);}
        \foreach \xx in {-1,...,5}{\draw[toricgrid] (\xx,-4) -- (\xx,1);}
        \draw[gray!70, dashed] (-0.16,-4.16) rectangle (4.16,0.16);
        \draw[toricx, very thick] (0,0) -- (1,0) node[op]{$1$};
        \draw[toricx, very thick] (1,0) -- (2,0) node[op]{$\lambda$};
        \draw[toricx, very thick] (2,0) -- (3,0) node[op]{$\lambda-1$};
        \draw[toricx, very thick] (3,0) -- (4,0) node[op]{$-1$};
        \draw[toricx, very thick] (0,0) -- (0,-1) node[op]{$1$};
        \draw[toricx, very thick] (0,-1) -- (0,-2) node[op]{$\lambda$};
        \draw[toricx, very thick] (0,-2) -- (0,-3) node[op]{$\lambda-1$};
        \draw[toricx, very thick] (0,-3) -- (0,-4) node[op]{$-1$};
        \fill (0,0) circle (2.6pt);
        \node[font=\small] at (2,-5.1) {(a) vertex term, $X$ type};
        \end{scope}
        \begin{scope}[xshift=8.6cm]
        \foreach \yy in {-4,...,1}{\draw[toricgrid] (-4,\yy) -- (2,\yy);}
        \foreach \xx in {-4,...,2}{\draw[toricgrid] (\xx,-4) -- (\xx,1);}
        \draw[gray!70, dashed] (-3.16,-4.16) rectangle (1.16,0.16);
        \draw[toricz, very thick] (0,-3) -- (1,-3) node[op]{$1$};
        \draw[toricz, very thick] (0,-2) -- (1,-2) node[op]{$\lambda$};
        \draw[toricz, very thick] (0,-1) -- (1,-1) node[op]{$\lambda-1$};
        \draw[toricz, very thick] (0,0) -- (1,0) node[op]{$-1$};
        \draw[toricz, very thick] (0,-4) -- (0,-3) node[op]{$-1$};
        \draw[toricz, very thick] (-1,-4) -- (-1,-3) node[op]{$-\lambda$};
        \draw[toricz, very thick] (-2,-4) -- (-2,-3) node[op]{$1-\lambda$};
        \draw[toricz, very thick] (-3,-4) -- (-3,-3) node[op]{$1$};
        \fill (0,-3) circle (2.6pt);
        \node[font=\small] at (-1,-5.1) {(b) plaquette term, $Z$ type};
        \end{scope}
        \end{tikzpicture}
        \caption{\textbf{Checks of the lifted-Hamming toric code.}  One $\mathbb{Z}_q$ qudit sits on each edge, and the number on an edge is the exponent of the Pauli operator applied to it, with $\lambda$ a root of $\lambda^2-\lambda+2$ in $\mathbb{Z}_q$ as in Example~\ref{eg:liftedhamming}.  (a) The $X$ check at a vertex (dot) applies $X^{1},X^{\lambda},X^{\lambda-1},X^{-1}$ along four edges to the east and four edges to the south.  (b) The $Z$ check at a plaquette corner (dot) applies the same exponents to four horizontal edges to the north and the negated exponents to four vertical edges to the west, in the row just below the dot.  Every check occupies the dashed $4\times4$ window, for every $q$.}
        \label{fig:leetoric}
    \end{figure}

    The code is nontrivial exactly when $7$ divides $L$.  Its string operators carry the elements of the annihilator of $v$ in $\mathbb{Z}_q[x]/(x^L-1)$, and $v$ is a unit of this ring unless $7\mid L$, because $v$ reduces to an irreducible cubic modulo two.  For every $L=7m$, the annihilator is a free module of rank three (for odd $m$, it is the Galois-ring factor of $\mathbb{Z}_q[x]/(x^L-1)$ that belongs to $v$, i.e., $\mathbb{Z}_q[x]/(v)\cong\mathrm{GR}(q,3)$), and the code encodes $k=2\cdot3^2=18$ qudits of dimension $q$ (verified numerically for $L=7,14$ and $q\le64$).  The annihilator consists of the codewords of the dual $\mathcal{H}_q^\perp$ of the lifted Hamming code, repeated with period seven, so a row or column carrying a minimum-weight codeword of $\mathcal{H}_q^\perp$ is a logical operator.  Hence
    \begin{equation}
        d\le\frac{d_qL}{7},\qquad d_q=4,\,6,\,8,\,12,\,14,\,24\quad\text{for}\quad q=2,\,4,\,8,\,16,\,32,\,64, \label{eq:leetoricd}
    \end{equation}
    where $d_q$ is the minimum Lee weight of $\mathcal{H}_q^\perp$.  For $q=4$ and $L=7$, exhaustive enumeration of all errors of Lee weight at most five shows that \eqref{eq:leetoricd} is attained, so the code is a $\qcode{98,18,6}_4$ code.  Its Hamming distance is four, because the code reduced modulo two is the hypergraph product of the binary $[7,3,4]$ simplex code with itself.  Thus, the octacode's ratio $3/2$ of Lee to Hamming distance survives in two dimensions, against $\qcode{98,2,7}_4$ for the $\mathbb{Z}_4$ toric code on the same torus.  Whether the distance equals $d_qL/7$ for all $L$ and $q$ is open.

    Increasing $r$ in $q=2^r$ keeps the check support at eight edges, while both the check Lee weight $2(2+|\lambda|_{\mathrm L}+|\lambda-1|_{\mathrm L})$ and the string weight $d_q$ per period grow: $|\lambda|_{\mathrm L}\gtrsim\sqrt q$ because, for the centered representative of $\lambda$, the integer $\lambda^2-\lambda+2$ is positive and divisible by $q$, hence at least $q$, and $d_q\ge\sqrt q$ because an annihilator word $c$ evaluated at a primitive seventh root of unity has algebraic norm at least $q^3$ and at most $|c|_{\mathrm L}^6$ in absolute value.  The strings thus grow heavier without bound in $q$ at fixed geometry, at the price of check Lee weights that grow at least as $\sqrt q$.  By the definition of Section~\ref{sec:LDPC}, which requires $w_{\mathrm C}$ to be independent of $q$, this family is therefore Lee-LDPC only at fixed $q$.

    For $q=4$, the topological order of the code is that of nine $\mathbb{Z}_4$ toric codes, but its lattice translations permute the anyons.  It is thus a symmetry-enriched topological order \cite{Barkeshli2019SET} whose enriching symmetry is translation~\cite{Watanabe2023toric,Chen2025anyon}.  To describe the permutation, we use the Galois ring $\mathrm{GR}(4,3)=\mathbb{Z}_4[x]/(h(x))$, where $h(x)=-v(x)=x^3-(\lambda-1)x^2-\lambda x-1$ (for $q=4$, $h=x^3+3x^2+2x+3$) is the monic cubic factor of $x^7-1$ over $\mathbb{Z}_4$ that reduces to the irreducible $x^3+x^2+1$ modulo two \cite[Sec.~II]{Hammons1994}, \cite[Ch.~6]{Wan1997quaternary}.  This ring is to $\mathbb{Z}_4$ what the field $\mathbb{F}_8$ is to $\mathbb{Z}_2$.  It has $64$ elements, reduction modulo two maps it onto $\mathbb{F}_8$, and its unit group contains a unique cyclic subgroup of order $2^3-1=7$, generated by a root $\xi$ of $h$.  The powers of $\xi$ together with $0$ form the Teichm\"uller set, a canonical set of representatives of the residue field $\mathbb F_8$, closed under multiplication but not generally under addition, and $\xi$ is called a Teichm\"uller unit.  For $q=4$, the anyons created at the ends of $Z$-type strings, i.e., the configurations of violated vertex terms modulo those created by single $Z$ errors, form the module $\mathbb{Z}_4[x^{\pm1},y^{\pm1}]/(v(x),\bar v(y))$ with $\bar v(y)=v(y^{-1})$, where $x$ and $y$ act as the two translations~\cite{Liang2024extracting}.  Since $v=-h$, this module is $\mathrm{GR}(4,3)[y]/(\bar v(y))\cong\mathrm{GR}(4,3)^3\cong\mathbb{Z}_4^9$, with $x$ acting as $\xi$ and with $\bar v$ splitting into three linear factors with roots $\xi^{6},\xi^{5},\xi^{3}$.  The logical string operators have a nondegenerate commutation pairing, so the anyons braid as in nine copies of the $\mathbb{Z}_4$ toric code.  A translation by one site along $x$, however, multiplies every anyon charge by $\xi$.  Since $\xi-1$ is a unit of $\mathrm{GR}(4,3)$, no nontrivial anyon returns to itself under fewer than seven translations, while $\xi^7=1$ makes translation by seven sites act trivially, which is why the code is trivial unless $7\mid L$.  Translations along $y$ act as $\xi^{6}$, $\xi^{5}$, and $\xi^{3}$ on the three factors of the charge module, so each factor is invariant along exactly one of the directions $(1,1)$, $(2,1)$, and $(4,1)$, whose first components $1,2,4$ are the quadratic residues modulo seven.  This is the mechanism of Wen's plaquette model~\cite{Wen2003plaquette}, in which a translation exchanges the $e$ and $m$ anyons of the toric code, except that here the permutation has order seven instead of two.  Coarse-graining by seven sites removes the permutation and leaves nine decoupled toric codes.  The same algebra applies for larger $q=2^r$ with $\mathrm{GR}(q,3)$ in place of $\mathrm{GR}(4,3)$, so we expect nine $\mathbb{Z}_q$ toric codes with the same order-seven action for every $q=2^r$, although we verified the braiding only at $q=4$.  The same construction accepts the lifted Golay code of Example~\ref{eg:liftedgolay} in place of the lifted Hamming code, on tori with $23\mid L$, where it yields $121$ toric codes enriched by translations of order $23$.
\end{eg}

\subsection{Helical repetition codes and their hypergraph products}
Since classical Lee codes are already well-studied, we will mostly just focus on a single example of a classical code which is very relevant to our discussion on high-threshold codes in Section~\ref{sec:memory}: \begin{defn}[Helical repetition code $H(a,-b)$]
    Pick two natural numbers $a>b\ge 1$ with $\gcd(a,q)=\gcd(b,q)=\gcd(a,b)=1$. The $(a,-b)$ helical repetition code with open boundary conditions (OBC) has parity-check matrix $H_{\mathrm{OBC}}\in\mathbb{Z}_q^{(n-1)\times n}$: 
\begin{equation}
        H_{\text{OBC}}(a,-b):= \left(\begin{array}{cccccc} a &\ -b &\ 0 &\ \cdots &\ 0 &\ 0 \\ 0 &\ a &\ -b &\ \cdots &\ 0 &\ 0 \\ \vdots &\ \vdots &\ \vdots &\ \ddots &\ \vdots &\ \vdots  \\ 0 &\ 0 &\ 0 &\ \cdots &\ a &\ -b\end{array}\right).
    \end{equation}
    The helical repetition code with periodic boundary conditions (PBC) has parity-check matrix $H_{\mathrm{PBC}}\in \mathbb{Z}_q^{n\times n}$, which adds one extra parity check at the bottom: 
    \begin{equation}\label{eq:helicalPBC}
        H_{\text{PBC}}(a,-b):= \left(\begin{array}{cccccc} a &\ -b &\ 0 &\ \cdots &\ 0 &\ 0 \\ 0 &\ a &\ -b &\ \cdots &\ 0 &\ 0 \\ \vdots &\ \vdots &\ \vdots &\ \ddots &\ \vdots &\ \vdots  \\ 0 &\ 0 &\ 0 &\ \cdots &\ a &\ -b  \\ -b &\ 0 &\ 0 &\ \cdots &\ 0 &\ a \end{array}\right).
    \end{equation}
\end{defn}
\noindent For $b=1$, the PBC checks are the classical couplings of the generalized $q$-state clock model of Ref.~\cite{HuWatanabe2023clock}, whose ground-state degeneracy in the ordered phase, $\gcd(a^n-1,q)$, depends on the chain length.  The helical repetition code has some conceptual similarities with the ``quantum breakdown model" introduced in a recent set of papers \cite{lian1,lian2,lian3}.  However, these papers focused on the quantum mechanics of low-energy states.  We expect that our commentary on the thermodynamics of the helical repetition code in Section~\ref{sec:memory} will be directly relevant to these models.  
\begin{prop} \label{prop:helicalproperties}
    The codespace of the helical repetition code with open boundary conditions is \begin{equation} \label{eq:helical_codeword}
        \mathcal{C} = \mathrm{span}\left\lbrace \left(\begin{array}{cccc} b^{n-1} &\ ab^{n-2} &\ \cdots &\ a^{n-1} \end{array}\right) \right\rbrace .
    \end{equation}
     These are also the codewords with periodic boundary conditions if $a^n=b^n \text{ (mod $q$)}$. If $c=\gcd(q,a^n-b^n)>1$, then $q/c$-multiples of the generator in \eqref{eq:helical_codeword} remain valid codewords.
     For sufficiently large $q$, the Lee distance of this code grows exponentially: $d>a^{n-1}$.
\end{prop}

Proposition \ref{prop:helicalproperties} can be verified by inspection, but it is a key illustration of the potential power of Lee-LDPC codes.  In particular, the exponentially growing distance allows us to verify the existence of passive local decoders with extremely high thresholds, for sufficiently large $q$.  This connection is developed further in Section~\ref{sec:memory}.

The coprimality assumptions of the definition above imply $\gcd(a,b,q)=1$, and some such condition is needed, as the Smith normal form shows. Let $g=\gcd(a,b)$ and $h=\gcd(g,q)$. For the OBC helical code, the Smith form of
$H_{\mathrm{OBC}}(a,-b)$ is $(gI_{n-1}\ 0)$, and hence
\begin{equation}
    \ker_{\mathbb Z_q} H_{\text{OBC}}(a,-b)
    \cong
    \mathbb Z_h^{\,n-1} \oplus \mathbb Z_q  .
\end{equation}
Thus $h=1$ is precisely the condition under which the codespace of the OBC helical code is a single $\mathbb Z_q$ factor, i.e., $\ker_{\mathbb Z_q}H_{\mathrm{OBC}}(a,-b)\cong\mathbb Z_q$. For the PBC helical code,
\begin{equation}
    \ker_{\mathbb Z_q} H_{\text{PBC}}(a,-b)
    \cong
    \mathbb Z_h^{\,n-1}
    \oplus
    \mathbb Z_{\gcd(q,(a^n-b^n)/g^{n-1})}.
\end{equation}
Consequently, the codespaces of PBC and OBC helical codes coincide if and only if $q$ divides $(a^n-b^n)/g^{n-1}$; when $h=1$, this reduces to
$a^n=b^n \pmod q$. If $h>1$, the codespaces of both codes contain additional factors $\mathbb Z_h^{\,n-1}$, which contain codewords of Hamming weight one, such as $(q/h)e_j$.

Applying the hypergraph product of Definition~\ref{def:hgp} to a horizontal $(a_h,-b_h)$ and a vertical $(a_v,-b_v)$ helical repetition code gives local codes on a two-dimensional grid.  For $b_h=b_v=1$ and periodic boundary conditions in both directions, these coincide, up to relabeling, with the generalized $\mathbb{Z}_q$ toric codes of Ref.~\cite{Watanabe2023toric}.  Every check acts on at most four neighboring qudits, with Lee weight at most $a_h+b_h+a_v+b_v$, independent of $n$ and $q$.  Section~\ref{sec:memory} studies a mixed-boundary and an open--open version of this construction, and shows that their $Z$-error sectors inherit the long memory of the classical helical codes.

\section{High-temperature memory}\label{sec:memory}

This section asks whether the Lee metric can support long-lived classical and quantum passive memories at arbitrarily high fixed temperature as the local dimension $q$ increases. We use the memory time of single-site Metropolis dynamics as the dynamical diagnostic. We develop two mechanisms that can lead to exponentially large memory time at arbitrarily high temperature. Section \ref{subsec:energybarrier} shows that the PBC helical code has soundness and linear confinement, while the associated HGP code has an analogous approximate confinement bound; both lead to energy-dominated free-energy barriers. Section \ref{subsec:entropiceffect} studies the OBC helical code and the corresponding HGP code, which exhibit entropy-dominated free-energy barriers; the classical OBC code additionally has a glassy free-energy landscape. Section~\ref{subsec:quantumobstruction} explains why the distance and energy-barrier bounds do not by themselves decide whether both error sectors can have long memory times at arbitrarily high fixed temperature.

Before turning to these results, we define the memory time and the mixing time for the single-site Metropolis dynamics and explain the relation between them. Let $\Omega$ be the configuration space, and let $X_t$ be the state of the system at time $t$. For a quantum error sector, \(\Omega\) is understood modulo the corresponding stabilizer subgroup, and vectors denote representatives of these stabilizer classes. We
use $X_0=0$, where $0$ denotes the zero codeword for a classical code and the
zero error in the error sector under consideration for a quantum code. Fix a syndrome-based decoder, so the correction applied to a configuration is determined entirely by its syndrome. We write the probability of decoding failure as
\begin{equation}
    p_{\mathrm{fail}}(t)
    :=
    \Pr(\text{the decoder fails on }X_t \mid X_0=0).
\end{equation}

For $0<\varepsilon<1$, the memory time is defined as the first time at which
this probability reaches $\varepsilon$:
\begin{equation}
    t_{\mathrm{mem}}(\varepsilon)
    :=
    \inf\{t\ge 0:p_{\mathrm{fail}}(t)\ge\varepsilon\}.
\end{equation}

Let $\pi$ be the Gibbs distribution. For $0<\delta<1$, the mixing time is the
first time at which, regardless of the initial state $x$, the probability of
finding the system in any set $B\subseteq\Omega$ differs from its Gibbs
probability by at most $\delta$:
\begin{equation}
    t_{\mathrm{mix}}(\delta)
    :=
    \inf\left\{
        t\ge 0:
        \max_{x\in\Omega}
        \max_{B\subseteq\Omega}
        \left|
            \Pr(X_t\in B\mid X_0=x)-\pi(B)
        \right|
        \le\delta
    \right\}.
\end{equation}

Since a syndrome-based decoder outputs a correction determined only by the syndrome, adding a codeword leaves the chosen correction unchanged; hence, for any nontrivial code, the decoding-failure probability in the Gibbs distribution is at least \(1/2\). Taking $B$ to be the
set of configurations on which decoding fails therefore gives
\begin{equation} \label{eq:pfailgibbs}
    p_{\mathrm{fail}}\!\left(t_{\mathrm{mix}}(\delta)\right)
    \ge
    \frac{1}{2}-\delta,
\end{equation}
and consequently, for $0<\delta<1/2$,
\begin{equation} \label{eq:tmemletmix}
    t_{\mathrm{mem}}\!\left(\frac{1}{2}-\delta\right)
    \le
    t_{\mathrm{mix}}(\delta).
\end{equation}
In particular,
\begin{equation} \label{eq:tmemletmixquarter}
    t_{\mathrm{mem}}\!\left(\frac{1}{4}\right)
    \le
    t_{\mathrm{mix}}\!\left(\frac{1}{4}\right),
\end{equation}
so any lower bound on the memory time is also a lower bound on the mixing time.

\subsection{Memory from energy barriers} \label{subsec:energybarrier}

The periodic helical repetition code provides an unusual route to long-lived memory in one dimension by allowing the local Hilbert space dimension $q$ to grow exponentially in system size. Unlike conventional $\mathbb{Z}_2$ expander-code constructions with geometrically nonlocal connectivity~\cite{Sipser1996expander,rakovszky2023goodldpcI,rakovszky2024goodldpcII,Hong2025quantummemory}, it is built from nearest-neighbor checks on a one-dimensional ring, yet exhibits linear confinement~\cite{Quintavalle2021singleshot} and, more strongly, linear soundness~\cite{Campbell2019singleshot} in the Lee metric. In this subsection, we show how these properties give rise to a global free-energy bottleneck, leading to exponentially large memory time for any fixed $\beta>0$, which suggests an unconventional form of one-dimensional symmetry breaking. Finally, we show that the same mechanism yields long-lived memory in one quantum error sector of a mixed-boundary helical HGP code.

\begin{prop}[Soundness and confinement in invertible-$H$ codes]
\label{prop:helicalPBCsound}
Let $H\in\mathbb{Z}^{n\times n}$ be invertible over $\mathbb{Q}$, and let $q=|\det H|\geq 2$. The corresponding code over $\mathbb{Z}_q$ contains exactly $q$ distinct codewords.
For $\mathbf{z}_0\in\mathbb{Z}_q^n$, define its reduced distance to the codespace by
\begin{equation}
    d_{\rm r}(\mathbf{z}_0,\mathcal{C})
    :=\min_{\mathbf{c}\in\mathcal{C}}
    |\mathbf{z}_0-\mathbf{c}|_{\mathrm{L}}.
\end{equation}
Then
\begin{equation}\label{eq:invertibleH_soundness}
    \frac{d_{\rm r}(\mathbf{z}_0,\mathcal{C})}
         {\|H^{-1}\|_{1\to1}}
    \leq |H\mathbf{z}_0|_{\mathrm{L}}
    \leq \|H\|_{1\to1}
         d_{\rm r}(\mathbf{z}_0,\mathcal{C}).
\end{equation}
\end{prop}

\begin{proof}
Let $s_1,\ldots,s_n$ be the positive Smith invariant factors of $H$. Since $\prod_{i=1}^n s_i=|\det H|=q$, each $s_i$ divides $q$. Hence, by Proposition~\ref{prop:classicalcodespace},
\begin{equation}
    |\mathcal{C}|
    =\prod_{i=1}^n\gcd(s_i,q)
    =\prod_{i=1}^n s_i
    =q.
\end{equation}
Fix $\mathbf{z}_0\in\mathbb{Z}_q^n$, and choose a nearest codeword $\mathbf{c}\in\mathcal{C}$. Let $\mathbf{x}\in\mathbb{Z}^n$ be the centered integer lift of $\mathbf{z}_0-\mathbf{c}$, with entries in $(-q/2,q/2]$, so that $|\mathbf{x}|_1=|\mathbf{x}|_{\mathrm{L}}$. Then
\begin{equation}
    |\mathbf{x}|_{\mathrm{L}}
    =|\mathbf{z}_0-\mathbf{c}|_{\mathrm{L}}
    =d_{\rm r}(\mathbf{z}_0,\mathcal{C}).
\end{equation}
Since $H\mathbf{c}=0$ mod $q$, the upper bound follows from
\begin{equation}
    |H\mathbf{z}_0|_{\mathrm{L}}
    =|H\mathbf{x}|_{\mathrm{L}}
    \leq \|H\|_{1\to1}|\mathbf{x}|_{\mathrm{L}}
    =\|H\|_{1\to1}d_{\rm r}(\mathbf{z}_0,\mathcal{C}).
\end{equation}

For the lower bound, let $\mathbf{r}\in\mathbb{Z}^n$ be the centered representative of $H\mathbf{x}$ mod $q$, so that
\begin{equation}
    r_i\in(-q/2,q/2]\cap\mathbb{Z}.
\end{equation}
There exists $\mathbf{s}\in\mathbb{Z}^n$ such that
\begin{equation}
    H\mathbf{x}=q\mathbf{s}+\mathbf{r}.
\end{equation}
By construction,
\begin{equation}
    |\mathbf{r}|_1=|H\mathbf{x}|_{\mathrm{L}}=|H\mathbf{z}_0|_{\mathrm{L}},
\end{equation}
and
\begin{equation}
    \mathbf{x}=qH^{-1}\mathbf{s}+H^{-1}\mathbf{r}.
\end{equation}
Note that because we chose $q=|\det(H)|$, $qH^{-1}$ is an integer matrix, meaning that $H^{-1}\mathbf{r}\in\mathbb{Z}^n$. Since $H(qH^{-1}\mathbf{s})=q\mathbf{s}$, the first term reduces modulo $q$ to a codeword. Since $\mathbf{x}$ is already the minimum displacement from the codespace, we have
\begin{equation}
|\mathbf{x}|_{\mathrm{L}}\leq|H^{-1}\mathbf{r}|_{\mathrm{L}}\leq|H^{-1}\mathbf{r}|_1\leq \lVert H^{-1}\rVert_{1\to 1}|\mathbf{r}|_1=\lVert H^{-1}\rVert_{1\to 1}|H\mathbf{z}_0|_{\mathrm{L}}.
\end{equation}
This proves the lower bound.
\end{proof}

\begin{cor}
Consider the $(a,-b)$ helical repetition code with periodic boundary condition defined on $\mathbb{Z}_q$ on a one-dimensional chain of length $n$. Let
    \begin{equation}
        C=a+b,\quad  \kappa=|a-b|.
    \end{equation}
The codespace contains exactly $q$ distinct codewords, with 
\begin{equation}
    q=|\det(H)|=|a^n-b^n|.
\end{equation} 
Then every configuration $\mathbf{z}_0\in \mathbb{Z}^n_q$ satisfies
    \begin{equation} \label{eq:PBC_linearconfinement}
        \kappa d_{\rm r}(\mathbf{z}_0,\mathcal{C})\leq|H\mathbf{z}_0|_{\mathrm{L}}\leq Cd_{\rm r}(\mathbf{z}_0,\mathcal{C}).
    \end{equation}
\end{cor}
\begin{proof}
    The parity-check matrix of the $(a,-b)$ PBC helical code can alternatively be written as
    \begin{equation}
        H=aI-bS,
        \label{eq:helicalH}
    \end{equation}
    where $S$ is the cyclic shift matrix with $S^n=I$. Then 
    \begin{equation}
       \lVert H\rVert_{1\to1}=a+b=:C, \quad |\det(H)|=|a^n-b^n|=:q,
    \end{equation}
    which means that $H$ is invertible over $\mathbb{Q}$. The code is therefore confined and sound. The inverse matrix
    \begin{equation}
        H^{-1}=\frac{1}{a^n-b^n}\sum_{k=0}^{n-1}a^{n-1-k}b^k S^k.
    \end{equation}
    Since each $S^k$ is a permutation matrix, there is only a single 1 in each column. Therefore, the maximum column weight is
    \begin{equation}
        \lVert H^{-1} \rVert_{1\to1}=\frac{1}{|a^n-b^n|}\sum_{k=0}^{n-1}a^{n-1-k}b^k=\frac{1}{|a-b|}=:\kappa^{-1}.
    \end{equation}
Proposition~\ref{prop:helicalPBCsound} immediately gives
Eq.~\eqref{eq:PBC_linearconfinement}.
\end{proof}
\begin{prop}[Exponentially large memory time of the PBC helical code]\label{prop:PBCslowmixing}
Consider the PBC $(a,-b)$ helical repetition code defined on $\mathbb{Z}_q$ on a one-dimensional chain of length $n$. Set $q=a^n-b^n\ge2$, $C=a+b$, and $\kappa=a-b$. Choose a Lee ball of radius $R$ satisfying
    \begin{equation}
        R=\lfloor\frac{q}{2C}-\frac{1}{2}\rfloor.
    \end{equation}
    Under the single-site Metropolis dynamics, in which each step changes one site by $+1$ or $-1$, the memory time satisfies
    \begin{equation}\label{eq:PBCtmix}
        t_{\text{mem}}(1/4)\gtrsim\frac{1}{8}\frac{(n-1)!}{(2R)^{n-1}}\exp\left(\beta\kappa R\right).
    \end{equation}
For finite $n$, this bound defines a finite-size inverse-temperature threshold $\beta_c$. For $\beta>\beta_c$, the memory time is exponentially large in $q$, with $\beta_c$ approximately determined by
    \begin{equation}\label{eq:PBCT_c}
        \beta_c\kappa R\approx (n-1)\log{2R}-\log{(n-1)!},
    \end{equation}
    which gives 
    \begin{equation}
        \beta_c=O(n^2/q).
    \end{equation}
In particular, for any fixed $\beta>0$, the memory time is exponentially large in $q$.
\end{prop}

\begin{proof}
    We first show that the distance $d$ of the PBC helical code is lower bounded by
    \begin{equation}
        d\geq \frac{q}{C}.
    \end{equation}
    Take the centered lift $\mathbf{r}$ of a nonzero codeword $\mathbf{c}\in\mathcal{C}$ with $r_i\in(-q/2,q/2]\cap\mathbb{Z}$ such that $|\mathbf{c}|_{\mathrm{L}}=|\mathbf{r}|_1$. Since $H$ is invertible, $H\mathbf{r}=q\mathbf{s}$ with $\mathbf{s}\neq0$. Thus
    \begin{equation}
    q\leq|H\mathbf{r}|_1\leq(a+b)|\mathbf{r}|_1=C|\mathbf{c}|_{\mathrm{L}}.
    \end{equation}
    Choose a ball of radius $R=\lfloor\frac{q}{2C}-\frac{1}{2}\rfloor$,
    \begin{equation} \label{eq:PBCballprooftwo}
        B=\{\mathbf{x}\in\mathbb{Z}^n_q:|\mathbf{x}|_{\mathrm{L}}\leq R\}.
    \end{equation}
    Since the code distance for the PBC helical code is lower bounded by $q/C$, the balls $B$ and $B+\mathbf{v}$ are disjoint for any nonzero codeword $\mathbf{v}$. 
    Every state $\mathbf{x}\in B$ satisfies 
    \begin{equation}
        |\mathbf{x}|_{\mathrm{L}}=d_{\rm r}(\mathbf{x},\mathcal{C}),
    \end{equation}
    and has energy bounded by
    \begin{equation} \label{eq:PBCconfprooftwo}
        \kappa|\mathbf{x}|_{\mathrm{L}}\leq|H\mathbf{x}|_{\mathrm{L}}\leq C|\mathbf{x}|_{\mathrm{L}}.
    \end{equation}
    Consider the stationary probability flow
    \begin{equation}
        Q(B,B^c):=\sum_{\substack{\mathbf{x}\in B \\ \mathbf{y}\notin B}}\pi(\mathbf{x})P(\mathbf{x},\mathbf{y}),
    \end{equation}
    which counts the probability flow across the boundary $\partial B=\{\mathbf{x}:|\mathbf{x}|_{\mathrm{L}}=R\}$. From the soundness of the code, we obtain
    \begin{align}
        Q(B,B^c)&\leq \sum_{\mathbf{x}\in\partial B}\pi(\mathbf{x})=\frac{1}{Z}\sum_{\mathbf{x}:|\mathbf{x}|_{\mathrm{L}}=R}e^{-\beta|H\mathbf{x}|_{\mathrm{L}}}\\
        &<2^n\binom{n+R-1}{n-1}\frac{e^{-\beta\kappa R}}{Z}=\frac{2^n R^{n-1}}{(n-1)!}\prod_{j=1}^{n-1}\left(1+\frac{j}{R}\right)\frac{e^{-\beta\kappa R}}{Z}\\
        &\approx\frac{2^n R^{n-1}}{(n-1)!}\frac{e^{-\beta\kappa R}}{Z}, \label{eq:boundary flow}
    \end{align}
    where the second inequality follows from the crude counting bound on $|\partial B|$ obtained by choosing nonnegative weights summing to $R$ and assigning at most two signs to each site. And the last approximation follows from $R=O(\exp(n))$.

We next convert the stationary probability flow across the boundary of $B$ into a bound on the
memory time for a chain initialized at the zero codeword~(cf.\ Ref.~\cite{Chesi2010thermodynamic} for a quantum analogue). Let $X_s$ denote the configuration after $s$
updates, with $X_0=\mathbf{0}$, and let $P^s(\mathbf{x},\mathbf{y})$ denote the corresponding $s$-step transition probability.
Reversibility implies that
\begin{equation}
    \pi(\mathbf{x})P^s(\mathbf{x},\mathbf{y})
    =
    \pi(\mathbf{y})P^s(\mathbf{y},\mathbf{x})
\end{equation}
for every $s\ge 0$. Taking $\mathbf{x}=\mathbf{0}$ and using $P^s(\mathbf{y},\mathbf{0})\le 1$, we obtain
\begin{equation}
    P^s(\mathbf{0},\mathbf{y})
    =
    \frac{\pi(\mathbf{y})}{\pi(\mathbf{0})}
    P^s(\mathbf{y},\mathbf{0})
    \le
    \frac{\pi(\mathbf{y})}{\pi(\mathbf{0})}.
\end{equation}
It follows that the probability of crossing from $B$ to $B^c$ between steps $s$ and $s+1$ satisfies
\begin{equation}
\begin{aligned}
    \mathbb{P}_{\mathbf{0}}(X_s\in B,\;X_{s+1}\notin B)
    &=
    \sum_{\substack{\mathbf{x}\in B\\ \mathbf{y}\notin B}}
    P^s(\mathbf{0},\mathbf{x})P(\mathbf{x},\mathbf{y}) \\
    &\le
    \frac{1}{\pi(\mathbf{0})}
    \sum_{\substack{\mathbf{x}\in B\\ \mathbf{y}\notin B}}
    \pi(\mathbf{x})P(\mathbf{x},\mathbf{y}) =
    \frac{Q(B,B^c)}{\pi(\mathbf{0})}.
\end{aligned}
\end{equation}
Let
\begin{equation}
    \tau_{B^c}
    =
    \inf\{s\ge 0:X_s\notin B\}
\end{equation}
be the first time at which a given trajectory leaves $B$. If $\tau_{B^c}\le t$,
then the trajectory must cross from $B$ to $B^c$ during one of the first $t$
steps. A union bound therefore gives
\begin{equation}
\begin{aligned}
    \mathbb P_{\mathbf 0}(\tau_{B^c}\le t)
    &\le
    \sum_{s=0}^{t-1}
    \mathbb P_{\mathbf 0}(X_s\in B,X_{s+1}\notin B)
    \\
    &\le
    t\,\frac{Q(B,B^c)}{\pi(\mathbf 0)}.
\end{aligned}
\end{equation}

We can choose a decoder such that every configuration in $B$ is decoded to the zero codeword. 
The decoding-failure probability therefore
satisfies
\begin{equation}
    p_{\mathrm{fail}}(t)
    \le
    \mathbb P_{\mathbf 0}(X_t\notin B)
    \le
    \mathbb P_{\mathbf 0}(\tau_{B^c}\le t)
    \le
    t\,\frac{Q(B,B^c)}{\pi(\mathbf 0)}.
\end{equation}
Thus $p_{\mathrm{fail}}(t)<1/4$ for every integer $t$ satisfying
\begin{equation}
    t
    <
    \frac{\pi(\mathbf 0)}{4Q(B,B^c)},
\end{equation}
which yields
\begin{equation} \label{eq:memorytimebound}
    t_{\mathrm{mem}}(1/4)
    \ge
        \frac{\pi(\mathbf 0)}{4Q(B,B^c)}.
\end{equation}

We now combine this estimate with the boundary-flow bound obtained in \eqref{eq:boundary flow}. Together with $\pi(\mathbf 0)=1/Z$, the memory time is therefore lower bounded by \eqref{eq:PBCtmix}.
\end{proof}

Figure~\ref{fig:PBC_tmix} shows that the mean hitting time to the region closer to a different codeword than to the initial codeword grows approximately exponentially with \(q\) for several values of \(\beta\). The mean hitting time gives a lower bound on the mixing time up to a constant factor, provided that this region carries a constant fraction of the Gibbs weight~\cite{PeresSousi2015,Oliveira2012}. This scaling is consistent with the exponentially large memory time established in Proposition~\ref{prop:PBCslowmixing}.

From the perspective of statistical physics, the helical repetition code serves as a simple ``counterexample" to the lore that a discrete global symmetry cannot be broken spontaneously in one spatial dimension~\cite{landau1980statistical,BergamaschiChen2025fastmixing}.   We remind the reader of the argument for this lore, considering a model with periodic boundary conditions as we have been discussing in this subsection.   If we had a finite alphabet of size $q$ on each site, then creating a pair of domain walls between two sectors related by the broken symmetry would only cost a finite energy $E_{\mathrm{dw}}$, but we would gain an entropy $2 \log n$ (with periodic boundary conditions) for where to place the two domain walls.  The free-energy cost $F_{\mathrm{dw}} = E_{\mathrm{dw}}-2T\log n <0$ at any $T>0$, so long as $n$ is large enough; therefore we will always nucleate some low but finite density of domain walls, destroying long-range order.

One can evade this conclusion in a rather trivial way by letting the domain wall tension itself grow with system size. For example, in a one-dimensional Ising chain with coupling $J$, the energy cost of creating a domain wall pair is $E_{\text{dw}}=4J$. If $J\gg T\log{n}$, domain walls are suppressed in one dimension. Equivalently, the Ising correlation length $\xi=(\log\coth{\beta J})^{-1}\sim e^{2\beta J}$ is parametrically larger than $n$ under this scaling of $J$, leading to long-range order over the chain. 

In the PBC helical repetition code, $q$ grows exponentially with $n$. This also invalidates Peierls's argument~\cite{Peierls1936} against spontaneous symmetry breaking in one dimension. However, the mechanism here is not simply that the local cost of creating domain walls is large: a single-site update of the PBC helical code changes the energy by at most $O(a+b)$, and thus the energy landscape is locally flat. Consequently, the finite-temperature Gibbs measure is not sharply localized on the exact codewords themselves, but instead supports substantial local fluctuations within the neighborhood of each codeword sector. The large barrier arises when a sequence of local updates escapes such a neighborhood. From the physics perspective, the helical repetition code is symmetric under shifts of the state of the system by codewords, $\mathbf{x} \rightarrow \mathbf{x} + \mathbf{c}$ with $\mathbf{c}\in\mathcal{C}$, as is any linear error-correcting code. For the PBC code with $q=a^n-b^n$, this symmetry group is $\mathcal{C}\cong\mathbb{Z}_q$. For $b=1$, Hu and Watanabe studied the breaking of this symmetry at zero temperature in a generalized quantum clock model~\cite{HuWatanabe2023clock}. Linear soundness produces a free-energy bottleneck between different symmetry-related sectors, which is of the form
\begin{equation}
    \Delta F\gtrsim \Omega(q)-T\,O(n\log{q}),
\end{equation}
where the entropy term comes from counting the configurations on the shell of the Lee ball. At any finite temperature, the energetic term dominates as $n\to\infty$, giving measure condensation in the Gibbs measure.

\begin{figure}
    \centering
    \includegraphics[width=0.5\linewidth]{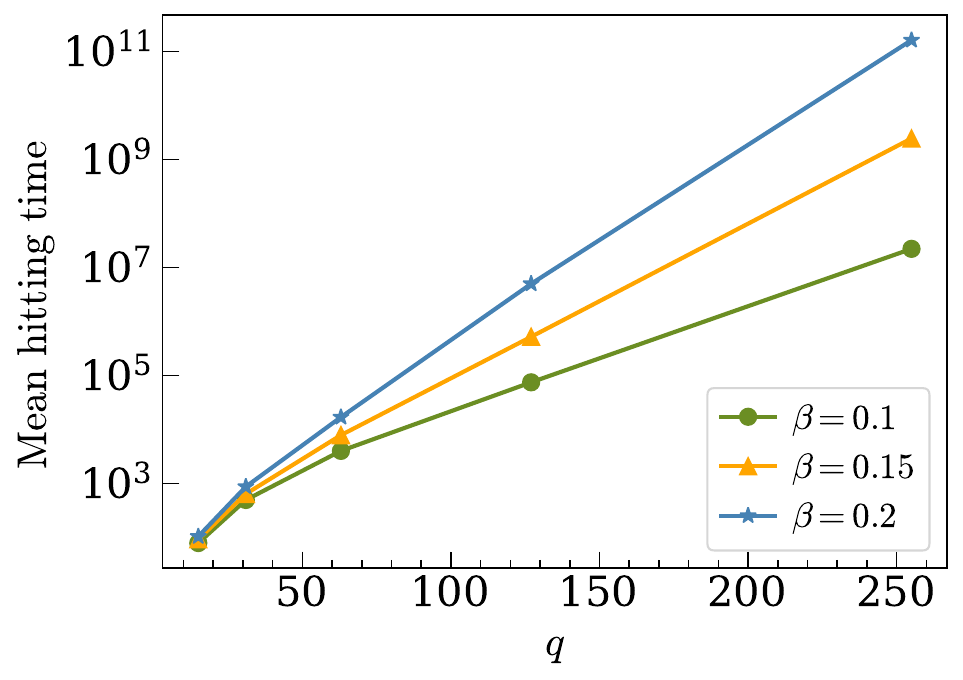}
    \caption{Numerically estimated mean hitting time of the PBC helical repetition code $H_{\mathrm{PBC}}(2,-1)$ as a function of $q=2^{n}-1$ at different inverse temperatures $\beta$. Starting from the all-zero codeword, we evolve the system under the Monte Carlo Metropolis algorithm and record the first time when the state becomes closer, in Lee distance, to the neighboring codeword $\mathbf{v}=(1,2,\dots,2^{n-1})$ than to the initial codeword. The data are plotted on a semilog scale, with a logarithmic vertical axis.}
    \label{fig:PBC_tmix}
\end{figure}

We next embed the same mechanism of the PBC helical code into a quantum hypergraph product code with mixed boundary conditions. We show that one of the error sectors can still exhibit slow mixing with exponentially large memory time at arbitrarily high temperature.

\begin{prop}[Exponentially large memory time for $Z$ errors in mixed-boundary helical HGP codes] \label{prop:PBCOBCHGP}
Fix two pairs of coprime positive integers
\begin{equation}
    a_h>b_h\ge 1,
    \qquad
    a_v>b_v\ge 1.
\end{equation}
Let $H_h=H_{\mathrm{PBC}}(a_h,-b_h)\in \mathbb{Z}^{n\times n}$ be the parity-check matrix of a helical
repetition code with periodic boundary conditions, and let $H_v=H_{\mathrm{OBC}}(a_v,-b_v)\in \mathbb{Z}^{(n-1)\times n}$ be the parity-check matrix of a helical repetition code with open boundary conditions. 
Assume that $q$ satisfies
\(
    a_h^n \mathrel{\mathchar"3211} b_h^n \pmod{q}
\)
and
\(
    \gcd(q,a_hb_ha_vb_v)=1.
\)

Construct the CSS hypergraph-product code over $\mathbb Z_q$ with check matrices
\begin{equation}
    H_X =
    \left(
    H_h\otimes I_n
    \,\middle|\,
    I_n\otimes H_v^{\mathsf T}
    \right), \quad
    H_Z =
    \left(
    -I_n\otimes H_v
    \,\middle|\,
    H_h^{\mathsf T}\otimes I_{n-1}
    \right).
\end{equation}
Then $H_XH_Z^{\mathsf T}=0$, so this defines a CSS code over $\mathbb Z_q$ on
\(
    N=2n^2-n
\)
physical qudits.

Consider the single-site Metropolis dynamics on $Z$-error configurations, in which each time step changes one coordinate by $+1$ or $-1$. There exist constants
$\kappa_h,\eta,\rho>0$, depending only on $a_h,b_h,a_v,b_v$, such that the memory time satisfies
\begin{equation} \label{eq:tmixPBCOBCHGP}
    t_{\mathrm{mem}}(1/4)
    \ge
    \frac{1}{4}
    \exp\!\left[
        \beta\bigl(\kappa_h( \rho q-1) - \eta n(n-1)\bigr)
        -N\log\!\bigl(2( \rho q+N)\bigr)
    \right].
\end{equation}
In particular, for any fixed inverse temperature $\beta>0$, along any family of
parameters satisfying
\begin{equation}
    q\gg n^2\log (q+n),
\end{equation}
the memory time in the $Z$-error sector is exponentially large.
\end{prop}

The proof of the proposition is given in Appendix \ref{app:proofPBCOBCHGP}. The essential part of the proof is that, as for the classical PBC code~(cf.\ Ref.~\cite{ZhaoDohertyKim2025} for qubit hypergraph products), the $Z$-error sector has an $O(q)$ energy barrier and approximate linear confinement around the codewords. The boundary of a Lee ball centered at the zero error configuration provides a free-energy barrier that leads to the bound~\eqref{eq:tmixPBCOBCHGP} on the memory time.

\subsection{Memory from ``entropic'' effects} \label{subsec:entropiceffect}

In the preceding subsection, we showed that soundness and linear confinement produce an $O(q)$ energy barrier, leading to large memory time at any fixed finite temperature when $q$ grows sufficiently rapidly with $n$. In this subsection, we establish an exponentially large memory time through a qualitatively different mechanism: the code has only a small $O(n)$ energy barrier that does not depend on $q$, but the energy-flat logical direction is exponentially long, so reaching another logical sector along it takes exponentially many local updates, while leaving it transversally costs an energy of order $q$. Entropic barriers were proposed earlier for $\mathbb{Z}_q$ quantum doubles with a grid of defect lines~\cite{BrownAlShimaryPachos2014}.  Komar et al.\ ruled out entropic protection in Abelian quantum doubles under Davies dynamics by an Arrhenius bound whose generalized energy barrier~\cite{Temme2017thermalization} is at most the energy of two anyons~\cite{Komar2016energybarrier}.  In our codes, by contrast, a single violated check can cost an energy of up to $q/2$, which grows with $n$. We first introduce a classical code exhibiting this mechanism and then present the corresponding quantum code.

\begin{prop}[Arbitrarily high-temperature memory for open-boundary helical repetition codes] \label{prop:classicalOBC}
Fix coprime integers $a>b\ge 1$. For each $n$, consider the length-$n$
open-boundary $(a,-b)$ helical repetition code over $\mathbb Z_q$, defined by
the parity-check matrix
$H_{\mathrm{OBC}}(a,-b)\in \mathbb Z_q^{(n-1)\times n}$.
Let
\begin{equation}
    \pi(x)\propto \exp\bigl(-\beta |H_{\mathrm{OBC}}(a,-b)x|_{\mathrm{L}}\bigr),
    \qquad
    x\in \mathbb Z_q^n,
\end{equation}
and consider the corresponding Metropolis dynamics, in which each time step changes one coordinate by $+1$ or $-1$.

Then there exist positive constants $c,C>0$, depending only on $a$ and $b$,
such that for all sufficiently large $n$, whenever
\begin{equation}
    q
    \ge
    Cn\max\left\{
        1,\,
        \frac{1}{\beta}\log\left(2+\frac{1}{\beta}\right)
    \right\},
\end{equation}
the memory time satisfies
\begin{equation}
    t_{\mathrm{mem}}(1/4)\ge \exp(cn).
\end{equation}
In particular, this family of codes has exponentially large memory time at every finite temperature whenever $q$ grows superlinearly with $n$.
\end{prop}

\begin{proof}

The complete proof of the proposition is given in
Appendix~\ref{app:proofclassicalOBC}. Here, we only describe the geometric
picture underlying the proof. When the model is lifted to $\mathbb Z$, the kernel
of $H_{\mathrm{OBC}}(a,-b)$ is generated by the integer codeword
\begin{equation}
    v=\left(b^{n-1},ab^{n-2},\ldots,a^{n-1}\right).
\end{equation}
Taking $\mathbb Rv$ as the axis and an $O(q)$-radius neighborhood in the
transverse directions defines a cylinder in $\mathbb Z^n$. We take the length of the cylinder equal to that of $v$, so that the cylinder contains exactly one codeword.

Translation along the axis leaves the syndrome unchanged, so the energy is
completely flat in this direction. By contrast, the energy grows at least
linearly with the distance from the axis, producing linear confinement in the
transverse directions. The side of the cylinder therefore has energy of order
$q$, while the bulk extends along the energy-flat axis.

Reducing each coordinate modulo $q$ folds the cylinder into $\mathbb Z_q^n$.
The cylinder can be chosen sufficiently narrow that it does not self-intersect
under this reduction and retains the energy landscape over $\mathbb Z$. Its
boundary consists of two caps and the transverse side. Physically, a state
inside the cylinder has only two possible escape routes. It can travel along
the low-energy axis toward a different codeword, but this direction remains
exponentially long even after folding, so reaching a cap under the local
dynamics takes an exponentially long time. Alternatively, it can escape through the
transverse side, but linear confinement forces it to cross a large energy
barrier. 
Therefore, for any fixed $\beta>0$, when $q$ satisfies the conditions of Proposition \ref{prop:classicalOBC}, the memory time is exponentially large.

\end{proof}

Beyond the exponentially large memory time established in Proposition~\ref{prop:classicalOBC}, in Appendix~\ref{app:freeenergylandscape}, we study the equilibrium free-energy landscape of the corresponding integer model along the logical direction. Adding any integer multiple of the codeword $v$ to a configuration leaves its syndrome and energy unchanged, so it suffices to consider a single period of this translation. Within one period, the value of
the last site ranges over $a^{n-1}$ integers. For each such value, we fix the
last site and sum the Gibbs weights of all compatible configurations. The
resulting fixed-endpoint partition functions, after normalization, give the
equilibrium Gibbs distribution along the logical direction. 

In the appendix, we prove that, for every fixed $\beta>0$, these $a^{n-1}$
fixed-endpoint partition functions are far from uniform: the ratio between the
largest and their average grows exponentially with $n$. Nevertheless, the
normalized Gibbs probability of every individual endpoint value decreases
exponentially with $n$. The Gibbs weight is therefore neither approximately
uniform along the logical direction nor concentrated on one or finitely many
endpoint values. Any set of endpoint values carrying a fixed positive fraction of the total Gibbs weight must contain exponentially many values.

This result differs from the usual outcome of order by disorder~\cite{Villain1980order}, in which thermal fluctuations select particular states from an energy-degenerate manifold.
The resulting distribution of Gibbs weight is reminiscent of spin-glass shattering, but arises without disorder. Since neighboring values remain connected by local updates, the model also does not exhibit fragmentation. It therefore realizes an equilibrium structure in which endpoint values are distinguished by their Gibbs weights without being reduced to a finite set or divided into geometrically separated sectors.

We next show that the HGP code constructed from two OBC helical codes can, through the same mechanism, have an exponentially large memory time in one error sector at arbitrarily high fixed temperature.

\begin{prop} [Exponentially large memory time for $Z$ errors in open-boundary helical HGP codes] \label{prop:OBCOBCHGP}

Fix two pairs of coprime positive integers
\begin{equation}
    a_h>b_h\ge 1,
    \qquad
    a_v>b_v\ge 1,
    \qquad
    a_h>a_v.
\end{equation}
Let
\begin{equation}
    H_h=H_{\mathrm{OBC}}(a_h,-b_h)\in\mathbb Z^{(n-1)\times n},
    \qquad
    H_v=H_{\mathrm{OBC}}(a_v,-b_v)\in\mathbb Z^{(n-1)\times n}
\end{equation}
be the parity-check matrices of OBC helical repetition codes. Assume
that
\begin{equation}
    \gcd(q,a_hb_ha_vb_v)=1
\end{equation}
and
\begin{equation}
    q>2(a_h+b_h+a_v+b_v).
\end{equation}
Construct the CSS hypergraph-product code over $\mathbb Z_q$ on
\(
    N=n^2+(n-1)^2
\)
physical qudits with check matrices
\begin{equation}
    H_X=
    \left(
        H_h\otimes I_n
        \,\middle|\,
        I_{n-1}\otimes H_v^{\mathsf T}
    \right),
\end{equation}
\begin{equation}
    H_Z=
    \left(
        -I_n\otimes H_v
        \,\middle|\,
        H_h^{\mathsf T}\otimes I_{n-1}
    \right).
\end{equation}

Consider the nearest-neighbor single-site Metropolis dynamics on $Z$-error
configurations, in which each time step changes one coordinate by $+1$ or $-1$,
with energy
\begin{equation}
    E_Z(z)=|H_Xz|_{\mathrm{L}}.
\end{equation}
Then there exist positive constants $c,\gamma,B >0$,
depending only on $a_h,b_h,a_v,b_v$, such that
\begin{equation} \label{eq:tmixOBCOBCHGP}
    t_{\mathrm{mem}}(1/4)
    \ge
    c\min\left\{
        \e^{\gamma n},
        \exp\!\left[
            \beta\left(\frac{q}{2}-B\right)
            -
            4n^2\log\left(2+\frac{q}{n^2}\right)
        \right]
    \right\}.
\end{equation}
In particular, for any fixed inverse temperature $\beta>0$, along any family of
parameters satisfying
\begin{equation}
    q\gg n^2\log\left(2+\frac{q}{n^2}\right),
\end{equation}
the memory time of $Z$-error Metropolis dynamics is exponential in $n$.
    
\end{prop}
\begin{proof} Given in Appendix \ref{app:proofOBCOBCHGP}. The proof is similar to the classical case: we first lift the code to $\mathbb Z$ and define a generalized cylinder with some gauge-invariant quantities for each codeword in the $Z$-error sector. The generalized cylinder is exponentially long, so that reaching a cap requires exponentially many local updates. Reaching the side of the cylinder is suppressed by the large energy it costs. Consequently, the side is reached with exponentially small probability within an exponentially long time. Reduction modulo $q$ folds the cylinder into $\mathbb Z_q^N$, but the cylinder does not intersect itself and the energy of every configuration inside the cylinder is also preserved, so its boundary remains a valid bottleneck, which gives the memory-time bound \eqref{eq:tmixOBCOBCHGP}.
\end{proof}

\subsection{On the possibility of high-temperature memory}\label{subsec:quantumobstruction}

The bounds established in Section~\ref{subsec:quantumbounds} might appear to rule out reliable decoding of quantum Lee-LDPC stabilizer codes at sufficiently high temperature. Theorem~\ref{thm:baddistance} shows that every noncommuting logical pair admits a noncommuting logical witness of Lee weight $O(n)$. Theorem~\ref{thm:energybarrier} shows that every nontrivial quantum CSS Lee-LDPC code has a nontrivial logical operator with an $O(n)$ energy barrier. A tempting argument would be that, at a sufficiently high error rate, the number of physical errors occurring before any correction can greatly exceed the code distance, so that any code with distance $O(n)$ should fail when the ratio between the number of errors and the distance is sufficiently large. However, this argument is not valid, and these bounds do not, by themselves, constitute a no-go theorem for passive quantum memories with arbitrarily high thresholds.

We discuss two independent limitations to Theorem~\ref{thm:baddistance} in this subsection. First, the transference bound depends on the $\mathbb{Z}_q$ Pauli algebra of the elementary physical errors. Second, bounds on distance and energy barriers alone do not determine the behavior of a specified stochastic dynamics.

We now illustrate the first limitation with a counterexample, in which the distances of a quantum code with bounded-weight stabilizers can be $O(q)$ for a more general ``nice error basis" \cite{Knill1996nice}.  Unfortunately we are not sure how to realize this construction in the lab, but it does immediately lead to a quantum memory up to temperature $T_c \sim q$:

\begin{eg}\label{eg:highTmemory}
Consider a stabilizer code on two $q$-dimensional qudits and one qubit, with Hilbert space $\mathbb{C}^q \otimes \mathbb{C}^q \otimes \mathbb{C}^2$. Let $q$ be odd, and let $X_i$ and $Z_i$ denote the Pauli operators on the $i$-th tensor factor. Consider the stabilizer code generated by
\begin{equation}
S_X=X_1X_2,\qquad S_Z=Z_1Z_2^{-1}.
\end{equation}
This code encodes one qubit, with $X_3$ and $Z_3$ as logical $X$ and logical $Z$, respectively. Take the elementary physical errors to be
\begin{equation}
E_1=X_1X_3,\qquad E_2=Z_1Z_3,\qquad E_3=X_2X_3,\qquad E_4=Z_2Z_3.
\end{equation}
We define the weight of an operator as the minimum number of elementary physical errors required to generate it, up to phase. Since
\begin{equation}
S_X=E_1E_3,\qquad S_Z=E_2E_4^{-1},
\end{equation}
both stabilizer generators have weight two. Moreover,
\begin{equation}
E_1^q=E_3^q=X_3,\qquad E_2^q=E_4^q=Z_3,
\end{equation}
and it can be checked that the distance of the code is $q$.
\end{eg}

The elementary physical errors in this example do not have the $\mathbb{Z}_q$ Pauli algebra assumed in Theorem \ref{thm:baddistance}, so the theorem cannot be applied here. The restriction to odd $q$ is also essential to this code. If $q$ is even, the elementary physical errors cannot generate a nontrivial logical operator, and the resulting error dynamics cannot connect different logical sectors.  

\begin{prop}\label{prop:highTmemory}
    The code of Example~\ref{eg:highTmemory} is a high-temperature memory for the elementary physical errors $E_1,\ldots,E_4$. For any $\epsilon>0$, if \begin{equation}
        \frac{1}{\beta} < \frac{q}{\log^{1+\epsilon}q}, \label{eq:beta_q_epsilon}
    \end{equation}
    then for sufficiently large $q$ the memory time is parametrically longer than it is at infinite temperature. In particular, the memory time scales as \begin{equation}\label{eq:pathology}
        t_{\mathrm{mem}} = \mathrm{\Theta}\left(  \frac{\mathrm{e}^{\beta (q-1)/2}}{2(q-1)}\right).
    \end{equation}
\end{prop}
\begin{proof}
    We can bound the memory time assuming that we undergo random $E_{1,2,3,4}^\pm 1$ errors with Metropolis sampling rates, using the same bottleneck theorem as in Proposition \ref{prop:PBCslowmixing}.  Consider the set of errors of total ``Lee weight'' exactly equal to $\frac{q-1}{2}$.   We may use the stabilizers to reduce the error set to only $E_1$ and $E_2$ up to stabilizer-equivalent errors.   There are then exactly $2(q-1)$ such errors, as one may easily count.  The Gibbs probability of this set of error configurations is \begin{equation}
        \pi_\beta(\text{bottleneck}) = \frac{1}{Z(\beta)} \cdot 2(q-1)\e^{-\beta (q-1)/2}.
    \end{equation}
    The Markov chain bottleneck theorem~\cite[Thm.~7.4]{LevinPeresWilmer2017} implies \eqref{eq:pathology}.

    At $\beta=0$, there is no bottleneck, and the mixing time is algebraic:  $t_{\mathrm{mem}} = \mathrm{\Theta}(q^2)$ is limited by the random walk time across a two-dimensional torus of $E_{1,2}$ errors.  When $\beta$ obeys  \eqref{eq:beta_q_epsilon}, the memory time grows faster than any power law in $q$.
\end{proof}

We now show the second limitation: $O(n)$ distance and an $O(n)$ energy barrier do not by themselves preclude an arbitrarily high threshold when $q$ is allowed to grow. Although we have not found a quantum example, the one-dimensional repetition code over $\mathbb{Z}$ provides a classical counterexample to such a no-go argument. With the parity-check matrix
\begin{align}
H_n=
\begin{pmatrix}
1 & -1 &        &        \\
  & 1  & -1     &        \\
  &    & \ddots & \ddots \\
  &    &        & 1 & -1
\end{pmatrix}
\in\mathbb Z^{(n-1)\times n},
\end{align}
the corresponding length-$n$ code is
\begin{align}
\mathcal C_n=\ker_{\mathbb Z}H_n
=\{z(1,\ldots,1):z\in\mathbb Z\},
\end{align}
whose distance in the $\ell_1$ metric is $n$.
Starting from the zero codeword, consider the infinite-temperature dynamics in which each elementary time step selects one site uniformly at random and changes its value by $+1$ or $-1$ with equal probability. No correction is performed during storage. After $T$ steps, let $S_T=\frac{1}{n}\sum_{i=1}^n x_i(T)$ denote the average value of the configuration $x(T)$. A global decoder then maps the configuration to the codeword
\begin{align}
    \operatorname{round}\!\left(S_T\right)(1,\ldots,1),
\end{align}
with any fixed convention at half-integers. The initial codeword is recovered whenever
\(
|S_T|<1/2.
\)
If $T=un$, where $u$ is any fixed finite constant, a standard concentration bound~\cite{probability} gives
\begin{align}
\Pr[\text{logical error}]
\leq
2\exp\!\left(-\frac{n^2}{8T}\right)
=
2\exp\!\left(-\frac{n}{8u}\right).
\end{align}
Consequently, for every fixed $u$, however large, the decoder returns the correct codeword with probability approaching $1$ as $n\to\infty$. The memory time of the code is therefore
\(
t_{\mathrm{mem}}=\Theta(n^2).
\)
Consequently, even without any correction, this code can be regarded as a passive memory with an arbitrarily high threshold, even though, with energy \(E(x)=\lVert H_nx\rVert_1\), its energy barrier is just $1$. The initial codeword can thus be recovered even when the number of errors greatly exceeds the distance, although the code has only $O(n)$ distance and an $O(1)$ energy barrier.

Our results thus neither establish nor rule out the existence of quantum Lee-LDPC stabilizer codes with an arbitrarily high passive-memory threshold arising from a mechanism other than those of Sections~\ref{subsec:energybarrier} and~\ref{subsec:entropiceffect}, which protect only the $Z$-error sector of our quantum codes. Resolving this question requires a direct analysis of specified local thermal dynamics. We leave this problem for future work.

\section{Outlook} \label{sec:outlook}
We have presented a comprehensive analysis of many aspects of quantum error correction using quantum codes on qudits subject to small clock and shift errors.  The resulting Lee metric used to analyze the performance of these codes captures the robustness of these codes to such errors more faithfully than the Hamming metric used in the existing literature.

From a coding-theoretic perspective, these codes can use a Hilbert space of the same dimension more effectively to protect against more errors, simply because their ``error graph'' has lower connectivity than that of a qubit code subject to arbitrary single-qubit Pauli errors.  

From a physics perspective, this same reduction in the error graph connectivity could -- in certain instances -- lead to very unusual properties, such as high-temperature self-correction (thermal spontaneous symmetry breaking) in local one-dimensional models, in contradiction with the naive lore of statistical physics \cite{vanhove}.  This issue will be discussed further in Ref.~\cite{andy_dlr} for a closely related model, albeit with a strictly infinite on-site state space ($\mathbb{Z}^2$, rather than $\mathbb{Z}_q$).  Our work therefore highlights the potential applications in error correction of a recent set of works studying unconventional ``modulated'' symmetries (such as the helical symmetries of our codes) that mix translation with flavor symmetries; see, e.g., Refs.~\cite{Sala2022modulated,Pace:2024tgk}.  On the other hand, the transference inequalities provided a fundamental no-go theorem against the immediate application of these ideas to \emph{quantum} memories in Lee-LDPC codes.

It is, at least conceptually, straightforward to extend the Lee-distance framework developed here to more general ``nice error bases'' in the sense of Knill~\cite{Knill1996nice}.  The more general setting does lead to some very intriguing possibilities, such as the high-temperature memory of Proposition~\ref{prop:highTmemory}.  A more systematic understanding of what is possible is an exciting goal for future work.

From a practical quantum-engineering perspective, a major limitation of the present analysis, which future work should address, is that it does not consider the performance of Lee-LDPC codes under circuit-level noise.  In particular, one would hope to identify, if possible, syndrome extraction circuits that do not convert small Lee-weight errors into larger errors.  Again, this analysis may need to be revisited for any error model more complicated than the clock-and-shift model analyzed in this paper. 

\section*{Acknowledgments}
We thank Sarah Budai, Wesley C. Campbell, Zohar Komargodski, Richard Kueng, Vincent Mourik, Marvin Qi, and Martin Ringbauer
for useful discussions.
JG, YH, TL, and AL were supported by the Air Force Office of Scientific Research under Grant FA9550-24-1-0120.
VVA acknowledges NSF grants OMA-2120757 (QLCI) and CIF-2330909.
This
manuscript was prepared with the assistance of ChatGPT, developed by OpenAI, and Claude, developed by
Anthropic. These tools were used to teach us about the geometry of numbers \cite{transference}, improve clarity, and generate
figures, in accordance with the authors’ instructions. All
content, claims, and conclusions have been reviewed and
verified by the authors to ensure accuracy and originality. Certain products, commercial and otherwise, are
mentioned in this publication. These mentions are for
informational purposes only, and do not imply recommendation or endorsement by NIST.

\begin{appendix}

\makeatletter


\let\@sectioncntformat\@seccntformat

\let\@hangfrom@section\@hang@from


\renewcommand\section{\@ifstar{\AndySectionStar}{\@startsection{section}{1}{\z@}%

  {-3.5ex \@plus -1ex \@minus -.2ex}%

  {2.3ex \@plus .2ex}%

  {\normalfont\Large\bfseries}}}

\renewcommand\subsection{\@ifstar{\AndySubsectionStar}{\@startsection{subsection}{2}{\z@}%

  {-3.25ex\@plus -1ex \@minus -.2ex}%

  {1.5ex \@plus .2ex}%

  {\normalfont\large\bfseries}}}

\renewcommand\subsubsection{\@ifstar{\AndySubsubsectionStar}{\@startsection{subsubsection}{3}{\z@}%

  {-3.25ex\@plus -1ex \@minus -.2ex}%

  {1.5ex \@plus .2ex}%

  {\normalfont\normalsize\bfseries}}}

\makeatother

\renewcommand{\thesubsection}{\thesection.\arabic{subsection}}
\renewcommand{\theequation}{\thesection.\arabic{equation}}
\renewcommand{\thesubsubsection}{\thesubsection.\arabic{subsubsection}}
\makeatletter
\renewcommand\section{\@ifstar{\AndySectionStar}{\@startsection {section}{1}{-\parindent}%
  {-3.5ex \@plus -1ex \@minus -.2ex}%
  {2.3ex \@plus .2ex}%
  {\normalfont\Large\bfseries}}}
\renewcommand\subsection{\@ifstar{\AndySubsectionStar}{\@startsection{subsection}{2}{-\parindent}%
  {-3.25ex\@plus -1ex \@minus -.2ex}%
  {1.5ex \@plus .2ex}%
  {\normalfont\large\bfseries}}}
\renewcommand\subsubsection{\@ifstar{\AndySubsubsectionStar}{\@startsection{subsubsection}{3}{-\parindent}%
  {-3.25ex\@plus -1ex \@minus -.2ex}%
  {1.5ex \@plus .2ex}%
  {\normalfont\normalsize\bfseries}}}
\makeatother

\section{Background and relation to previous work} \label{app:background}
This appendix gives further motivation and background for the Lee metric, and it describes the extent to which our formalism has appeared in the prior literature.

There is a natural connection, which we elucidate in Section~\ref{subsec:gkp}, between codes in the Lee metric and $q$-dimensional Gottesman--Kitaev--Preskill (GKP) qudits~\cite{Gottesman2001GKP}.  In a GKP qudit, a displacement by $j$ steps of the underlying grid in position or in momentum acts as $X^j$ or $Z^j$, respectively.\footnote{Such finite-dimensional GKP codes naturally arise in the context of geometrically quantizing a compact torus phase space, or alternatively, studying the lowest Landau level on a finite-size spatial torus.}  Error correction of GKP qudits with $q=3$ and $q=4$ has recently surpassed break-even~\cite{Brock2025qudit}, so such codes are of recent interest.  GKP codes are sensible when the dominant noise is weak displacement noise, i.e., mostly small powers of $X$ and $Z$, making the Lee metric relevant.    Similarly, the $q$ coherent states $|\alpha\omega^j\rangle$, with $\omega=\mathrm{e}^{2\pi \mathrm{i}/q}$ and $j\in\mathbb{Z}_q$, span a qudit that generalizes the cat codes~\cite{Cochrane1999cat,Leghtas2013cat,Mirrahimi2014cat,Grimsmo2020rotation}.  On these states, the annihilation operator acts as $\alpha Z$, and a rotation of phase space by $2\pi j/q$ acts as $X^j$.  The loss of $k$ photons thus acts as a multiple of $Z^k$, and small rotations act as small powers of $X$.  Multimode GKP codes defined by lattices protect several oscillators against small displacements~\cite{Harrington2001rates,Conrad2022GKP,Royer2022multimode}, while multimode generalizations of cat codes protect them against photon loss and small phase-space rotations~\cite{jain2024quantumh,xu2025lettinga}.  Concatenations of GKP qudits with qudit stabilizer codes have also been studied~\cite{Schmidt2022HigherGKP}.   

In the classical coding theory literature, the Lee metric~\cite{Lee1958}  was introduced a long time ago to analyze the robustness of linear codes over $\mathbb{Z}_q$ to small shift errors.   The Lee weight of an entry $x\in\mathbb{Z}_q$ is $\min(x,q-x)$, i.e., its distance from zero around a cycle of length $q$.  The Lee weight of a vector is the sum of the Lee weights of its entries.  The Lee and Hamming weights agree for $q\le3$, but they differ for every $q\ge4$.  Classical Lee codes include the perfect codes of Golomb and Welch~\cite{GolombWelch1970}, the negacyclic codes of Berlekamp~\cite{Berlekamp1968algebraic}, and Lee-metric BCH codes~\cite{RothSiegel1994}.  The Gray map sends Lee weight over $\mathbb{Z}_4$ to Hamming weight over $\mathbb{Z}_2$~\cite{Hammons1994}.  Related error models with limited magnitude arise in multilevel flash memories~\cite{Cassuto2010}.  One of the central results of this paper is that the behavior of classical and quantum codes in the Lee metric can qualitatively differ, in particular regarding the distance of Lee-LDPC codes. 

Topological phases of matter give a separate reason to study qudits of composite dimension.  In two dimensions, translation-invariant topological Pauli stabilizer codes on qudits of prime dimension are equivalent to stacks of toric codes~\cite{Haah2021classification}.  On qudits of composite dimension, Pauli stabilizer codes realize every Abelian anyon theory that admits a gapped boundary~\cite{Ellison2022twisted}.  The primary example is a double-semion code on $\mathbb{Z}_4$ qudits, obtained by condensing a boson in the $\mathbb{Z}_4$ toric code.  Pauli topological subsystem codes on such qudits go further and realize every Abelian anyon theory~\cite{Ellison2023subsystem}.  Their examples on $\mathbb{Z}_4$ qudits include the chiral semion theory, which admits neither a gapped boundary nor a topological stabilizer code.  Every composite $q$ is at least four, so the Lee and Hamming metrics differ on all of these qudits.  If such codes are built from GKP qudits with $q=4$~\cite{Brock2025qudit}, their dominant errors are small shifts, and their distances are naturally measured in the Lee metric.  Moreover, the toric code naturally contains low Lee-weight stabilizers.  Despite this fact, with the notable recent exception of Ref.~\cite{Vijay2025information}, authors have not treated small vs.\ large shift errors differently when studying the robustness of such toric codes or quantum double models.

The stabilizer formalism extends to qudits of any dimension~\cite{Knill1996nice, AshikhminKnill2001, Gottesman1999higher, Rains1999nonbinary, Ketkar2006nonbinary}, but the distance of a qudit code is usually measured in the Hamming metric.
Several lines of work nevertheless exploit the structure that the Lee metric captures, in addition to the GKP codes described above.   Some codes for qudit amplitude damping are built from classical codes that correct a single error in the Lee metric~\cite{Grassl2018amplitude}.  Ling and Sol\'e used the Gray map to build nonadditive qubit codes from $\mathbb{Z}_4$-linear codes, with the qubit distance set by Lee distances of the $\mathbb{Z}_4$ codes~\cite{LingSole2008}.  Grassl and R\"otteler built nonadditive codes from the binary Goethals and Preparata codes~\cite{GrasslRoetteler2008,GrasslRoetteler2008ITW}, variants of which are Gray images of linear codes over $\mathbb{Z}_4$~\cite{Hammons1994}.  CSS codes have also been studied in the Mannheim metric over the Gaussian integers~\cite{Ozen2012Mannheim}.  A subfield metric that distinguishes diagonal from non-diagonal qudit errors was proposed for asymmetric quantum codes~\cite{Grassl2024subfield}.  Homological rotor codes store logical qubits in torsion~\cite{Vuillot2024rotor,XuWangAlbert2024rotor}, a feature that also appears in our codes for composite $q$.  

Closest to our work, Matsumoto proposed a quantum Lee weight for CSS codes over prime fields and derived a Gilbert--Varshamov-type bound for it~\cite{Matsumoto2021Lee}.  However, to our knowledge, no general theory of stabilizer codes in the Lee metric exists that covers non-CSS codes, composite $q$, fault-tolerant gates, and LDPC codes;  it is this gap that our paper fills.

\section{Proof of Theorem~\ref{thm:quantumcodespace}}\label{app:codespace}
Throughout this appendix, with a slight abuse of terminology, we refer to integer representatives of logical classes simply as logical operators.
The lift of the chain complex over $\mathbb{Z}_q$ to the chain complex of $D_{1,2,3}$ over $\mathbb{Z}$ appears to be a standard method in homological algebra known as the ``mapping cone"~\cite[Sec.~1.5]{weibel1994homological}.   We will go through the proof of the theorem ``slowly,'' however, since these results are likely unfamiliar to our intended audience. Although logical bases of composite-dimensional qudit codes have been obtained from normal forms over $\mathbb{Z}_q$~\cite{Gheorghiu2014standard,SarkarYoder2024pauli}, we are not aware of this integer lift in the literature on quantum error correction.   In what follows we will mostly present only ``one half'' of the proof, since it is easy to see that switching the roles of $X$ and $Z$ simply corresponds to taking the cochain complex.
\begin{lem}\label{lem:A1}
    $\mathcal{C}=\ker_{\mathbb{Z}_q} (H_X)/\mathrm{im}_{\mathbb{Z}_q}(H_Z^{\mathsf{T}}) \cong \ker_{\mathbb{Z}}(D_1)/\mathrm{im}_{\mathbb{Z}}(D_2)$.
\end{lem}
\begin{proof}
    To show the equivalence between these two spaces, we proceed as follows: \begin{subequations}\begin{align}
        \ker_{\mathbb{Z}}(D_1) &= \left\lbrace \left(\begin{array}{c} z \\ t\end{array}\right) \in \mathbb{Z}^{n+m_X} : H_Xz=qt \right\rbrace , \label{eq:kerZD1} \\
        \mathrm{im}_{\mathbb{Z}}(D_2) &= \left\lbrace \left(\begin{array}{c} H_Z^{\mathsf{T}}a+qb \\ q^{-1}H_XH_Z^{\mathsf{T}}a+H_Xb\end{array}\right) \in \mathbb{Z}^{n+m_X} : \left(\begin{array}{c} a \\ b\end{array}\right) \in \mathbb{Z}^{m_Z+n}  \right\rbrace . \label{eq:imZD2}
    \end{align}\end{subequations}
    Define the projection $\pi : \mathbb{Z}^{n+m_X} \rightarrow \mathbb{Z}_q^n$ via \begin{equation}
        \pi(z,t)  =z \; \text{(mod $q$)}.
    \end{equation}
    It follows immediately from \eqref{eq:kerZD1} that $\pi(\ker_{\mathbb{Z}}(D_1)) = \ker_{\mathbb{Z}_q}(H_X)$, and from \eqref{eq:imZD2} that $\pi(\mathrm{im}_{\mathbb{Z}}(D_2)) = \mathrm{im}_{\mathbb{Z}_q}(H_Z^{\mathsf{T}})$.    Therefore there is an induced homomorphism $\overline{\pi}:  \ker_{\mathbb{Z}}(D_1)/\mathrm{im}_{\mathbb{Z}}(D_2) \rightarrow \ker_{\mathbb{Z}_q} (H_X)/\mathrm{im}_{\mathbb{Z}_q}(H_Z^{\mathsf{T}}) $.

To complete the proof, we show that \(\overline{\pi}\) is an isomorphism. To prove injectivity, suppose that the class of \((z,t)\in\ker_{\mathbb{Z}}(D_1)\) is mapped by \(\overline{\pi}\) to the trivial class in \(\ker_{\mathbb{Z}_q}(H_X)/\mathrm{im}_{\mathbb{Z}_q}(H_Z^{\mathsf T})\). Then there exist \(a\in\mathbb{Z}^{m_Z}\) and \(b\in\mathbb{Z}^n\) such that
\begin{equation}
z=H_Z^{\mathsf T}a+qb.
\end{equation}
Since \(H_Xz=qt\), it follows that
\begin{equation}
t=q^{-1}H_XH_Z^{\mathsf T}a+H_Xb.
\end{equation}
Therefore,
\begin{equation}
\begin{pmatrix}
z\\
t
\end{pmatrix}
=
D_2
\begin{pmatrix}
a\\
b
\end{pmatrix},
\end{equation}
so \((z,t)\) is trivial in \(\ker_{\mathbb{Z}}(D_1)/\mathrm{im}_{\mathbb{Z}}(D_2)\), and \(\overline{\pi}\) is injective. To prove surjectivity, take any element of \(\ker_{\mathbb{Z}_q}(H_X)\) and choose an integer lift \(z\). Since \(H_Xz=0\pmod q\), there exists \(t\in\mathbb{Z}^{m_X}\) such that \(H_Xz=qt\). Hence \((z,t)\in\ker_{\mathbb{Z}}(D_1)\) and maps to the chosen element. Thus \(\overline{\pi}\) is an isomorphism.
\end{proof}

Now, because $\mathbb{Z}$ is torsion-free,\footnote{This means that $\ell x=0$ with a nonzero integer $\ell$ implies $x=0$.} we can readily identify what $\mathcal{C}$ is.  Take the SNF of $D_2$: \begin{equation}
    D_2 = \widetilde{U} \left(\begin{array}{ccccccc} {\color{red} d_1} &\ 0 &\ \cdots &\ 0 &\ 0 &\ \cdots &\ 0\\ 0 &\ {\color{red} d_2} &\ \cdots  &\ 0 &\ 0 &\ \cdots  &\ 0 \\ \vdots &\ \vdots &\ {\color{red} \ddots}  &\ \vdots &\ \vdots &\ \cdots  &\ \vdots \\ 0 &\ 0 &\ \cdots  &\ {\color{red} d_r} &\ 0 &\ \cdots  &\ 0 \\  0 &\ 0 &\ \cdots  &\ 0 &\ 0 &\ \cdots  &\ 0 \\ \vdots &\ \vdots &\ \cdots  &\ \vdots &\ \vdots &\ \cdots  &\ \vdots \\ 0 &\ 0 &\ \cdots  &\ 0 &\ 0 &\ \cdots  &\ 0 \\ \end{array}\right) \widetilde{V}.
\end{equation}
The directions highlighted in {\color{red} red}, up to the unimodular change of basis $\widetilde{U}$,  are the only possible contributions to $\ker_{\mathbb{Z}}(D_1)$. Otherwise, by Lemma~\ref{lem:A1}, $\mathcal{C}$ would be an infinite group (it would contain a subgroup isomorphic to $\mathbb{Z}$ if $D_1$ annihilated any vector orthogonal to $\mathrm{im}_{\mathbb{Z}}(D_2)$). As a result, $r=n$ because $\operatorname{rank} (D_1) = m_X$ and $\ker_{\mathbb{Z}}(D_1) \cong \mathbb{Z}^n$. 
Now consider, for example, the null vector \begin{equation}
    v_1 = \widetilde{U} \left(\begin{array}{c} {\color{red}1} \\ 0 \\ \vdots \\ 0 \end{array}\right) \in \ker_{\mathbb{Z}}(D_1).
\end{equation}
While $v_1 \in \ker_{\mathbb{Z}}(D_1)$, we see that the smallest positive multiple of $v_1$ contained in $\mathrm{im}_{\mathbb{Z}}(D_2)$ is $d_1v_1$. Therefore $\mathcal{C}$ contains a cyclic factor $\mathbb{Z}_{d_1}$ generated by the class of $v_1$.  
Repeating this argument for each $d_i$ we arrive at \eqref{eq:quantumcodespace}; indeed the above arguments show that  \begin{equation}
    \ker_{\mathbb{Z}}(D_1) = \mathrm{im}_{\mathbb{Z}}\left( \widetilde{U} \left(\begin{array}{cc} I_{n\times n} &\ 0_{n\times m_Z} \\ 0_{m_X\times n} &\ 0 _{m_X\times m_Z} \end{array}\right) \right). \label{eq:kerZSNF}
\end{equation}
Moreover, we can see that each $d_i$ must divide $q$, because on $\mathbb{Z}_q$, any $\overline{z} \in \ker_{\mathbb{Z}_q}(H_X) $ must obey $q\overline{z}=0\in \mathrm{im}_{\mathbb{Z}_q}(H_Z^{\mathsf{T}})$.  Hence, e.g., $qv_1 \in \mathrm{span}_{\mathbb{Z}}(d_1 v_1)$, which holds only if $d_1 \mid q$ ($d_1$ is a divisor of $q$); an analogous statement holds for the other $d_j$.

The above discussion also immediately shows us that the $Z$-logical operators, i.e., non-trivial elements of $\ker_{\mathbb{Z}}(D_1)$, are the torsional elements contained in \eqref{eq:kerZSNF} but not in $\mathrm{im}_{\mathbb{Z}}(D_2)$.   They are generated by the $j^{\mathrm{th}}$ column of $\widetilde{U}$ if and only if $d_j>1$, which implies \eqref{eq:zLogical}.

The $X$-logicals come from swapping the roles of $H_X$ and $H_Z$ in the chain complex.  Notice that this corresponds to taking the cochain complex $D_3^{\mathsf{T}} D_2^{\mathsf{T}}=D_2^{\mathsf{T}}D_1^{\mathsf{T}}=0$ \emph{and} permuting the rows and columns of $D_2^{\mathsf{T}}$: \begin{equation}
    D_{2,X\leftrightarrow Z} = \Sigma_{\mathrm{L}} D_2^{\mathsf{T}} \Sigma_{\mathrm{R}},
\end{equation}
where the permutation matrices $\Sigma_{\mathrm{L,R}}$ are \begin{subequations}
    \begin{align}
        \Sigma_{\mathrm{L}} &= \left(\begin{array}{cc} 0_{n\times m_Z} &\ I_n \\ I_{m_Z} &\ 0_{m_Z\times n}\end{array}\right), \\
        \Sigma_{\mathrm{R}} &= \left(\begin{array}{cc} 0_{n\times m_X} &\ I_n \\ I_{m_X} &\ 0_{m_X\times n}\end{array}\right).
    \end{align}
\end{subequations}
The permutation matrices and transpose operation do not change the Smith decomposition matrix $\widetilde{D}$, while due to the transpose and permutations the role of $\widetilde{U}$ is replaced by $\Sigma_{\mathrm{L}}\widetilde{V}^{\mathsf{T}}$.  This leads to the identification in \eqref{eq:xLogical}.

Lastly let us derive \eqref{eq:logicalConjugates}.  Define the two matrices $\widetilde{Z}, \widetilde{X} \in \mathbb{Z}^{n\times n}$ by:
\begin{subequations} \label{eq:matrixZXdef}
    \begin{align}
        \widetilde{Z}_{j,k} &:= \widetilde{U}_{j,k}, \\ 
         \widetilde{X}_{j,k} &:= \widetilde{V}_{k,j+m_Z},
    \end{align}
\end{subequations}
which encode the $Z/X$-logicals and stabilizers acting on the $n$ physical qudits alone (i.e., restricted to the $n$ qudit coordinates).  Plugging the definition \eqref{eq:matrixZXdef} into \eqref{eq:D2SNF} and comparing with \eqref{eq:D2def} we see that \begin{equation}
   \widetilde{Z} \widetilde{D}\widetilde{X}^{\mathsf{T}} = qI_n. \label{eq:takedeterminant}
\end{equation}  
Viewed as $\mathbb{Q}$-valued matrices, each of $\widetilde{Z}$, $\widetilde{D}$, and $\widetilde{X}$ is invertible.\footnote{Take the determinant of both sides of \eqref{eq:takedeterminant}.  If, e.g., $\det(\widetilde{Z})=0$, namely $\widetilde{Z}$ is not invertible, there is a contradiction.}  Inverting and rearranging the sides of \eqref{eq:takedeterminant} we obtain \begin{equation}
    q \widetilde{D}^{-1} = \widetilde{X}^{\mathsf{T}} \widetilde{Z},
\end{equation}
which, componentwise, is identical to \eqref{eq:logicalConjugates} if we choose $z_{\mathrm{L}}^{(j)} = \widetilde{Z}_{*,j}\; \text{(mod $q$)}$ and $x_{\mathrm{L}}^{(j)} = \widetilde{X}_{*,j}\; \text{(mod $q$)}$.

\section{Proof of Proposition~\ref{prop:decom}}\label{app:codespace_integer}
From this point on, we work exclusively with integer representatives in $\mathbb Z^n$. The
code and its logical operators are defined only over $\mathbb Z_q$. With a slight abuse of terminology, we will refer to the
chosen integer representatives in $\mathbb Z^n$ simply as logical operators, and
to the corresponding bases in $\mathbb Z^n$ simply as logical bases. Strictly
speaking, these vectors become logical operators and logical bases only after
reduction modulo $q$.

From the SNF of $H_{X/Z}$ \begin{align}\label{eq:SNF-H2}
  U_X H_X V_X &= D_X
  = \operatorname{diag}\bigl(d_1^X,\ldots,d_{r_X}^X,0,\ldots,0\bigr), \nonumber\\
  U_Z H_Z V_Z &= D_Z
  = \operatorname{diag}\bigl(d_1^Z,\ldots,d_{r_Z}^Z,0,\ldots,0\bigr),
\end{align}
if we define \begin{subequations}\label{eq:vXZwXZ}
    \begin{align}
        v_i^{X/Z} &= \left(V_{X/Z}\right)_{*,i}, \\
        w_i^{X/Z} &= \left(V_{X/Z}^{-1}\right)^\mathsf{T}_{*,i},
    \end{align}
\end{subequations}we can decompose 
\begin{subequations}
\begin{align}\label{eq:decom_Zn}
    \mathbb{Z}^n &= \ker_{\mathbb{Z}} H_X \oplus C_X = \operatorname{span}_{\mathbb Z}
    \left\{v_{r_X+1}^X,\ldots,v_n^X\right\} \oplus \operatorname{span}_{\mathbb Z}
    \left\{v_1^X,\ldots,v_{r_X}^X\right\} \\
    &=\ker_{\mathbb{Z}} H_Z \oplus C_Z = \operatorname{span}_{\mathbb Z}
    \left\{v_{r_Z+1}^Z,\ldots,v_n^Z\right\} \oplus \operatorname{span}_{\mathbb Z}
    \left\{v_1^Z,\ldots,v_{r_Z}^Z\right\} \\
    &=\operatorname{Sat}\!\left(\operatorname{im}_{\mathbb Z}
\left(H_X^{\mathsf T}\right)\right) \oplus R_X = \operatorname{span}_{\mathbb Z}
    \left\{w_1^X,\ldots,w_{r_X}^X\right\} \oplus \operatorname{span}_{\mathbb Z}
    \left\{w_{r_X+1}^X,\ldots,w_n^X\right\}  \\
    &=\operatorname{Sat}\!\left(\operatorname{im}_{\mathbb Z}
\left(H_Z^{\mathsf T}\right)\right) \oplus R_Z = \operatorname{span}_{\mathbb Z}
    \left\{w_1^Z,\ldots,w_{r_Z}^Z\right\} \oplus \operatorname{span}_{\mathbb Z}
    \left\{w_{r_Z+1}^Z,\ldots,w_n^Z\right\} ,
\end{align}
\end{subequations}
where $\operatorname{Sat}\!\left(\operatorname{im}_{\mathbb Z}
\left(H_{X/Z}^{\mathsf T}\right)\right)$ is the saturation of $\operatorname{im}_{\mathbb Z}
\left(H_{X/Z}^{\mathsf T}\right)$, which is defined by adding all vectors in $\mathbb{Z}$ proportional to any generator in $\operatorname{im}_{\mathbb Z}
\left(H_{X/Z}^{\mathsf T}\right)$ -- for example, adding $(2k+1)v$ if $2v$ is a generator of $\operatorname{im}_{\mathbb Z}
\left(H_{X/Z}^{\mathsf T}\right)$.

Since \(H_X H_Z^{\mathsf T}=0\) on $\mathbb{Z}$, we have $\operatorname{Sat}\!\left(
    \operatorname{im}_{\mathbb Z} H_X^{\mathsf T}
    \right)
    \subseteq
    \ker_{\mathbb Z} H_Z$ and $\operatorname{Sat}\!\left(
    \operatorname{im}_{\mathbb Z} H_Z^{\mathsf T}
    \right)
    \subseteq
    \ker_{\mathbb Z} H_X $.
It is always possible to choose the unimodular matrices $V_{X/Z}$ so that
\begin{subequations}
\begin{align}
    \ker_{\mathbb Z} H_Z
    =
    \operatorname{Sat}\!\left(
    \operatorname{im}_{\mathbb Z} H_X^{\mathsf T}
    \right)
    \oplus
    \operatorname{span}_{\mathbb Z}
    \left\{
    w_{r_X+1}^X,\ldots,w_{n-r_Z}^X
    \right\}, \label{eq:decom_kerHZ} \\
    \ker_{\mathbb Z} H_X
    =
    \operatorname{Sat}\!\left(
    \operatorname{im}_{\mathbb Z} H_Z^{\mathsf T}
    \right)
    \oplus
    \operatorname{span}_{\mathbb Z}
    \left\{
    w_{r_Z+1}^Z,\ldots,w_{n-r_X}^Z
    \right\} \label{eq:decom_kerHX}. 
\end{align}
\end{subequations}

With these decompositions, we now have \begin{subequations}\begin{align}
\ker_{\mathbb Z}(D_1)
&=
\left\{\begin{pmatrix} z_0 \\ 0 \end{pmatrix} : z_0\in\ker_{\mathbb Z}H_X\right\}
\oplus
\left\{\begin{pmatrix} z_1 \\ t \end{pmatrix} : z_1\in C_X,\ H_Xz_1=qt\right\} , \label{eq:kerZD1_2} \\
\operatorname{im}_{\mathbb Z}(D_2)
&=
\left\{\begin{pmatrix} H_Z^{\mathsf T}a+q b_0 \\ 0 \end{pmatrix} : a\in\mathbb Z^{m_Z},\ b_0\in\ker_{\mathbb Z}H_X\right\}
\oplus
\left\{\begin{pmatrix} q b_1 \\ H_Xb_1 \end{pmatrix} : b_1\in C_X\right\} , \label{eq:imZD2_2}
\end{align}\end{subequations}
and therefore, \begin{equation} \label{eq:decomD1D2}
\frac{\ker_{\mathbb Z}(D_1)}{\operatorname{im}_{\mathbb Z}(D_2)}
\cong
\left[
    \frac{\ker_{\mathbb Z} H_X}
    {\operatorname{im}_{\mathbb Z}(H_Z^{\mathsf T})+q\ker_{\mathbb Z} H_X}
\right]
\oplus
\left[
    \frac{\{z_1\in C_X : H_X z_1\in q\mathbb Z^{m_X}\}}
    {qC_X}
\right].
\end{equation}
From \eqref{eq:decom_kerHX} we can see \begin{equation} \label{eq:decom1}
    \frac{\ker_{\mathbb Z} H_X}
    {\operatorname{im}_{\mathbb Z}(H_Z^{\mathsf T})+q\ker_{\mathbb Z} H_X} \cong \frac{\operatorname{span}_{\mathbb Z}
    \left\{
    w_{r_Z+1}^Z,\ldots,w_{n-r_X}^Z
    \right\}}{q \operatorname{span}_{\mathbb Z}
    \left\{
    w_{r_Z+1}^Z,\ldots,w_{n-r_X}^Z
    \right\}}
    \oplus
    \frac{\operatorname{Sat}\!\left(
    \operatorname{im}_{\mathbb Z} H_Z^{\mathsf T}
    \right)}{\operatorname{im}_{\mathbb Z}(H_Z^{\mathsf T}) + q \operatorname{Sat}\!\left(
    \operatorname{im}_{\mathbb Z} H_Z^{\mathsf T}
    \right)}.
\end{equation}
The first factor in \eqref{eq:decom1} is the free part generated by \(w_{r_Z+1}^Z,\ldots,w_{n-r_X}^Z\), reduced modulo \(q\), and is therefore isomorphic to \(\mathbb Z_q^{\,n-r_X-r_Z}\). We take \(w_{r_Z+1}^Z,\ldots,w_{n-r_X}^Z\) as part of the logical \(Z\)-operator basis, and denote this block by $z_L^f$. For the second factor, since \begin{equation}
\frac{\operatorname{Sat}\!\left(
    \operatorname{im}_{\mathbb Z} H_Z^{\mathsf T}
    \right)}{\operatorname{im}_{\mathbb Z}(H_Z^{\mathsf T}) + q \operatorname{Sat}\!\left(
    \operatorname{im}_{\mathbb Z} H_Z^{\mathsf T}\right)} \cong \bigoplus_{i=1}^{r_Z}  \frac{\operatorname{span}_{\mathbb Z}\left\{ w_{i}^Z \right\}}{d_i^Z \operatorname{span}_{\mathbb Z}\left\{ w_{i}^Z \right\}+q \operatorname{span}_{\mathbb Z}\left\{ w_{i}^Z \right\}}
    \cong
    \bigoplus_{i=1}^{r_Z}
      \mathbb{Z}_{\gcd(d_{i}^Z,\,q)},
\end{equation}
the \(w_i^Z\)-direction contributes a nontrivial cyclic factor $\mathbb{Z}_{\gcd(d_{i}^Z,\,q)}$ precisely when \(\gcd(d_i^Z,q)>1\). We include the corresponding vectors \(w_i^Z\) in the logical \(Z\)-operator basis, and denote this block by \(z_L^{ Z}\). Moreover, each basis element in $z_L^Z$ lies in the real span of the $Z$-stabilizers.
Indeed, $d_i^Z w_i^Z \in \operatorname{im}_{\mathbb Z}(H_Z^{\mathsf T})$, so
$w_i^Z \in \operatorname{span}_{\mathbb R}(H_Z^{\mathsf T})$.

As for the $z_1$ part in \eqref{eq:decomD1D2}, every \(z_1\in C_X\) can be written uniquely as \(z_1=\sum_{i=1}^{r_X}\lambda_i v_i^X\). Since \(U_XH_Xv_i^X=d_i^X e_i\), the condition \(H_Xz_1\in q\mathbb Z^{m_X}\) is equivalent to \(q\mid d_i^X\lambda_i\) for every \(i\). Thus \(\lambda_i\) must be divisible by \(q/\gcd(d_i^X,q)\). Therefore, each \(\gcd(d_i^X,q)>1\) contributes a nontrivial cyclic factor $\mathbb{Z}_{\gcd(d_{i}^X,\,q)}$, and we include \(\frac{q}{\gcd(d_i^X,q)}v_i^X\) in the logical \(Z\)-operator basis. We denote this block of basis vectors by $z_L^{ X}$.

Moreover, each such basis element satisfies $H_X z_{L,i}^X = q c_i$ for some nonzero
integer vector $c_i$. Indeed, for
$z_{L,i}^X=\frac{q}{\gcd(d_i^X,q)}v_i^X$, one has
\begin{equation}
    H_X z_{L,i}^X
    =
    q\,\frac{d_i^X}{\gcd(d_i^X,q)}\,U_X^{-1}e_i
    \ne 0.
\end{equation}

Together, the three blocks \(z_L^f\), \(z_L^{ Z}\), and \(z_L^{ X}\) generate all logical $Z$ operators described in \eqref{eq:decomD1D2}, so we get the decomposition of the codespace \begin{equation}
    \mathcal{C} \cong \mathbb{Z}_q^{\,n - r_X - r_Z}
    \;\oplus\;
    \bigoplus_{i=1}^{r_Z}
      \mathbb{Z}_{\gcd(d_{i}^Z,\,q)}
    \;\oplus\;
    \bigoplus_{j=1}^{r_X}
      \mathbb{Z}_{\gcd(d_{j}^X,\,q)},
\end{equation} 
and the three blocks form a logical \(Z\)-operator basis. Interchanging \(X\) and \(Z\) gives the corresponding logical \(X\)-operator basis. We denote its three blocks by \(x_L^f\), \(x_L^{ X}\), and \(x_L^{ Z}\). Similarly, each basis element in $x_L^X$ lies in the real span of the $X$-stabilizers, and each basis element in $x_L^Z$ satisfies $H_Z x_{L,j}^Z = q c_j$ for some nonzero integer vector $c_j$.

We now check how the \(X\)- and \(Z\)-operator bases pair with each other. First consider the block \(x_L^f\). By construction, the block $x_L^f$ consists of the vectors
$w_{r_X+1}^X,\ldots,w_{n-r_Z}^X$.
Since each element of $z_L^X$ is proportional to some $v_i^X$ with $i \le r_X$, and since from \eqref{eq:vXZwXZ} we have \(\langle w_i^X, v_j^X\rangle=\delta_{ij}\), the block \(x_L^f\) is orthogonal to \(z_L^{ X}\). Moreover, \(x_L^f\subseteq\ker_{\mathbb Z}H_Z\), while \(z_L^{ Z}\subseteq\operatorname{Sat}(\operatorname{im}_{\mathbb Z}H_Z^{\mathsf T})\). Hence \(x_L^f\) is also orthogonal to \(z_L^{ Z}\). 
Therefore, the only \(Z\)-operator block that can have a nonzero pairing with \(x_L^f\) is the free block \(z_L^f\). Since the relevant saturated stabilizer lattices are primitive orthogonal sublattices of \(\mathbb Z^n\), the induced pairing between the two free blocks is perfect over \(\mathbb Z\). Hence its Gram matrix is unimodular, and after unimodular basis changes within the two free blocks, which we absorb into the trailing columns of $V_X$ and $V_Z$ (this preserves \eqref{eq:decom_Zn}--\eqref{eq:decom_kerHX}), we may choose the bases so that
\begin{equation} \label{eq:pairing_01}
    \left\langle x_{L,i}^f, z_{L,j}^f \right\rangle = \left\langle w_{r_X+i}^X, w_{r_Z+j}^Z \right\rangle = \delta_{ij},
    \qquad i,j=1,\ldots,n-r_X-r_Z .
\end{equation}
Then each element of \(x_L^f\) pairs nontrivially only with the corresponding element of \(z_L^f\), and the pairing is exactly \(1\).

By a similar analysis, the block \(x_L^{ X}\) is orthogonal to both \(z_L^f\) and \(z_L^{ Z}\), so it can pair nontrivially only with \(z_L^{ X}\). By construction, \begin{subequations} \label{eq:pairing_02}
\begin{align}
\left\langle x_{L,i}^{ X}, z_{L,j}^{ X} \right\rangle
&=
\left\langle w_i^X, \frac{q}{\gcd(d_j^X,q)} v_j^X \right\rangle
=
\frac{q}{\gcd(d_i^X,q)}\delta_{ij}, \\
\left\langle x_{L,i}^{ Z}, z_{L,j}^{ Z} \right\rangle
&=
\left\langle \frac{q}{\gcd(d_i^Z,q)} v_i^Z, w_j^Z \right\rangle
=
\frac{q}{\gcd(d_i^Z,q)}\delta_{ij}.
\end{align}
\end{subequations}

However, the block \(z_L^{ X}\) does not only pair with \(x_L^{ X}\). In general, \(z_L^{ X}\) may have nonzero pairings with \(x_L^{ Z}\), which is the only possible off-diagonal pairing between blocks. It is because\begin{equation}
\left\langle x_{L,i}^{ Z}, z_{L,j}^{X} \right\rangle
=
\frac{q}{\gcd(d_i^{Z},q)} \frac{q}{\gcd(d_j^{X},q)}\left\langle v_i^Z, v_j^X \right\rangle,
\end{equation}
and $\left\langle v_i^Z, v_j^X \right\rangle$ may not be zero.
In order to get dual bases, we may replace \begin{equation}
x_{L,i}^{Z}
\rightarrow
x_{L,i}^{ Z}
-
\sum_j
\frac{q}{\gcd(d_i^Z,q)}
\left\langle v_i^Z, v_j^X \right\rangle
x_{L,j}^{ X},
\end{equation}
where $j$ runs over the indices with $\gcd(d_j^X,q)>1$. Since $H_Zx_{L,j}^X=0$ and $x_L^X$ is orthogonal to $z_L^f$ and $z_L^Z$, the replacement preserves $H_Zx_{L,i}^Z=qc_i$ and all other pairings. All other basis elements are kept unchanged, including the entire $z_L^X$ block. After this replacement, the pairing between $x_L^Z$ and $z_L^X$ vanishes, and all nonzero pairings are now described by \eqref{eq:pairing_01} and \eqref{eq:pairing_02}. Therefore, after reducing modulo $q$, the bases formed by $x_L^f$, $x_L^X$, $x_L^Z$ and $z_L^f$, $z_L^X$, $z_L^Z$ are dual bases of logical operators satisfying the analogue of \eqref{eq:logicalConjugates}, with $q/d_j$ replaced by $1$ on the free block and by $q/\gcd(d_i^{X/Z},q)$ on the torsion blocks.

\section{Proof of Theorem \ref{thm:baddistance}}\label{app:baddistance}

In this appendix, since by the proof of Theorem~\ref{thm:baddistance} (using Proposition~\ref{prop:commteoverZ}) it remains only to treat the case $q>w_{\mathrm C}^2$, in which $H_XH_Z^{\mathsf T}=0$ over $\mathbb{Z}$, we continue to work with the chosen integer representatives
fixed in Appendix \ref{app:codespace_integer} and keep the slight abuse of terminology: we will refer to the
representatives and to the corresponding bases in $\mathbb Z^n$ simply as
logical operators and logical bases; strictly speaking, they become the actual
logical operators and logical bases of the $\mathbb Z_q$ code only after
reduction modulo $q$.

Choose \(r_X\) linearly independent original rows of \(H_X\) and \(r_Z\) linearly independent original rows of \(H_Z\) that span their respective real row spaces. Let \(S_X\) and \(S_Z\) be the matrices formed by these rows. Thus,
\begin{equation}
\operatorname{row}_{\mathbb{R}}(S_X)
=
\operatorname{span}_{\mathbb{R}}(H_X^{\mathsf T}),
\qquad
\operatorname{row}_{\mathbb{R}}(S_Z)
=
\operatorname{span}_{\mathbb{R}}(H_Z^{\mathsf T}),
\end{equation}
and every row of \(S_X\) and \(S_Z\) has Lee weight at most \(w_{\mathrm C}\).

Let $L_X$ and $L_Z$ be an arbitrary noncommuting logical pair and expand them
in the dual bases formed by $x_L^f$, $x_L^X$, $x_L^Z$ and $z_L^f$, $z_L^X$, $z_L^Z$ in Appendix \ref{app:codespace_integer}. We refer to the non-free blocks $x_L^X, z_L^X, x_L^Z, z_L^Z$
as the torsion part because they come from the nontrivial invariant factors of the integer lattices. Note that, unlike the convention of Section~\ref{sec:css}, this includes blocks with $\gcd(d_i^X,q)=q$ or $\gcd(d_i^Z,q)=q$. 
If $L_X$ has a nonzero coefficient on some
$x_{L,j}^{Z}$, the dual basis element $z_{L,j}^{Z}$ has nonzero pairing
with $L_X$.
By construction, \(z_{L,j}^{Z}\) lies in \(\operatorname{span}_{\mathbb R}(H_Z^{\mathsf T})\), so it can be written as a real linear combination of the rows of \(S_Z\). Subtracting an integer combination of these rows, we may choose all remaining coefficients in \([-1/2,1/2]\). Since \(S_Z\) contains \(r_Z\le n\) original \(Z\)-check rows, the resulting representative has Lee weight at most
\(
w_{\mathrm C}r_Z/2\le w_{\mathrm C}n/2.
\) 
Similarly, if $L_Z$ has a nonzero coefficient
on some $z_{L,j}^{X}$, then the dual basis element $x_{L,j}^{X}$ has a nonzero
pairing with $L_Z$, and $x_{L,j}^{X}$ has reduced Lee weight at most $O(n)$. Therefore, unless one of the two cases above already gives the desired $O(n)$
upper bound, we may assume that $L_X$ has no component
along $x_L^Z$, and $L_Z$ has no component along $z_L^X$. Under this assumption,
the torsion part cannot contribute to their pairing, so the nonzero pairing
between $L_X$ and $L_Z$ must come entirely from the free part. We now restrict to this residual case: a noncommuting logical pair with no
$x_L^Z$ component in its $X$-type operator and no $z_L^X$ component in its
$Z$-type operator will be denoted by $L_X^{\mathrm{res}}$ and
$L_Z^{\mathrm{res}}$.

We can reduce $L_X^{\mathrm{res}}$ and $L_Z^{\mathrm{res}}$ by allowing real stabilizer coefficients. Using the matrices \(S_X\) and \(S_Z\) chosen above, define
\begin{equation}
    P_X = I - S_X^{\mathsf T}(S_X S_X^{\mathsf T})^{-1}S_X,
    \qquad
    P_Z = I - S_Z^{\mathsf T}(S_Z S_Z^{\mathsf T})^{-1}S_Z .
\end{equation}
Thus $P_X L_X^{\mathrm{res}}$ and $P_Z L_Z^{\mathrm{res}}$ are the projected forms of $L_X^{\mathrm{res}}$ and $L_Z^{\mathrm{res}}$ obtained
by minimizing the Euclidean norm $|L_X^{\mathrm{res}}-S_X^{\mathsf T}\alpha|_2$ (resp.\ $|L_Z^{\mathrm{res}}-S_Z^{\mathsf T}\beta|_2$) over real stabilizer coefficients $\alpha$ (resp.\ $\beta$), and this projection preserves their pairing:
\begin{equation}
    (P_X L_X^{\mathrm{res}})\cdot(P_Z L_Z^{\mathrm{res}})=L_X^{\mathrm{res}}\cdot L_Z^{\mathrm{res}}.
\end{equation}

Let $k_f = n-r_X-r_Z$ and \begin{equation}
L_X^f =
\begin{pmatrix}
x_{L,1}^f \\
\vdots \\
x_{L,k_f}^f
\end{pmatrix},
\qquad
L_Z^f =
\begin{pmatrix}
z_{L,1}^f \\
\vdots \\
z_{L,k_f}^f
\end{pmatrix}.
\end{equation}
Every logical operator in the present case is then obtained, over $\mathbb Z_q$, by
reducing modulo $q$ an integer linear combination of the rows of $L_X^f$ or
$L_Z^f$. After projection, the remaining torsion components vanish, so the projected residual operators are represented by the lattices generated by the free blocks
\begin{equation}
    \Lambda := (L_X^f P_X)^{\mathsf T}\mathbb Z^{k_f},
    \qquad
    \Lambda^* := (L_Z^f P_Z)^{\mathsf T}\mathbb Z^{k_f}.
\end{equation}
Since the free blocks were chosen to be dual, their pairing remains
\begin{equation}
    (L_X^f P_X)(L_Z^f P_Z)^{\mathsf T}=I_{k_f}.
\end{equation}
Moreover, since
\begin{equation}
    \Lambda,\Lambda^*
    \subset
    \operatorname{span}_{\mathbb R}(H_X^{\mathsf T})^\perp
    \cap
    \operatorname{span}_{\mathbb R}(H_Z^{\mathsf T})^\perp,
\end{equation}
and this common space has dimension $k_f$, $\Lambda$ and $\Lambda^*$ are full-rank Euclidean dual lattices in this space.

Naturally, all possible choices of $L_X^{\mathrm{res}}$ and $L_Z^{\mathrm{res}}$ are represented by vectors $P_X L_X^{\mathrm{res}}$ and $P_Z L_Z^{\mathrm{res}}$ in
$\Lambda$ and $\Lambda^*$, respectively, although different choices may be
represented by the same vector.

\begin{lem} \label{lem:reduced_weight}
For any $(w_{\mathrm B},w_{\mathrm C})$-Lee-LDPC code, if $u\in\Lambda$ satisfies $|u|_2=\ell$, then the reduced Lee weight of the corresponding $X$-type
logical operator is at most
$\sqrt{n}\ell+\frac{w_{\mathrm C}}{2}n$. The same statement holds for every
$z\in\Lambda^*$.
\end{lem}
\begin{proof}
For $u\in\Lambda$, choose $m\in\mathbb Z^{k_f}$ with
$u=(L_X^fP_X)^{\mathsf T}m$. The corresponding logical $X$ operator is then $x=(L_X^f)^{\mathsf T}m$. Since
$x-u\in\operatorname{span}_{\mathbb R}(H_X^{\mathsf T})$,
$x=u+S_X^{\mathsf T}\alpha$ for some real $\alpha$. Rounding $\alpha$ to the
nearest integer vector $\beta$, we obtain the vector $x-S_X^{\mathsf T}\beta$, which represents
the same logical operator and has Lee weight at most
$|u|_1+\frac{w_{\mathrm C}}{2}n$. Because
$|u|_1\le\sqrt n\,|u|_2=\sqrt n\,\ell$, the reduced Lee weight is at most
$\sqrt n\,\ell+\frac{w_{\mathrm C}}{2}n$. The proof for $z\in\Lambda^*$ is identical.
\end{proof}

Now we come to the heart of the proof, which is showing that if $\Lambda$ has large directions (i.e., there are $X$-logicals of large Lee weight), then $\Lambda^*$ has a conjugate short direction (i.e., there are noncommuting low-weight $Z$-logicals).   This generalization of the standard transference inequalities~\cite{transference} is captured by the following lemma: 

\begin{lem}\label{lem:transcendence_p}
    Let \(\Lambda\subset\mathbb R^k\) be a full-rank lattice, and let \(\Lambda^*\) be its Euclidean dual lattice.
    For $v\in \Lambda$, define \begin{equation}
        D_q(v) := \max_{0<a<q}\mathrm{dist}\left(\frac{av}{q}, \Lambda \right) = \max_{0<a<q}\left( \min_{\lambda\in \Lambda} \left|\frac{av}{q}-\lambda \right|_2 \right).
    \end{equation}
    If $v\in \Lambda \setminus q\Lambda$, or equivalently, $D_q(v)>0$, there exists a $w \in \Lambda^*$ such that $v\cdot w \not\mathrel{\mathchar"3211} 0 \pmod{q}$ and \begin{equation}
       |w|_2 D_q(v) \le 2k.
    \end{equation}
\end{lem}
\begin{proof}
    The proof is rather technical; we provide it in Appendix \ref{app:transcendence_p}.
\end{proof}

\begin{lem}\label{lem:smallDv}
    In the setting of Lemma~\ref{lem:transcendence_p}, if there exists $u\in \Lambda^*$ with $v\cdot u \not\mathrel{\mathchar"3211} 0 \pmod{q}$, we can always find a vector $v'\in \Lambda$ with $v' \cdot u \not\mathrel{\mathchar"3211} 0 \pmod{q}$ and $|v'|_2 \le 4D_q(v)$.
\end{lem}
\begin{proof}
According to the definition of $D_q(v)$, for each $a=1,\ldots,q-1$, there exists $x_a \in \Lambda$ with $|x_a-\frac{av}{q}|_2\le D_q(v)$. We also set $x_0 = 0$ and $x_q=v$.
Define
\begin{equation}
    \Delta_a := x_{a+1}-2x_a+x_{a-1}, \quad d_a := x_{a+1}-x_a,
\end{equation}
so $|\Delta_a|_2\le 4D_q(v)$. We now show there is at least one $\Delta_a$ that has a nonzero pairing with $u$ modulo $q$. Suppose, on the contrary, that \begin{equation}
    \Delta_a \cdot u \mathrel{\mathchar"3211} 0
    \quad (\mathrm{mod}\ q)
    \qquad
    \text{for all } a=1,\ldots,q-1.
\end{equation}
Since $\Delta_a =d_a-d_{a-1} $, all $d_a\cdot u$ are congruent modulo $q$. Let $c$ denote their common
residue. Summing over $a$ gives
\begin{equation}
    v\cdot u
    =
    (x_q-x_0)\cdot u
    =
    \sum_{a=0}^{q-1} d_a\cdot u
    \mathrel{\mathchar"3211}
    \sum_{a=0}^{q-1} c
    =
    qc
    \mathrel{\mathchar"3211} 0
    \quad (\mathrm{mod}\ q),
\end{equation}
which contradicts $v\cdot u \not\mathrel{\mathchar"3211} 0 \pmod{q}$.
Thus $\Delta_a\cdot u\not\mathrel{\mathchar"3211} 0 \pmod{q}$ for some $a$.
Taking $v':=\Delta_a$, we obtain $v'\in\Lambda$,
$v'\cdot u\not\mathrel{\mathchar"3211} 0 \pmod{q}$, and
$|v'|_2\le 4D_q(v)$, as claimed.
\end{proof}

For any logical pair $L_X^{\mathrm{res}}, L_Z^{\mathrm{res}}$ with $L_X^{\mathrm{res}}\cdot L_Z^{\mathrm{res}}\not\mathrel{\mathchar"3211} 0 \pmod{q}$, if $D_q(P_X L_X^{\mathrm{res}})\le\sqrt{k_f/2}$, then by Lemma~\ref{lem:smallDv} there exists $W_X\in \Lambda$ with $W_X\cdot L_Z^{\mathrm{res}} \not\mathrel{\mathchar"3211} 0 \pmod{q}$ and $|W_X|_2\le 2\sqrt{2k_f}$. On the other hand, if $D_q(P_X L_X^{\mathrm{res}}) \ge \sqrt{k_f/2}$, from Lemma \ref{lem:transcendence_p}, there exists $W_Z \in \Lambda^*$ with $L_X^{\mathrm{res}} \cdot W_Z \not\mathrel{\mathchar"3211} 0 \pmod{q}$ and $|W_Z|_2\le 2\sqrt{2k_f}$. As a result, there is always a vector that has a nonzero pairing with one of the logical pair $L_X^{\mathrm{res}}$ and $L_Z^{\mathrm{res}}$ and has Euclidean norm at most $2\sqrt{2k_f}$ after the projection $P_{X/Z}$. With Lemma \ref{lem:reduced_weight}, the reduced weights of the corresponding logical operators are bounded by $ 2\sqrt{2nk_f}+\frac{w_{\mathrm C}}{2}n \le \left(2\sqrt{2}+\frac{w_{\mathrm C}}{2}\right)n$.

\subsection{Proof of Lemma \ref{lem:transcendence_p}}\label{app:transcendence_p}

For $s>0$ and $t\in\mathbb{R}^k$, define \cite{2017thesis}
\begin{equation}
    \rho_{s,\Lambda}(t)
    :=
    \sum_{\lambda\in\Lambda}
    \mathrm{e}^{-\pi |t+\lambda|_2^2/s^2}.
\end{equation}
We use the following useful fact: 
\begin{lem}[{\cite[Corollary~1.3.5]{2017thesis}}]
    Let $\Lambda\subset\mathbb{R}^k$ be a full-rank lattice, let
    $t\in\mathbb{R}^k$, and let $s>0$. If $r>s\sqrt{k/(2\pi)}$, then
    \begin{equation}
        \sum_{\substack{\lambda\in\Lambda:\\ |t+\lambda|_2\ge r}}
        \mathrm{e}^{-\pi |t+\lambda|_2^2/s^2}
        \le
        \rho_{s,\Lambda}(0)
        \mathrm{e}^{-
            \pi\left(
                r/s-\sqrt{k/(2\pi)}
            \right)^2}.
    \end{equation}
\end{lem}

Choose $a_*\in\{0,\ldots,q-1\}$ such that
\begin{equation}
    \mathrm{dist}\!\left(\frac{a_*v}{q},\Lambda\right)
    =
    D_q(v).
\end{equation}
For each $a\in\{0,\ldots,q-1\}$, let $a'\in\{0,\ldots,q-1\}$ be the residue of $a_*-a$ modulo $q$. The map $a\mapsto a'$ is an
involution. Moreover, $a+a'=a_*+mq$ for some $m\in\mathbb Z$, so
\begin{equation}
    \frac{(a+a')v}{q}-\frac{a_*v}{q}=mv\in\Lambda.
\end{equation}
If both $a$ and $a'$ satisfied
\begin{equation}
    \mathrm{dist}\!\left(\frac{av}{q},\Lambda\right)
    <
    \frac{D_q(v)}{2},
    \qquad
    \mathrm{dist}\!\left(\frac{a'v}{q},\Lambda\right)
    <
    \frac{D_q(v)}{2},
\end{equation}
we would have
\begin{align}
    D_q(v)
    &=
    \mathrm{dist}\!\left(\frac{a_*v}{q},\Lambda\right)
    =
    \mathrm{dist}\!\left(\frac{(a+a')v}{q},\Lambda\right) \notag\\
    &\le
    \mathrm{dist}\!\left(\frac{av}{q},\Lambda\right)
    +
    \mathrm{dist}\!\left(\frac{a'v}{q},\Lambda\right)
    <
    D_q(v),
\end{align}
which is impossible. Thus every two-element orbit of the involution contains
at most one residue whose distance is smaller than $D_q(v)/2$, while a fixed
point cannot have distance smaller than $D_q(v)/2$. Consequently, at least half
of the indices $a=0,\ldots,q-1$ satisfy
\begin{equation}\label{eq:half_far_q}
    \mathrm{dist}\!\left(\frac{av}{q},\Lambda\right)
    \ge
    \frac{D_q(v)}{2}.
\end{equation}

Since $v\notin q\Lambda$, one has $D_q(v)>0$. 
Take $s=D_q(v)/\left(2\sqrt{k}\right)$. Then, for every $a$ satisfying \eqref{eq:half_far_q}, we have
\begin{align} \label{eq:ratio_rho_1_p}
    \frac{\rho_{s,\Lambda}(av/q)}{\rho_{s,\Lambda}(0)}
    &\le
    \mathrm{e}^{-
        \pi\left(
            D_q(v)/(2s)-\sqrt{k/(2\pi)}
        \right)^2} \notag\\
    &=
    \mathrm{e}^{-
        \pi k\left(
            1-1/\sqrt{2\pi}
        \right)^2}.
\end{align}

Poisson summation gives, for every $t\in\mathbb{R}^k$,
\begin{equation}\label{eq:poisson_transference_q}
\begin{aligned}
    \rho_{s,\Lambda}(t)
    &=
    \frac{1}{\det(\Lambda)}
    \sum_{w'\in\Lambda^*}
    \int_{\mathbb{R}^k}
    \mathrm{e}^{-\pi|t+y|_2^2/s^2}
    \mathrm{e}^{-2\pi \ii\,w'\cdot y}\,\mathrm{d}y \\
    &=
    \frac{s^k}{\det(\Lambda)}
    \sum_{w'\in\Lambda^*}
    \mathrm{e}^{-\pi s^2|w'|_2^2}
    \mathrm{e}^{2\pi \ii\,w'\cdot t}.
\end{aligned}
\end{equation}
At $t=0$, this becomes
\begin{equation}\label{eq:rho_zero_dual_q}
    \rho_{s,\Lambda}(0)
    =
    \frac{s^k}{\det(\Lambda)}
    \sum_{w'\in\Lambda^*}
    \mathrm{e}^{-\pi s^2|w'|_2^2}
    =
    \frac{s^k}{\det(\Lambda)}\rho_{1/s,\Lambda^*}(0).
\end{equation}
Since $\rho_{s,\Lambda}(t)>0$, taking the absolute value of
\eqref{eq:poisson_transference_q} and using the triangle inequality gives
\begin{equation}\label{eq:rho_max_at_zero_q}
    \rho_{s,\Lambda}(t)
    \le
    \rho_{s,\Lambda}(0).
\end{equation}
Combining \eqref{eq:ratio_rho_1_p} and \eqref{eq:rho_max_at_zero_q} yields
\begin{align}\label{eq:constant_deficit_q}
    1-
    \frac{1}{q}
    \sum_{a=0}^{q-1}
    \frac{\rho_{s,\Lambda}(av/q)}{\rho_{s,\Lambda}(0)}
    &\ge
    \frac{1}{2}
    \left[
        1-
        \mathrm{e}^{-
            \pi k\left(
                1-1/\sqrt{2\pi}
            \right)^2}
    \right] \notag\\
    &\ge
    \frac{1}{2}
    \left[
        1-
        \mathrm{e}^{-
            \pi\left(
                1-1/\sqrt{2\pi}
            \right)^2}
    \right]
    >
    0.339.
\end{align}

We next express the same deficit on the dual lattice. Averaging \eqref{eq:poisson_transference_q} over $a=0,\ldots,q-1$ gives
\begin{align}
    \rho_{s,\Lambda}(0)
    -
    \frac{1}{q}
    \sum_{a=0}^{q-1}
    \rho_{s,\Lambda}(av/q)
    &=
    \frac{s^k}{\det(\Lambda)}
    \sum_{w'\in\Lambda^*}
    \exp\!\left(-\pi s^2|w'|_2^2\right)
    \left[
        1
        -
        \frac{1}{q}
        \sum_{a=0}^{q-1}
        \exp\!\left(
            \frac{2\pi \ii a\, w'\cdot v}{q}
        \right)
    \right] \notag \\
    &=  \frac{s^k}{\det(\Lambda)}
    \sum_{\substack{w'\in\Lambda^* \\ w' \cdot v \not\mathrel{\mathchar"3211}0 \bmod q}}
    \exp\!\left(-\pi s^2|w'|_2^2\right).
\end{align}

Let $w_0=\min(|w'|_2: w'\in\Lambda^*, w'\cdot v \not\mathrel{\mathchar"3211} 0\pmod{q})$, it follows that if $w_0>\frac{1}{s} \sqrt{\frac{k}{2\pi}}$, then
\begin{align} 
    \sum_{w'\in \Lambda^*: w'\cdot v\not\mathrel{\mathchar"3211} 0\pmod{q}}
\mathrm{e}^{-\pi s^2 |w'|_2^2} &\le \sum_{w'\in \Lambda^* : |w'|_2 \ge w_0} \mathrm{e}^{-\pi s^2 |w'|_2^2} \notag \\
&\le \rho_{1/s,\Lambda^*}(0) \exp\left[ -\pi \left(w_0 s - \sqrt{\frac{k}{2\pi}}\right)^2 \right].\label{eq:odd_w_sum_p}
\end{align}

Combining \eqref{eq:rho_zero_dual_q} and \eqref{eq:constant_deficit_q}--\eqref{eq:odd_w_sum_p}, we see that if $w_0>\frac{1}{s}\sqrt{\frac{k}{2\pi}}$, then
\begin{equation}
0.339 < \frac{\rho_{s,\Lambda}(0)-\frac{1}{q}\sum_{a=0}^{q-1}\rho_{s,\Lambda}(av/q)}{\rho_{s,\Lambda}(0)}
\le \exp\left[-\pi\left(w_0s-\sqrt{\frac{k}{2\pi}}\right)^2\right],
\end{equation}
which leads to
\begin{equation}
w_0D_q(v) < \sqrt{\frac{2}{\pi}}\,k+2\sqrt{\frac{k}{\pi}\log\frac{1}{0.339}}
\le \left(\sqrt{\frac{2}{\pi}}+2\sqrt{\frac{1}{\pi}\log\frac{1}{0.339}}\right)k < 2k.
\end{equation}

On the other hand, if $w_0 \le \frac{1}{s}\sqrt{\frac{k}{2\pi}}$,
then
\begin{equation}
w_0D_q(v) \le \sqrt{\frac{2}{\pi}}\,k < 2k.
\end{equation}

\section{Short noncommuting witnesses in non-CSS codes} \label{app:noncss}

This appendix extends Theorem~\ref{thm:baddistance} to general stabilizer codes. Lemma~\ref{lem:nonCSSdecomposition} replaces Proposition~\ref{prop:decom}, and the proof of Theorem~\ref{thm:noncsswitness} follows Appendix~\ref{app:baddistance}. Throughout this appendix, we use the notation of Definition~\ref{def:stabilizer}: a stabilizer code is specified by a matrix $H\in\mathbb{Z}_q^{m\times 2n}$ with rows $h_a=(x_a\mid z_a)$, $H$ satisfies $HJH^{\mathsf T}=0$ modulo $q$ for the symplectic matrix $J$, the logical group is $\mathcal{L}=\ker(HJ)/\operatorname{im}(H^{\mathsf T})$, and the Lee distance is the joint Lee distance $d_\diamond$ of \eqref{eq:ddiamond}.

For row vectors $x$ and $y$, we write $x\cdot y$ for the ordinary Euclidean dot
product, and $\langle x,y\rangle := xJy^{\mathsf T}$ for the symplectic pairing.

\begin{defn} [Quantum non-CSS Lee-LDPC code]
A quantum non-CSS stabilizer code with stabilizer matrix
$H\in\mathbb Z_q^{m\times 2n}$ is a $(w_B,w_C)$-Lee-LDPC code if
\begin{equation}
    \sum_{j=1}^{2n} |H_{ij}|_1 \le w_C
    \qquad
    \text{for each row } i,
\end{equation}
and
\begin{equation}
    \sum_{i=1}^{m} |H_{ij}|_1 \le w_B
    \qquad
    \text{for each column } j.
\end{equation}

\end{defn}

\begin{prop} \label{prop:HJH0} For any quantum \((w_B,w_C)\)-Lee-LDPC non-CSS stabilizer code over \(\mathbb Z_q\), if \(q>w_C^2\), then, using the integer lift convention above,
\(
HJH^{\mathsf T}=0
\)
over \(\mathbb Z\).
\end{prop}

\begin{proof}
Each entry of \(HJH^{\mathsf T}\) is the symplectic pairing of two rows of \(H\), and hence has absolute value at most \(w_C^2\). Stabilizer commutation implies that it vanishes modulo \(q\), so for \(q>w_C^2\) it must vanish over \(\mathbb Z\).
\end{proof}

Over $\mathbb{Z}_q$, the Smith normal form brings any qudit stabilizer group to a standard form by Clifford gates~\cite{Gheorghiu2014standard}, and the alternating Smith normal form yields bases of logical operators~\cite{SarkarYoder2024pauli}.  The Smith normal form over $\mathbb{Z}$ used below also places one member of each torsion pair in the real span of the stabilizers.

\begin{lem}[Free/torsion decomposition with short torsion representatives] \label{lem:nonCSSdecomposition}
Consider any stabilizer code with stabilizer matrix $H$. Let $UHV=D=\operatorname{diag}(d_1,\ldots,d_r,0,\ldots,0)$ be the Smith normal
form of $H$ over $\mathbb Z$, and let $g_i=\gcd(d_i,q)$. If $H J H^\mathsf{T}=0$ over $\mathbb{Z}$, there exist dual logical bases of $\mathcal C$, written as
\begin{equation}
    \{u_1^f,\ldots,u_{k_f}^f,\ u_i^{\mathrm{tor}}\}_{i:g_i>1}
    \qquad\text{and}\qquad
    \{v_1^f,\ldots,v_{k_f}^f,\ v_i^{\mathrm{tor}}\}_{i:g_i>1},
\end{equation}
with $k_f=n-r$, such that
\begin{equation}
    \mathcal C
    \cong
    \bigoplus_{\mu=1}^{k_f}
    \left(
        \mathbb Z_q u_\mu^f
        \oplus
        \mathbb Z_q v_\mu^f
    \right)
    \oplus
    \bigoplus_{i:g_i>1}
    \left(
        \mathbb Z_{g_i}u_i^{\mathrm{tor}}
        \oplus
        \mathbb Z_{g_i}v_i^{\mathrm{tor}}
    \right).
\end{equation}
In particular, $\mathcal C$ splits into a free part of rank $2k_f$ and a torsion
part whose nontrivial blocks are indexed precisely by the Smith factors
satisfying $g_i>1$, and the only nonzero pairings of the dual bases are
\begin{equation} \label{eq:noncsspairings}
    \left\langle u_\mu^f, v_\nu^f \right\rangle
    =
    \delta_{\mu\nu},
    \qquad
    \left\langle u_i^{\mathrm{tor}}, v_j^{\mathrm{tor}} \right\rangle
    =
    \frac{q}{g_i}\delta_{ij}.
\end{equation}
Here, as in Appendix~\ref{app:baddistance}, ``free'' and ``torsion'' refer to the integer decomposition. A block with $g_i=q$ is counted as torsion even though it is not torsion in the sense of Section~\ref{sec:css}.

Moreover, if $H$ is Lee-LDPC, then the dual bases may be chosen so that for
every torsion pair
\begin{equation}
    \left(u_i^{\mathrm{tor}},v_i^{\mathrm{tor}}\right),
\end{equation}
at least one of the two logical classes admits a representative of reduced Lee
weight $O(n)$.

\end{lem}

\begin{proof}
First, we take the Smith normal form of $H$ and write
\begin{equation}
    UHV = D = \operatorname{diag}(d_1,\ldots,d_r,0,\ldots,0).
\end{equation}
We define $v_i$ to be the $i$-th column of $V$, and $w_i$ to be the $i$-th column of $(V^{-1})^{\mathsf T}$. 
Let
\begin{equation}
S=\operatorname{Sat}\!\left(\operatorname{im}(H^{\mathsf{T}})\right)=\operatorname{span}_{\mathbb Z}\{w_1,\ldots,w_r\}.
\end{equation}
Since $d_iw_i\in\operatorname{im}(H^{\mathsf T})$ for $i\le r$ and $HJH^{\mathsf{T}}=0$ over $\mathbb{Z}$, we have $d_id_j\langle w_i,w_j\rangle=0$, so $S$ is a primitive isotropic sublattice of the unimodular symplectic lattice $\mathbb{Z}^{2n}$. Therefore, the basis $w_1,\ldots,w_r$ of $S$ can be extended to an integral symplectic basis of $\mathbb{Z}^{2n}$. We choose this extension so that
\begin{equation}
\ker_{\mathbb{Z}}(HJ)=S\oplus \mathcal W,
\end{equation}
where
\begin{equation}
\mathcal W=\operatorname{span}_{\mathbb{Z}}\{w_{r+1},\ldots,w_{2n-r}\},
\end{equation}
and, with $k_f=n-r$,
\begin{equation}
\{w_{r+1},\ldots,w_{r+k_f},w_{r+k_f+1},\ldots,w_{2n-r}\}
\end{equation}
is a Darboux basis of $\mathcal W$. After choosing this extension, redefine \(V\) so that \(w_i\) is the \(i\)-th column of \((V^{-1})^{\mathsf T}\) for all \(i\). Since the first \(r\) vectors \(w_1,\ldots,w_r\) are unchanged, \(UHV=D\) still holds.
Therefore,
\begin{subequations}
\begin{equation}
\left\langle w_{r+\mu},w_{r+\nu}\right\rangle=0,
\end{equation}
\begin{equation}
\left\langle w_{r+k_f+\mu},w_{r+k_f+\nu}\right\rangle=0,
\end{equation}
\begin{equation}
\left\langle w_{r+\mu},w_{r+k_f+\nu}\right\rangle=\delta_{\mu\nu},
\end{equation}
\end{subequations}

We now fix the part of the torsion basis that we will need
later.
For each $i$ satisfying $g_i>1$, we directly set
\begin{equation}
    u_i^{\mathrm{tor}} := w_i\in\mathcal C.
\end{equation}

These vectors will serve as one side of the torsion basis. We make this choice because $d_iw_i$ lies in $\operatorname{im}(H^{\mathsf T})$, so the order of $u_i^{\mathrm{tor}}$ is exactly controlled by $g_i$. In particular, $\left\langle u_i^{\mathrm{tor}},u_j^{\mathrm{tor}}\right\rangle=0$ because $S$ is isotropic.
Moreover, $w_i$ lies in the real span of
$H^{\mathsf T}$, and this is the source of the later $O(n)$ weight bound.

To choose the other side of the torsion basis, we first define the raw torsion
vectors
\begin{equation}
    \tilde v_i^{\mathrm{tor}}
    :=
    \frac{q}{g_i}\,v_iJ.
\end{equation}

These vectors define logical operators, because
\begin{equation}
    HJ(\tilde v_i^{\mathrm{tor}})^{\mathsf T}
    =
    \frac{q}{g_i}\,Hv_i^{\mathsf T}
    =
    \frac{q}{g_i}\,d_i(U^{-1}e_i)
    =
    0
    \quad(\mathrm{mod}\ q).
\end{equation}

Moreover, they pair with the previously fixed torsion vectors
$u_j^{\mathrm{tor}}=w_j$ as
\begin{equation}
    \left\langle u_j^{\mathrm{tor}},\tilde v_i^{\mathrm{tor}}\right\rangle
    =
    \frac{q}{g_i}\delta_{ij}.
\end{equation}

We next modify the entire $\tilde v^{\mathrm{tor}}$-block, while keeping the
whole $u^{\mathrm{tor}}$-block fixed. Since $\tilde v_j^{\mathrm{tor}}$ has
order $g_j$ in $\mathcal C$, we write
\begin{equation}
    \left\langle \tilde v_i^{\mathrm{tor}}, \tilde v_j^{\mathrm{tor}} \right\rangle
    =
    c_{ij}\frac{q}{g_j}
    \quad(\mathrm{mod}\ q),
    \qquad
    i<j.
\end{equation}

Using these coefficients, we replace the raw torsion vectors by
\begin{equation}
    v_i^{\mathrm{tor}}
    :=
    \tilde v_i^{\mathrm{tor}}
    -
    \sum_{j>i} c_{ij}u_j^{\mathrm{tor}}.
\end{equation}
Therefore,
\begin{equation}
    \left\langle u_k^{\mathrm{tor}}, v_i^{\mathrm{tor}} \right\rangle
    =
    \left\langle u_k^{\mathrm{tor}}, \tilde v_i^{\mathrm{tor}} \right\rangle
    =
    \frac{q}{g_i}\delta_{ki}.
\end{equation}

On the other hand, for every $i<j$,
\begin{equation}
\begin{aligned}
    \left\langle v_i^{\mathrm{tor}}, v_j^{\mathrm{tor}} \right\rangle
    &=
    \left\langle \tilde v_i^{\mathrm{tor}}, \tilde v_j^{\mathrm{tor}} \right\rangle
    -
    c_{ij}
    \left\langle u_j^{\mathrm{tor}}, \tilde v_j^{\mathrm{tor}} \right\rangle \\
    &=
    \left\langle \tilde v_i^{\mathrm{tor}}, \tilde v_j^{\mathrm{tor}} \right\rangle
    -
    c_{ij}\frac{q}{g_j}
    \mathrel{\mathchar"3211} 0
    \quad(\mathrm{mod}\ q).
\end{aligned}
\end{equation}

Since the symplectic pairing is alternating, we also have
\(
    \left\langle v_i^{\mathrm{tor}}, v_i^{\mathrm{tor}} \right\rangle = 0
\).
Therefore
\begin{equation}
    \left\langle v_i^{\mathrm{tor}}, v_j^{\mathrm{tor}} \right\rangle = 0
\end{equation}
for all $i,j$.
Thus the torsion sector is organized into dual pairs
\begin{equation}
    \left(u_i^{\mathrm{tor}},v_i^{\mathrm{tor}}\right),
    \qquad
    g_i>1.
\end{equation}

We now come to the free part. 
We choose the free logical basis from the Darboux basis of $\mathcal W$. For $\mu=1,\ldots,k_f$, define
\begin{equation}
u_\mu^f=w_{r+\mu},\qquad v_\mu^f=w_{r+k_f+\mu}.
\end{equation}
Since $\mathcal W\subseteq\ker_{\mathbb{Z}}(HJ)$, all these vectors are integer logical operators. By the choice of the Darboux basis,
\begin{equation}
\left\langle u_\mu^f,u_\nu^f\right\rangle
=
\left\langle v_\mu^f,v_\nu^f\right\rangle
=0,
\qquad
\left\langle u_\mu^f,v_\nu^f\right\rangle
=\delta_{\mu\nu}.
\end{equation}

Since \(\mathcal W \subseteq S^{\perp_J}\), both free blocks pair trivially with every
\(u_i^{\mathrm{tor}} \in S\). Moreover, for every \(a>r\),
\begin{equation}
\left\langle w_a,v_i^{\mathrm{tor}}\right\rangle
=
\frac{q}{g_i}w_a\cdot v_i
= 0.
\end{equation}
Therefore, both free blocks also pair trivially with the adjusted vectors
\(v_i^{\mathrm{tor}}\). The only pairings among these generators that can be nonzero are those given in \eqref{eq:noncsspairings}.

The pairing relations show that the free generators have order \(q\), the torsion generators indexed by \(i\) have order \(g_i\), and these generators are independent in \(\mathcal C\). Hence they generate a subgroup of order $q^{2k_f}\prod_{i=1}^{r}g_i^2$.
Since \(J\) is unimodular, \(HJ\) has the same Smith factors as \(H\). Therefore,
\begin{equation}
|\mathcal C|
=
\frac{\left|\ker_{\mathbb{Z}_q}(HJ)\right|}
{\left|\operatorname{im}_{\mathbb{Z}_q}(H^{\mathsf T})\right|}
=
q^{2(n-r)}\prod_{i=1}^{r}g_i^2
=
q^{2k_f}\prod_{i=1}^{r}g_i^2.
\end{equation}
The constructed generators therefore form the claimed basis of \(\mathcal C\).

Now assume that $H$ is the stabilizer matrix of a $(w_B,w_C)$-Lee-LDPC code. For every $i$ with
$g_i>1$, the vector $u_i^{\mathrm{tor}}=w_i$ lies in
$\operatorname{span}_{\mathbb R}(H^{\mathsf T})$. Hence, after subtracting a
suitable integer linear combination of stabilizers, we may choose a representative
of the same logical class whose stabilizer coefficients all lie in
$[-1/2,1/2]$. Its reduced Lee weight is therefore at most one half of the total
Lee weight of $r$ linearly independent rows of $H$ spanning $\operatorname{span}_{\mathbb R}(H^{\mathsf T})$. Since each such row has Lee weight at
most $w_C$, this gives
\begin{equation}
    \left|u_i^{\mathrm{tor}}\right|_L
    \le
    \frac{w_C}{2}r
    \le
    \frac{w_C}{2}n.
\end{equation}
Therefore, in every torsion pair
\begin{equation}
    \left(u_i^{\mathrm{tor}},v_i^{\mathrm{tor}}\right),
    \qquad
    \gcd(d_i,q)>1,
\end{equation}
one of the two logical classes admits a representative of reduced Lee weight at
most $\frac{w_C}{2}n$.
\end{proof}

\begin{thm}[Short noncommuting witnesses in quantum non-CSS Lee-LDPC codes] \label{thm:noncsswitness}
For any quantum $(w_B,w_C)$-Lee-LDPC non-CSS stabilizer code on $\mathbb Z_q$,
let $L_1,L_2\in\mathcal C$ be nontrivial noncommuting logical operators, namely
\begin{equation}
    \langle L_1,L_2\rangle \ne 0
    \quad(\mathrm{mod}\ q).
\end{equation}
Then there must exist either
$\widetilde L_1\in\mathcal C$ or $\widetilde L_2\in\mathcal C$ such that
\begin{equation}
    \left\langle \widetilde L_1,L_2\right\rangle \ne 0
    \quad(\mathrm{mod}\ q)
    \qquad
    \left[
    \text{or }
    \left\langle L_1,\widetilde L_2\right\rangle \ne 0
    \quad(\mathrm{mod}\ q)
    \right],
\end{equation}
with
\begin{equation}
    \min\left(
        |\widetilde L_1|_{\mathrm L}^{\mathrm{red}},
        |\widetilde L_2|_{\mathrm L}^{\mathrm{red}}
    \right)
    \le Cn,
\end{equation}
where $C$ is bounded independently of $n$ and $q$:
\begin{equation} \label{eq:noncssconstant}
C \le
\begin{cases}
w_C^2,
&
q\le w_C^2,
\\[0.4em]
4\sqrt{2}+\dfrac{1}{2}w_C,
&
q>w_C^2.
\end{cases}
\end{equation}
\end{thm}
\begin{proof}
When $q\le w_C^2$, the proof is trivial. For $q>w_C^2$, Porposition \ref{prop:HJH0} gives $H J H^\mathsf{T}=0$ over $\mathbb{Z}$, so we can apply Lemma~\ref{lem:nonCSSdecomposition} to $L_1$ and $L_2$.
We first treat the case in which one of the two logical operators has a nonzero
component along the torsion $v$-block. If, for some $i$ with $g_i>1$,
the coefficient of $v_i^{\mathrm{tor}}$ in $L_1$ is nonzero, then by the
construction of the dual bases,
\begin{equation}
    \left\langle u_i^{\mathrm{tor}},L_1\right\rangle \ne 0
    \quad(\mathrm{mod}\ q).
\end{equation}

The same conclusion holds with $L_2$ in place of $L_1$. On the other hand,
Lemma \ref{lem:nonCSSdecomposition} shows that the logical class of $u_i^{\mathrm{tor}}$ admits a
representative of reduced Lee weight at most $\frac{w_C}{2}n$. Therefore, unless
one of the two cases above already gives the desired $O(n)$ upper bound, we may
assume that both $L_1$ and $L_2$ have zero component along every torsion vector
$v_i^{\mathrm{tor}}$. Under this assumption, the torsion part cannot contribute
to their pairing, so the nonzero pairing between $L_1$ and $L_2$ must come
entirely from the free part. We now restrict to this residual case, and write the chosen noncommuting logical pair as $L_1^{\mathrm{res}}$ and $L_2^{\mathrm{res}}$. We take each $L_j^{\mathrm{res}}$ to be the integer vector $\sum_\mu(a_\mu u_\mu^f+b_\mu v_\mu^f)+\sum_i c_iu_i^{\mathrm{tor}}$, with integer coefficients depending on $j$. Since every $u_i^{\mathrm{tor}}$ lies in $S\subset\operatorname{span}_{\mathbb R}(H^{\mathsf T})$, the projection $P$ defined below annihilates it, so $PL_j^{\mathrm{res}}$ lies in the lattice $\Lambda$ defined below.

Choose $r$ linearly independent rows of $H$ that span its real row space, and let $B$ be the matrix formed by these rows, so every row of $B$ has Lee weight at most $w_C$. Define
\begin{equation}
    P = I - B^{\mathsf T}(BB^{\mathsf T})^{-1}B.
\end{equation}
Thus $PL_1^{\mathrm{res}}$ and $PL_2^{\mathrm{res}}$ are obtained by projecting $L_1^{\mathrm{res}}$ and $L_2^{\mathrm{res}}$ onto
$\operatorname{span}_{\mathbb R}(H^{\mathsf T})^\perp$, and this projection preserves their pairing:
\begin{equation}
    \left\langle PL_1^{\mathrm{res}},PL_2^{\mathrm{res}}\right\rangle
    =
    \left\langle L_1^{\mathrm{res}},L_2^{\mathrm{res}}\right\rangle
    \ne 0
    \quad(\mathrm{mod}\ q).
\end{equation}

Let
\begin{equation}
    \Lambda
    :=
    P\,\operatorname{span}_{\mathbb Z}
    \{u_1^f,\ldots,u_{k_f}^f,v_1^f,\ldots,v_{k_f}^f\}.
\end{equation}
Because \(\operatorname{span}_{\mathbb R}(H^{\mathsf T})\subseteq\ker_{\mathbb R}(HJ)\), the restriction of \(P\) to \(\ker_{\mathbb R}(HJ)\) has kernel \(\operatorname{span}_{\mathbb R}(H^{\mathsf T})\) and image \(\ker_{\mathbb R}(H)\cap\ker_{\mathbb R}(HJ)\). Therefore, this real space has dimension \(2n-2r=2k_f\). Moreover, \(P\) preserves the symplectic pairing on \(\ker_{\mathbb R}(HJ)\): for every \(x\in\ker_{\mathbb R}(HJ)\), the difference \(x-Px\) lies in \(\operatorname{span}_{\mathbb R}(H^{\mathsf T})\), which pairs trivially with \(\ker_{\mathbb R}(HJ)\). The full pairing relations in \eqref{eq:noncsspairings} therefore show that the projected free basis is a Darboux basis, so \(\Lambda\) is a full-rank lattice in this real space. This real space is also \(J\)-invariant. Since all symplectic pairings in \(\Lambda\) are integral, \(\Lambda J\subseteq\Lambda_{\mathrm{Euc}}^*\). The matrix of Euclidean pairings between the projected free basis and its \(J\)-image is, up to sign, the standard symplectic matrix and is therefore unimodular. Hence the Euclidean dual lattice of \(\Lambda\) is
\begin{equation} \label{eq:noncssdual}
    \Lambda_{\mathrm{Euc}}^*=\Lambda J.
\end{equation}

We apply Lemmas~\ref{lem:transcendence_p} and \ref{lem:smallDv} within the $2k_f$-dimensional space $\ker_{\mathbb R}(H)\cap\ker_{\mathbb R}(HJ)$, so $k=2k_f$. Now consider the projected residual vector $PL_1^{\mathrm{res}}\in\Lambda$. If
\(
    D_q(PL_1^{\mathrm{res}})
    \le
    \sqrt{k_f}
\),
then Lemma~\ref{lem:smallDv}, applied with $u:=PL_2^{\mathrm{res}}J^{\mathsf T}\in\Lambda_{\mathrm{Euc}}^*$, for which $PL_1^{\mathrm{res}}\cdot u=\langle PL_1^{\mathrm{res}},PL_2^{\mathrm{res}}\rangle\not\equiv0 \pmod{q}$, gives a vector $W\in\Lambda$ such that
\begin{equation}
    \left\langle W,PL_2^{\mathrm{res}}\right\rangle
    \ne 0
    \quad(\mathrm{mod}\ q)
\end{equation}
and
\begin{equation}
    |W|_2\le 4\sqrt{k_f}.
\end{equation}

On the other hand, if
$D_q(PL_1^{\mathrm{res}})\ge \sqrt{k_f}$, then, since
$\Lambda_{\mathrm{Euc}}^*=\Lambda J$, Lemma~\ref{lem:transcendence_p} gives a vector
$Y\in\Lambda J$ such that
\begin{equation}
    PL_1^{\mathrm{res}}\cdot Y\ne 0
    \quad(\mathrm{mod}\ q)
\end{equation}
and
\begin{equation}
    |Y|_2D_q(PL_1^{\mathrm{res}})
    \le
    4k_f.
\end{equation}
Write $Y=WJ$ with $W\in\Lambda$. Then
\begin{equation}
    \left\langle PL_1^{\mathrm{res}},W\right\rangle
    \ne 0
    \quad(\mathrm{mod}\ q)
\end{equation}
and
\begin{equation}
    |W|_2=|Y|_2\le 4\sqrt{k_f}.
\end{equation}
Thus $W$ gives a short noncommuting witness for $L_1^{\mathrm{res}}$.

In either case, $W$ corresponds to a logical operator $\widetilde L\in\mathcal C$
whose reduced Lee weight satisfies
\begin{equation} \label{eq:noncssreducedweight}
    |\widetilde L|_{\mathrm L}^{\mathrm{red}}
    \le
    4\sqrt{2nk_f}
    +
    \frac{w_C}{2}n \le \left(4\sqrt{2}+\frac{w_C}{2}\right)n,
\end{equation}
by the same rounding argument as in the proof of Lemma~\ref{lem:reduced_weight}. 
\end{proof}

\section{Proof of Proposition \ref{prop:PBCOBCHGP}} \label{app:proofPBCOBCHGP}

We prove the desired lower bound on the memory time in two steps. First, we show that every nontrivial logical $Z$ operator has Lee weight proportional to $q$; second, we show that states of small reduced Lee weight satisfy approximate linear confinement, which leads to a small bottleneck ratio when $q$ is large.

We begin with the distance lower bound, which can be read off from the Smith normal forms of the relevant parity-check matrices. Over $\mathbb{Z}$, the matrix $H_Z$ has rank
\(
r_Z = n(n-1),
\)
and all of its nonzero Smith factors are equal to $1$. The matrix $H_X$ has rank
\(
r_X = n^2,
\)
and its nonzero Smith factors consist of $n^2-1$ factors equal to $1$, together with one remaining factor
\(
D_h = a_h^n - b_h^n.
\)
According to Proposition \ref{prop:decom} and Corollary \ref{cor:reducedweightbound}, we immediately see that the codespace $\mathcal{C} \cong \mathbb{Z}_{\gcd(D_h,q)} = \mathbb{Z}_q$, and the Lee distances satisfy\begin{equation}
d_Z \ge \frac{q}{\alpha}, \quad d_X \le (a_h+b_h+a_v+b_v)N/2,
\end{equation}
where $\alpha = \max (a_h,b_h,a_v,b_v)$.
Thus the $q$-scale bottleneck proved below is specific to the $Z$-error sector.
In particular, the same argument cannot produce large memory time at arbitrarily high temperature on the $X$-error side.

Next, we prove the approximate linear confinement. We first give the idea. For any $Z$ error, by
adding $Z$-stabilizers, we choose a representative in which the variables on part of the qudits are uniformly bounded. After this gauge fixing, the leading part of the
$X$-syndrome reduces to the syndromes of $n$ independent classical PBC helical codes, one per horizontal row, as studied in Proposition \ref{prop:PBCslowmixing}. Applying the classical linear
confinement bound to each row and treating the bounded variables as an
$O(n^2)$ correction gives the desired estimate.

With the indexing convention described in Fig.~\ref{fig:PBC_OBC_hgp_tanner}, the physical qudits are divided into two blocks. We call the qudits represented by filled circles black qudits and label them by pairs $(r,i)$, where $i\in\mathbb Z/n\mathbb Z$ is the horizontal periodic coordinate and $r=1,\ldots,n$ is the vertical layer index. We call the qudits represented by open circles white qudits and label them by pairs $(b,i)$, where $i\in\mathbb Z/n\mathbb Z$ is the same horizontal coordinate and $b=1,\ldots,n-1$ labels the vertical position.

For a $Z$-type error, we write its entries on the black qudits as $v_{r,i}$ and
its entries on the white qudits as $t_{b,i}$. Thus a $Z$-type error is written
as $z=(v,t)$, with $v=(v_{r,i})\in\mathbb Z_q^{n\times n}$ and
$t=(t_{b,i})\in\mathbb Z_q^{(n-1)\times n}$. This black/white terminology will
be used throughout the proof.

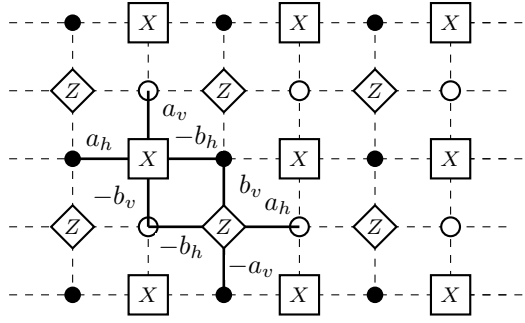
\begin{figure}[t]
\centering
\begin{tikzpicture}[
    x=2cm,y=1.8cm,
    blackq/.style={circle,fill=black,inner sep=2.2pt},
    whiteq/.style={circle,draw,fill=white,inner sep=2.5pt,line width=0.8pt},
    xcheck/.style={rectangle,draw,fill=white,minimum size=5.2mm,inner sep=0pt,line width=0.8pt},
    zcheck/.style={diamond,draw,fill=white,minimum size=5.8mm,inner sep=0pt,aspect=1,line width=0.8pt},
    bg/.style={dashed,thin},
    stub/.style={dashed,thin},
    hi/.style={line width=1.0pt},
    lbl/.style={fill=white,inner sep=1pt,font=\small}
]


\foreach \i in {0,1,2}{
  \foreach \j in {0,1,2}{
    \draw[bg] (\i+0.5,\j) -- (\i,\j);
    \draw[bg] (\i+0.5,\j) -- (\i+1,\j);
    \ifnum\j>0
      \draw[bg] (\i+0.5,\j) -- (\i+0.5,\j-0.5);
    \fi
    \ifnum\j<2
      \draw[bg] (\i+0.5,\j) -- (\i+0.5,\j+0.5);
    \fi
  }
}

\draw[bg] (2,0.5) -- (2.5,0.5);
\draw[bg] (2,1.5) -- (2.5,1.5);

\foreach \i in {0,1,2}{
  \foreach \j in {0,1}{
    \draw[bg] (\i,\j+1/2) -- (\i,\j);
    \draw[bg] (\i,\j+1/2) -- (\i,\j+1);
    \ifnum\i>0
      \draw[bg] (\i,\j+1/2) -- (\i-0.5,\j+1/2);
    \fi
    \ifnum\i<2
      \draw[bg] (\i,\j+1/2) -- (\i+0.5,\j+1/2);
    \fi
  }
}

\foreach \j in {0,1,2}{
  \draw[stub] (-0.42,\j) -- (0,\j);
}
\foreach \j in {0,1}{
  \draw[stub] (-0.42,\j+0.5) -- (0,\j+0.5);
}

\foreach \j in {0,1,2}{
  \draw[stub] (2.5,\j) -- (2.92,\j);
}
\foreach \j in {0,1}{
  \draw[stub] (2.5,\j+0.5) -- (2.92,\j+0.5);
}


\foreach \i in {0,1,2}{
  \foreach \j in {0,1,2}{
    \node[blackq] at (\i,\j) {};
  }
}

\foreach \i in {0,1,2}{
  \foreach \j in {0,1}{
    \node[whiteq] at (\i+0.5,\j+0.5) {};
  }
}

\draw[hi] (0.5,1) -- (0,1);
\draw[hi] (0.5,1) -- (1,1);
\draw[hi] (0.5,1) -- (0.5,0.5);
\draw[hi] (0.5,1) -- (0.5,1.5);

\node[lbl] at (0.18,1.12) {$a_h$};
\node[lbl] at (0.82,1.16) {$-b_h$};
\node[lbl] at (0.28,0.72) {$-b_v$};
\node[lbl] at (0.68,1.34) {$a_v$};

\draw[hi] (1,0.5) -- (1,0);
\draw[hi] (1,0.5) -- (1,1);
\draw[hi] (1,0.5) -- (0.5,0.5);
\draw[hi] (1,0.5) -- (1.5,0.5);

\node[lbl] at (1.18,0.2) {$-a_v$};
\node[lbl] at (1.18,0.8) {$b_v$};
\node[lbl] at (0.72,0.35) {$-b_h$};
\node[lbl] at (1.36,0.65) {$a_h$};

\foreach \i in {0,1,2}{
  \foreach \j in {0,1,2}{
    \node[xcheck] at (\i+0.5,\j) {\scriptsize $X$};
  }
}

\foreach \i in {0,1,2}{
  \foreach \j in {0,1}{
    \node[zcheck] at (\i,\j+0.5) {\scriptsize $Z$};
  }
}

\end{tikzpicture}
\caption{
Local Tanner-graph schematic for the mixed helical HGP code built from the horizontal PBC
$(a_h,-b_h)$ helical code and the vertical OBC $(a_v,-b_v)$ helical code.
Filled and open circles denote the two types of physical qudits in the HGP code. We label the filled-circle qudits by pairs (\(r,i\)), with \(i\in\mathbb Z/n\mathbb Z\) and \(r=1,\ldots,n\), and label the open-circle qudits by pairs (\(b,i\)), with \(i\in\mathbb Z/n\mathbb Z\) and \(b=1,\ldots,n-1\). The coordinate \(i\) is periodic in the horizontal direction, while \(r\) labels positions in the open vertical direction.
Squares denote $X$-checks, and diamonds denote $Z$-checks.
For an interior $X$-check, the coefficients on its (left, right, lower, upper) edges are $(a_h,-b_h,-b_v,a_v)$;
for an interior $Z$-check, they are $(-b_h,a_h,-a_v,b_v)$.
Thick solid edges show the coefficients of one interior $X$-check and one interior $Z$-check, and all other Tanner edges are dashed. The dashed stubs on the left and right indicate the horizontal periodic identification,
while the vertical direction has open boundaries.
}
\label{fig:PBC_OBC_hgp_tanner}
\end{figure}

We label by $(r,i)$ the $X$-check (square) between the black qudits $(r,i)$ and $(r,i+1)$, and by $(b,i)$ the $Z$-check (diamond) between the black qudits $(b,i)$ and $(b+1,i)$. With this convention, the local components of the $X$-syndrome are
\begin{equation}
    (H_X z)_{r,i}    =    a_h v_{r,i}    -    b_h v_{r,i+1}    -    b_v t_{r-1,i}
    +    a_v t_{r,i},    \qquad    t_{0,i}=t_{n,i}=0.
\end{equation} 

For the gauge fixing, we only need the following elementary observation. Consider a $Z$ error $z=(v,t)$ such that $|z|_{\mathrm{L}}$ equals its reduced Lee weight. Suppose
a white qudit has entry $t_{b,i}$ with $|t_{b,i}|_{\mathrm{L}}\ge a_h$. We add or subtract
the adjacent $Z$-stabilizer whose coefficient on this qudit is $a_h$, choosing
the sign so that $t_{b,i}$ is shifted toward zero. This decreases the Lee weight of
$t_{b,i}$ by $a_h$. The same stabilizer changes the entry of one neighboring white qudit by
$b_h$, so the total Lee weight on the white qudits decreases by at least
\begin{equation}
    a_h-b_h>0.
\end{equation}
Iterating this reduction, we obtain an equivalent $Z$ error representative $z'=(v',t')$ modulo
$Z$-stabilizers such that every white-qudit entry has Lee weight less than
$a_h$, so the
total Lee weight of the white block satisfies
\begin{equation}
    |t'|_{\mathrm{L}} \le (a_h-1)n(n-1).
\end{equation}
This is the gauge choice used below. During this gauge-fixing procedure, each step decreases the total Lee weight on
the white qudits by at least $a_h-b_h$, so the procedure can be applied at most
\begin{equation}
    \frac{|z|_{\mathrm{L}}}{a_h-b_h}
\end{equation}
times. Each step changes the Lee weight of black qudits by at most $a_v+b_v$. Therefore,
\begin{equation}
    |v'|_{\mathrm{L}}
    \le
    \left(1+\frac{a_v+b_v}{a_h-b_h}\right)|z|_{\mathrm{L}}.
\end{equation}
Thus, provided
\begin{equation}
    |z|_{\mathrm{L}}
    <
    \frac{a_h-b_h}{2a_h(a_h-b_h+a_v+b_v)}\,q,
\end{equation}
we get 
\begin{equation}
    |v'|_{\mathrm{L}}
    <
    \frac{q}{2a_h}.
\end{equation}
If we ignore the contribution from the white qudits $t'$, the $X$-syndrome only comes from black qudits $v'$ in the same row, so black qudits and $X$-stabilizers form a collection of $n$ independent one-dimensional PBC helical codes, one for each horizontal row.
For each horizontal row, let $\widetilde v'_{r,:}$ be its centered
integer lift. Write $H_h=a_hI-b_hS$, where $S$ is the cyclic shift $(Sw)_i=w_{i+1}$ (cf.~\eqref{eq:helicalH}), so $|Sw|_1=|w|_1$ for every $w\in\mathbb Z^n$. Since
\begin{equation}
    |v'_{r,:}|_{\mathrm{L}}
    \le |v'|_{\mathrm{L}}
    <
    \frac{q}{2a_h},
\end{equation}
every coordinate of $H_h\widetilde v'_{r,:}$ has absolute value strictly
smaller than $q/2$. Hence no syndrome coordinate wraps modulo $q$, and
\begin{align}
    |H_hv'_{r,:}|_{\mathrm{L}}
    =
    |H_h\widetilde v'_{r,:}|_1 
    \ge
    a_h|\widetilde v'_{r,:}|_1
    -
    b_h|S\widetilde v'_{r,:}|_1 
    =
    (a_h-b_h)|v'_{r,:}|_{\mathrm{L}} .
\end{align}
Summing over the horizontal rows gives
\begin{equation}
    |H_X(v',0)|_{\mathrm{L}}
    \ge
    (a_h-b_h)|v'|_{\mathrm{L}} .
\end{equation}

Adding back the white-qudit terms gives \begin{align}
    |H_X z|_{\mathrm{L}}
    =
    |H_X z'|_{\mathrm{L}}\ge
    |H_X(v',0)|_{\mathrm{L}}
    -
    |H_X(0,t')|_{\mathrm{L}}
    \ge
    (a_h-b_h)|v'|_{\mathrm{L}}
    -
    (a_v+b_v)|t'|_{\mathrm{L}}.
\end{align}
Since $|z|_{\mathrm{L}} \le |z'|_{\mathrm{L}} = |v'|_{\mathrm{L}} + |t'|_{\mathrm{L}}$, we have $|v'|_{\mathrm{L}} \ge |z|_{\mathrm{L}} - |t'|_{\mathrm{L}}$.
Combining these estimates gives
\begin{align}
    |H_X z|_{\mathrm{L}}
    \ge
    (a_h-b_h)|z|_{\mathrm{L}}
    -
    (a_h-1)(a_h-b_h+a_v+b_v)n(n-1).
\end{align}

We now define the bottleneck set. Choose an integer
\begin{equation}
    R
    =
    \left\lceil\frac{a_h-b_h}{2a_h(a_h-b_h+a_v+b_v)}\,q\right\rceil-1,
\end{equation}
and
define
\begin{equation}
    A_R=\{z:\ |z|_{\mathrm{L}}^{\text{red}}\le R\}.
\end{equation}
Let $\partial A_R$ be the set of configurations in $A_R$ from which one
single-qudit $\pm1$ update can leave $A_R$, so every element of $\partial A_R$ has
reduced Lee weight exactly $R$. 
The stationary probability flow $Q(A_R,A_R^c)$ out of $A_R$ is bounded by the Gibbs weight of $\partial A_R$:
\begin{equation}
    Q(A_R,A_R^c)\le \pi(\partial A_R).
\end{equation}

To estimate the upper bound of $Q(A_R,A_R^c)$, we count minimum-weight representatives rather than states. A given state may admit more than one such representative, but this only overcounts the boundary and therefore still gives a valid upper bound.
Since every state in $\partial A_R$ has reduced Lee weight $R$, every
minimum-weight representative on the boundary has Lee weight $R$. Because
$R<q/2$, each such representative lifts uniquely to an integer vector in
$\mathbb Z^N$ of $\ell_1$-weight $R$. Therefore the number of boundary states is
at most the number of such representatives, which is bounded by
\begin{equation} \label{eq:appEboundarycount}
    2^N \binom{R+N-1}{N-1}
    \le
    \exp\!\left[N\log(2(R+N))\right].
\end{equation}
For every $z\in\partial A_R$, the confinement bound gives
\begin{equation}
    |H_X z|_{\mathrm{L}}
    \ge
    (a_h-b_h)R
    -
    (a_h-1)(a_h-b_h+a_v+b_v)n(n-1),
\end{equation}
which leads to
\begin{align}
    \frac{Q(A_R,A_R^c)}{\pi(0)}
    \le
    \exp\!\left[
        -\beta\bigl((a_h-b_h)R
        -(a_h-1)(a_h-b_h+a_v+b_v)n(n-1)\bigr)
        +N\log(2(R+N))
    \right].
\end{align}
Since $2R < d_Z$, we can choose a decoder such that all states in $A_R$ are decoded to the zero codeword. Therefore, applying \eqref{eq:memorytimebound} with the set $A_R$ in place of the ball used there, we get:
\begin{align}
    t_{\mathrm{mem}}(1/4)
    \ge
    \frac{\pi(0)}{4Q(A_R,A_R^c)}
    \ge
    \frac{1}{4}
    \exp\!\left[
        \beta\bigl(\kappa_h( \rho q-1) - \eta n(n-1)\bigr)
        -N\log\!\bigl(2( \rho q+N)\bigr)
    \right],
\end{align}
where\begin{equation}
    \kappa_h := a_h-b_h,
    \qquad
    \eta := (a_h-1)(a_h-b_h+a_v+b_v),
    \qquad
    \rho :=
    \frac{a_h-b_h}{2a_h(a_h-b_h+a_v+b_v)}.
\end{equation}
This completes the proof.

\section{Proof of Proposition \ref{prop:classicalOBC}} \label{app:proofclassicalOBC}

Throughout this appendix, $H:=H_{\mathrm{OBC}}(a,-b)$. From Proposition~\ref{prop:helicalproperties}, the codewords are generated by the integer vector \begin{equation}
    v=\left(b^{n-1},ab^{n-2},a^2b^{n-3},\ldots,a^{n-1}\right),
\end{equation}
so \begin{equation}
\|v\|_2=\sqrt{\frac{a^{2n}-b^{2n}}{a^2-b^2}}.
\end{equation}
Let $\mathbb Z_q v\subset \mathbb Z_q^n$ be the corresponding codeword
orbit. We will use a Lee tube around one segment of this orbit as
the bottleneck set to derive a lower bound on the memory time. The only geometric point needed below is that this tube does
not overlap itself when the radius is chosen below half the minimum transverse
separation of the wrapped orbit. 

We first bound the minimum transverse separation of this wrapped orbit from below. In
integer coordinates, the different affine lifts of the wrapped orbit are
parallel lines in $\mathbb Z^n$ with direction $v$. Their relative displacement
is of the form $qm$, where $m\in\mathbb Z^n$. Thus the $\ell_1$ distance
between two such lifts is
\begin{equation}
    q\inf_{\lambda\in\mathbb R}\|m-\lambda v\|_1.
\end{equation}

We claim that, unless the two lifts are the same, this $\ell_1$ distance is at least
$q/a$. Indeed, if $m\notin \mathbb Z v$, then $Hm\ne0$ on $\mathbb{Z}$. Hence some component of $Hm$ is a nonzero integer.
For any $\lambda\in\mathbb R$, setting $z=m-\lambda v$, we have
$Hz=Hm$, and therefore, for some $j$,
\begin{equation}
    1\le |(Hz)_j|
    =
    |az_j-bz_{j+1}|
    \le
    a|z_j|+b|z_{j+1}|
    \le
    a\|z\|_1.
\end{equation}
It follows that distinct affine lifts of the wrapped orbit have transverse $\ell_1$ distance
\begin{equation}
    q\inf_{\lambda\in\mathbb R}\|m-\lambda v\|_1\ge \frac{q}{a}.
\end{equation}

We now define the bottleneck set. For an integer vector $\widetilde x\in\mathbb Z^n$, let
\begin{equation}
    \tau(\widetilde x):=\frac{v\cdot\widetilde x}{\|v\|_2^2},
    \qquad
    r(\widetilde x):=\operatorname{dist}_1(\widetilde x,\mathbb R v)
\end{equation}
be its axial and radial coordinates relative to the helical axis $\mathbb R v$.

Set
\begin{equation}
    R=\left\lfloor\frac{q}{2a}\right\rfloor-1.
\end{equation}
Define the lifted cylinder
\begin{equation}
    \widetilde A
    =
    \left\{
        \widetilde x\in\mathbb Z^n:
        -\frac{1}{2}\le \tau(\widetilde x)<\frac{1}{2},\ 
        r(\widetilde x)\le R
    \right\},
\end{equation}
and let $A\subset\mathbb Z_q^n$ be the reduction of $\widetilde A$ modulo $q$. Let $\partial A$ denote the set of configurations in $A$ that can leave $A$ after one single-site update.

Since $\tau(\widetilde x+v)=\tau(\widetilde x)+1$, there is exactly one codeword inside the cylinder. Moreover, the $\ell_1$ distance between distinct affine lifts of
the wrapped orbit is at least $q/a$, whereas $2R+1<q/a$. Therefore the radius-$R$
tube around one lift cannot intersect the radius-$R$ tube around any other lift.
In particular, the reduction map from $\widetilde A$ to $A$ is a one-to-one map.

Although $\widetilde A$ is defined as the cylinder centered at the zero codeword, the same construction applies around every codeword. Indeed, addition of $v$ is a
symmetry of the code, and, as $mv$ ranges over the codewords, the distinct
translates $\widetilde A+mv$ are disjoint and have identical structure. As a result, the reduction modulo $q$ is still a one-to-one map, and the sets $A +mv \bmod q$ are also pairwise disjoint and have identical structure.
Moreover, because $Hv=0$ and $\tau(\widetilde x+v)=\tau(\widetilde x)+1$, every integer syndrome fiber $H\widetilde x=y$ intersects each folded cylinder $A +mv \bmod q$ in at most one integer representative. We can therefore choose a decoder that maps every configuration in $A+mv \bmod q$ to the codeword $mv \bmod q$. A lower bound on the time
required to leave $A$ from the zero codeword therefore gives a lower bound on the memory time. 

We analyze the two possible exit routes separately.
Since $A$ is a cylinder, its boundary naturally splits into the side and the two caps. For each $x\in A$, let $\widetilde x\in\widetilde A$ denote the unique
lifted representative of $x$. We define the side and cap sets by
\begin{equation}
    \partial_{\mathrm{side}}A
    :=
    \left\{
        x\in\partial A :
        \exists j,\ \exists \sigma\in\{\pm1\},\
        r(\widetilde x+\sigma e_j)>R
    \right\},
\end{equation}
and
\begin{equation}
    \partial_{\mathrm{cap}}A
    :=
    \left\{
        x\in\partial A :
        \exists j,\ \exists \sigma\in\{\pm1\},\
        \tau(\widetilde x+\sigma e_j)<-\frac{1}{2}
        \text{ or }
        \tau(\widetilde x+\sigma e_j)\ge \frac{1}{2}
    \right\}.
\end{equation}

We first analyze the case where the state exits through one of the caps.
Consider any $\widetilde x\in\widetilde A$. Set
\begin{equation}
    \Delta :=
    \max_{1\le j\le n}\frac{v_j}{\|v\|_2^2} =
    \frac{a^{n-1}}{\|v\|_2^2}
    =
    \frac{a^{n-1}(a^2-b^2)}{a^{2n}-b^{2n}}.
\end{equation}
A single-site update changes the axial coordinate by at most $\Delta$, since
\begin{equation}
    \left|\tau(\widetilde x\pm e_j)-\tau(\widetilde x)\right|
    =
    \frac{v_j}{\|v\|_2^2}
    \le
    \Delta.
\end{equation}
Because $2R+1<q/a$, and because a translation by $qkv$ with $k\neq0$ shifts $\tau$ by $qk$, the only lift of $(\widetilde x\pm e_j)\bmod q$ that can lie in $\widetilde A$ is $\widetilde x\pm e_j$ itself. Hence each accepted move inside $A$ changes the lift $\widetilde x$ by $\pm e_j$. Therefore, it takes at least \begin{equation}
    T_{\text{cap}} = \frac{1}{2\Delta} = \exp\left(O(n)\right)
\end{equation}
steps to reach any caps from the zero codeword. 

To estimate the time required to leave $A$ from the side, it suffices to convert the geometric
condition of being far from the helical axis into an energy lower bound. 
We claim that, whenever $r(\widetilde x)<q/(2a)$, one has
\begin{equation}
|H\widetilde x|_{\mathrm{L}} \ge (a-b)r(\widetilde x).
\end{equation}

To prove this bound, define a linear map $T:\mathbb R^{n-1}\to\mathbb R^n$ by
\begin{equation}
    T(y)_n:=0,
    \qquad
    T(y)_j:=
    \frac{y_j+bT(y)_{j+1}}{a}
    \quad (1\le j\le n-1).
\end{equation}
Then $HT(y)=y$. Expanding the recursion gives
\begin{equation}
    \|T(y)\|_1 = \| \sum_{k=j}^{n-1} \frac{b^{k-j}}{a^{k-j+1}}y_k \|_1
    \le
    \sum_{k=1}^{n-1}|y_k|
    \sum_{m=0}^{\infty}\frac{b^m}{a^{m+1}}
    =
    \frac{1}{a-b}\|y\|_1.
\end{equation}

Now take $y=H\widetilde x$. Since $HT(H\widetilde x)=H\widetilde x$, we have
\begin{equation}
    H\bigl(\widetilde x-T(H\widetilde x)\bigr)=0,
\end{equation}
so $\widetilde x-T(H\widetilde x)\in\mathbb Rv$. It follows that
\begin{equation}
    r(\widetilde x)
    \le
    \|T(H\widetilde x)\|_1
    \le
    \frac{1}{a-b}\|H\widetilde x\|_1.
\end{equation}

Since there is no wrap-around modulo $q$, we obtain
\begin{equation}
    |H\widetilde x|_{\mathrm{L}}
    =
    \|H\widetilde x\|_1
    \ge
    (a-b)r(\widetilde x),
\end{equation}
which is the desired transverse linear confinement estimate.

We now apply this to the side boundary. If
$x\in\partial_{\mathrm{side}}A$, and
$\widetilde x\in\widetilde A$ is its unique lifted representative, then
since $r(\widetilde x+\sigma e_j)>R$ for some $j,\sigma$ and $r$ is $1$-Lipschitz with respect to $\ell_1$, we have
\(
    r(\widetilde x)>R-1.
\)
Since also $r(\widetilde x)\le R=\lfloor q/(2a)\rfloor-1$, the confinement estimate applies and
yields
\begin{equation}
    |H\widetilde x|_{\mathrm{L}}\ge (a-b)(R-1).
\end{equation}
Consequently, writing $\pi_{\mathrm{side}}:=\pi(\partial_{\mathrm{side}}A)$, the relative weight $W_{\mathrm{side}}:=\pi_{\mathrm{side}}/\pi(0)$ satisfies
\begin{equation}
    W_{\mathrm{side}}
    \le
    \#(\partial_{\mathrm{side}}A)
    \exp\bigl(-\beta(a-b)(R-1)\bigr).
\end{equation}
It therefore remains to bound the number of side-boundary states.

For each $x\in\partial_{\mathrm{side}}A$ with lift $\widetilde x\in\widetilde A$, consider
the unit cube
\begin{equation}
    Q_{\widetilde x}:=\widetilde x+[-1/2,1/2)^n.
\end{equation}
These cubes are pairwise disjoint. If $z\in Q_{\widetilde x}$ then
\begin{equation}
    \|z-\widetilde x\|_1\le \frac{n}{2},
\end{equation}
and hence
\begin{equation}
    r(z)\le r(\widetilde x)+\frac{n}{2}\le R+\frac{n}{2},
\end{equation}
while
\begin{equation}
    r(z)\ge r(\widetilde x)-\frac{n}{2}>R-1-\frac{n}{2}.
\end{equation}
Thus every such cube is contained in the radial annulus
\begin{equation}
    R-1-\frac{n}{2}<r\le R+\frac{n}{2}.
\end{equation}
We also need to control the axial direction. For $z\in Q_{\widetilde x}$,
\begin{equation}
    |\tau(z)-\tau(\widetilde x)|
    \le
    \frac{1}{2}\frac{\|v\|_1}{\|v\|_2^2} = \frac{1}{2} \frac{a+b}{a^n+b^n}
    \ll \frac{1}{2}.
\end{equation}
Since every $\widetilde x\in\widetilde A$ satisfies
\(
    -1/2\le \tau(\widetilde x)<1/2,
\)
all the cubes $Q_{\widetilde x}$ are contained in an axial slab of length less than $2$.
It follows that the total number of side-boundary states is bounded by the
volume of a slab of axial length $2$ times the $(n-1)$-dimensional
$\ell_1$-annulus above:
\begin{equation} \label{eq:appFsidecount}
    \#(\partial_{\mathrm{side}}A)
    \le
    \frac{2^n \|v\|_2}{(n-1)!}
    \left[
        \left(R+\frac{n}{2}\right)^{n-1}
        -
        \left(R-1-\frac{n}{2}\right)^{n-1}
    \right].
\end{equation}
By the mean value theorem,
\begin{equation}
    \left(R+\frac{n}{2}\right)^{n-1}
    -
    \left(R-1-\frac{n}{2}\right)^{n-1}
    \le
    (n-1)(n+1)
    \left(R+\frac{n}{2}\right)^{n-2}.
\end{equation}
Hence
\begin{equation} \label{eq:appFWside}
    \frac{\pi_{\mathrm{side}}}{\pi(0)} = W_{\mathrm{side}}
    \le
    \frac{2^n(n+1)\|v\|_2}{(n-2)!}
    \left(R+\frac{n}{2}\right)^{n-2}
    \exp\bigl(-\beta(a-b)(R-1)\bigr).
\end{equation}

For $t\ge 1$, the probability that the dynamics exits $A$ through the side
boundary within $t$ steps is bounded by
\begin{equation}
\begin{aligned}
    \Pr(\text{side exit by time }t)
    &\le
    \sum_{k=0}^{t-1}
    \sum_{x\in\partial_{\mathrm{side}}A}
    P^k(0,x)P(x,A^c)
    \\
    &\le
    \frac{t}{\pi(0)}
    \sum_{x\in\partial_{\mathrm{side}}A}\pi(x)
    =
    t\frac{\pi_{\mathrm{side}}}{\pi(0)},
\end{aligned}
\end{equation}
where the second inequality follows from reversibility,
\begin{equation}
    P^k(0,x)
    =
    \frac{\pi(x)}{\pi(0)}P^k(x,0)
    \le
    \frac{\pi(x)}{\pi(0)}.
\end{equation}
Consequently, the probability of a side exit remains below $1/4$ for at least
\begin{equation}
    T_{\mathrm{side}}
    =
    \frac{\pi(0)}{4\pi_{\mathrm{side}}}
\end{equation}
steps, up to integer rounding. Using the bound on $W_{\mathrm{side}}$, we obtain
\begin{equation}
    T_{\mathrm{side}}
    \ge
    \frac{(n-2)!}{2^{n+2}(n+1)\|v\|_2}
    \left(R+\frac{n}{2}\right)^{-(n-2)}
    \exp\!\left(\beta(a-b)(R-1)\right).
\end{equation}
Therefore, for every $t<\min\{T_{\mathrm{cap}},T_{\mathrm{side}}\}$,
a cap exit is impossible and the probability of a side exit is less than $1/4$.
Since the decoder is correct throughout $A$, only at this point do we conclude
that
\begin{equation}
    p_{\mathrm{fail}}(t)
    \le
    \Pr(A\text{ has been exited by time }t)
    <\frac{1}{4}.
\end{equation}
Hence,
\begin{align}
    t_{\mathrm{mem}}(1/4)
    &\ge
    \min\left\{
        \frac{1}{2\Delta},
        \frac{\pi(0)}{4\pi_{\mathrm{side}}}\right\} \notag \\
        &\ge    \min\left\{
        \frac{a^{2n}-b^{2n}}{2a^{n-1}(a^2-b^2)},
        \frac{(n-2)!}
        {2^{n+2}(n+1)\left(\frac{q}{2a}-1+\frac{n}{2}\right)^{n-2}\|v\|_2}
        \exp\!\left[
            \beta(a-b)\left(\frac{q}{2a}-3\right)
        \right]
    \right\}.
\end{align}

It remains to choose the constants $c$ and $C$ in the statement of this proposition. Set
\begin{equation}
    c=\frac{1}{2}\log a
\end{equation}
and
\begin{equation}
    C
    =
    4a
    +
    \frac{8a}{(a-b)\log 2}
    \left(
        1+3\log2+\frac{1}{2}\log a
    \right).
\end{equation}
It can be checked that for sufficiently large $n$, whenever
\begin{equation}
    q
    \ge
    Cn\max\left\{
        1,\,
        \frac{1}{\beta}\log\left(2+\frac{1}{\beta}\right)
    \right\},
\end{equation}
the memory time satisfies
\begin{equation}
    t_{\mathrm{mem}}(1/4)\ge \exp(cn).
\end{equation}
This proves the proposition.

\section{Free-energy landscape of the classical helical repetition code with open boundary conditions} \label{app:freeenergylandscape}

As shown in the proof of Proposition \ref{prop:classicalOBC}, the energy landscape of an OBC helical code over $\mathbb Z$ is strongly anisotropic. Fix coprime integers
$a>b\ge 1$, and define the energy of a configuration $\widetilde x\in\mathbb Z^n$
by
\begin{equation}
    E(\widetilde x)
    =
    |H_{\mathrm{OBC}}(a,-b)\widetilde x|_1.
\end{equation}
The primitive codeword generator is
\begin{equation}
    v
    =
    \left(
        b^{n-1},
        ab^{n-2},
        a^2b^{n-3},
        \ldots,
        a^{n-2}b,
        a^{n-1}
    \right).
\end{equation}
Since $H_{\mathrm{OBC}}(a,-b)v=0$, if we extend the same energy function to
$\mathbb R^n$, it is constant on every affine line parallel to $\mathbb Rv$.
On the integer configuration space, this flat direction gives only the
discrete symmetry $E(\widetilde x+mv)=E(\widetilde x)$.
By contrast, the proof of Proposition \ref{prop:classicalOBC} established transverse
linear confinement:
\begin{equation}
    E(\widetilde x)
    \ge
    (a-b)\operatorname{dist}_1(\widetilde x,\mathbb Rv).
\end{equation}
In Proposition \ref{prop:classicalOBC}, we studied the memory time of the single-site
Metropolis dynamics initialized at the zero codeword. Here we study a
different, static question: how is the Gibbs weight of the model over $\mathbb Z$
distributed along this energy-flat direction?

Because the zero mode is completely flat, the unrestricted partition function
on $\mathbb Z^n$ diverges. A normalized equilibrium distribution along the axis
can therefore be defined only after quotienting by the primitive translation
$\widetilde x\sim\widetilde x+v$. If the configuration space were instead $\mathbb R^n$, every
configuration could be written uniquely as
\begin{equation}
    \widetilde x=tv+w,
\end{equation}
with $w\perp v$ and, after taking this quotient, $0\le t<1$. The marginal distribution of $t$ would therefore be uniform.

On $\mathbb Z^n$, however, the integrality constraint makes the admissible
transverse configurations depend on $t$, so their total Gibbs weight need not be
uniform along the flat direction.
One could study this variation using affine hyperplanes orthogonal to $v$. We
instead use the fixed-endpoint slices
\begin{equation}
    \Sigma_r
    =
    \{\widetilde x\in\mathbb Z^n:\widetilde x_n=r\},
    \qquad
    r\in\mathbb Z.
\end{equation}
These slices are not exactly orthogonal to $v$. However, $v_n=a^{n-1}$ is the
largest component of $v$, so the hyperplanes with fixed $\widetilde x_n$ are the closest to being orthogonal to $v$ among all coordinate hyperplanes. 
A primitive codeword period contains precisely $a^{n-1}$ distinct endpoint
slices. 

For each $r\in\{0,1,\ldots,a^{n-1}-1\}$, define the fixed-endpoint partition function by \begin{equation} \label{eq:Zbetar}
    Z_\beta(r)
    =
    \sum_{\widetilde x_1,\ldots,\widetilde x_{n-1}\in\mathbb Z}
    \exp\!\left[
        -\beta E(\widetilde x_1,\ldots,\widetilde x_{n-1},r)
    \right].
\end{equation}
The corresponding free-energy landscape along the logical direction is $F_\beta(r):=-\beta^{-1}\log Z_\beta(r)$, defined up to an additive constant.

We now prove that $Z_\beta (r)$ is strongly nonuniform. In particular, any set of slices carrying a fixed positive fraction of the total Gibbs weight must contain exponentially many slices, and every individual slice has exponentially small probability.

Let
\begin{equation}
    z
    =
    (\widetilde x_1,\ldots,\widetilde x_{n-1})^{\mathsf T}
    \in \mathbb Z^{n-1},
    \qquad
    e_{n-1}
    =
    (0,\ldots,0,1)^{\mathsf T}
    \in \mathbb Z^{n-1},
\end{equation}
and define
\begin{equation}
    M
    =
    \begin{pmatrix}
        a & -b & 0 & \cdots & 0 \\
        0 & a & -b & \ddots & \vdots \\
        \vdots & \ddots & \ddots & \ddots & 0 \\
        0 & \cdots & 0 & a & -b \\
        0 & \cdots & \cdots & 0 & a
    \end{pmatrix}
    \in \mathbb Z^{(n-1)\times(n-1)}.
\end{equation}

At fixed endpoint $\widetilde x_n=r$, the bond-difference vector is
\begin{equation}
    y
    =
    Mz-br e_{n-1}
    =
    (a\widetilde x_1-b\widetilde x_2,\ldots,a\widetilde x_{n-1}-br)^{\mathsf T}.
\end{equation}

Equation \eqref{eq:Zbetar} becomes
\begin{equation}
    Z_\beta(r)
    =
    \sum_{y\in\Lambda-br e_{n-1}}
    e^{-\beta|y|_1},
\end{equation}
where $\Lambda:=M\mathbb Z^{n-1}$.
This series converges because its summation set is contained in
$\mathbb Z^{n-1}$ and
$\sum_{y\in\mathbb Z^{n-1}}e^{-\beta|y|_1}<\infty$.
Define
\begin{equation}
    Q(y)
    :=
    \sum_{j=1}^{n-1}
    a^{n-1-j}b^{j-1}y_j.
\end{equation}
The functional $Q$ expresses compatibility with integer site coordinates as one
congruence. Then
\begin{equation}
    \Lambda-br e_{n-1}
    =
    \left\{
        y\in\mathbb Z^{n-1}:
        Q(y)\equiv -b^{n-1}r \pmod{a^{n-1}}
    \right\}.
\end{equation}
As $r$ ranges over one endpoint period, the sets in the last equation are disjoint and partition $\mathbb Z^{n-1}$, so the total Gibbs weight of one endpoint period is
\begin{equation}
    Z_\beta^{\mathrm{tot}}
    =
    \sum_{y\in\mathbb Z^{n-1}}
    e^{-\beta|y|_1}
    =
    \left(
        \coth\frac{\beta}{2}
    \right)^{n-1},
\end{equation}
and therefore the average Gibbs weight
\begin{equation} \label{eq:Zbetabar}
    \overline Z_\beta := a^{-(n-1)} Z_\beta^{\mathrm{tot}}
    =
    \left(
        \frac{1}{a}\coth\frac{\beta}{2}
    \right)^{n-1}.
\end{equation}

By the set identity above,
\begin{equation}
    Z_\beta(r)
    =
    \sum_{y\in\mathbb Z^{n-1}}
    e^{-\beta|y|_1}
    \mathbf{1}_{\{Q(y)\equiv -b^{n-1}r\;(\mathrm{mod}\;a^{n-1})\}}.
\end{equation}
The indicator has the character-sum expansion
\begin{equation}
    \mathbf{1}_{\{Q(y)\equiv -b^{n-1}r\;(\mathrm{mod}\;a^{n-1})\}}
    =
    \frac{1}{a^{n-1}}
    \sum_{k=0}^{a^{n-1}-1}
    \exp\!\left[
        -\frac{2\pi i k}{a^{n-1}}
        \left(
            Q(y)+b^{n-1}r
        \right)
    \right].
\end{equation}
The character sum is finite, so its order may be exchanged with the sum over
$y$. The coefficient of $y_j$ in $Q(y)/a^{n-1}$ is $b^{j-1}/a^j$, and the sum
over $y$ therefore factorizes into
\begin{equation}
    \prod_{j=1}^{n-1}
    \sum_{y_j\in\mathbb Z}
    e^{-\beta|y_j|}
    \exp\!\left(
        -\frac{2\pi i k b^{j-1}y_j}{a^j}
    \right),
\end{equation}
which is a product of values of the Fourier transform of the Gibbs weight of a single integer bond, \begin{equation} \label{eq:Gbeta}
    G_\beta(t)
    :=
    \sum_{m\in\mathbb Z}
    e^{-\beta|m|}e^{2\pi i m t} =
    \frac{\sinh\beta}
    {\cosh\beta-\cos(2\pi t)}.
\end{equation}
In particular, $G_\beta(0)=\coth(\beta/2)$, so $Z_\beta^{\mathrm{tot}}=G_\beta(0)^{n-1}$.
Therefore, for every $n\ge 2$,
\begin{equation} \label{eq:Zbetar_2}
    Z_\beta(r)
    =
    \frac{1}{a^{n-1}}
    \sum_{k=0}^{a^{n-1}-1}
    \exp\!\left(
        -\frac{2\pi i k b^{n-1}r}{a^{n-1}}
    \right)
    \prod_{j=1}^{n-1}
    G_\beta\!\left(
        \frac{k b^{j-1}}{a^j}
    \right).
\end{equation}
When $r=0$, all phases equal one, so
\begin{equation}
    Z_\beta(r)
    \le
    Z_\beta(0)
    \qquad
    \text{for every } r \pmod{a^{n-1}}.
\end{equation}
Dividing the exact Fourier formula \eqref{eq:Zbetar_2} by \eqref{eq:Zbetabar} gives
\begin{equation}
    \frac{Z_\beta(r)}{\overline{Z}_\beta}
    =
    \sum_{k=0}^{a^{n-1}-1}
    \exp\!\left(
        -\frac{2\pi i k b^{n-1}r}{a^{n-1}}
    \right)
    \prod_{j=1}^{n-1}
    \frac{
        G_\beta\!\left(kb^{j-1}/a^j\right)
    }{
        G_\beta(0)
    }.
\end{equation}

We next express the sum at $r=0$ as repeated applications of the same operation.
For a one-periodic function $f$, define
\begin{equation}
    (\mathcal L_\beta f)(t)
    =
    \sum_{s=0}^{a-1}
    \frac{
        G_\beta\!\left((bt+s)/a\right)
    }{
        G_\beta(0)
    }
    f\!\left(
        \frac{bt+s}{a}
    \right).
\end{equation}
Expanding $\mathcal L_\beta^{\,n-1}$ and using the coprimality of $a$ and $b$
to reindex the resulting $a^{n-1}$ terms by the index $k$ gives
\begin{equation}
    \frac{Z_\beta(0)}{\overline Z_\beta}
    =
    \left(\mathcal L_\beta^{\,n-1}1\right)(0).
\end{equation}
Indeed, starting from $t_0=0$, the branch sequence $(s_1,\ldots,s_{n-1})$ visits the points $t_j=a^{-j}\sum_{l\le j}a^{l-1}b^{j-l}s_l$. The map sending this sequence to $k\equiv\sum_{l}a^{l-1}b^{1-l}s_l\pmod{a^{n-1}}$, where $b^{-1}$ denotes the inverse of $b$ modulo $a^{n-1}$, is a bijection onto $\{0,\ldots,a^{n-1}-1\}$ and satisfies $t_j\equiv kb^{j-1}/a^j\pmod 1$. Since $G_\beta$ is one-periodic, the two expansions agree term by term.
This identity reduces the large-$n$ problem to the asymptotic behavior
of repeated applications of $\mathcal L_\beta$. We obtain this behavior from
the Ruelle--Perron--Frobenius theorem for contractive iterated function systems \cite[Corollary~4.4]{ye2013ruelle}.

\begin{lem}[Ruelle--Perron--Frobenius asymptotics]
For every fixed $\beta>0$, there exist $\rho_{a,b}(\beta)>0$ and
$C_{a,b}(\beta)>0$ such that
\begin{equation} \label{eq:RPFtheorem}
    \left(\mathcal L_\beta^m 1\right)(0)
    =
    C_{a,b}(\beta)\rho_{a,b}(\beta)^m[1+o(1)]
\end{equation}
as $m\to\infty$.
\end{lem}

\begin{proof}
We apply the Ruelle--Perron--Frobenius theorem for contractive iterated function systems. It applies to a finite family of $C^1$ strict contractions on a compact convex set $K$, provided that the derivatives never vanish and the weights are strictly positive and Dini-continuous. The theorem states that, after normalization by the corresponding power of the spectral radius, the iterates converge uniformly on $K$ to a nonzero positive rank-one limit. We now show that this theorem applies to $\mathcal L_\beta$.

Let
\begin{equation}
    K
    =
    \left[
        0,\frac{a-1}{a-b}
    \right].
\end{equation}

For every $t\in K$ and $s\in\{0,\ldots,a-1\}$,
\begin{equation}
    0
    \le
    \frac{bt+s}{a}
    \le
    \frac{b(a-1)/(a-b)+a-1}{a}
    =
    \frac{a-1}{a-b}.
\end{equation}

Thus for each $s\in\{0,\ldots,a-1\}$, the map $t\longmapsto \frac{bt+s}{a}$ sends $K$ into itself. The images of $K$ under these maps form a chain of overlapping intervals from
$0$ to $(a-1)/(a-b)$, since $b(a-1)/(a-b)\ge 1$. Their union is therefore $K$.

For any $t,t'\in K$ and any $s$, the branch $t\mapsto(bt+s)/a$ satisfies
\begin{equation}
    \left|
        \frac{bt+s}{a}
        -
        \frac{bt'+s}{a}
    \right|
    =
    \frac{b}{a}|t-t'|.
\end{equation}
Since $a>b$, each map is a strict contraction with contraction ratio $b/a<1$. Each map is $C^1$, and its derivative is $b/a>0$, so the derivative never vanishes.

For the map corresponding to $s$, the weight is $G_\beta((bt+s)/a) / G_\beta(0)$. As a function of $t$, this weight is strictly positive and analytic on $K$, and is therefore Dini-continuous. The Ruelle--Perron--Frobenius theorem therefore applies to $\mathcal L_\beta$. Let $\rho_{a,b}(\beta)$ be the spectral radius of $\mathcal L_\beta$. The theorem gives $\rho_{a,b}(\beta)>0$ and shows that $\rho_{a,b}(\beta)^{-m}\mathcal L_\beta^m 1$ converges uniformly on $K$ to a strictly positive continuous function. Let $C_{a,b}(\beta)>0$ be the value of this limiting function at $0$. Evaluating the uniform convergence at $0$ gives \eqref{eq:RPFtheorem}.
\end{proof}

\begin{lem}
For every $\beta>0$,\begin{equation}
    1<\rho_{a,b}(\beta)<a.
\end{equation}
\end{lem}

\begin{proof}
We first prove that $\rho_{a,b}(\beta)>1$. Separating the $s=0$ term in the
definition of $\mathcal L_\beta$ gives
\begin{equation}
    \left(\mathcal L_\beta^{m+1}1\right)(0)
    -
    \left(\mathcal L_\beta^m1\right)(0)
    =
    \sum_{s=1}^{a-1}
    \frac{G_\beta(s/a)}{G_\beta(0)}
    \left(\mathcal L_\beta^m1\right)(s/a).
\end{equation}

For each $s=1,\ldots,a-1$, retain in
$\left(\mathcal L_\beta^m1\right)(s/a)$ only the path that chooses the branch $0$ (i.e., $t\mapsto bt/a$) at
every subsequent step. Its argument after $j$ steps is
$sb^j/a^{j+1}$, and hence
\begin{equation}
    \left(\mathcal L_\beta^m1\right)(s/a)
    \ge
    \prod_{j=1}^{m}
    \frac{G_\beta(sb^j/a^{j+1})}{G_\beta(0)}
    \ge
    \prod_{j=1}^{\infty}
    \frac{G_\beta(sb^j/a^{j+1})}{G_\beta(0)}.
\end{equation}
The second inequality follows from
$0<G_\beta(t)/G_\beta(0)\le 1$. Moreover,
\begin{equation}
    0
    \le
    1-\frac{G_\beta(t)}{G_\beta(0)}
    =
    \frac{1-\cos(2\pi t)}
    {\cosh\beta-\cos(2\pi t)}
    \le
    \frac{2\pi^2t^2}{\cosh\beta-1},
\end{equation}
and therefore
\begin{equation}
    \sum_{j=1}^{\infty}
    \left[
        1-
        \frac{G_\beta(sb^j/a^{j+1})}{G_\beta(0)}
    \right]
    \le
    \frac{2\pi^2s^2}{a^2(\cosh\beta-1)}
    \sum_{j=1}^{\infty}
    \left(\frac{b}{a}\right)^{2j}
    <
    \infty.
\end{equation}

The corresponding infinite-product criterion can be verified directly. Since
$sb^j/a^{j+1}\to 0$, the factors converge to $1$ and are therefore eventually
at least $1/2$. For $1/2\le x\le 1$,
\begin{equation}
    -\log x\le 2(1-x).
\end{equation}
Thus the sum of the negative logarithms of these factors is finite. The finitely
many preceding factors are also strictly positive, so the logarithm of the
infinite product does not tend to negative infinity. Consequently,
\begin{equation}
    \prod_{j=1}^{\infty}
    \frac{G_\beta(sb^j/a^{j+1})}{G_\beta(0)}
    >
    0.
\end{equation}

Substituting this bound into the increment above, with the prefactor
$G_\beta(s/a)/G_\beta(0)$ supplying the $j=0$ factor, gives
\begin{equation}
    \left(\mathcal L_\beta^{m+1}1\right)(0)
    -
    \left(\mathcal L_\beta^m1\right)(0)
    \ge
    \sum_{s=1}^{a-1}
    \prod_{j=0}^{\infty}
    \frac{G_\beta(sb^j/a^{j+1})}{G_\beta(0)}
    >
    0.
\end{equation}

The right-hand side is independent of $m$. Iterating the inequality therefore
yields
\begin{equation}
    \left(\mathcal L_\beta^m1\right)(0)
    \ge
    1
    +
    m
    \sum_{s=1}^{a-1}
    \prod_{j=0}^{\infty}
    \frac{G_\beta(sb^j/a^{j+1})}{G_\beta(0)}
    \longrightarrow
    \infty.
\end{equation}

If $\rho_{a,b}(\beta)<1$, then by \eqref{eq:RPFtheorem} $(\mathcal L_\beta^m1)(0)\to0$,
whereas if $\rho_{a,b}(\beta)=1$, it tends to $C_{a,b}(\beta)$. Both cases
contradict the preceding divergence. Hence
\begin{equation}
    \rho_{a,b}(\beta)>1.
\end{equation}

We now prove that $\rho_{a,b}(\beta)<a$. Because $\sum_{s=0}^{a-1}e^{2\pi iks/a}$ vanishes unless $a\mid k$, we have
\begin{equation}
\begin{aligned}
    \sum_{s=0}^{a-1}
    G_\beta\!\left(\frac{bt+s}{a}\right)
    &=
    \sum_{k\in\mathbb Z}
    e^{-\beta|k|}
    e^{2\pi i kbt/a}
    \sum_{s=0}^{a-1}
    e^{2\pi i ks/a}
    \\
    &=
    a\sum_{k\in\mathbb Z}
    e^{-a\beta|k|}
    e^{2\pi i kbt}
    =
    aG_{a\beta}(bt).
\end{aligned}
\end{equation}

Because $G_{a\beta}$ is maximized at $0$, we have
\begin{equation}
    \sum_{s=0}^{a-1}
    \frac{G_\beta((bt+s)/a)}{G_\beta(0)}
    =
    a\frac{G_{a\beta}(bt)}{G_\beta(0)}
    \le
    a\frac{\coth(a\beta/2)}{\coth(\beta/2)}.
\end{equation}

This row-sum bound implies inductively that, for every $m\ge 0$,
\begin{equation}
    \left(\mathcal L_\beta^m1\right)(t)
    \le
    \left[
        a\frac{\coth(a\beta/2)}{\coth(\beta/2)}
    \right]^m.
\end{equation}
Therefore, \begin{equation}
    \rho_{a,b}(\beta)
    =
    \lim_{m\to\infty}
    \left[
        \left(\mathcal L_\beta^m1\right)(0)
    \right]^{1/m}
    \le
    a\frac{\coth(a\beta/2)}{\coth(\beta/2)}
    <
    a.\qedhere
\end{equation}
\end{proof}

The preceding asymptotics and the expression for $Z_\beta(0)/\overline{Z}_\beta$ give
\begin{equation}
\frac{Z_\beta(0)}{Z_\beta^{\mathrm{tot}}}
= \frac{1}{a^{n-1}}\frac{Z_\beta(0)}{\overline{Z}_\beta}
= C_{a,b}(\beta)\left(\frac{\rho_{a,b}(\beta)}{a}\right)^{n-1}[1+o(1)].
\end{equation}
Since $Z_\beta(r)\le Z_\beta(0)$ and $\rho_{a,b}(\beta)<a$, the Gibbs probability of every fixed-endpoint slice decays exponentially in $n$.

To study how the Gibbs weight is distributed among all fixed-endpoint slices, equip $\mathbb{Z}^{n-1}$ with the probability measure $\mathbb{P}_\beta$ defined by
\begin{equation}
\mathbb{P}_\beta(y)=\frac{e^{-\beta|y|_1}}{Z_\beta^{\mathrm{tot}}}
=\prod_{j=1}^{n-1}\frac{e^{-\beta|y_j|}}{G_\beta(0)},\qquad y\in\mathbb{Z}^{n-1}.
\end{equation}
Thus $y_1,\ldots,y_{n-1}$ are independent and identically distributed under $\mathbb{P}_\beta$. We need to study
\begin{equation}
\frac{Z_\beta(r)}{Z_\beta^{\mathrm{tot}}}
=\mathbb{P}_\beta\!\left(Q(y)\equiv-b^{n-1}r\pmod{a^{n-1}}\right).
\end{equation}
Since the same $r$ may be produced by many different $y$, we remove this multiplicity by assigning to each $y$ an $a$-ary digit block $B$ that, as shown below, is uniquely determined by $r$. Starting with the bond $y_{n-1}$ adjacent to the fixed endpoint, set $C_{n-1}=0$. For $j=n-1,n-2,\ldots,1$, divide the integer $bC_j-y_j$ by $a$. Denote the quotient by $C_{j-1}$ and the remainder by $B_j$. Thus,
\begin{equation}
bC_j-y_j=aC_{j-1}+B_j,\qquad C_{j-1}\in\mathbb{Z},\qquad B_j\in\{0,1,\ldots,a-1\}.
\end{equation}
The quotient $C_{j-1}$ is the carry passed to the next step, and the remainder $B_j$ is recorded as the $j$-th digit. Since Euclidean division has a unique quotient and remainder, this procedure uniquely determines $(B_1,\ldots,B_{n-1})$ from $y$.
Multiplying the last identity by $a^{n-1-j}b^{j-1}$ and summing over $j=1,\ldots,n-1$ cancels all intermediate carry terms and gives
\begin{equation}
\begin{aligned}
\sum_{j=1}^{n-1}a^{n-1-j}b^{j-1}B_j
&=b^{n-1}C_{n-1}-a^{n-1}C_0-Q(y)\\
&=-a^{n-1}C_0-Q(y).
\end{aligned}
\end{equation}
Consequently, $Q(y)\equiv-b^{n-1}r\pmod{a^{n-1}}$ if and only if
\begin{equation}
\sum_{j=1}^{n-1}a^{n-1-j}b^{j-1}B_j\equiv b^{n-1}r\pmod{a^{n-1}}.
\end{equation}
It remains to verify that this digit block records exactly the endpoint $r$. Suppose that $(B_1,\ldots,B_{n-1})$ and $(B'_1,\ldots,B'_{n-1})$ correspond to the same $r$. Then
\begin{equation}
\sum_{j=1}^{n-1}a^{n-1-j}b^{j-1}(B_j-B'_j)\equiv0\pmod{a^{n-1}}.
\end{equation}
Reducing this congruence modulo $a$ gives
\begin{equation}
b^{n-2}(B_{n-1}-B'_{n-1})\equiv0\pmod{a}.
\end{equation}
Since $\gcd(a,b)=1$ and $B_{n-1},B'_{n-1}\in\{0,1,\ldots,a-1\}$, it follows that $B_{n-1}=B'_{n-1}$.
Repeating the same argument gives
\begin{equation}
B_j=B'_j,\qquad j=1,\ldots,n-1.
\end{equation}
The map from digit blocks $B$ to $r\in\{0,1,\ldots,a^{n-1}-1\}$ is therefore bijective, and $Z_\beta(r)/Z_\beta^{\mathrm{tot}}$ is the probability of the corresponding block $B$.

Next, before we discuss the probability of the finite vector $B=(B_1,\ldots,B_{n-1})$, we first consider infinite sequences $(y_j)_{j\in\mathbb{Z}}$ and $(B_j)_{j\in\mathbb{Z}}$. We will use the following consequence of the Shannon–McMillan–Breiman theorem~\cite{algoet1988sandwich}.

\begin{lem}[Typical sets from the Shannon--McMillan--Breiman theorem] \label{lem:SMB}
Define the left shift $\sigma$ on $\{0,\ldots,a-1\}^{\mathbb{Z}}$ by $(\sigma u)_j=u_{j+1}$. 
The distribution of $B$ is stationary if
\begin{equation}
\mathbb{P}_\beta(B\in\sigma^{-1}E)=\mathbb{P}_\beta(B\in E)
\end{equation}
for every measurable subset $E\subseteq\{0,\ldots,a-1\}^{\mathbb{Z}}$.
A stationary distribution is ergodic if every measurable subset $E$ satisfying $\sigma^{-1}E=E$
also satisfies
\begin{equation}
\mathbb{P}_\beta(B\in E)\in\{0,1\}.
\end{equation}
Suppose that the distribution of $B$ is stationary and ergodic. For every
\begin{equation}
u=(u_1,\ldots,u_d)\in\{0,\ldots,a-1\}^d,
\end{equation}
define
\begin{equation}
p_d(u)=\mathbb{P}_\beta(B_1=u_1,\ldots,B_d=u_d).
\end{equation}
The Shannon entropy of the length-$d$ block is
\begin{equation}
H(B_1,\ldots,B_d)=-\sum_{u\in\{0,\ldots,a-1\}^d}p_d(u)\log p_d(u),
\end{equation}
and the entropy rate of $B$ is
\begin{equation}
h_\beta=\lim_{d\to\infty}\frac{1}{d}H(B_1,\ldots,B_d).
\end{equation}
For every $\eta>0$, define
\begin{equation}
A_d=\left\{u\in\{0,\ldots,a-1\}^d:p_d(u)\ge e^{-d(h_\beta+\eta)}\right\}.
\end{equation}
Then
\begin{equation}
\mathbb{P}_\beta\bigl((B_1,\ldots,B_d)\in A_d\bigr)\longrightarrow1,
\qquad |A_d|\le e^{d(h_\beta+\eta)}.
\end{equation}
\end{lem}

\begin{proof}
The Shannon--McMillan--Breiman theorem gives
\begin{equation}
-\frac{1}{d}\log p_d(B_1,\ldots,B_d)\xrightarrow[d\to\infty]{\mathrm{a.s.}}h_\beta.
\end{equation}
By the definition of $A_d$, $(B_1,\ldots,B_d)\in A_d$
is equivalent to
\begin{equation}
-\frac{1}{d}\log p_d(B_1,\ldots,B_d)\le h_\beta+\eta.
\end{equation}
Almost-sure convergence therefore implies that this event eventually occurs for almost every realization, and hence
\begin{equation}
\mathbb{P}_\beta\bigl((B_1,\ldots,B_d)\in A_d\bigr)\longrightarrow1.
\end{equation}
Moreover, every $u\in A_d$ has probability at least $e^{-d(h_\beta+\eta)}$. Therefore,
\begin{equation}
1\ge\sum_{u\in A_d}p_d(u)\ge|A_d|e^{-d(h_\beta+\eta)},
\end{equation}
which gives the stated cardinality bound.
\end{proof}

To apply the preceding lemma, it remains to construct an infinite random vector $B$ whose distribution is stationary and ergodic and to show that its entropy rate is strictly smaller than $\log a$. For the infinite construction, we use the same notation $\mathbb{P}_\beta$ for the product measure on $\mathbb{Z}^{\mathbb{Z}}$ under which the coordinates $y_j$ are independent and identically distributed, with $\mathbb{P}_\beta(y_j=k)=e^{-\beta|k|}/G_\beta(0)$ for every $j,k\in\mathbb{Z}$. We write $\mathbb{E}_\beta$ for expectation with respect to this measure.

\begin{lem}
    There exist integer-valued random variables $(C_j)_{j\in\mathbb{Z}}$ and digits $(B_j)_{j\in\mathbb{Z}}$ satisfying $bC_j-y_j=aC_{j-1}+B_j$, where $B_j\in\{0,1,\ldots,a-1\}$, such that the distribution of $B=(B_j)_{j\in\mathbb{Z}}$ is stationary and ergodic. Moreover, $C_j$ is independent of $(y_k)_{k\le j}$ for every $j$, and $h_\beta<\log a$.
\end{lem}

\begin{proof}
Viewed from right to left, the carries $C_j$ form a Markov chain. Given $C_j=c$, sample $y_j$ independently with the distribution above and set
\begin{equation}
C_{j-1}=\left\lfloor\frac{bc-y_j}{a}\right\rfloor,\qquad B_j=bc-y_j-aC_{j-1}.
\end{equation}
For any $c,c'\in\mathbb{Z}$, choosing $y_j=bc-ac'$ gives $C_{j-1}=c'$, so the chain is irreducible; taking $c'=c$ also gives a self-loop. Furthermore,
\begin{equation}
\mathbb{E}_\beta\!\left[|C_{j-1}|-|C_j|\mid C_j=c\right]\le-\left(1-\frac{b}{a}\right)|c|+\frac{\mathbb{E}_\beta|y_j|}{a}+1.
\end{equation}
Since $\beta>0$, the distribution of $y_j$ has exponentially decaying tails, and hence $\mathbb{E}_\beta|y_j|<\infty$. Since $b<a$, the right-hand side is at most $-1$ for all sufficiently large $|c|$. The Foster--Lyapunov criterion therefore implies that the Markov chain is positive recurrent.

An irreducible positive recurrent Markov chain has a unique stationary distribution. Let $\mu$ denote this stationary distribution. On any finite interval of indices, sample the carry at the right endpoint according to $\mu$, independently sample the $y_j$ on that interval, and apply the recursion from right to left. The invariance of $\mu$ under the carry transition and the independence of the $y_j$ make these finite-dimensional distributions consistent under restriction and invariant under translations of the indices. By the Kolmogorov extension theorem, they define a stationary bi-infinite joint process $(C_j,y_j)_{j\in\mathbb{Z}}$. We continue to denote its distribution and expectation by $\mathbb{P}_\beta$ and $\mathbb{E}_\beta$, respectively. By construction, the $y_j$ retain their original independent and identically distributed law, $C_j\sim\mu$ for every $j$, and $C_j$ is independent of $(y_k)_{k\le j}$.

The joint process $(C_j,y_j)_{j\in\mathbb{Z}}$, read from right to left, is a stationary Markov chain on $\mathbb{Z}^2$. Starting from $(C_j,y_j)=(c,y)$, the recursion fixes $C_{j-1}$. For any prescribed state $(c',y')\in\mathbb{Z}^2$, choosing $y_{j-1}=bC_{j-1}-ac'$ and $y_{j-2}=y'$ gives $(C_{j-2},y_{j-2})=(c',y')$. Both choices have positive probability because the distribution of each $y_j$ has full support on $\mathbb{Z}$. Hence the joint Markov chain is irreducible. A stationary irreducible Markov chain on a countable state space is ergodic with respect to the shift. Therefore, $(C_j,y_j)_{j\in\mathbb{Z}}$ is ergodic.

Each digit $B_j=bC_j-y_j-aC_{j-1}$ is obtained from the joint process by the same measurable rule at every index, so the distribution of $B=(B_j)_{j\in\mathbb{Z}}$ is stationary. The inverse image of every shift-invariant measurable subset of digit sequences is shift-invariant for the joint process, so the distribution of $B$ is also ergodic.

For $d\ge 1$, define the weighted digit sum represented by a length-$d$ digit block by
\begin{equation}
V_d=\sum_{j=1}^{d}a^{d-j}b^{j-1}B_j.
\end{equation}
The map from $(B_1,\ldots,B_d)$ to $V_d\pmod{a^d}$ is bijective. Summing the carry identities over this block gives
\begin{equation}
V_d=b^dC_d-a^dC_0-\sum_{j=1}^d a^{d-j}b^{j-1}y_j.
\end{equation}
We next show that one finite digit block is not uniformly distributed. For a positive integer $L$, the term containing $C_0$ disappears after exponentiation because $C_0\in\mathbb{Z}$. The independence of $C_L$ from $y_1,\ldots,y_L$ then gives
\begin{equation}
\mathbb{E}_\beta e^{2\pi iV_L/a^L}
=\mathbb{E}_\beta e^{2\pi i(b/a)^LC_L}\prod_{j=1}^L\mathbb{E}_\beta e^{-2\pi ib^{j-1}y_j/a^j}.
\end{equation}
For every real $\theta$, the marginal distribution of $y_j$ gives $\mathbb{E}_\beta e^{i\theta y_j}=(\cosh\beta-1)/(\cosh\beta-\cos\theta)>0$. Thus every factor in the product is strictly positive. Since $C_L$ has distribution $\mu$,
\begin{equation}
\left|\mathbb{E}_\beta e^{2\pi i(b/a)^LC_L}-1\right|
\le\sum_{c\in\mathbb{Z}}\mu(c)\min\left\{2,2\pi\left(\frac{b}{a}\right)^L|c|\right\}\longrightarrow 0.
\end{equation}
Hence $\mathbb{E}_\beta e^{2\pi iV_L/a^L}\neq 0$ for all sufficiently large $L$. Fix one such $L$. The residue $V_L\pmod{a^L}$ is not uniformly distributed. Using stationarity, subadditivity of block entropy, and the bijection between digit blocks and residues, we obtain
\begin{equation}
h_\beta\le\frac{1}{L}H(B_1,\ldots,B_L)=\frac{1}{L}H(V_L\bmod a^L)<\log a.
\end{equation}
This completes the proof.
\end{proof}

Choose $0<\eta<\log a-h_\beta$, and let $A_d$ be the set given by Lemma \ref{lem:SMB}. Then
\begin{equation}
\mathbb{P}_\beta\bigl((B_1,\ldots,B_d)\in A_d\bigr)\longrightarrow 1,\qquad |A_d|\le e^{d(h_\beta+\eta)}.
\end{equation}

We continue to work with the stationary construction above. For each $n\ge 2$, set $d=n-1$ and consider the first $d$ coordinates $y=(y_1,\ldots,y_d)$. Let $r\in\{0,\ldots,a^d-1\}$ be the unique residue satisfying
\begin{equation}
Q(y)\equiv-b^dr\pmod{a^d}.
\end{equation}

Since $\gcd(a,b)=1$, the integer $b^d$ is invertible modulo $a^d$; let $b^{-d}$ denote its multiplicative inverse. The carry identity above and the endpoint congruence give
\begin{equation}
r\equiv b^{-d}V_d-C_d\pmod{a^d}.
\end{equation}

Let $S_d\subseteq\{0,\ldots,a^d-1\}$ be the set of standard representatives of the residues in
\begin{equation}
\left\{b^{-d}\sum_{j=1}^d a^{d-j}b^{j-1}u_j-c\pmod{a^d}:(u_1,\ldots,u_d)\in A_d,\ c\in\mathbb{Z},\ |c|\le d\right\}.
\end{equation}

If $(B_1,\ldots,B_d)\in A_d$ and $|C_d|\le d$, then the preceding congruence implies that $r\in S_d$. Therefore,
\begin{equation}
\begin{aligned}
\mathbb{P}_\beta(r\in S_d)&\ge\mathbb{P}_\beta\bigl((B_1,\ldots,B_d)\in A_d,\ |C_d|\le d\bigr)\\
&\ge\mathbb{P}_\beta\bigl((B_1,\ldots,B_d)\in A_d\bigr)-\mathbb{P}_\beta(|C_d|>d).
\end{aligned}
\end{equation}

Since $C_d\sim\mu$ for every $d$, and $\mu$ is a probability distribution on $\mathbb{Z}$, 
\begin{equation}
\mathbb{P}_\beta(|C_d|>d)\xrightarrow{d\to\infty} 0.
\end{equation}
It follows that
\begin{equation}
\mathbb{P}_\beta(r\in S_d)\xrightarrow{d\to\infty} 1.
\end{equation}

For each digit block in $A_d$, there are exactly $2d+1$ integers $c$ with $|c|\le d$, so each block contributes at most $2d+1$ elements to $S_d$. Consequently,
\begin{equation}
\frac{|S_d|}{a^d}\le(2d+1)\frac{|A_d|}{a^d}\le(2d+1)e^{-d(\log a-h_\beta-\eta)}\xrightarrow{d\to\infty} 0.
\end{equation}

Under the infinite construction, $(y_1,\ldots,y_{n-1})$ has exactly the same product distribution as the original finite vector. Since the endpoint $r$ is determined by $y$,
its distribution is the same in the finite and infinite constructions. Therefore,
\begin{equation}
\sum_{r\in S_{n-1}}\frac{Z_\beta(r)}{Z_\beta^{\mathrm{tot}}}\xrightarrow{n\to\infty} 1,\qquad \frac{|S_{n-1}|}{a^{n-1}}\xrightarrow{n\to\infty} 0.
\end{equation}

In summary, for every fixed $\beta>0$, the Gibbs distribution $Z_\beta(r)$ exhibits two distinct scales. Since $1<\rho_{a,b}(\beta)<a$, the Gibbs weight of the $r=0$ slice is exponentially larger than the average slice weight, while the Gibbs probability of every individual slice still decays exponentially with $n$. Consequently, exponentially many slices are required to carry any fixed positive fraction of the total Gibbs probability. On the other hand, there exists a set $S_{n-1}\subseteq\{0,\ldots,a^{n-1}-1\}$ whose total Gibbs probability tends to $1$, while $|S_{n-1}|/a^{n-1}$ tends to $0$ exponentially fast. Thus, the Gibbs weight is concentrated on a set containing exponentially many slices but constituting an exponentially small fraction of all fixed-endpoint slices.

A physical interpretation for this phenomenon is as follows.   Each $Z_\beta(r)$ is dominated by configurations that locally ``look like'' the helical repetition code codewords, but which have to ``jump'' between different codeword sectors at certain points.  As a simple illustration, take $a=2$ and $b=1$.  We see that $Z_\beta(0)=Z_\beta(2^{n-1})$: we can map the two partition functions into each other by adding or subtracting the codeword $v=(1,2,4,\ldots,2^{n-1})$.    Now compare $Z_\beta(0)$ with $Z_\beta(2^{n-\ell})$ for some $\ell>1$ with $\ell=O(1)$.  There is a simple way to find an $O(1)$ energy configuration:  simply take $\widetilde x_1=\cdots = \widetilde x_{\ell-1}=0$ followed by $\widetilde x_i = 2^{i-\ell}$ for $i\ge \ell$.   The presence of the ``slip'' near $i=\ell$ means that $Z_\beta(2^{n-\ell})$ will be smaller than $Z_\beta(0)$, but not by much.   The configurations in the set $S_{n-1}$ which dominate the Gibbs measure will correspond to the choices of $r$ which are similar to the above, and exhibit a low-density of these ``slips".   This point is discussed further for a related model in Ref.~\cite{andy_dlr}.

\section{Proof of Proposition \ref{prop:OBCOBCHGP}} \label{app:proofOBCOBCHGP}

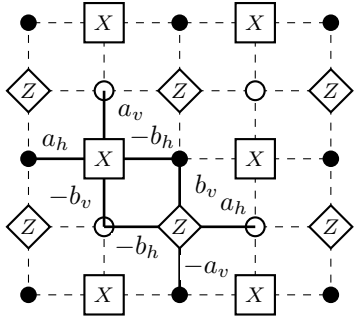
\begin{figure}[t]
\centering
\begin{tikzpicture}[
    x=2cm,y=1.8cm,
    blackq/.style={circle,fill=black,inner sep=2.2pt},
    whiteq/.style={circle,draw,fill=white,inner sep=2.5pt,line width=0.8pt},
    xcheck/.style={rectangle,draw,fill=white,minimum size=5.2mm,inner sep=0pt,line width=0.8pt},
    zcheck/.style={diamond,draw,fill=white,minimum size=5.8mm,inner sep=0pt,aspect=1,line width=0.8pt},
    bg/.style={dashed,thin},
    stub/.style={dashed,thin},
    hi/.style={line width=1.0pt},
    lbl/.style={fill=white,inner sep=1pt,font=\small}
]


\foreach \i in {0,1}{
  \foreach \j in {0,1,2}{
    \draw[bg] (\i+0.5,\j) -- (\i,\j);
    \draw[bg] (\i+0.5,\j) -- (\i+1,\j);
    \ifnum\j>0
      \draw[bg] (\i+0.5,\j) -- (\i+0.5,\j-0.5);
    \fi
    \ifnum\j<2
      \draw[bg] (\i+0.5,\j) -- (\i+0.5,\j+0.5);
    \fi
  }
}


\foreach \i in {0,1,2}{
  \foreach \j in {0,1}{
    \draw[bg] (\i,\j+1/2) -- (\i,\j);
    \draw[bg] (\i,\j+1/2) -- (\i,\j+1);
    \ifnum\i>0
      \draw[bg] (\i,\j+1/2) -- (\i-0.5,\j+1/2);
    \fi
    \ifnum\i<2
      \draw[bg] (\i,\j+1/2) -- (\i+0.5,\j+1/2);
    \fi
  }
}


\foreach \i in {0,1,2}{
  \foreach \j in {0,1,2}{
    \node[blackq] at (\i,\j) {};
  }
}

\foreach \i in {0,1}{
  \foreach \j in {0,1}{
    \node[whiteq] at (\i+0.5,\j+0.5) {};
  }
}

\draw[hi] (0.5,1) -- (0,1);
\draw[hi] (0.5,1) -- (1,1);
\draw[hi] (0.5,1) -- (0.5,0.5);
\draw[hi] (0.5,1) -- (0.5,1.5);

\node[lbl] at (0.18,1.12) {$a_h$};
\node[lbl] at (0.82,1.16) {$-b_h$};
\node[lbl] at (0.28,0.72) {$-b_v$};
\node[lbl] at (0.68,1.34) {$a_v$};

\draw[hi] (1,0.5) -- (1,0);
\draw[hi] (1,0.5) -- (1,1);
\draw[hi] (1,0.5) -- (0.5,0.5);
\draw[hi] (1,0.5) -- (1.5,0.5);

\node[lbl] at (1.18,0.2) {$-a_v$};
\node[lbl] at (1.18,0.8) {$b_v$};
\node[lbl] at (0.72,0.35) {$-b_h$};
\node[lbl] at (1.36,0.65) {$a_h$};

\foreach \i in {0,1}{
  \foreach \j in {0,1,2}{
    \node[xcheck] at (\i+0.5,\j) {\scriptsize $X$};
  }
}

\foreach \i in {0,1,2}{
  \foreach \j in {0,1}{
    \node[zcheck] at (\i,\j+0.5) {\scriptsize $Z$};
  }
}

\end{tikzpicture}
\caption{
Local Tanner-graph schematic for the helical HGP code built from the horizontal OBC
$(a_h,-b_h)$ helical code and the vertical OBC $(a_v,-b_v)$ helical code.
Filled and open circles denote the two types of physical qudits in the HGP code. We label the filled-circle qudits by pairs (\(i,j\)), with  \(i,j=1,\ldots,n\), and label the open-circle qudits by pairs (\(k,l\)), with \(k,l=1,\ldots,n-1\). Here the first and second coordinates label positions in the vertical and horizontal directions, respectively.
Squares denote $X$-checks, and diamonds denote $Z$-checks.
For an interior $X$-check, the incident coefficients are $(a_h,-b_h,-b_v,a_v)$;
for an interior $Z$-check, they are $(-a_v,b_v,-b_h,a_h)$. 
}
\label{fig:OBC_OBC_hgp_tanner}
\end{figure}

We now give the rough idea of the proof: we first lift the HGP code to the integer lattice $\mathbb Z^N$.
In the integer model, we define cylinders analogous to the one used in Proposition \ref{prop:classicalOBC}: each cylinder is centered at a codeword and imposes an upper cutoff on the syndrome energy.
We then estimate the lower bound on memory time by calculating the time needed to reach the boundary of the cylinder starting from any codeword.
Finally, we transfer the integer cylinders back to the original code over $\mathbb Z_q$. For this step, we prove that, for the cylinder chosen above, reduction modulo $q$ is one-to-one on the entire cylinder. We also prove that the ordinary $\ell_1$ energy of the corresponding integer syndrome agrees exactly with the Lee energy over $\mathbb Z_q$ for all points in the cylinder. Therefore, the dynamics within the cylinder is preserved after reduction modulo $q$, which gives the same lower bound on the memory time as over $\mathbb Z$. Throughout the proof, for any element of $\mathbb Z_q$, we always choose its representative in $(-q/2,q/2]\cap\mathbb Z.$

We prove the proposition in three steps.

\noindent \textbf{Step 1: Logical generator.}

In this step, we derive the generator of all logical $Z$ operators. We first explain the indexing convention shown in Fig.~\ref{fig:OBC_OBC_hgp_tanner}: the physical qudits are the circular vertices. There are two types of
such vertices. We call the filled circles black qudits and the open circles white qudits. 

A black qudit is labeled by $(i,j)$, where $i=1,\ldots,n$ is a vertical vertex and $j=1,\ldots,n$ is a horizontal vertex. A white qudit is labeled by $(k,l)$, where $k=1,\ldots,n-1$ is a vertical edge and $l=1,\ldots,n-1$ is a horizontal
edge.

For a $Z$-type error, we write
\begin{equation}
    z=(v,t),
\end{equation}
where $v=(v_{i,j})\in\mathbb Z_q^{n\times n}$ records the error values on the black qudits, and $t=(t_{k,l})\in\mathbb Z_q^{(n-1)\times(n-1)}$ records the error values on the white qudits.

The $X$-checks are the square vertices in Fig.~\ref{fig:OBC_OBC_hgp_tanner}. We label them by pairs
$(i,l)$ with the above convention. The $X$-syndrome of a $Z$-type error is
\begin{equation}
    H_Xz=vH_h^{\mathsf T}+H_v^{\mathsf T}t.
\end{equation}

We now lift the code to $\mathbb Z$: the matrices $H_X$ and $H_Z$ are kept
unchanged as integer matrices, while all configurations, stabilizers, syndromes,
and logical operators are now considered over $\mathbb Z$. We then identify the
generator of the integer logical $Z$ operators.

From Proposition~\ref{prop:helicalproperties}, the codewords of the
OBC helical repetition code $H_{\mathrm{OBC}}(a,-b)$ are generated by the integer
vector
\begin{equation}
    u(a,b)
    =
    \left(b^{n-1},ab^{n-2},a^2b^{n-3},\ldots,a^{n-1}\right).
\end{equation}
We write $u_h=u(a_h,b_h)$ and $u_v=u(a_v,b_v)$ for the horizontal and vertical kernel vectors.
Both $u_h$ and $u_v$ are primitive integer vectors. Hence we may choose an
integer vector $r_v\in\mathbb Z^n$ such that
\(
    u_v^{\mathsf T}r_v=1.
\)
Define
\begin{equation}
    \ell_Z=(r_vu_h^{\mathsf T},0),
\end{equation}
where the first component is the black-qudit block and the second component is the white-qudit block.
We claim that every zero-syndrome integer $Z$-error vector can be written as an
integer multiple of $\ell_Z$ plus an integer $Z$-stabilizer. We now prove this decomposition.
First, $\ell_Z\in\ker_{\mathbb Z}H_X$, since
\begin{equation}
    H_X\ell_Z
    =
    (r_vu_h^{\mathsf T})H_h^{\mathsf T}
    =
    r_v(H_hu_h)^{\mathsf T}
    =
    0.
\end{equation}
Let
\(
    (v,t)\in\ker_{\mathbb Z}H_X,
\)
which means
\begin{equation}
    vH_h^{\mathsf T}+H_v^{\mathsf T}t=0.
\end{equation}
Multiplying this equation on the left by $u_v^{\mathsf T}$ gives
\(
    u_v^{\mathsf T}vH_h^{\mathsf T}=0.
\)
Equivalently, the row vector $u_v^{\mathsf T}v$ lies in the left kernel of
$H_h^{\mathsf T}$, so there is a unique integer $M$ such that
\(
    u_v^{\mathsf T}v=Mu_h^{\mathsf T}.
\)
Consider the difference
\(
    (v,t)-M\ell_Z.
\)
Its black block satisfies
\begin{equation}
    u_v^{\mathsf T}(v-Mr_vu_h^{\mathsf T})
    =
    u_v^{\mathsf T}v
    -
    M(u_v^{\mathsf T}r_v)u_h^{\mathsf T}
    =
    0.
\end{equation}
Therefore every column of $v-Mr_vu_h^{\mathsf T}$ lies in
\begin{equation}
    \{x\in\mathbb Z^n:\ u_v^{\mathsf T}x=0\}
    =
    \operatorname{im}_{\mathbb Z}H_v^{\mathsf T}.
\end{equation}
Thus there exists an integer matrix $A\in\mathbb Z^{(n-1)\times n}$ such that
\begin{equation}
    v-Mr_vu_h^{\mathsf T}
    =
    -H_v^{\mathsf T}A.
\end{equation}
We can use $A$ to generate the $Z$-stabilizer
\(
    H_Z^{\mathsf T}A
    =
    (-H_v^{\mathsf T}A,\ AH_h^{\mathsf T}).
\)
Subtracting this stabilizer from $(v,t)-M\ell_Z$, we obtain
\begin{equation}
    (v,t)-M\ell_Z-H_Z^{\mathsf T}A
    =
    (0,\ t-AH_h^{\mathsf T}).
\end{equation}
This remaining vector is still in $\ker_{\mathbb Z}H_X$. Hence
\begin{equation}
    H_v^{\mathsf T}(t-AH_h^{\mathsf T})=0.
\end{equation}
Since $H_v^{\mathsf T}$ is injective over $\mathbb Z$, we get
\begin{equation}
    t-AH_h^{\mathsf T}=0.
\end{equation}
Therefore
\begin{equation}
    (v,t)-M\ell_Z = H_Z^{\mathsf T}A.
\end{equation}
This proves that every integer $Z$-logical operator is represented by an integer
multiple of $\ell_Z$, modulo integer $Z$-stabilizers.

\noindent \textbf{Step 2: Integer cylinder and memory-time estimate.}

In this step, we still work on $\mathbb Z$. 

We define the integer syndrome energy by
\begin{equation}
    E_Z(z)=|H_Xz|_1.
\end{equation}
We use the unnormalized Gibbs weight
\begin{equation}
W(z)=e^{-\beta E_Z(z)},\qquad W(D)=\sum_{z\in D}W(z)
\end{equation}
for any set $D$ of integer states.
We also define the weighted logical coordinate
\begin{equation}
    T(v,t)=u_v^{\mathsf T}vu_h.
\end{equation}
This coordinate is invariant when applying $Z$-stabilizers. 
Moreover, for the primitive logical generator
\(
    \ell_Z=(r_vu_h^{\mathsf T},0),
\)
we have
\begin{equation}
    T(\ell_Z)
    =
    u_v^{\mathsf T}(r_vu_h^{\mathsf T})u_h
    =
    \|u_h\|_2^2.
\end{equation}
We denote this fundamental length by
\(
    L:=\|u_h\|_2^2.
\)
Thus translating by one copy of the primitive logical operator shifts $T$ by $L$.

Define the integer cylinder
\begin{equation}
    \widetilde C
    =
    \{z\in\mathbb Z^N:\ -\frac{L}{2}\le T(z)<\frac{L}{2},\ E_Z(z)\le S\},
\end{equation}
where $S$ is an energy cutoff, to be chosen in Step 3. The interval condition selects one fundamental
interval along the integer logical direction, while the energy condition
restricts the syndrome energy.

Let $z$ and $z'$ be two gauge-fixed
integer vectors with the same syndrome. By Step 1, there exist an integer $M$
and an integer stabilizer coefficient matrix $A$ such that
\begin{equation}
    z-z'=M\ell_Z+H_Z^{\mathsf T}A.
\end{equation}
Therefore
\begin{equation}
    T(z)-T(z')=ML.
\end{equation}
Consequently, among all gauge-fixed vectors with a fixed syndrome, at most one satisfies $T(z)\in[-L/2,L/2)$. It is therefore possible to choose a decoder that maps every state in the cylinder to the zero codeword, and the memory time can be lower bounded by the time needed to reach the boundary of the cylinder.

We next describe the boundary of $\widetilde C$ relevant for single-site
Metropolis updates. Let
\begin{equation}
    \Delta
    :=
    \max_{\text{single-site update }e}|T(e)|=\max_{i,j}(u_v)_i(u_h)_j
\end{equation}
be the maximum possible change in $T$ in a single step and
\begin{equation}
    C_{\mathrm{step}}
    :=
    \max_{\text{single-site update }e}|H_Xe|_1\le \max\{a_h+b_h,\ a_v+b_v\}
\end{equation}
be the maximum change in energy in a single step.

We decompose the boundary into two parts: the cap and side of the cylinder. The cap is
\begin{equation}
    \widetilde C^{\mathrm{cap}}
    =
    \left\{
        z\in\widetilde C:
        T(z)\in [-\frac{L}{2},-\frac{L}{2}+\Delta)\cup [\frac{L}{2}-\Delta,\frac{L}{2})
    \right\}.
\end{equation}
The side is
\begin{equation}
    \widetilde C^{\mathrm{side}}
    =
    \left\{
        z\in\widetilde C:
        S-C_{\mathrm{step}}\le E_Z(z)\le S
    \right\}.
\end{equation}
If $z\in\widetilde C$ and can leave $\widetilde C$ with a single-site update, then $z$ must lie in
\begin{equation}
    \widetilde C^{\mathrm{cap}}\cup \widetilde C^{\mathrm{side}}.
\end{equation}

We now estimate the time needed to reach the boundary of the cylinder starting from the zero codeword.
The time needed to reach the caps is bounded by \begin{equation}
    \tau_{\mathrm{cap}} \ge \frac{L}{2\Delta}.
\end{equation}
As for the side, since a fixed integer syndrome corresponds to at most one gauge-fixed vector in $\widetilde C$, the number of side representatives is bounded by the number of
integer syndromes with $\ell_1$ weight in the shell
\(
    S-C_{\mathrm{step}}\le |\sigma|_1\le S.
\)
Let $m=n(n-1)$ be the number of $X$-checks. Since distinct states in $\widetilde C^{\mathrm{side}}$ have distinct integer syndromes, the number of states in $\widetilde C^{\mathrm{side}}$ is at most
\begin{equation}
    \#\widetilde C^{\mathrm{side}}
    \le
    (C_{\mathrm{step}}+1)2^m
    \binom{S+m-1}{m-1}.
\end{equation}
Using
\begin{equation}
    \binom{S+m-1}{m-1}
    \le
    \left(
        \frac{e(S+m-1)}{m-1}
    \right)^{m-1},
\end{equation}
we obtain
\begin{equation}
    \#\widetilde C^{\mathrm{side}}
    \le
    (C_{\mathrm{step}}+1)\exp\!\left[
        m\log 2
        +(m-1)\log\frac{e(S+m-1)}{m-1}
    \right].
\end{equation}
Therefore,
we obtain an upper bound on the unnormalized Gibbs weight of the side over $\mathbb Z$:
\begin{equation}
    W(\widetilde C^{\mathrm{side}})
    \le
    (C_{\mathrm{step}}+1)
    \exp\!\left[
        -\beta(S-C_{\mathrm{step}})
        +m\log 2
        +(m-1)\log\frac{e(S+m-1)}{m-1}
    \right],
\end{equation}
where we used the fact that every representative in $\widetilde C^{\mathrm{side}}$ has energy at
least $S-C_{\mathrm{step}}$.

The integer Metropolis dynamics is reversible with respect to the unnormalized weight $W$. Hence, for every $k\ge 0$, \begin{equation}
W(0)P^k(0,x)
=
W(x)P^k(x,0).
\end{equation}
Since $W(0)=1$ and \(P^k(x,0)\le1\), we have $P^k(0,x)\le W(x)$. Therefore, for $\tau\ge 1$, the probability that the chain leaves $\widetilde C$ from a state of $\widetilde C^{\mathrm{side}}$ within $\tau$ steps is bounded by
\begin{equation}
\begin{aligned}
    \Pr_\text{side}(\tau)
    &\le
    \sum_{k=0}^{\tau-1}
    \sum_{x\in\widetilde C^{\mathrm{side}}}
    P^k(0,x)P(x,\widetilde C^c)
    \\
    &\le \tau\sum_{x\in\widetilde C^{\mathrm{side}}}W(x)
    = \tau W(\widetilde C^{\mathrm{side}}).
\end{aligned}
\end{equation}
Consequently, the probability of leaving through the side remains below $1/4$ for at least
\begin{equation}
    \tau_{\mathrm{side}}
    =
    \frac{1}{4W(\widetilde C^{\mathrm{side}})}
\end{equation}
steps.
Therefore, for every $\tau<\min\{\tau_{\mathrm{cap}},\tau_{\mathrm{side}}\}$, leaving $\widetilde C$ through a cap is impossible and the probability of leaving through the side is at most $1/4$.
The memory time of the code over $\mathbb Z$ is bounded by
\begin{align} \label{eq:memorytime_OBCHGP_Z}
    \widetilde t_{\mathrm{mem}}(1/4)
    \ge
    \min\left\{
        \frac{L}{2\Delta},
        \frac{1}{4W(\widetilde C^{\mathrm{side}})}\right\}.
\end{align}

\noindent \textbf{Step 3: Reduction modulo $q$ and the memory-time bound.}

We now transfer the integer cylinder to $\mathbb Z_q$. For an integer vector
$z$, let $\pi_q(z)$ denote the gauge-fixed vector obtained by reducing $z$
entrywise modulo $q$. 
Define
\begin{equation}
    C:=\pi_q(\widetilde C)\subset \mathbb Z_q^N.
\end{equation}

We first show that $\pi_q$, followed by gauge fixing, is injective on $\widetilde C$.
Suppose that $z,z'\in\widetilde C$ and
\(
    \pi_q(z)=\pi_q(z').
\)
Since both $z$ and $z'$ lie in the cylinder, their integer syndrome energies are
at most $S$. We choose
\begin{equation}
    S=\lfloor \frac{q-1}{2}\rfloor.
\end{equation}
Thus every entry of $H_Xz$ and $H_Xz'$ lies in $(-q/2,q/2]$, and reduction modulo $q$ does not change the syndrome and energy. Hence $H_Xz=H_Xz'$ over $\mathbb Z$. As shown in Step 2, a fixed integer syndrome corresponds to at most one gauge-fixed vector in $\widetilde C$, which forces $z=z'$.
This proves the claimed injectivity. 

The injectivity also means that after the reduction modulo $q$, the cylinder does not intersect with itself. Moreover, because $z \in C$ and $z+M\ell_Z \bmod q$ (with $M\ell_Z$ not a $\mathbb Z_q$-stabilizer) are different states with the same syndrome, they cannot both lie in $C$. As a result, we can define a cylinder $C+M\ell_Z \bmod q$ centered at each distinct logical operator $M\ell_Z$, and the cylinders are disjoint. We can choose a decoder that maps all states in $C+M\ell_Z \bmod q$ to $M\ell_Z$. Since the energy of every $z$ in the cylinder is the same as $\pi_q(z)$, and each cylinder has exactly the same structure because addition of $\ell_Z$ is a symmetry of the code. Therefore, the single-site Metropolis dynamics within $C$ is exactly the same as that within $\widetilde C$, and the memory time of the code over $\mathbb Z_q$ has the same lower bound as in \eqref{eq:memorytime_OBCHGP_Z}:\begin{align}
    t_{\mathrm{mem}}(1/4)
    &\ge
    \min\left\{
        \frac{L}{2\Delta},
        \frac{
            \exp\!\left[
                \beta(S-C_{\mathrm{step}})
                -m\log 2
                -(m-1)\log\frac{e(S+m-1)}{m-1}
            \right]
        }{
            4(C_{\mathrm{step}}+1)
        }
    \right\} \notag \\
    & \ge     c\min\left\{
        e^{\gamma n},
        \exp\!\left[
            \beta\left(\frac{q}{2}-B\right)
            -
            4n^2\log\left(2+\frac{q}{n^2}\right)
        \right]
    \right\},
\end{align}
where
\begin{equation}
    \gamma=\log(a_h/a_v),
    \qquad
    B=C_{\mathrm{step}}+1,
    \qquad
    c=\min\left\{
        \frac{a_v}{2a_h},
        \frac{1}{4(C_{\mathrm{step}}+1)}
    \right\}.
\end{equation}
Here we used
$L/(2\Delta)\ge (a_v/(2a_h))e^{\gamma n}$, $S\ge q/2-1$, and the
syndrome-shell entropy bound
\begin{equation}
    m\log 2
    +(m-1)\log\frac{e(S+m-1)}{m-1}
    \le
    4n^2\log\left(2+\frac{q}{n^2}\right).
\end{equation}

In particular, for every fixed inverse temperature $\beta>0$, the memory time of the $Z$-error Metropolis dynamics is exponentially large in $n$,
provided $q$ grows fast enough compared with the entropy scale
$n^2\log(2+q/n^2)$.

\end{appendix}
\phantomsection
\section*{References}
\addcontentsline{toc}{section}{References}
\makeatletter
\let\bibsection\relax
\makeatother
\bibliography{thebib}
\end{document}